\PassOptionsToPackage{names,dvipsnames}{xcolor} %

\documentclass[acmsmall]{acmart}
\usepackage{booktabs}   %
\usepackage{subcaption} %
\usepackage[english]{babel}
\usepackage[inference]{semantic} 
\usepackage{amsmath}\allowdisplaybreaks
\usepackage{mathpartir}
\usepackage{stmaryrd} 
\usepackage{xspace}
\usepackage{latexsym}
\usepackage{ifthen}
\usepackage{mathtools}
\usepackage[names,dvipsnames]{xcolor}
\usepackage{color}
\usepackage{listings}
\usepackage{verbatim}
\usepackage[colorinlistoftodos]{todonotes}
\usepackage{tikz}
\usetikzlibrary{positioning,shadows,arrows,calc,backgrounds,fit,shapes,snakes,shapes.multipart,decorations.pathreplacing,shapes.misc,patterns}
\usepackage{tikzscale}
\usepackage[T1]{fontenc}
\usepackage[scaled=.83]{beramono}
\usepackage{epigraph}
\usepackage{booktabs}
\usepackage{dashrule}
\usepackage{float}
\usepackage{etoolbox}
\usepackage{wrapfig}
\usepackage{nameref}
\usepackage{scalerel}
\usepackage{hyperref}
\usepackage[capitalize]{cleveref}
\usepackage{enumitem}
\usepackage{bm}
\usepackage{mdframed}
\usepackage{centernot}

\hypersetup{ pdfpagemode=UseOutlines, colorlinks=true, linkcolor=black, citecolor=blue } %

\AtBeginEnvironment{tikzpicture}{\catcode`$=3 } %

\newcommand{\mi}[1]{\ensuremath{\mathit{#1}}}

\newcommand{\mtt}[1]{\ensuremath{\mathtt{#1}}}
\newcommand{\mf}[1]{\ensuremath{\mathbf{#1}}}

\newcommand{\mc}[1]{\ensuremath{\mathcal{#1}}}
\newcommand{\ms}[1]{\ensuremath{\mathsf{#1}}}
\newcommand{\mb}[1]{\ensuremath{\mathbb{#1}}}

\newcommand{\isdef}[0]{\ensuremath{\mathrel{\overset{\makebox[0pt]{\mbox{\normalfont\tiny\sffamily def}}}{=}}}}

\newcommand{\relmiddle}[1]{\mathrel{}\middle#1\mathrel{}}

\newcommand\bnfdef{\ensuremath{\mathrel{::=}}}

\newcommand{\OB}[1]{\ensuremath{\overline{#1}}}

\newcommand{\myset}[2]{\ensuremath{\left\{#1 ~\relmiddle|~ #2\right\}}}

\newcommand{\divrt}[0]{\ensuremath{\trg{\Uparrow}}\xspace}
\newcommand{\divrc}[0]{\ensuremath{\com{\Uparrow}}\xspace}

\newcommand{\term}[0]{\ensuremath{^{\Downarrow}}\xspace}
\newcommand{\termsl}[0]{\ensuremath{^{\src{\Downarrow}}}\xspace}
\newcommand{\termt}[0]{\ensuremath{^{\trg{\Downarrow}}}\xspace}
\newcommand{\termc}[0]{\ensuremath{^{\com{\Downarrow}}}\xspace}

\newcommand*{\QEDA}{\hfill\ensuremath{\blacksquare}}%

\AtEndEnvironment{problem}{\null\hfill\QEDA}
\AtEndEnvironment{example}{}%

\Crefname{lstlisting}{Listing}{Listings}
\Crefname{problem}{Problem}{Problems}

\Crefname{equation}{Rule}{Rules}

\newenvironment{proofsketch}{\trivlist\item[]\emph{Proof Sketch}.\xspace}{\unskip\nobreak\hskip 1em plus 1fil\nobreak$\Box$\parfillskip=0pt\endtrivlist}
\newcommand{\compskel}[3]{\ensuremath{\bl{\left\llbracket \src{#1} \right\rrbracket^{#2}_{#3}}}}

\newcommand{\compap}[1]{\compskel{#1}{\LA}{\LC}}
\newcommand{\compup}[1]{\compskel{#1}{\LU}{\LP}}

\newcommand{\compgen}[1]{\compskel{#1}{\S}{\T}}

\newcommand{\compgenm}[1]{\compskel{\fbox{\fbox{#1}}}{\S}{\T}}
\newcommand{\compai}[1]{\compskel{#1}{\LA}{\LI}}

\newcommand{\compskelbag}[3]{\ensuremath{\bl{\lbag \src{#1} \rbag^{#2}_{#3}}}}
\newcommand{\compfac}[1]{\compskelbag{{#1}}{\LU}{\LP}}

\newcommand{\funname}[1]{\mtt{#1}}
\newcommand{\fun}[2]{\ensuremath{{\bl{\funname{#1}(}#2{\bl{)}}}}\xspace}
\newcommand{\dom}[1]{\fun{dom}{#1}}

\newcommand{\backtrskel}[3]{\ensuremath{\bl{\left\langle\!\left\langle {#1} \right\rangle\!\right\rangle^{#2}_{#3}}}}
\newcommand{\backtr}[1]{\backtrskel{#1}{}{}}
\newcommand{\backtrup}[1]{\backtrskel{#1}{\LP}{\LU}}
\newcommand{\backtrfac}[1]{\backtrskel{\trg{\fbox{\trg{#1}}}}{\LP}{\LU}}

\renewcommand{\S}[0]{\src{{S}}\xspace}
\newcommand{\T}[0]{\trg{{T}}\xspace}

\newcommand{\langlett}[0]{L} %
\newcommand{\LA}[0]{\src{\langlett^{\tau}}\xspace}
\newcommand{\LC}[0]{\trg{\langlett^{\pi}}\xspace}
\newcommand{\LU}[0]{\src{\langlett^{U}}\xspace}
\newcommand{\LP}[0]{\trg{\langlett^{P}}\xspace}

\newcommand{\LI}[0]{\oth{\langlett^{I}}\xspace}

\newcommand{\Cs}[1]{\src{C#1}\xspace} %
\renewcommand{\P}[1]{\com{C#1}\xspace}%

\newcommand{\contextletter}[0]{C}
\newcommand{\ctx}[1]{\ensuremath{\mb{#1}}}
\newcommand{\ctxs}[1]{\src{\ctx{\contextletter}#1}\xspace}
\newcommand{\ctxt}[1]{\trg{\ctx{\contextletter}#1}\xspace}%
\newcommand{\ctxc}[1]{\com{\ctx{\contextletter}#1}\xspace}%
\newcommand{\hole}[1]{\ensuremath{\left[#1\right]}}

\newcommand{\Bools}[0]{\src{{Bool}}\xspace}
\newcommand{\Nats}[0]{\src{{Nat}}\xspace}

\newcommand{\trues}[0]{\src{{true}}\xspace}
\newcommand{\falses}[0]{\src{{false}}\xspace}
\newcommand{\units}[0]{\src{{unit}}\xspace}
\newcommand{\truet}[0]{\trg{{true}}\xspace}
\newcommand{\falset}[0]{\trg{false}\xspace}

\newcommand{\Refs}[1]{\src{Ref~#1}\xspace}
\newcommand{\UNS}[0]{\src{UN}\xspace}
\newcommand{\vdashun}[0]{\vdash_\UNS}
\newcommand{\vdashatts}[0]{\vdash_{\src{att}}}
\newcommand{\vdashattt}[0]{\vdash_{\trg{att}}}
\newcommand{\vdashatto}[0]{\vdash_{\oth{att}}}

\newcommand{\lskip}{\ensuremath{{skip}}}
\newcommand{\skips}{\ensuremath{\src{skip}}}
\newcommand{\skipt}{\ensuremath{\trg{skip}}}
\newcommand{\skipo}{\ensuremath{\oth{skip}}}

\newcommand{\srce}[0]{\src{\emptyset}\xspace}
\newcommand{\trge}[0]{\trg{\emptyset}\xspace}
\newcommand{\othe}[0]{\oth{\emptyset}\xspace}
\newcommand{\come}[0]{\com{\emptyset}\xspace}

\newcommand{\SInit}[1]{\ensuremath{{\Omega_0}\left({#1}\right)}\xspace}
\newcommand{\SInits}[1]{\ensuremath{\src{\Omega_0}\left(\src{#1}\right)}\xspace}
\newcommand{\SInitt}[1]{\ensuremath{\trg{\Omega_0}\left(\trg{#1}\right)}\xspace}
\newcommand{\SInito}[1]{\ensuremath{\oth{\Omega_0}\left(\oth{#1}\right)}\xspace}

\newcommand{\TInitt}[1]{\ensuremath{\trg{\Theta_0}\left(\trg{#1}\right)}\xspace}

\newcommand{\neutcol}[0]{black}
\newcommand{\stlccol}[0]{RoyalBlue}
\newcommand{\ulccol}[0]{RedOrange}
\newcommand{\othercol}[0]{CarnationPink}
\newcommand{\commoncol}[0]{black}    %

\newcommand{\col}[2]{\ensuremath{{\color{#1}{#2}}}}

\newcommand{\src}[1]{\ms{\col{\stlccol}{#1}}}
\newcommand{\trg}[1]{\mf{\col{\ulccol }{#1}}}
\newcommand{\trgb}[1]{\ensuremath{\bm{\col{\ulccol }{#1}}}}
\newcommand{\oth}[1]{\mi{\col{\othercol }{#1}}}
\newcommand{\bl}[1]{\col{\neutcol }{#1}}
\newcommand{\com}[1]{\mi{\col{\commoncol }{#1}}}

\newcommand{\fails}[0]{\src{fail}\xspace}

\newcommand{\lowloc}[1]{\fun{low\text{-}loc}{#1}}
\newcommand{\highloc}[1]{\fun{high\text{-}loc}{#1}}
\newcommand{\highcap}[1]{\fun{high\text{-}cap}{#1}}
\newcommand{\sech}[1]{\fun{secure}{#1}}
\newcommand{\monh}[1]{\fun{mon\text{-}care}{#1}}

\newcommand{\strip}[1]{\relev{#1}}
\newcommand{\relev}[1]{\fun{relevant}{#1}}
\newcommand{\agree}[0]{\mathop{\raisebox{1mm}{$\frown$}}}

\newcounter{typerule}
\crefname{typerule}{rule}{rules}

\newcommand{\typeruleInt}[5]{%
	\def\thetyperule{#1}%
	\refstepcounter{typerule}%
	\label{tr:#4}%
  \ensuremath{\begin{array}{c}#5 \inference{#2}{#3}\end{array}} 
}
\newcommand{\typerule}[4]{%
  \typeruleInt{#1}{#2}{#3}{#4}{\textsf{\scriptsize ({#1})} \\      }
}
\newcommand{\typerulenolabel}[4]{%
	\def\thetyperule{#1}%
	\refstepcounter{typerule}%
  \ensuremath{\begin{array}{c} \inference{#2}{#3}\end{array}} 
}

\pgfdeclarelayer{background}
\pgfdeclarelayer{veryback}
\pgfdeclarelayer{veryback2}
\pgfdeclarelayer{veryback3}
\pgfdeclarelayer{back2}
\pgfdeclarelayer{foreground}
\pgfsetlayers{veryback3,veryback2,veryback,background,back2,main,foreground}

\newcommand{\myfig}[3]{\begin{figure} [!ht]
#1
\caption{\label{fig:#2}#3}
\end{figure}}

\newcommand{\etal}[0]{\textit{et al.}\xspace} 

\newcommand{\BREAK}[0]{
\botrule
\begin{center}$\spadesuit$\end{center}
\botrule}

\newcommand{\mytoprule}[1]{\vspace{1mm}\noindent\hrulefill\ \raisebox{-0.5ex}{\fbox{\ensuremath{#1}}} \hrulefill\hrulefill\hrulefill\vspace{0.5mm}}

\def\botrule{\vspace{0mm}\hrule\vspace{2mm}}

\newcommand{\hl}[1]{\colorbox{yellow}{#1}}
\newcounter{line}

\newcommand{\asm}[1]{\mtt{#1}}

\newcommand{\xto}[1]{\ensuremath{~\mathrel{\xrightarrow{~#1~}}~}}
\newcommand{\Xto}[1]{\ensuremath{~\mathrel{\xRightarrow{~#1~}}~}}

\newcommand{\xtos}[1]{\src{\xto{#1}}}
\newcommand{\Xtos}[1]{\src{\Xto{#1}}}

\newcommand{\xtot}[1]{\trg{\xto{#1}}}
\newcommand{\Xtot}[1]{\trg{\Xto{#1}}}
\newcommand{\nxtot}[1]{\trg{\centernot{\xto{#1}}}}

\newcommand{\Xtoc}[1]{\com{\Xto{#1}}}

\newcommand{\xtoo}[1]{\oth{\xto{#1}}}
\newcommand{\Xtoo}[1]{\oth{\Xto{#1}}}

\newcommand{\xtol}[1]{\ensuremath{\xrightarrow{~#1~}\low}}
\newcommand{\Xtol}[1]{\ensuremath{\xRightarrow{~#1~}\Low}}

\newcommand{\low}[0]{\ensuremath{\!\!\!\!\to} }
\newcommand{\Low}[0]{\ensuremath{\!\!\!\!\Rightarrow} }

\newcommand{\unk}[0]{\ensuremath{unk}}

\newcommand{\stu}[0]{\ensuremath{^{\times}\xspace}}

\newcommand{\stut}[0]{\ensuremath{^{\trg{\times}}\xspace}}

\definecolor{mygreen}{rgb}{0,0.6,0}
\definecolor{mygray}{rgb}{0.5,0.5,0.5}
\definecolor{mymauve}{rgb}{0.58,0,0.82}

\newcommand{\lst}[1]{\ensuremath{\text{\texttt{#1}}}}

\lstdefinelanguage{Java} %
{morekeywords={abstract, all, and, as, assert, but, disj, else, exactly, extends, fact, for, fun, iden, if, iff, implies, in, Int, int, let, lone, module, no, none, not, one, open, or, part, pred, run, seq, set, sig, some, sum, then, univ, package, class, public, private, null, return, new, interface, extern, object, implements, System, static, super, try , catch, throw, throws, Unit, var, val, of, principal, trust},
sensitive=true,
keywordstyle=\bfseries\color{green!40!black},
commentstyle=\itshape\color{purple!40!black},
morecomment=[l][\small\itshape\color{purple!40!black}]{//},
identifierstyle=\color{blue},
stringstyle=\color{orange},
basicstyle=\small,
basicstyle={\small\ttfamily},
numbers=left,
numberstyle=\tiny\color{mygray},
tabsize=2,
numbersep=3pt,
breaklines=true,
lineskip=-2pt,
stepnumber=1,
captionpos=b,
breaklines=true,
breakatwhitespace=false,
showspaces=false,
showtabs=false,
float=!h,
columns=fullflexible,escapeinside={(*@}{@*)},
moredelim=**[is][\color{red!60}]{@}{@},
literate={->}{{$\to$}}1 {^}{{$\mspace{-3mu}\widehat{\quad}\mspace{-3mu}$}}1
{<}{$<$ }2 {>}{$>$ }2 {>=}{$\geq$ }2 {=<}{$\leq$ }2
{<:}{{$<\mspace{-3mu}:$}}2 {:>}{{$:\mspace{-3mu}>$}}2
{=>}{{$\Rightarrow$ }}2 {+}{$+$ }2 {++}{{$+\mspace{-8mu}+$ }}2
{<=>}{{$\Leftrightarrow$ }}2 {+}{$+$ }2 {++}{{$+\mspace{-8mu}+$ }}2
{\~}{{$\mspace{-3mu}\widetilde{\quad}\mspace{-3mu}$}}1
{!=}{$\neq$ }2 {*}{${}^{\ast}$}1 %
{\#}{$\#$}1
}

\lstdefinelanguage{Asm}
{morekeywords={abstract, all, and, as, assert, but, check, disj, else, exactly, extends, fact, for, fun, iden, if, iff, implies, in, Int, int, let, lone, module, no, none, not, one, open, or, part, pred, run, seq, set, sig, some, sum, then, univ, package, class, public, private, null, return, new, interface, extern, object, implements, System, static, super, try , catch, throw, throws, Unit, var, val, principal, trust, label, load, add, addi, into, test},
sensitive=true,
identifierstyle=\color{black},
keywordstyle=\bfseries,
commentstyle=\itshape\color{purple!40!black},
morecomment=[l][\small\itshape\color{purple!40!black}]{//},
stringstyle=\color{orange},
basicstyle=\small,
basicstyle={\small},
numbers=left,
numberstyle=\tiny\color{mygray},
tabsize=2,
numbersep=3pt,
breaklines=true,
lineskip=-2pt,
stepnumber=1,
captionpos=b,
breaklines=true,
breakatwhitespace=false,
showspaces=false,
showtabs=false,
float=!h,
columns=fullflexible,escapeinside={(*@}{@*)},
moredelim=**[is][\color{red!60}]{@}{@},
literate={->}{{$\to$}}1 {^}{{$\mspace{-3mu}\widehat{\quad}\mspace{-3mu}$}}1
{<}{$<$ }2 {>}{$>$ }2 {>=}{$\geq$ }2 {=<}{$\leq$ }2
{<:}{{$<\mspace{-3mu}:$}}2 {:>}{{$:\mspace{-3mu}>$}}2
{=>}{{$\Rightarrow$ }}2 {+}{$+$ }2 {++}{{$+\mspace{-8mu}+$ }}2
{<=>}{{$\Leftrightarrow$ }}2 {+}{$+$ }2 {++}{{$+\mspace{-8mu}+$ }}2
{\~}{{$\mspace{-3mu}\widetilde{\quad}\mspace{-3mu}$}}1
{!=}{$\neq$ }2 {*}{${}^{\ast}$}1 %
{\#}{$\#$}1
}
\DeclareMathOperator\ceq{\ensuremath{\mathrel{\simeq_{\mi{ctx}}}}}
\DeclareMathOperator\nceq{\mathrel{\nsimeq_{\mi{ctx}}}}

\DeclareMathOperator\ceqs{\src{\ceq}}
\DeclareMathOperator\ceqt{\trg{\ceq}}

\DeclareMathOperator\nceqs{\src{\nceq}}
\DeclareMathOperator\nceqt{\trg{\nceq}}

\def\teqaux#1{\vcenter{\hbox{\ooalign{\hfil
       \raise6pt \hbox{\scriptsize{T}}\hfil\cr\hfil
       $=$}}}}

\def\ceqwaux#1{\vcenter{\hbox{\ooalign{\hfil
       \raise6pt \hbox{\scriptsize{w-b}}\hfil\cr\hfil
       $\ceq$}}}}

\def\praux#1{\vcenter{\hbox{\ooalign{\hfil
       \raise4pt \hbox{$\subset$}\hfil\cr\hfil
       $\sim$}}}}
\def\pr{\mathrel{\mathpalette\praux{}}}

\DeclareMathOperator\refrel{\pr}

\newcommand{\labelfont}[1]{\ensuremath{\asm{#1}}}
\newcommand{\cl}[2]{\ensuremath{\labelfont{call}~ #1~ #2{?}}}
\newcommand{\cb}[2]{\ensuremath{\labelfont{call}~ #1~ #2{!}}}
\newcommand{\rt}[1]{\ensuremath{\labelfont{ret}~ #1{!}}}
\newcommand{\rb}[1]{\ensuremath{\labelfont{ret}~ #1{?}}}

\newcommand{\clh}[3]{\ensuremath{\labelfont{call}~ #1~ #2~ #3{?}}}
\newcommand{\cbh}[3]{\ensuremath{\labelfont{call}~ #1~ #2~ #3{!}}}
\newcommand{\rth}[2]{\ensuremath{\labelfont{ret}~ #2{!}}}%
\newcommand{\rbh}[2]{\ensuremath{\labelfont{ret}~ #2{?}}}%
\newcommand{\mnh}[1]{\ensuremath{\labelfont{mon}~#1}}

\newcommand{\wrl}[1]{\ensuremath{\labelfont{write}(#1)}}

\newcommand{\traces}[3]{\ensuremath{\ms{TR}^{#2}_{#3}\left(#1\right)}}

\newcommand{\trt}[1]{\trg{\traces{#1}{}{}}}

\DeclareMathOperator\relate{\bl{\approx}}
\DeclareMathOperator\relatebeta{\bl{\approx_{\beta}}}
\DeclareMathOperator\relatephi{\bl{\approx_{\varphi}}}

\DeclareMathOperator\relatet{\bl{\triplesim}}
\DeclareMathOperator\relatetbeta{\bl{\triplesim_{\beta}}}
\DeclareMathOperator\relatetphi{\bl{\triplesim_{\varphi}}}

\newcommand*{\triplesim}{\mathrel{\vcenter{\offinterlineskip\hbox{$\sim$}\vskip-.35ex\hbox{$\sim$}\vskip-.35ex\hbox{$\sim$}}}}

\DeclareMathOperator\altrelatet{\bl{\sim}}
\DeclareMathOperator\altrelatetbeta{\bl{\sim_{\beta}}}
\DeclareMathOperator\notaltrelatetbeta{\bl{\not\sim_{\beta}}}

\newcommand{\alphaseq}[0]{\OB{\alpha}}
\newcommand{\targetaction}[0]{\omega}
\newcommand{\alss}[0]{\src{\alphaseq}\xspace}
\newcommand{\alst}[0]{\trgb{\OB{\targetaction}}\xspace}
\newcommand{\alsc}[0]{\com{\alphaseq}\xspace}

\newcommand{\at}[0]{\trgb{\targetaction}\xspace}

\newcommand{\divt}[0]{\trg{\uparrow}\xspace}
\newcommand{\tert}[0]{\trg{\downarrow}\xspace}

\newcommand{\set}[1]{\setb{#1}}%

\newcommand{\setb}[1]{\ensuremath{\left\{#1\right\}}}

\newcommand{\card}[1]{\ensuremath{|\!|{#1}|\!|}}

\newcommand{\proc}[2]{\ensuremath{(#1)_{#2} }}
\Crefname{corollary}{Corollary}{Corollaries}
\Crefname{informal}{Definition}{Definition}
\Crefname{assumption}{Assumption}{Assumptions}
\crefname{assumption}{Assumption}{Assumptions}
\Crefname{property}{Property}{Properties}
\crefname{property}{Property}{Properties}

\Crefname{paragraph}{Section}{Sections}

\newcommand{\pair}[1]{\ensuremath{\left\langle#1\right\rangle}}
\newcommand{\projone}[1]{\ensuremath{#1.1}}
\newcommand{\projtwo}[1]{\ensuremath{#1.2}}

\newcommand{\Bool}[0]{\mtt{Bool}\xspace}

\newcommand{\wrong}[0]{\trg{wrong}}
\newcommand{\wrongo}[0]{\oth{wrong}}

\newcommand{\newty}[2]{\ensuremath{\src{new}_{#2}~#1}}

\newcommand{\fork}[1]{\ensuremath{\left(\parallel #1\right)}}
\newcommand{\op}[0]{\ensuremath{\oplus}}
\newcommand{\bop}[0]{\ensuremath{\otimes}}

\newcommand{\ret}[0]{{return;}}%
\newcommand{\letin}[3]{{let}~#1=#2~{in}~#3}
\newcommand{\letnew}[3]{{let}~#1=\new{#2}~{in}~#3}
\newcommand{\letnewt}[3]{\trg{let}~#1=\newt{#2}~\trg{in}~#3}
\newcommand{\letnews}[3]{\src{let}~#1=\news{#2}~\src{in}~#3}
\newcommand{\letnewo}[3]{\oth{let}~#1=\newo{#2}~\oth{in}~#3}
\newcommand{\letins}[3]{\src{let}~#1\src{=}#2~\src{in}~#3}
\newcommand{\letint}[3]{\trg{let}~#1\trg{=}#2~\trg{in}~#3}
\newcommand{\letino}[3]{\oth{let}~#1\oth{=}#2~\oth{in}~#3}
\newcommand{\letnewty}[4]{\src{let}~#1=\newty{#2}{#3}~\src{in}~#4}
\newcommand{\letiso}[3]{\oth{let}~#1=\newiso{#2}~\oth{in}~#3}
\newcommand{\call}[1]{{call}~#1}
\newcommand{\calls}[1]{\src{call}~#1}

\newcommand{\ifte}[3]{{if}~#1~{then}~#2~{else}~#3}
\newcommand{\iftes}[3]{\src{if}~#1~\src{then}~#2~\src{else}~#3}
\newcommand{\ifzte}[3]{{ifz}~#1~{then}~#2~{else}~#3}
\newcommand{\ifztet}[3]{\trg{ifz}~#1~\trg{then}~#2~\trg{else}~#3}
\newcommand{\ifzteo}[3]{\oth{ifz}~#1~\oth{then}~#2~\oth{else}~#3}
\newcommand{\new}[1]{{new}~#1}
\newcommand{\news}[1]{\src{new}~#1}
\newcommand{\newt}[1]{\trg{new}~#1}
\newcommand{\newo}[1]{\oth{new}~#1}
\newcommand{\newiso}[1]{\oth{newiso}~#1}
\newcommand{\hide}[1]{\trg{\trg{hide}~#1}}
\newcommand{\lethide}[3]{\trg{let}~#1=\hide{#2}~\trg{in}~#3}
\newcommand{\with}[1]{\trg{~{with}~#1}}
\newcommand{\myendorse}[4]{\src{endorse}~#1~\src{=}~#2~\src{as}~#3~\src{in}~#4}
\newcommand{\destruct}[5]{\trg{destruct}~#1~\trg{=}~#2~\trg{as}~#3~\trg{in}~#4~\trg{or}~#5}
\newcommand{\destructo}[5]{\oth{destruct}~#1~\oth{=}~#2~\oth{as}~#3~\oth{in}~#4~\oth{or}~#5}
\newcommand{\letatom}[3]{\trg{let}~#1~\trg{=}~\trg{newhide}~#2~\trg{in}~#3}

\newcommand{\formatCompilers}[1]{\mf{\mi{#1}}\xspace}
\newcommand{\facomp}[0]{\formatCompilers{FAC}}
\newcommand{\rscomp}[0]{\formatCompilers{RSC}}
\newcommand{\pfrscomp}[0]{\formatCompilers{PF\text{-}RSC}}

\newcommand{\ccomp}[0]{\formatCompilers{CC}\xspace}

\newcommand{\redto}[0]{\ensuremath{\ \hookrightarrow\low\ }}
\newcommand{\red}[0]{\ensuremath{\ \hookrightarrow\ }}

\newcommand{\monred}[0]{\rightsquigarrow}

\newcommand{\redapp}[1]{\ensuremath{\!\!^{#1}\ }}

\DeclareMathOperator\redtos{\src{\redto}}
\DeclareMathOperator\reds{\src{\red}}

\DeclareMathOperator\redtot{\trg{\redto}}
\DeclareMathOperator\redt{\trg{\red}}

\DeclareMathOperator\redtoo{\oth{\redto}}

\AtEndEnvironment{example}{\null\hfill$\boxdot$}

\makeatletter
\xdef\@thefnmark{\@empty}

\newcommand{\Thmref}[1]{\Cref{#1}~(\nameref{#1})}
\makeatother

\newcommand{\subst}[2]{\ensuremath{\bl{\left[#1\relmiddle/#2\right]}}} %
\newcommand{\subs}[2]{\subst{\src{#1}}{\src{#2}}}
\newcommand{\subt}[2]{\subst{\trg{#1}}{\trg{#2}}}
\newcommand{\subo}[2]{\subst{\oth{#1}}{\oth{#2}}}

\renewcommand{\emptyset}[0]{\varnothing}

\newcounter{hps}
\crefname{hps}{}{}

\newcommand{\proven}[1]{\ensuremath{\checkmark}} \renewcommand{\contextletter}[0]{A}
\renewcommand{\ctx}[1]{\ensuremath{{#1}}}
\renewcommand{\set}[1]{\text{\ensuremath{\widehat{#1}}}}
\renewcommand{\set}[1]{\setb{#1\cdots}}
\colorlet{NAVYBLUE}{NavyBlue}%

\theoremstyle{definition}
\newtheorem{property}{Property}

\renewcommand{\citet}[1]{\cite{#1}}

\setcopyright{acmcopyright}
\acmJournal{TOG}
\acmYear{2019}\acmVolume{37}\acmNumber{4}\acmArticle{111}\acmMonth{8}
\acmDOI{10.1145/1122445.1122456}

\begin{document}

\title{Robustly Safe Compilation, an Efficient Form of Secure Compilation} 

\author{Marco Patrignani}
\orcid{0000-0003-3411-9678}             %
\affiliation{
  \position{Visiting Assistant Professor}
  \department{Computer Science}             %
  \institution{Stanford University}
  \city{Stanford}
  \country{USA}
}
\affiliation{
  \position{Junior Research Group Leader}
  \department{}             %
  \institution{CISPA Helmholz Center for Information Security }           %
  \city{Saarbr\"ucken}
  \country{Germany}
}
\email{mp @ cs.stanford.edu}         %

\author{Deepak Garg}
\affiliation{%
  \institution{Max Planck Institute for Software Systems}
  \city{Saarbr\"ucken}
  \country{Germany}
}
\email{dg @ mpi-sws.org}




\begin{abstract}
Security-preserving compilers generate compiled code that withstands target-level attacks such as alteration of control flow, data leaks or memory corruption. 
Many existing security-preserving compilers are proven to be fully abstract, meaning that they reflect and preserve observational equivalence. 
Fully abstract compilation is strong and useful but, in certain cases, comes at the cost of requiring expensive runtime constructs in compiled code.
These constructs may have no relevance for security, but are needed to accommodate differences between the source and target languages that fully abstract compilation necessarily needs.

As an alternative to fully abstract compilation, this paper explores a different criterion for secure compilation called robustly safe compilation or \rscomp. 
Briefly, this criterion means that the compiled code preserves relevant safety properties of the source program against all adversarial contexts interacting with the compiled program. 
We show that \rscomp can be proved more easily than fully abstract compilation and also often results in more efficient code. 
We also present two different proof techniques for establishing that a compiler attains \rscomp and, to illustrate them, develop three illustrative robustly-safe compilers that rely on different target-level protection mechanisms.
We then proceed to turn one of our compilers into a fully abstract one and through this example argue that proving \rscomp can be simpler than proving fully abstraction.

\begin{center}
\small\it 
  To better explain and clarify notions, this paper uses syntax highlighting in a way that colourblind and black-\&-white readers can benefit from~\citet{patrignani2020use}.
  For a better experience, please print or view this paper in colour.%
    \footnote{%
		Specifically, in this paper we use a \src{blue}, \src{sans\text{-}serif} font for elements of the \src{source} language, an \trg{orange}, \trg{bold} font for elements of the \trg{first} \trg{two} \trg{target} languages and a \oth{pink} \oth{italics} font for elements of the \oth{third} \oth{target} language.
		Elements common to all languages are typeset in a \com{\commoncol}, \com{italic} font (to avoid repeating similar definitions twice), thus, \src{C} is a source-level component, \trg{C} and \oth{C} are target-level components and \com{C} is generic notation for either a source-level or a target-level component.
    } 
\end{center}
\end{abstract}

\maketitle

\section{Introduction}\label{sec:intro}

Low-level adversaries, such as those written in C or assembly can
attack co-linked code written in a high-level language in ways that
may not be feasible in the high-level language itself. For example,
such an adversary may manipulate or hijack control flow, cause buffer
overflows, or directly access private memory, all in contravention to
the abstractions of the high-level language.
Specific countermeasures such as Control Flow Integrity~\cite{abadicfi} or Code Pointer Integrity~\cite{cpi} have been devised to address some of these attacks \emph{individually}.
An alternative approach is to devise a \emph{security-preserving compiler}, which seeks to defend against entire \emph{classes} of such attacks.
Security-preserving compilers often achieve security by relying on different protection mechanisms, e.g., cryptographic primitives~\cite{Abadi:2000:APC:325694.325734,Abadi:2002:SIC:570966.570969,Bugliesi:2007:SIT:1190215.1190253,Corin:2008:SCS:1454415.1454419}, types~\cite{ahmedCPS,Ahmed:2008:TCC:1411203.1411227}, address space layout randomisation~\cite{Abadi:2012,Jagadeesan:2011:LMV:2056311.2056556}, protected module architectures~\cite{KULeuven-358154,scoo-j,mfac,scoo} (also know as enclaves~\cite{intel}), tagged architectures~\cite{catalin,catalinRSC}, etc.
Once designed, the question researchers face is how to formalise that such a compiler is indeed secure, and how to prove this. Basically, we want a criterion that specifies secure compilation.
A widely-used criterion for compiler security is fully abstract compilation (\facomp)~\cite{abadiFa,gcFA,faEHM}, which has been shown to preserve many interesting security properties like confidentiality, integrity, invariant definitions, well-bracketed control flow and hiding of local state~\cite{Jagadeesan:2011:LMV:2056311.2056556,scoo-j,KULeuven-358154,scsurvey}.

Informally, a compiler is fully abstract if it preserves and reflects observational equivalence of source-level components (i.e., partial programs) in their compiled counterparts.
Most existing work instantiates observational equivalence with contextual equivalence: co-divergence of two components in any larger context they interact with.
Fully abstract compilation is a very strong property, which preserves \emph{all} source-level abstractions.

Unfortunately, preserving \emph{all} source-level abstractions also has downsides.
In fact, while \facomp preserves many relevant security properties, it also preserves a plethora of other non-security ones, and the latter may force inefficient checks in the compiled code.
For example, when the target is assembly, two observationally equivalent components must compile to code of the same size~\cite{KULeuven-358154,scoo-j}, else full abstraction is trivially violated.
This requirement is security-irrelevant in most cases. %
Additionally, \facomp is not well-suited for source languages with undefined behaviour (e.g., \lst{C} and \lst{LLVM})~\cite{catalin} and, if used na\"ively, it can fail to preserve even simple safety properties~\cite{schp} (though, fortunately, no \emph{existing} work falls prey to this na\"ivety).

Motivated by this, recent work started investigating alternative secure compilation criteria that overcome these limitations.
These security-focussed criteria take the form of preservation of hyperproperties or classes of hyperproperties, such as hypersafety properties or safety properties~\cite{rhpc-arx,rhc}.
This paper investigates one of these criteria, namely, \emph{Robustly Safe Compilation} (\rscomp) which has clear security guarantees and, as we show, can often be attained more efficiently than FAC.

Informally, a compiler attains \rscomp if it is correct and it preserves \emph{robust safety} of source components in the target components it produces. 
Robust safety is an important security notion that has been widely adopted to formalise security, e.g., of communication protocols~\cite{refty-sec-impl,catalin-rs,autysec}.
Before explaining \rscomp, we explain robust safety as a language property.

\paragraph{Robust Safety as a Language Property}
Informally, a program property is a safety property if it encodes that ``bad'' sequences of events do not happen when the program executes~\cite{enfoschneider,AlpernS85}.
A program is \emph{robustly safe} if it has relevant (specified) safety properties \emph{despite} active attacks from adversaries~\cite{dg-rs,davidcaps,catalin-rs}.
As the name suggests, robust safety relies on the notions of safety and robustness which we now explain.

\subparagraph{\mf{Safety}.}
As mentioned, safety asserts that ``no bad sequence of events happens'', so we can specify a safety property by the set of \emph{finite observations} which characterise all bad sequences of events.
A whole program has a safety property if its behaviours exclude these bad observations.
Many security properties can be encoded as safety, including integrity, weak secrecy and functional correctness.
\begin{example}[Integrity]\label{ex:integrity}
  Integrity ensures that an attacker does not tamper with invariants on state.
  For example, consider the function \lst{charge\_account( amount )} in the snippet below, which deducts \lst{amount} from an account as part of an electronic card payment. 
  A card PIN is required if \lst{amount} is larger than 10 euros.
  So the function checks whether \lst{amount} $>10$, requests the PIN if this is the case, and then changes the account balance.
  We expect this function to have a safety (integrity) property on the account balance: a reduction of more than 10 euros to the account balance must be preceded by a call to \lst{request\_pin()}. 
  Here, the relevant observation is a trace (sequence) of account balances and calls to \lst{request\_pin()}.
  Bad observations for this safety property are those where an account balance is at least 10 euros less than the previous one, without a call to \lst{request\_pin()} in between.
  Note that this function seems to have this safety property, but it may not have the safety property \emph{robustly}: a target-level adversary may transfer control directly to  the ``else'' branch of the check \lst{amount} $>10$ after setting \lst{amount} to more than $10$, to violate the safety property.

\begin{minipage}{.9\textwidth}
\begin{lstlisting}
function charge_account( amount : Int ){
  if amount > 10 { request_pin(); }
  charge_account(amount);
  return;
}
\end{lstlisting}
\end{minipage}
\end{example}

\begin{example}[Weak Secrecy]\label{ex:weak-sec}
  Weak secrecy asserts that a program secret never flows \emph{explicitly} to the attacker.
  For example, consider code that manages \lst{network\_h}, a handler (socket descriptor) for a sensitive network interface. %
  This code does not expose \lst{network\_h} directly to external code but it provides an API to use it. This API makes some security checks internally.
  If the handler is directly accessible to outer code, then it can be misused in insecure ways (since the security checks may not be made).
  If the code has weak secrecy with respect to \lst{network\_h} then we know that the handler is never passed to an attacker.
  In this case we can define bad observations as those where \lst{network\_h} is passed to external code (e.g., as a parameter, as a return value on or on the heap).
\end{example}

\begin{example}[Partial Correctness]\label{ex:correct}
  Program correctness can also be formalised as a safety property.
  Consider a program that computes the \lst{n}th Fibonacci number. 
  The program reads \lst{n} from an input source and writes its output to an output source. 
  Correctness of this program is a safety property. 
  Observations here are pairs of an input (read by the program) and the corresponding output (produced by the program) so, for example, outputting \lst{13} is only allowed if \lst{7} were passed as input.
  A bad observation is one where the input is \lst{n} (for some \lst{n}) but the output is different from the \lst{n}th Fibonacci number, e.g., input \lst{4} and output \lst{5} as well as input \lst{3} and output \lst{6} are bad observations.
\end{example}

These examples not only illustrate the expressiveness of safety
properties, but also show that safety properties as we capture here are quite \emph{coarse-grained}, since they are only concerned with (sequences of) relevant events like calls to specific functions, changes to specific heap variables, inputs, and outputs. 
In the model of observable events that we use, we can see that safety properties do not specify or constrain how the program computes between these events, leaving the programmer and the compiler considerable flexibility in optimizations. 
This gives us confidence in the model of events we choose.
However, safety properties are  not a panacea for security, and there are security properties that are not safety. 
For example, noninterference~\cite{TSInfoFlow,zdancewicThesis}, the standard information flow property, is not safety.
Nonetheless, many interesting security properties are safety. In fact, many non-safety properties including noninterference can be conservatively approximated as safety properties~\cite{seinfl}. Hence, safety properties are a meaningful goal to pursue for secure compilation.

\subparagraph{\mf{Robustness.}}
We often want to reason about properties of a component of interest that hold irrespective of any other components the component interacts with.
These other components may be the libraries the component is linked against, or the language runtime. 
Often, these surrounding components are modelled as the \emph{program context} whose hole the component of interest fills.
When the component of interest links to a context, we have a whole program that can run.
A property holds \emph{robustly} for a component if it holds in \emph{any} context that the component of interest can be linked to.

From a security perspective, the context represents the attacker in the threat model we consider.
The implications of this fact are that the attacker's power is limited to what can be expressed by the language semantics. 
Concretely, all of the attackers we consider have no control over the code section of the component and they cannot tamper with the protection mechanisms the compiler uses.

\paragraph{Robust Safety Preservation as a Compiler Property}
A compiler attains robustly safe compilation or {\rscomp} if it maps any source component that has a safety property \emph{robustly} to a compiled component that has the \emph{same} safety property robustly.
Thus, safety has to hold robustly in the target language, which often does not have the powerful abstractions (e.g., typing) that the source language has.
Hence, the compiler must insert enough defensive runtime checks into the compiled code to prevent the more powerful \emph{target} contexts from launching attacks (violations of safety properties) that source contexts could not launch.
This is unlike correct compilation, which either considers only those target contexts that behave like source contexts~\cite{pils,compocompcert,compocompcert2} or considers only whole programs~\cite{Leroy-Compcert-CACM}.

As mentioned, safety properties are usually quite coarse-grained. This means that {\rscomp} still allows the compiler to optimise code internally, as long as the sequence of observable events is not affected.
For example, when compiling the \lst{fibonacci} function of \Cref{ex:correct}, the compiler can do any internal optimisation such as caching intermediate results, as long as the end result is correct.
Crucially, however, cached results must be protected from tampering by a (target-level) attacker, else the output can be incorrect, breaking {\rscomp}.

A {\rscomp}-attaining compiler focuses only on preserving security (as captured by robust safety) instead of contextual equivalence (typically captured by full abstraction).
So, such a compiler can produce code that is more efficient than code compiled with a fully abstract compiler as it does not have to preserve \emph{all} source abstractions (we illustrate this later).

Finally, robust safety scales naturally to thread-based concurrency~\cite{autysec,sec-typ-prot,ot4jc}.
Thus \rscomp also scales naturally to thread-based concurrency (we demonstrate this too).
This is unlike \facomp, where thread-based concurrency can introduce additional undesired observations that also need to be preserved.

\smallskip

{\rscomp} is a very recently proposed criterion for security-preserving compilers.
Recent work~\cite{rhpc-arx,rhc,rc-rel} defines {\rscomp} abstractly in terms of preservation of program behaviours, but the development is limited to the definition only. 
Other recent work~\cite{catalinRSC} defines a form of \rscomp for source languages with undefined behaviour and where attackers are components that become compromised as execution progresses.
Our goal in this paper is to examine how {\rscomp} can be realized and established, and to show that in certain cases it leads to compiled code that is more efficient than what {\facomp} leads to. 
To this end, we consider a specific setting where observations are values in specific (sensitive) heap locations at cross-component calls. 
We define robust safety and \rscomp for this specific setting (\Cref{sec:rsc}). 
Unlike previous work~\cite{rhpc-arx,rhc,catalin-rs} which assumed that the domain of traces (behaviours) is the same in the source and target languages, our \rscomp definition allows for different trace domains in the source and target languages, as long as they can be suitably related. 
This relation is analogous to that found in recent work~\cite{rc-rel} that studied the necessary properties of trace relations in order to preserve security through compilation.
The second contribution of our paper is two proof techniques to establish \rscomp. 
\begin{itemize}
  \item 
	The first technique is an adaption of trace-based backtranslation, an existing technique for proving {\facomp}~\cite{mfac,KULeuven-358154,catalinRSC}.
	To illustrate this technique, we build a compiler from an untyped source language to an untyped target language with support for fine-grained memory protection via so-called capabilities~\cite{cheri,MMach} (\Cref{sec:rsc-instance}). 
	Here, we guarantee that if a source program is robustly safe, then so is its compilation. 
  \item 
	The second proof technique shows that if source programs are \emph{verified} for robust safety, then one can simplify the proof of {\rscomp} so that no backtranslation is needed.
	In this case, we develop a compiler from a \emph{typed} source language where the types already enforce robust safety, to a target language similar to that of the first compiler (\Cref{sec:comp-effi}). 
	In this instance, both languages also support shared-memory concurrency.
	Here, we guarantee that all compiled target programs are robustly safe. 
\end{itemize}
To argue that {\rscomp} is general and is not limited to compilation targets based on capabilities, we also develop a third compiler.
\begin{itemize}
	\item 
		This compiler starts from the same source language as our second compiler but targets an untyped concurrent language with support for \emph{coarse-grained memory isolation}, modelling recent hardware extensions such as Intel's SGX~\cite{intel} (\Cref{sec:rsc-iso}).
\end{itemize}
The final contribution of this paper is a comparison between \rscomp and \facomp (\Cref{sec:facomp}). For this,
\begin{itemize}
	\item 
		We first introduce \facomp and discuss its advantages and limitations.
	\item 
		Then, we present a series of code examples that describe different ways in which a fully abstract compiler introduces inefficiencies in compiled code in order to attain \facomp.
		We then sketch a fourth compiler by turning the first one into a fully abstract one and show how the changes introduced to attain \facomp make compiled code inefficient.
	\item 
		Finally, we argue that this compiler attains \facomp and highlight how the proof is significantly more complex than before.
\end{itemize}
Finally, the paper  discusses related work (\Cref{sec:rw}) and concludes (\Cref{sec:conc}).

This paper supersedes and extends the work of Patrignani and Garg~\citet{PatrignaniG18} by providing full details of the languages and compilers formalisations.
Additionally, it describes how the \rscomp theory scales to different protection mechanisms (\Cref{sec:rsc-iso}) and it presents in much more detail the comparison with \facomp.
For the sake of brevity and clarity, we limit proofs to sketches, the interested reader will find full proofs and additional lemmas in the companion technical report~\cite{rsc-techrep}.

\section{Robustly Safe Compilation}\label{sec:rsc}
This section first discusses robust safety as a language (not a compiler) property (\Cref{sec:rs}) and then presents \rscomp as a compiler property along with an informal discussion of techniques to prove it (\Cref{sec:rsc-theory}).

\subsection{Safety and Robust Safety}\label{sec:rs}
To explain robust safety, we first describe a general
\emph{imperative} programming model that we use.  Programmers write
\emph{components} on which they want to enforce safety properties
robustly. A component is a list of function definitions that can be
linked with other components (the context) in order to obtain a runnable
whole program (functions in ``other'' components are like \lst{extern}
functions in C). Additionally, every component declares a set of
``sensitive'' locations that contain all the data that is
safety-relevant. For instance, in \Cref{ex:integrity} this set may contain the account
balance and  in \Cref{ex:correct} it may contain the
I/O buffers. We explain the relevance of this set after we define safety
properties.

We want safety properties to specify that a component never executes a
``bad'' sequence of events. For this, we first need to fix a notion of
events. We have several choices here, e.g., our events could be inputs
and outputs, all syscalls, all changes to the heap (as in
CompCert~\cite{leroy2}), etc. Here, we make a specific choice motivated by
our interest in robustness: we define events as calls/returns that
cross a component boundary, together with the state of the heap at
that point. Consequently, our safety properties can constrain the
contents of the heap at component boundaries. This choice of component
boundaries as the point of observation is meaningful because, in our
programming model, control transfers to/from an adversary happen only
at component boundaries (more precisely, they happen at cross-component
function call and returns). This allows the compiler complete
flexibility in optimizing code within a component, while not reducing
the ability of safety properties to constrain observations of the
adversary.
In turn, safety properties regarding these kinds of boundary-crossing events are the only one that our criterion can preserve through compilers upholding our criterion.

Concretely, a component behaviour is a \emph{trace}, i.e., a sequence
of \emph{actions} recording component boundary interactions and, in
particular, the heap at these points.  \emph{Actions}, the items on a
trace, have the following grammar (notation-wise, we mainly indicate actions as \com{\alpha}, though to further disambiguate when source and target actions are mentioned, we will also use the \com{\targetaction} notation):
\begin{align*}
  \mi{Actions}~\com{\alpha},\com{\targetaction} \bnfdef&\ \com{\clh{f}{v}{H}} \mid \com{\cbh{f}{v}{H}} \mid \com{\rth{}{H}} \mid \com{\rbh{}{H}}
\end{align*}
These actions respectively capture call and callback to a function $f$ with parameter $v$ when the heap is $H$ as well as return and returnback with a certain heap $H$.
More precisely, a callback is a call from the component to the context, so it generates label \com{\cbh{f}{v}{H}} while a returnback is a return from such a callback, i.e., the context returning to the component, and it generates the label \com{\rbh{}{H}}.
We use ? and ! decorations to indicate whether the control flow of the action goes from the context to the component (?) or from the component to the context (!). 
Well-formed traces have alternations of \com{?} and \com{!} decorated actions, starting with \com{?} since execution starts in the context.
For a sequence of actions \alsc, \strip{\alsc} is the list of heaps \OB{H} mentioned in the actions of \alsc.
In the sequent, we separate list elements with $\cdot$, so $\OB{H}\cdot H$ indicates a non-empty list of heaps with at least one element ($H$).

Next, we need a representation of safety properties.  Generally,
properties are sets of traces, but safety properties specifically can
be specified as automata (or monitors in the sequel)~\cite{enfoschneider}.  We choose
this representation since monitors are less abstract than sets of
traces and they are closer to enforcement mechanisms used for
safety properties, e.g., runtime monitors.  Briefly, a safety property
is a monitor that transitions 
states in response to events of the program trace. At any point, the monitor may
refuse to transition (it gets \emph{stuck}), which encodes property
violation. While a monitor can transition, the property has not been
violated.  Schneider~\citet{enfoschneider} argues that all properties
codable this way are safety properties and that all enforceable safety
properties can be coded this way.

Formally, a monitor \com{M} in our setting consists of a set of abstract states \com{\set{\sigma}}, the transition relation \com{\monred}, an initial state \com{\sigma_0}, the set of heap locations that matter for the monitor, \com{\set{l}}, and the current state \com{\sigma_c} (we indicate a set of elements of class $e$ as \set{e}).
The transition relation \com{\monred} is a set of triples of the form $(\com{\sigma_s},\com{H},\com{\sigma_f})$ consisting of a starting state \com{\sigma_s}, a final state \com{\sigma_f} and a heap \com{H}. 
The transition $(\com{\sigma_s},\com{H},\com{\sigma_f})$ is interpreted as ``\emph{state \com{\sigma_s} transitions to \com{\sigma_f} when the heap is \com{H}}''. 
When determining the monitor transition in response to a program action, we restrict the program's heap to the location set \com{\set{l}}, i.e., to  the set of locations the monitor cares about. This heap restriction is written $H\big|_{\set{l}}$.
We assume determinism of the transition relation: for any \com{\sigma_s} and (restricted heap) \com{H}, there is at most one \com{\sigma_f} such that $(\com{\sigma_s},\com{H},\com{\sigma_f}) \in \com{\monred}$.

Given the behaviour of a program as a trace \alsc and a monitor \com{M} specifying a safety property, $\com{M}\vdash\alsc$ denotes that the trace satisfies the safety property.
Intuitively, to satisfy a safety property, the sequence of heaps in the actions of a trace, \emph{restricted to the locations that the monitor cares about}, must never get the monitor stuck (\Cref{tr:mon-valid}).
Every single restricted heap must allow the monitor to step according to its transition relation (\Cref{tr:mon-us}).
Note that we overload the $\monred$ notation here to also denote an auxiliary relation, the \emph{monitor small-step semantics} (\Cref{tr:ms-t-s-b} and \Cref{tr:ms-t-s}).
\begin{center}
  \typerule{Valid trace}{
    \com{M;\strip{\alsc}\monred M'}
  }{
    \com{M}\vdash\alsc
  }{mon-valid}
  \typerule{Monitor Step-base}{
  }{
  \com{M;\come\monred M}
  }{ms-t-s-b}
  \typerule{Monitor Step-ind}{
  \com{M;\OB{H}\monred M''}
  &
  \com{M'';H\monred M'}
  }{
  \com{M;\OB{H}\cdot H\monred M'}
  }{ms-t-s}
  \typerule{Monitor Step}{
    &
    \com{(\sigma_c,H\big|_{\set{l}},\sigma_f)}\in\com{\monred}
  }{
    \com{(\set{\sigma},\monred,\sigma_0,\set{l},\sigma_c)};
    H
    \monred
    \com{(\set{\sigma},\monred,\sigma_0,\set{l},\sigma_f)}
  }{mon-us}
\end{center}

With this setup in place, we can formalise safety, attackers and robust safety.
In defining (robust) safety for a component, we only admit monitors (safety properties) whose \com{\set{l}} agrees with the sensitive locations declared by the component. 
Making the set of safety-relevant locations explicit in the component and the monitor gives the compiler more flexibility by telling it precisely which locations need to be protected against target-level attacks (the compiler may choose to not protect the rest). 
At the same time, it allows for expressive modelling. 
For instance, in \Cref{ex:correct}, the safety-relevant locations could be the I/O buffers from which the program performs inputs and outputs, and the safety property can constrain the input and output buffers at corresponding call and return actions involving the Fibonacci function.
A whole program $\com{C}$ is safe for a monitor $\com{M}$, written $\com{M}\vdash \com{C} : \com{safe}$, if the monitor accepts any trace the program generates from its initial state (\SInit{C}).

An attacker $\com{A}$ is valid for a component $\com{C}$, written $\com{C}\vdash\com{A} : \com{atk}$, if $\com{A}$'s free locations (denoted \fun{locs}{\com{A}}) are disjoint from the locations that the component cares about (denoted \com{C.\mtt{locs}}).
This is a basic sanity check: if we allow an attacker to mention heap locations that the component cares about, the attacker will be able to modify those locations, causing all but trivial safety properties to \emph{not} hold robustly.

A component \com{C} is robustly safe with respect to monitor \com{M}, written $\com{M}\vdash \com{C} : \com{rs}$, if \com{C} composed with \emph{any} attacker is safe with respect to \com{M}.
As mentioned, for this setup to make sense, the monitor and the component must agree on the locations that are safety-relevant. This agreement is denoted \com{M\agree C}.
\begin{definition}[Safety, attacker and robust safety]\label{def:safe}
  \begin{align*}
    \com{M}\vdash \com{C} : \com{safe} \isdef&\ 
    \text{ if }
    \vdash \com{C}:\com{whole}
    \text{ then }
      \text{if }
      \SInit{\com{C}} \Xtoc{\alsc} \com{\_} 
      \text{ then }
      M\vdash\alsc
    \\
    \com{C}\vdash\com{A} : \com{atk} \isdef&\ \com{C}.\mtt{locs}= \com{\set{l}} \text{ and } \set{l}\cap\fun{locs}{\com{A}}=\come
    \\
    \com{M}\vdash \com{C} : \com{rs} \isdef&\  \forall \com{A}.
    \text{ if }
    M\agree C
    \text{ and }
    \com{C}\vdash \com{A}: \com{atk}
    \text{ then }
    \com{M}\vdash \com{A\hole{C}} : \com{safe}
  \end{align*}
\end{definition}

\subsection{Robustly Safe Compilation}\label{sec:rsc-theory}
Robustly-safe compilation ensures that robust safety properties \emph{and their meanings} are preserved across compilation.
But what does it means to preserve meanings across languages?
If a source safety property says \src{never} \src{write} \src{3} \src{to} \src{a} \src{location}, and we compile to an assembly language by mapping numbers to binary, the corresponding target property should say \trg{never} \trg{write} \trg{0x11} \trg{to} \trg{an} \trg{address}.

In order to relate properties across languages, we assume a relation $\relate : \src{v}\times\trg{v}$ between source and target values that is \emph{total} in the first component, so it maps any source value \src{v} to a target value \trg{v}: $\forall \src{v}. \exists \trg{v}. \src{v}\relate\trg{v}$.
This value relation is used to define a relation between heaps: $\src{H}\relate\trg{H}$, which intuitively holds when related locations point to related values.
This is then used to define a relation between actions: $\src{\alpha}\relate\at$, which holds when the two actions are the ``same'' modulo this relation, i.e., \src{\clh{\cdot}{\cdot}{\cdot}} only relates to \trg{\clh{\cdot}{\cdot}{\cdot}} and the arguments of the action (values and heap) are related.
Next, we require a relation $\src{M}\relate\trg{M}$ between source and target monitors, which means that the source monitor $\src{M}$ and the target monitor $\trg{M}$ enforce the same safety property, modulo the relation $\relate$ on actions (and thus on locations and values too) assumed above. 
The precise definition of this relation depends on the source and target languages; specific instances are shown in \Cref{sec:mon-rel-conc,sec:mon-rel-par}.%
\footnote{Accounting for the difference in the representation of
  safety properties sets us apart from recent work~\citet{rhpc-arx,rhc},
  which assumes that the source and target languages have the same trace alphabet. The latter works only in some settings.}
 
We denote a compiler from language \S to language \T by \compgen{\cdot}.
A compiler \compgen{\cdot} attains \rscomp, if it maps any component \src{C} that is robustly safe with respect to \src{M} to a component \trg{C} that is robustly safe with respect to \trg{M}, provided that $\src{M} \relate \trg{M}$.
\begin{definition}[Robustly Safe Compilation]\label{def:rsc}
\begin{align*}
    \vdash\compgen{\cdot}: \rscomp \isdef&\ 
    \forall \src{C},\src{M},\trg{M}\ldotp 
    \text{ if }
    \src{M}\vdash \src{C} : \src{rs}
    \text{ and }
    \src{M}\relate\trg{M}
    \text{ then }
    \trg{M}\vdash \compgen{\src{C}} : \trg{rs}
  \end{align*}
\end{definition}
A consequence of the universal quantification over monitors here is that the compiler cannot be property-sensitive.
A robustly-safe compiler preserves all robust safety properties, not just a specific one, e.g., it does not just enforce that \lst{fibonacci} is correct.
This seemingly strong goal is sensible as compiler writers will likely not know what safety properties individual programmers will want to preserve.

\paragraph{Remark \#1: Safety Through Assertions}
Some readers may wonder why we do not follow existing work and specify
safety as ``programmer-written assertions never
fail''~\cite{autysec,tydisa,cca,davidcaps}.  Unfortunately, this
approach does not yield a meaningful criterion for specifying a
compiler, since assertions in the compiled program (if any) are generated by the compiler itself.
Thus a compiler could just erase all assertions and the compiled code it generates would be trivially (robustly) safe -- no assertion can fail if there are no assertions in the first place!

\paragraph{Remark \#2: Compiling Monitors}
In our development, we assume that a source and a target monitor are related and do not actually compile a source monitor to obtain a related target monitor.
While such compilation is feasible, it is at odds with our view of monitors as \emph{specifications} of safety properties. 
Compiling monitors and, in particular, compiling monitors with the same compiler that we want to prove security of, leads to a circularity---we must understand the compiler to understand the target safety property, which, in turn, acts as the specification for the compiler! 
Consequently, we choose not to compile monitors and talk only of an abstract, compiler-independent relation between source and target monitors.

\subsubsection{Proving \rscomp}\label{sec:proving-rsc}
Proving that a compiler attains \rscomp can be done either by proving that a compiler satisfies \Cref{def:rsc} or by proving something \emph{equivalent}.
To this end, \Cref{def:rsc-eq} below presents an alternative, equivalent formulation of \rscomp.
We call this characterisation \emph{property-free} as it does not mention monitors explicitly (it mentions the \strip{\cdot} function for reasons we explain below).
\begin{definition}[Property-Free \rscomp]\label{def:rsc-eq}
  \begin{align*}
    \vdash\compgen{\cdot}:&\ \pfrscomp \isdef
    \forall \src{C},
    \trg{A},\alst.~
    \\
    \text{if }
    &\
    \trg{\compgen{C}}\vdash \trg{A}: \trg{atk} \text{ and }
    \vdash\trg{A\hole{\compgen{C}}}:\trg{whole} 
    \text{ and }
    \SInitt{A\hole{\compgen{\src{C}}}} \Xtot{\alst} \trg{\_}
    \\
    \text{ then }
    &\
    \exists\src{A},\alss.~
    \src{C}\vdash \src{A}: \src{atk} \text{ and }
    \vdash\src{A\hole{C}}:\src{whole} 
    \text{ and }
    \SInits{A\hole{C}}\Xtos{\alss} \src{\_}
    \\&\
    \text{ and }
    \strip{\alss}\relate\strip{\alst}
  \end{align*}
\end{definition}
\pfrscomp states that if the compiled code produces a behaviour in a target context, then the source code also produces a related behaviour in \emph{some} source context. In other words, target contexts cannot induce more (bad) behaviours in the compiled code than source contexts can in the source code.

\pfrscomp and \rscomp should, in general, be equivalent (\Cref{thm:rsc-prf-eq}). 
\begin{proposition}[\pfrscomp and \rscomp are equivalent]\label{thm:rsc-prf-eq}
    \begin{align*}
      \forall\compgen{\cdot}, \vdash\compgen{\cdot}: \pfrscomp \iff \vdash\compgen{\cdot}:\rscomp
    \end{align*}
\end{proposition}
As mentioned in \Cref{sec:intro}, a property is safety if it asserts that ``no bad sequence of events happens'', so a safety property specifies the set of \emph{bad prefixes} (i.e., finite traces) which characterise all bad sequences of events.
As such, a safety property implies that programs do not have any trace prefix from the set of bad prefixes. 
Hence, \emph{not} having a safety property robustly amounts to some context being able to induce a bad prefix. 
Consequently, preserving \emph{all} robust safety properties (\rscomp) amounts to ensuring that all target prefixes can be generated (by some context) in the source too (\pfrscomp). 
Formally, since \Cref{def:rsc} relies on the monitor relation, we can prove \Cref{thm:rsc-prf-eq} only after such a relation is finalised.
We give such a monitor relation and proof in Section~\ref{sec:comp-up} (see Theorem~\ref{thm:rsc-prf-eq2}).
However, in general this result should hold for any cross-language monitor relation that correctly relates safety properties.
If the proposition does not hold, then the relation does not capture how safety in one language is represented in the other.

Assuming \Cref{thm:rsc-prf-eq}, we can prove \pfrscomp for a compiler in place of \rscomp.
\pfrscomp can be proved with a \emph{backtranslation} technique.
This technique has been often used to prove full abstraction~\cite{max-embed,mfac,catalin,scoo-j,KULeuven-358154,rhpc-arx,catalinRSC,scsurvey,rhc} and it aims at building a source context starting from a target one.
In fact \pfrscomp, leads directly to a backtranslation-based proof technique since it can be rewritten (eliding irrelevant details) as:
\begin{itemize}
  \item[] $\forall\alst$ if $\exists\trg{A}.\ \trg{\SInitt{A\hole{\compgen{C}}}\Xtot{\alst}\_}$
  \item[] then $\exists\src{A},\alss.\ \src{\SInits{A\hole{{C}}}\Xtos{\alss}\_}$ and $\strip{\alss}\relate\strip{\alst}$
\end{itemize}
Essentially, given a target context $\trg{A}$, a compiled program $\compgen{C}$ and a target trace $\alst$ that $\trg{A}$ causes $\compgen{C}$ to have, we need to construct, or \emph{backtranslate} to, a source context $\src{A}$ that will cause the source program $\src{C}$ to simulate $\alst$. 
Such backtranslation based proofs can be quite difficult, depending on the features of the languages and the compiler. 
However, backtranslation for \rscomp (as we show in \Cref{sec:compup-proof}) is not as complex as backtranslation for \facomp (\Cref{sec:facomp-instance}).

A simpler proof strategy is also viable for \rscomp when we compile only those source programs that have been \emph{verified} to be robustly safe (e.g., using a type system). 
The idea is this: from the verification of the source program, we can find an invariant which is always maintained by the target code, and which, in turn, implies the robust safety of the target code. 
For example, if the safety property is that values in the heap always have their expected types, then the invariant can simply be that values in the target heap are always related to the source ones (which have their expected types). 
This is tantamount to proving type preservation in the target in the presence of an active adversary. 
This is harder than standard type preservation (because of the active adversary) but is still much easier than backtranslation as there is no need to map target constructs to source contexts syntactically.
We illustrate this proof technique in \Cref{sec:comp-effi}.

\subsubsection{\rscomp Implies Compiler Correctness}\label{sec:rsc-impl-cc}
As stated in \Cref{sec:intro}, \rscomp implies (a form of) compiler correctness.
While this may not be apparent from \Cref{def:rsc}, it is more apparent from its equivalent characterization in \Cref{def:rsc-eq}.
We elaborate this here.

Whether concerned with whole programs or partial programs, compiler correctness states that the behaviour of compiled programs \emph{refines} the behaviour of source programs~\cite{leroy2,compocompcert,compocompcert2,pils,kripke,realizability}.
So, if \trgb{\set{\alst}} and \src{\set{\alss}} are the sets of compiled and source behaviours, then a compiler should force $\trgb{\set{\alst}} \refrel \src{\set{\alss}}$, where $\refrel$ is the composition of $\subseteq$ and of the relation $\relate^{-1}$.

If we consider a source component $\src{C}$ that is whole, then it can only link against empty contexts, both in the source and in the target. 
Hence, in this special case, \pfrscomp simplifies to standard refinement of traces, i.e., whole program compiler correctness. 
Hence, assuming that the correctness criterion for a compiler is concerned with the same observations as safety properties (values in safety-relevant heap locations at component crossings in our illustrative setting), \pfrscomp implies whole program compiler correctness.

However, \pfrscomp (or, equivalently, \rscomp) does not imply, nor is implied by, any form of \emph{compositional compiler correctness} (CCC)~\cite{compocompcert,compocompcert2,pils}. 
CCC requires that the behaviours produced by a compiled component  linked against a target context that is related (in behaviour) to a source context can also be produced by the source component linked against the \emph{related} source context. 
In contrast, \pfrscomp allows picking \emph{any} source context to simulate the behaviours. 
Hence, \pfrscomp does not imply CCC. 
On the other hand, \pfrscomp universally quantifies over all target contexts, while CCC only quantifies over target contexts related to a source context, so CCC does not imply \pfrscomp either. 
Hence, compositional compiler correctness, if desirable, must be imposed in addition to \pfrscomp. 

We could remedy this and generalise our criterion even more by adding an additional parameter, a relation between source and target contexts that binds the quantified target and source contexts.
Our criterion chooses the weakest of these relations, where all source contexts are related to all target ones, in order to not impose any constraints on \trg{A} thus making the attackers in our threat model as powerful as possible.
Existing compositional compiler correctness criteria would instantiate this relation e.g., between a source context and its compilation~\cite{compocompcert,compocompcert2} or between a source context and something that behaves like its compilation~\cite{pils}.
As we focus on security, we choose not to pollute our definition with an additional parameter and leave this relation out.

Note that the lack of implications between \pfrscomp and CCC is unsurprising: the two criteria capture two very different aspects of compilation: security (against all contexts) and compositional preservation of behaviour (against well-behaved contexts).

\paragraph{Remark} Compiler correctness composes `vertically', that is, given two compilers (or, compiler passes) one from a source language to an intermediate one, and one from the intermediate to a target language, the compiler resulting of the composition of the two passes is still correct.
Like compiler correctness, \pfrscomp also composes vertically, i.e., if several compiler passes are all \pfrscomp, then the compiler resulting of the composition of those passes is also \pfrscomp.
Studying how \pfrscomp compiler passes interact with other passes (e.g., compiler passes that may be \facomp or any other criterion from \citet{rhc}) is an open research question.

\section{\rscomp via Trace-based Backtranslation}\label{sec:rsc-instance}

This section illustrates how to prove that a compiler attains \rscomp by means of a trace-based backtranslation technique~\cite{mfac,scoo-j,catalinRSC}.
To present such a proof, we first introduce our source language \LU, an untyped, first-order imperative language with abstract references and hidden local state (\Cref{sec:src}).
Then, we present our target language \LP, an untyped imperative target language with a concrete heap, whose locations are natural numbers that the context can compute.
\LP provides hidden local state via a fine-grained capability mechanism on heap accesses  (\Cref{sec:trg}).
Finally, we present the compiler \compup{\cdot} and prove that it attains \rscomp (\Cref{sec:comp-up}) by means of a trace-based backtranslation.
The section concludes with an example detailing why \rscomp preserves security (\Cref{ex:request}).

To avoid focussing on mundane details, we deliberately use source and target languages that are fairly similar.
However, they differ substantially in one key point: the heap model. 
This affords the target-level adversary attacks like guessing private locations and writing to them that do not obviously exist in the source (and makes our proofs nontrivial). 
We believe that (with due effort) the ideas here will generalize to languages with larger gaps and more features.

\subsection{The Source Language \LU}\label{sec:src} 
\myfig{
	\vspace{-1em}
  \begin{gather*}
  \begin{aligned}
  \mi{Components}~\src{C} \bnfdef&\ \src{\src{\ell_{root}} ; \OB{F} ; \OB{I}}
  &
  \mi{Contexts}~\src{A} \bnfdef&\ \src{ H ; \OB{F}\hole{\cdot}}
  \\
  \mi{Interfaces}~\src{I} \bnfdef&\ \src{f}
  &
  \mi{Functions}~\src{F} \bnfdef&\ \src{f(x)\mapsto s;\ret}
  \\
  \mi{Heaps}~\src{H} \bnfdef&\ \srce \mid \src{H ; \ell\mapsto v}
  &
  \mi{Values}~\src{v} \bnfdef&\ \src{\units} \mid 
  \src{\trues} %
  \mid
  \src{\falses}
  \mid \src{n}\in\mb{N} \mid \src{\pair{v,v}} \mid \src{\ell}  
  \end{aligned}
  \\
  \begin{aligned}
  \mi{Expressions}~\src{e} \bnfdef&\ \src{x} \mid \src{v} \mid \src{e \op e} \mid \src{e \bop e} \mid \src{\pair{e,e}} \mid \src{\projone{e}} \mid \src{\projtwo{e}} \mid \src{!e} 
  \\
  \mi{Statements}~\src{s} \bnfdef&\ \skips \mid \src{s;s} \mid \src{\letin{x}{e}{s}} \mid \src{\ifte{e}{s}{s}} 
  \mid \src{\call{f}~e} \mid \src{\letnew{x}{e}{s}} \mid \src{x := e} 
  \\
  \mi{Eval.\ Ctxs.}~\src{E} \bnfdef&\ \src{\hole{\cdot}} \mid \src{n \op E} \mid \src{E \op e} \mid \src{n \bop E} \mid \src{E \bop e} \mid \src{\pair{v,E}} \mid \src{\pair{E,e}} \mid \src{\projone{E}} \mid \src{\projtwo{E}} \mid \src{!E}
  \end{aligned}
  \\
  \begin{aligned}
  \mi{Mon.\ States}~\src{\sigma} \in&\ \src{\mc{S}}
  &
  \mi{Monitors}~\src{M} \bnfdef&\ \src{(\set{\sigma},\monred,\sigma_0,\ell_{root},\sigma_c)}
  \\
  \mi{Mon.\ Reds.}~\src{\monred} \bnfdef&\ \srce \mid \src{\monred;(s,H,s)}
  &
  \mi{Prog.\ States}~\src{\Omega}\bnfdef&\ \src{C, H\triangleright \proc{s}{\OB{f}} } 
  \\
  \mi{Labels}~\src{\lambda} \bnfdef&\ \src{\epsilon} \mid \src{\alpha}
  &
  \mi{Actions}~\src{\alpha} \bnfdef&\ \src{\clh{f}{v}{H}} \mid \src{\cbh{f}{v}{H}} \mid \src{\rth{v}{H}} \mid \src{\rbh{v}{H}}
  \end{aligned}
  \end{gather*}
}{lu-syn}{Syntax of \LU. We indicate a list of elements \com{e_1,\cdots,e_n} as \OB{e}.}
\LU is an untyped imperative while language~\cite{booknielson}. 
Its syntax is presented in \Cref{fig:lu-syn}. 
Components \src{C} are triples of function definitions, interfaces and a special location written $\src{\ell_{root}}$, which defines the locations that are monitored for safety, as explained below. %
We use a mnemonic `dot' notation to access sub-parts of elements that are tuples (such as components), so when accessing the functions sub-part of a component \src{C}, we will write $\src{C}.\mtt{funs}$.
Each function definition maps a function name and a formal argument to a body~$\src{s}$. %
An interface is a list of functions that the component relies on the context to provide (similar to C's \texttt{extern} declarations).
Attackers \src{A} (program contexts) are function definitions (and their heap) that represent untrusted code that a component interacts with.
A function body is a statement. 
Statements \src{s} are rather standard but their treatment is not. 
For example, statements define local variables but these are substitute and not looked up in an environment as in while languages~\cite{vasco}.
Additionally, statements manipulate the heap, do recursive function calls and branch conditionally.
Statements use effect-free expressions $\src{e}$, which contain arithmetic and comparison operations, pairing and projections, and location dereference.
Heaps \src{H} are maps from abstract locations \src{\ell} to values \src{v}.

We use a number of auxiliary functions to access parts of \LU that we now explain.
Function \fun{locs}{\cdot} returns the set of locations that are free in the argument of the function; that argument can be an expression, a list of functions or an attacker.
Function \fun{names}{\cdot} returns the names of the functions defined in the argument, which can be a list of functions or a list of interfaces.
Function \fun{fv}{\src{\OB{F}}} returns the free variables in the bodies of the list of functions supplied as argument.

As explained in \Cref{sec:rs}, safety properties are specified by monitors. 
Note that in place of the set \com{\set{l}} of safety-relevant locations,
the description of a monitor here (as well as a component above)
contains a \emph{single} location $\src{\ell_{root}}$. The
interpretation is that any location \emph{reachable} in the heap
starting from $\src{\ell_{root}}$ is relevant for safety. This set of
locations can change as the program executes, and hence this is more
flexible than statically specifying all of \com{\set{l}} upfront. This
representation of the set by a single location is made explicit in the
following monitor rule:

\begin{center}
	\typerule{\LU-Monitor Step}{
	    \src{M}= \src{(\set{\sigma},\monred,\sigma_0,\ell_{root},\sigma_c)}
	    &
	    \src{M'}= \src{(\set{\sigma},\monred,\sigma_0,\ell_{root},\sigma_f)}
	    \\
	    \src{(\sigma_c,H',\sigma_f)}\in\src{\monred}
	    &
	    \src{H'}\subseteq\src{H}
	    &
	    \dom{\src{H'}}=\fun{reach}{\src{\ell_{root}},\src{H}}
	}{
		\src{M;H\monred M'}
	}{ms-us}
\end{center}
\begin{align*}
  \fun{reach}{\src{\ell'},\src{H}} =&\ \myset{\src{\ell}}{\exists\src{e}.~ \src{H\triangleright e\redtos ~\ell} \wedge \src{\ell}\in\dom{\src{H}} \wedge \fun{locs}{\src{e}}=\src{\ell'}}
\end{align*}

Other than this small point, monitors, safety, robust safety and
\rscomp are defined as in \Cref{sec:rsc}. In particular, a monitor and
a component agree if they mention the same \src{\ell_{root}} and an attacker is valid for a component \src{C} if its code and heap do not mention the \src{\ell_{root}} location of \src{C}.
Note that checking this condition is sufficient because whether \src{A} is a valid attacker is a \emph{static condition}, checked before programs run. 
At this stage, \src{\ell_{root}} does not point to any other location, so the check is sufficient.
\src{\ell_{root}} may grow to point to other locations, but these will be dynamically-generated, and thus the attacker cannot possibly mention them statically in its code.
\begin{align*}
  \src{M\agree C}\isdef
    &\
      (\src{M} = \src{(\set{\sigma},\monred,\sigma_0,\ell_{root},\sigma_c))} \mbox{ and } (\src{C}=\src{(\ell_{root} ; \OB{F} ; \OB{I})})       %
  \\
    \src{C}\vdash\src{A}:\src{atk} \isdef&\ \src{C}=\src{(\ell_{root} ; \OB{F} ; \OB{I})}, \src{A}=\src{ H ; \OB{F'}} \text{ and }
    \src{\ell_{root}}\notin(\fun{locs}{\src{A}})
\end{align*}

The semantics of \LU relies on some auxiliary functions that we present in \Cref{fig:inout} before presenting the semantics itself.
\myfig{
	\begin{center}
		\typerule{\LU-Jump-Internal}{
		((\src{f'}\in\src{\OB{I}} \wedge \src{f}\in\src{\OB{I}}) \vee
		(\src{f'}\notin\src{\OB{I}} \wedge \src{f}\notin\src{\OB{I}}))
		}{
			\src{\OB{I}}\vdash\src{f,f'}:\src{internal}
		}{us-aux-intern}
		\typerule{\LU-Jump-IN}{
			\src{f}\in\src{\OB{I}} \wedge \src{f'}\notin\src{\OB{I}}
		}{
			\src{\OB{I}}\vdash\src{f,f'}:\src{in}
		}{us-aux-in}
		\typerule{\LU-Jump-OUT}{
			\src{f}\notin\src{\OB{I}} \wedge \src{f'}\in\src{\OB{I}}
		}{
			\src{\OB{I}}\vdash\src{f,f'}:\src{out}
		}{us-aux-out}

    \smallskip
    \hrule

    \typerule{\LU-Plug}{
    \src{A} \equiv \src{H ; \OB{F}\hole{\cdot}}
    &
    \src{C}\equiv\src{\src{\ell_{root}} ; \OB{F'} ; \OB{I}} 
    \\
    \vdash\src{C,\OB{F}}:\src{whole}
    &
    \src{main}\in\fun{names}{\src{\OB{F}}}
  }{
    \src{A\hole{C}} = \src{\src{\ell_{root}} ; H;\ell_{root}\mapsto0 ; \OB{F;F'}; \OB{I}}
  }{plug-us}
  \typerule{\LU-Whole}{
    \src{C}\equiv\src{\src{\ell_{root}} ; \OB{F'} ; \OB{I}} 
    &
    \fun{fv}{\src{\OB{F}}}\cup\fun{fv}{\src{\OB{F'}}}=\srce
    \\
    \fun{names}{\src{\OB{F}}}\cap\fun{names}{\src{\OB{F'}}}=\emptyset
    \\
    \fun{names}{\src{\OB{I}}}\subseteq \fun{names}{\src{\OB{F}}}\cup\fun{names}{\src{\OB{F'}}}
  }{
    \vdash\src{C,\OB{F}}:\src{whole}
  }{whole-us}
  \typerule{\LU-Initial State}{
    \src{P}\equiv\src{\src{\ell_{root}} ; H ; \OB{F} ; \OB{I}}
    &
    \src{C}\equiv\src{\src{\ell_{root}} ; \OB{F} ; \OB{I}}
    &
    \src{main(x)\mapsto s;\ret}\in\src{\OB{F}}
  }{
    \SInits{P} = \src{C ; H,\src{\ell_{root}\mapsto 0} \triangleright \proc{s\subs{0}{x}}{main}}
  }{ini-us}
	\end{center}
}{inout}{Auxiliary rules.
The first batch determines the direction of calls and returns.
The second batch defines plugging a component with an attacker, when a whole program is whole and how to calculate the starting state of a program, which starts computing from the main function. }

A program state \src{\Omega} includes the function bodies \src{C}, the heap \src{H}, a statement \src{s} being executed and a stack of function calls \src{\OB{f}} (often omitted in the semantics rules for simplicity and explained in \Cref{ex:fstack}). 
The initial state of a whole program is generated according to \Cref{tr:ini-us}.
A program consisting of a component \src{C} and attacker-provided functions \src{\OB{F}} is whole (as defined by \Cref{tr:whole-us}) if no function definition has free variables, if no function names is duplicated and if all import functions of the components are resolved.
Whole programs are typically the result of plugging a component and an attacker together as in \Cref{tr:plug-us}.
The stack of function calls is used to populate judgements of the form $\src{\OB{I}}\vdash\src{f,f'}:\src{internal/in/out}$ (\Cref{fig:inout}, top).
These judgements determine whether calls and returns are \src{internal} (within the attacker or within the component), directed from the attacker to the component (\src{in}) or directed from the component to the attacker (\src{out}).
This information is used to determine whether the semantics should generate a label (as in \Cref{tr:eus-call,tr:eus-ret,tr:eus-callback,tr:eus-retb}) or no label (as in \Cref{tr:eus-call-i,tr:eus-ret-i}) since internal calls should not be observable.

\myfig{
	\begin{center}
		\typerule{E\LU-ctx}{
			\src{H \triangleright e \redtos e'}
		}{
			\src{H \triangleright E\hole{e}} \redtos \src{E\hole{e'}}
		}{eus-cth}
		\typerule{E\LU-val}{
		}{
			\src{H \triangleright v} \redtos \src{v}
		}{eus-val}
		\typerule{E\LU-dereference}{
			\src{\ell\mapsto v } \in \src{H}
		}{
			\src{H \triangleright !\ell \redtos v }
		}{eus-de}
		\typerule{E\LU-op}{
			n\op n'=n''
		}{
			\src{H \triangleright n \op n' \redtos n''}
		}{eus-op}
		\typerule{E\LU-comp}{
			n\bop n'=b
		}{
			\src{H \triangleright n \bop n' \redtos b}
		}{eus-bop}
		\typerule{E\LU-p1}{
		}{
			\src{H \triangleright \projone{\pair{v,v'}} \redtos v}
		}{eus-p1}
		\typerule{E\LU-p2}{
		}{
			\src{H \triangleright \projtwo{\pair{v,v'}} \redtos v'}
		}{eus-p2}
	\end{center}
\hrule
	\begin{center}
	\typerule{E\LU-sequence}{
	}{
		\src{C, H \triangleright \skips;s} \xtos{\epsilon} \src{C, H \triangleright s}
	}{eus-seq}
	\typerule{E\LU-step}{
		\src{C, H \triangleright s} \xtos{\lambda} \src{C, H' \triangleright s'}
	}{
		\src{C, H \triangleright s;s''} \xtos{\lambda} \src{C, H' \triangleright s';s''}
	}{eus-step}
	\typerule{E\LU-if}{
		\src{H \triangleright e\redtos v}
		\\
		\src{v}\equiv\trues \Rightarrow \src{s''}=\src{s}
		&
		\src{v}\equiv\falses \Rightarrow \src{s''}=\src{s'}
	}{
		\src{C, H \triangleright \ifte{e}{s}{s'}} \xtos{\epsilon} \src{C, H \triangleright s''}
	}{eus-ift}
	\typerule{E\LU-letin}{
		\src{H \triangleright e\redtos v}
	}{
		\src{C, H \triangleright \letin{x}{e}{s}} \xtos{\epsilon} \src{C, H \triangleright s\subs{v}{x}}
	}{eus-letin}
	\typerule{E\LU-update}{
		\src{H \triangleright e\redtos v}
		\\
		\src{H}=\src{H_1; \ell\mapsto v' ; H_2}
		&
		\src{H'}=\src{H_1; \ell\mapsto v ; H_2}
	}{
		\src{C, H \triangleright \ell:= e} \xtos{\epsilon} \src{C, H' \triangleright \skips }
	}{eus-up}
    \typerule{E\LU-alloc}{
      \src{H \triangleright e\redtos v}
      &
      \src{\ell}\notin\dom{\src{H}}
    }{
      \begin{multlined}
        \src{C, H \triangleright \letnew{x}{e}{s}} \xtos{} 
        \\
        \src{C, H; \ell\mapsto v \triangleright s\subs{\ell}{x} }
      \end{multlined}
    }{eus-al}
    \typerule{E\LU-call}{
      \src{\OB{f'}} = \src{\OB{f''};f'}
      &
      \src{f(x)\mapsto s;\ret}\in\src{C}.\mtt{funs}
      \\
      \src{\OB{C}}.\mtt{intfs}\vdash\src{f',f}:\src{in}
      &
      \src{H \triangleright e\redtos v}
    }{
      \begin{multlined}
        \src{C, H \triangleright \proc{{\call{f}~e}}{\OB{f'}}} \xtos{\clh{f}{v}{H}} 
        \\
        \src{C, H \triangleright \proc{{s;\ret\subs{v}{x}}}{\OB{f'};f}}
      \end{multlined}
    }{eus-call}
  \typerule{E\LU-callback}{
    \src{\OB{f'}} = \src{\OB{f''};f'}
    &
    \src{f(x)\mapsto s;\ret}\in\src{\OB{F}}
    \\
    \src{\OB{C}}.\mtt{intfs}\vdash\src{f',f}:\src{out}
    &
    \src{H \triangleright e\redtos v}
  }{
    \begin{multlined}
    \src{C, H \triangleright \proc{{\call{f}~e}}{\OB{f'}}} \xtos{\cbh{f}{v}{H}} 
    \\
    \src{C, H \triangleright \proc{{s;\ret\subs{v}{x}}}{\OB{f'};f}}
    \end{multlined}
  }{eus-callback}
  \typerule{E\LU-return}{
      \src{\OB{f'}} = \src{\OB{f''};f'}
      &
      \src{\OB{C}}.\mtt{intfs}\vdash\src{f,f'}:\src{out}
    }{
        \src{C, H \triangleright \proc{{\ret}}{\OB{f'};f}} \xtos{\rth{v}{H}} 
        \src{C, H \triangleright \proc{\skips}{\OB{f'}}}
    }{eus-ret}
  \typerule{E\LU-retback}{
    \src{\OB{f'}} = \src{\OB{f''};f'}
    &
    \src{\OB{C}}.\mtt{intfs}\vdash\src{f,f'}:\src{in}
  }{
    \src{C, H \triangleright \proc{{\ret}}{\OB{f'};f}} \xtos{\rbh{v}{H}} 
    \src{C, H \triangleright \proc{\skips}{\OB{f'}}}
  }{eus-retb}
    \typerule{E\LU-call-internal}{
      \src{\OB{C}}.\mtt{intfs}\vdash\src{f,f'}:\src{internal}
      &
      \src{\OB{f'}} = \src{\OB{f''};f'}
       \\
      \src{f(x)\mapsto s;\ret}\in\src{C}.\mtt{funs}
      &
      \src{H \triangleright e\redtos v}
    }{
      \begin{multlined}
        \src{C, H \triangleright \proc{{\call{f}~e}}{\OB{f'}}} \xtos{\epsilon}
        \\
        \src{C, H \triangleright \proc{{s;\ret\subs{v}{x}}}{\OB{f'};f}}
      \end{multlined}
    }{eus-call-i}
	\typerule{E\LU-ret-internal}{
		\src{\OB{f'}} = \src{\OB{f''};f'}
		&
		\src{\OB{C}}.\mtt{intfs}\vdash\src{f,f'}:\src{internal}
	}{
		\src{C, H \triangleright \proc{{\ret}}{\OB{f'};f}} \xtos{\epsilon} \src{C, H \triangleright \proc{\skips}{\OB{f'}}}
	}{eus-ret-i}
	\end{center}
\hrule
	\begin{center}
		\typerule{E\LU-single}{
			\src{\Omega}\xtos{\alpha}\src{\Omega'}
		}{
			\src{\Omega}\Xtos{\alpha}\src{\Omega'}
		}{eus-tr-sin}
		\typerule{E\LU-silent}{
			\src{\Omega}\xtos{\epsilon}\src{\Omega'}
		}{
			\src{\Omega}\Xtos{}\src{\Omega'}
		}{eus-tr-silent}
		\typerule{E\LU-transitive}{
			\src{\Omega}\Xtos{\alss}\src{\Omega''}
			&
			\src{\Omega''}\Xtos{\alss'}\src{\Omega'}
		}{
			\src{\Omega}\Xtos{\alss\cdot\alss'}\src{\Omega'}
		}{eus-tr-trans}
	\end{center}
}{lu-sem}{Semantics of \LU. $\op$ includes $+,-,\times$. $\bop$ includes $==,<,>$ etc; \subs{v}{x} substitutes value \src{v} for variable \src{x}.}
\LU has a big-step semantics for expressions that relies on evaluation contexts, a small-step semantics for statements that has labels \src{\lambda} and a semantics that accumulates labels in traces by omitting silent actions \src{\epsilon} and concatenating the rest.
These semantics follow the judgements below. The rules defining these judgments are presented in \Cref{fig:lu-sem}:
\begin{align*}
	&\text{Expressions }\src{H\triangleright e\redtos v}
	&
	&\text{Statements }\src{\Omega\xtos{\lambda}\Omega'}
	&
	&\text{Traces }\src{\Omega\Xtos{\alss}\Omega'}
\end{align*}
Unlike existing work on compositional compiler correctness which only relies on having the component~\cite{compocompcert2}, our semantics relies on having both the component and the context (i.e., a whole program).

\begin{example}[Call semantics]\label{ex:fstack}
  To provide further insights on the semantics of this (and the following) language, this example shows the reduction for component \src{C_{base}} below plugged with an attacker \src{A_{base}} defining only function \src{main} (still below).
  \begin{align*}
    \src{C_{base}} =&\ \src{\ell_{root};}
    \
    \src{skipten (x) \mapsto} \src{ \ifte{x>=10}{ \skips }{\call{skipten}\ x+1};\ret}
    \
    \src{main}
    \\
    \src{A_{base}} =&\ \srce ; \src{main (x) \mapsto} \src{\call{skipten}~9;\ret}
  \end{align*}
  In the following, we indicate with \src{C} the component resulting by adding function \src{main} to the list of functions of \src{C_{base}} and leaving the rest unmodified (as according to \Cref{tr:plug-us}).  
    \begin{align*}
    &
    \src{C;\emptyset;\ell_{root} \triangleright \proc{\call{skipten}~9;\ret}{main}}
    \\
    \xtos{\clh{skipten}{9}{\emptyset}}
    &
    \src{C;\emptyset;\ell_{root} \triangleright \proc{\ifte{9>=10}{ \skips }{\call{skipten}\ 9+1};\ret\ret}{main\cdot skipten}}
    \\
    &
    \text{since } \src{\emptyset \triangleright 9>=10 \redtos \falses}
    \\
    \xtos{}
    &
    \src{C;\emptyset;\ell_{root} \triangleright \proc{{\call{skipten}\ 9+1};\ret\ret}{main\cdot skipten}}
    \\
    &
    \text{since } \src{\emptyset \triangleright 9+1 \redtos 10}
    \\
    &
    \text{(this call is internal to the component, so there is no label)}
    \\
    \xtos{}
    &
    \src{
      C;\emptyset;\ell_{root} \triangleright 
      \left(\begin{aligned}
          &
          \src{\ifte{10>=10}{ \skips }{\call{skipten}\ 10+1};}
          \\
          &\ 
          \src{\ret\ret\ret}
        \end{aligned}\right)_{main\cdot skipten\cdot skipten}
    }
    \\
    &
    \text{since } \src{\emptyset \triangleright 10>=10 \redtos \trues}
    \\
    \xtos{}
    &
    \src{
      C;\emptyset;\ell_{root} \triangleright 
      \proc{
          \skips;\ret\ret\ret
        }{main\cdot skipten\cdot skipten}
    }
    \\
    \xtos{}
    &
    \src{
      C;\emptyset;\ell_{root} \triangleright 
      \proc{
          \ret\ret\ret
        }{main\cdot skipten\cdot skipten}
    }
    \\
    &
    \text{(this return is internal to the component, so there is no label)}
    \\
    \xtos{}
    &
    \src{
      C;\emptyset;\ell_{root} \triangleright 
      \proc{
          \skips;\ret\ret
        }{main\cdot skipten}
    }
    \\
    \xtos{}
    &
    \src{
      C;\emptyset;\ell_{root} \triangleright 
      \proc{
          \ret\ret
        }{main\cdot skipten}
    }
    \\
    \xtos{\rth{v}{\emptyset}{}}
    &
    \src{
      C;\emptyset;\ell_{root} \triangleright 
      \proc{
          \skips;\ret
        }{main}
    }
    \\
    \xtos{}
    &
    \src{
      C;\emptyset;\ell_{root} \triangleright 
      \proc{
          \ret
        }{main}
    }
    \\
    \xtos{}
    &
    \src{
      C;\emptyset;\ell_{root} \triangleright 
          \skips;
    }
  \end{align*}
\end{example}

\subsection{The Target Language \LP}\label{sec:trg}
\myfig{
  \begin{gather*}
  	\begin{aligned}
    \mi{Components}~\trg{C} \bnfdef&\ \trg{k_{root} ; \OB{F} ; \OB{I}}
    &
    \mi{Heaps}~\trg{H} \bnfdef&\ \trge \mid \trg{H ; n\mapsto v:\trgb{\eta}} \mid \trg{H ; k}
    \\
    \mi{Expressions}~\trg{e} \bnfdef&\ \cdots \mid \trg{!e \with{e}}
    &
    \mi{Values}~\trg{v} \bnfdef&\ \trg{n}\in\mb{N} \mid \trg{\pair{v,v}} \mid \trg{k} %
  	\end{aligned}
    \\
    \mi{Statements}~\trg{s} \bnfdef\ \cdots \mid \trg{\lethide{x}{e}{s}} \mid \trg{\ifzte{e}{s}{s}}  \mid \trg{x := e \with{e}}
    \\
    \begin{aligned}
    \mi{Monitors}~\trg{M} \bnfdef&\ \trg{(\trgb{\set{\sigma}},\monred,\trgb{\sigma}_0,k_{root},\trgb{\sigma}_c)}
    &
    \mi{Tags}~\trg{\eta} \bnfdef&\ \trg{\bot} \mid \trg{k}
    \end{aligned}
    \\
    \mi{Actions}~\at \bnfdef\ \trg{\clh{f}{v}{H}} \mid \trg{\cbh{f}{v}{H}} \mid \trg{\rth{v}{H}} \mid \trg{\rbh{v}{H}}
  \end{gather*}
}{lp-syn}{Syntax of \LP. Elided bits ($\cdots$) and omitted ones are the same as in \LU (\Cref{fig:lu-syn}).}

\LP is an untyped, imperative language that follows the structure of \LU and it has similar expressions and statements (\Cref{fig:lp-syn}).  
However, there are critical differences (that make the compiler interesting).
The main difference is that heap locations in \LP are concrete natural numbers.
Upfront, an adversarial context can guess locations used as private state by a component and clobber them. 
To support hidden local state, a location can be ``hidden'' explicitly via the statement \trg{\lethide{x}{e}{s}}, which allocates a new \emph{capability} \trg{k}, an abstract token that grants access to the location \trg{n} to which \trg{e} points~\cite{nucalc}. 
Subsequently, all reads and writes to \trg{n} must be authenticated with the capability, so reading and writing a location take another parameter, the capability, as in \trg{!e \with{e}} and \trg{x:=e \with{e}}. 
In both cases, the \trg{e} after the \with{} is the capability.
Unlike locations, capabilities cannot be guessed. 
To make a location private, the compiler can make the capability of the location private. 
To bootstrap this hiding process, we assume that a component has one location that can only be accessed by it, a priori in the semantics (in our formalisation, we always focus on only one component and we assume that, for this component, this special location is at address \trg{0}).

\LP stores capabilities on the heap alongside locations, so a heap \trg{H} contain both capabilities \trg{k} as well as maps from natural numbers (locations) \trg{n} to values \trg{v} and a tag \trgb{\eta}. %
The tag \trgb{\eta} can be \trgb{\bot}, which means that \trg{n} is globally available (not protected) or a capability \trg{k}, which protects \trg{n}. 
A globally available location can be freely read and written but one that is protected by a capability requires the same capability to be supplied at the time of read/write (\Cref{tr:et-ac-t}, \Cref{tr:et-de-t}).

\LP has a big-step semantics for expressions, a labelled small-step semantics and a semantics that accumulates traces. These judgments follow similar judgements in the semantics of \LU (\Cref{fig:lp-sem}).

\myfig{
\begin{center}
  \typerule{E\LP-deref}{
    \trg{n\mapsto v:\trgb{\eta}}\in\trg{H} 
    &
    (\trgb{\eta}=\trgb{\bot})\text{ or }(\trgb{\eta}=\trg{k} \text{ and } \trg{v'}=\trg{k})
  }{
    \trg{H\triangleright {!n} \with{v'}} \redtot \trg{H\triangleright v}
  }{et-de-t}
\end{center}
\smallskip
  \hrule
\begin{center}
  	\typerule{E\LP-if}{
		\trg{H \triangleright e\redtot n}
		&
		\trg{n}\equiv\trg{0}\Rightarrow\trg{s''}=\trg{s}
		&
		\trg{n}\nequiv\trg{0}\Rightarrow\trg{s''}=\trg{s'}
	}{
		\trg{C, H \triangleright \ifzte{e}{s}{s'}} \xtot{\epsilon} \trg{C, H \triangleright s''}
	}{et-ift}
  \typerule{E\LP-new}{
    \trg{H} = \trg{H_1;n\mapsto (v',\trgb{\eta})}
    &
    \trg{H \triangleright e\redtot v}
    &
    \trg{H'}=\trg{H; n+1\mapsto v:\trgb{\bot} }
  }{
      \trg{C, H \triangleright \letnew{x}{e}{s}} \xtot{} 
      \trg{C,H'\triangleright s\subt{n+1}{x}}
  }{et-nu}
  \typerule{E\LP-hide}{
    \trg{H \triangleright e\redtot n}
    &
    \trg{H} = \trg{H_1;n\mapsto v:\trgb{\bot};H_2}
    &
    \trg{k}\notin\dom{\trg{H}}
    &
    \trg{H'}=\trg{H_1;n\mapsto v:k;H_2;k}
  }{
  		\trg{C, H \triangleright \lethide{x}{e}{s}} \xtot{}
  		\trg{C, H' \triangleright s\subt{k}{x}}
  }{et-hi}
  \typerule{E\LP-assign}{
    \trg{H \triangleright e\redtot v}
    &
    \trg{H \triangleright e'\redtot v'}
    &
    \trg{H} = \trg{H_1;n\mapsto \_:\trgb{\eta};H_2}
    &
    \trg{H'} = \trg{H_1;n\mapsto v:\trgb{\eta};H_2}
    \\
    (\trgb{\eta}=\trgb{\bot})\text{ or }(\trgb{\eta}=\trg{k} \text{ and } \trg{v'}=\trg{k})
  }{
  		\trg{C, H\triangleright n:=e \with{e'}} \xtot{}
  		\trg{C, H'\triangleright \skipt}
  }{et-ac-t}
\end{center}
}{lp-sem}{Expression and state semantics of \LP. Omitted rules are the same as in \LU (\Cref{fig:lu-sem}).}

A second difference between \LP and \LU is that \LP has no booleans, while \LU has them.
This makes the compiler and the related proofs interesting, as discussed in the proof of \Cref{thm:comp-up-rsc}.

In \LP, the locations of interest to a monitor are all those that can be reached from the address \trg{0}.
Location \trg{0} itself is protected with a capability \trg{k_{root}} that is assumed to occur only in the code of the component in focus, so a component is defined as $\trg{C} \bnfdef\ \trg{k_{root} ; \OB{F} ; \OB{I}}$.
We can now give a precise definition of component-monitor agreement for \LP as well as a precise definition of attacker, which must care about the \trg{k_{root}} capability.
In the following, we use auxiliary function \fun{caps}{\cdot} to return the capabilities of the argument (analogously to how \fun{locs}{\cdot} returned the location of its argument).
\begin{align*}
    \trg{M\agree C}\isdef&\
      (\trg{M} = \trg{(\trgb{\set{\sigma}},\monred,\trgb{\sigma}_0,k_{root},\trgb{\sigma}_c)}) \mbox{ and } (\trg{C} = \trg{(k_{root} ; \OB{F} ; \OB{I})})
    \\
  \trg{C}\vdash\trg{A}:\trg{atk} \isdef&\ \trg{C}.\trg{k_{root}} \notin \fun{caps}{\trg{A}}
  \end{align*}
A monitor and compiler agree if they agree on \trg{k_{root}}. An attacker is valid if it does not contain \trg{k_{root}} in its codebase a priori, though it may obtain \trg{k_{root}} during its interaction with the component, if the component is not carefully written.

\paragraph{Remark}
This language uses what is commonly referred to as `data' capabilities, i.e., capabilities only for heap-allocated resources.
Another kind of capabilities exist in the literature: `code' capabilities, which grant the permission to jump to certain functions and execute their code.
Since our programs do not have function pointers (nor higher-order functions), the code capabilities used and required by a program can be tracked statically, so we omit them entirely.
If we extended our language with function pointers or higher-order functions, we would have to introduce code capabilities and pass said capabilities around in order to invoke the right function.
We leave such an extension for future work.

\subsection{Compiler from \LU to \LP}\label{sec:comp-up}
We now present \compup{\cdot}, the compiler from \LU to \LP, detailing how it uses the capabilities of \LP to achieve \rscomp. 
Then, we prove that \compup{\cdot} attains \rscomp. 

\smallskip
\myfig{
	\begin{align*}
		\compup{
			\src{\ell_{root} ; \OB{F} ; \OB{I}}
		} &= \trg{
			k_{root} ;
			\compup{\OB{F}}; \compup{\OB{I}}
		}
		\\
		\compup{
			\src{f(x) \mapsto s;\ret}
		} &= \trg{
			f(x)\mapsto\compup{\src{s}};\ret
		}
		&
		\compup{f} &= \trg{f}
	\end{align*}
	\hrule
	\begin{minipage}[t]{.45\textwidth}
	\begin{align*}
		\compup{
			\trues
		}
		=&\ \trg{0}  
		\\
		\compup{
			\falses
		}
		=&\ \trg{1} 
		\\
		\compup{
			\src{n}
		}
		=&\ \trg{n}
		\\
		\compup{
			\src{x}
		}
		=&\  \trg{x}
		\\
		\compup{
			\src{\ell}
		}
		=&\  
				\trg{\pair{n,v}}
		\\
		\compup{
			\src{!e}
		}
		=&\ \trg{
			!\projone{\compup{e}} \with{\projtwo{\compup{e}}}
		}
		\end{align*}
	\end{minipage}
	\begin{minipage}[t]{.45\textwidth}
		\begin{align*}
		\compup{
			\src{\pair{e_1,e_2}}
		}
		=&\ \trg{\pair{\compup{\src{e_1}},\compup{\src{e_2}}}}
		\\
		\compup{
			\src{\projone{e}}
		}
		=&\ \trg{\projone{\compup{\src{e}}}}
		\\
		\compup{
			\src{\projtwo{e}}
		}
		=&\ \trg{\projtwo{\compup{\src{e}}}}
		\\
		\compup{
			\src{e \op e'}
		}
		=&\ \trg{\compup{\src{e}} \op \compup{\src{e'}}}
		\\
		\compup{
			\src{e \bop e'}
		}
		=&\ 
			\trg{
				\compup{\src{e}} \bop \compup{\src{e'}}
			}
		\end{align*}
	\end{minipage}
  \smallskip
	\hrule

	\begin{minipage}[t]{.40\textwidth}
		\begin{align*}
		\compup{
			\skips
		}
		=&\ \skipt  
		\\
		\compup{
			\src{s_u;s}
		}
		=&\ \trg{\compup{\src{s_u}} ; \compup{\src{s}}}
		\\
		\compup{
			\src{\call{f}~e}
		}
		=&\ \trg{\call{f}~\compup{\src{e}}}
		\\
		\compup{
			\src{\letin{x}{e}{s}}
		}
		=&\ \trg{
				\letin{\trg{x}}{\compup{ \src{e}} }{\compup{ \src{s}}}
			}
		\\
		\compup{
			\src{
				\begin{aligned}
					&
					\iftes{
						\src{e}
					}{
						\src{s_t}
					\\
					&\
					}{\src{s_e}}
				\end{aligned}
			}
		}
		=&\ \trg{
				\begin{aligned}
					&
					\ifztet{\compup{\src{e}}
					}{
					\compup{\src{s_t}}
					\\
					&\
					}{\compup{\src{s_e}}}
				\end{aligned}
		}
		\end{align*}
	\end{minipage}
	\begin{minipage}[t]{.45\textwidth}
		\begin{align*}
		\compup{
			\src{
				\begin{aligned}
					&
					\letnews{\src{x}
					}{
						\src{e}
						\\
						&\
					}{\src{s}}
				\end{aligned}
			}
		}
		=&\
		\trg{
			\begin{aligned}
				&
				\letnewt{\trg{x_{loc}}}{\compup{\src{e}}}{
				\\
				&\
					\lethide{\trg{x_{cap}}}{x_{loc}}{
					\\
					&\ \ 
						\letint{\trg{x} ~}{~ \trg{\pair{x_{loc},x_{cap}}}}{\compup{s}}
					}
				}
			\end{aligned}
		}
		\\
		\compup{
			\src{x := e}
		}
		=&\ \trg{
			\begin{aligned}[t]
				&
				\letint{\trg{x1}~}{~\trg{\projone{x}}}{
				\\
				&\
					\letint{\trg{x2}~}{~\trg{\projtwo{x}}}{
					\\
					&\ \ 
						\trg{x1 :=}\ \compup{e} \with{x2}
					}
				}
			\end{aligned}
		}
		\end{align*}
	\end{minipage}
}{compup}{\compup{\cdot}, compilation of components and functions, expressions and statements from \LU to \LP.}

\compup{\cdot} takes as input a \LU component \src{C} and returns a \LP component (\Cref{fig:compup}).
The compiler performs a simple pass on the structure of functions, expressions and statements, using the information of the intended cross-language relation ($\beta$) to compile values.
The only non-straightforward cases are the compilation of booleans and locations.
Concerning the former, the compiler codes source booleans \trues to \trg{0} and \falses to \trg{1}.
Concerning the latter, each \LU location is encoded as a pair of a \LP location and the capability to access the location.
Location update and dereference are compiled accordingly and thus project each pair to the location and the capability in order to use each part.

This compiler solely relies on the capability abstraction of the target language as a defence mechanism to attain \rscomp.
Unlike existing security-preserving compilers, \compup{\cdot} needs neither dynamic checks nor other constructs that introduce runtime overhead to attain \rscomp~\cite{fstar2js,mfac,scoo-j,KULeuven-358154,catalin}.

\subsubsection{Proof of \rscomp}\label{sec:compup-proof}
\compup{\cdot} attains \rscomp (\Cref{thm:comp-up-rsc}).
In order to set up this theorem, we need to instantiate the cross-language relation for values, which we write as $\relatebeta$ here.
The relation is parametrised by a partial bijection $\beta : \src{location}\times\trg{natural\ number}\times\trg{tag}$ from source heap locations to target heap locations such that:
\begin{itemize}
  \item if $(\src{\ell_1},\trg{n},\trgb{\eta})\in\beta$ and $(\src{\ell_2},\trg{n},\trgb{\eta})\in\beta$ then $\src{\ell_1}=\src{\ell_2}$;
  \item if $(\src{\ell},\trg{n_1},\trg{\trgb{\eta}_1})\in\beta$ and $(\src{\ell},\trg{n_2},\trg{\trgb{\eta}_2})\in\beta$ then $\trg{n_1}=\trg{n_2}$ and $\trg{\trgb{\eta}_1}=\trg{\trgb{\eta}_2}$.
\end{itemize}
The bijection determines 
when a source location and a target location (and its capability) are related.

On values, $\relatebeta$ is defined as follows: 
\begin{itemize}
	\item $\trues\relatebeta\trg{0}$;  
	\item $\falses\relatebeta\trg{n}$ for any $\trg{n}\neq\trg{0}$; 
	\item $\src{n}\relatebeta\trg{n}$; 
	\item $\src{\ell}\relatebeta\trg{\pair{n,k}}$ if $(\src{\ell},\trg{n},\trg{k})\in\beta$;
	\item $\src{\ell}\relatebeta\trg{\pair{n,\_}}$ if $(\src{\ell},\trg{n},\trgb{\bot})\in\beta$;
	\item $\src{\pair{v_1,v_2}}\relatebeta\trg{\pair{v_1,v_2}}$ if $\src{v_1}\relatebeta\trg{v_1}$ and $\src{v_2}\relatebeta\trg{v_2}$.
\end{itemize}
This relation is then used to define the heap, monitor state and action relations (\Cref{fig:mon-st-act-rel}).
Heaps are related, written $\src{H}\relatebeta\trg{H}$, when locations related in $\beta$ point to related values.
States are related, written $\src{\Omega}\relatebeta\trgb{\Omega}$, when they have related heaps.
The action relation $\src{\alpha}\relatebeta\at$ is defined following the intuition of \Cref{sec:rsc-theory}.
\myfig{
	\begin{center}
		\typerule{Heap relation}{
			\src{H}\relatebeta\trg{H_1;H_2}
			&
			\src{\ell}\relatebeta\trg{\pair{n,\eta}}
			&
			\src{v}\relatebeta\trg{v}
			&
			\trg{H} = \trg{H_1;n\mapsto v:\trgb{\eta};H_2}
		}{
			\src{H;\ell\mapsto v} \relatebeta \trg{H}
		}{hrel-i}
		\typerule{Empty relation}{
		}{
			\srce \relatebeta \trg{\OB{k}}
		}{hrel-b}
	\end{center}
	\hrule
	\begin{center}
		\typerule{Related states -- Whole}{
			\src{\Omega}=\src{M ; \OB{F},\OB{F'} ; \OB{I} ; H \triangleright s}
			&
			\trgb{\Omega}=\trg{M ; \OB{F},\compup{\OB{F'}} ; \OB{I} ; H \triangleright s}
			&
			\src{M}\relatebeta\trg{M}
			&
			\src{H}\relatebeta\trg{H}
		}{
			\src{\Omega}\relatebeta\trgb{\Omega}
		}{state-rel-whole}
	\end{center}
	\hrule
	\begin{center}
		\typerule{Call relation}{
			\src{f}\relate\trg{f}
			&
			\src{v}\relatebeta\trg{v}
			&
			\src{H}\relatebeta\trg{H}
		}{
			\src{\clh{f}{v}{H}}\relatebeta\trg{\clh{f}{v}{H}}
		}{rel-cl}
		\typerule{Callback relation}{
			\src{f}\relate\trg{f}
			&
			\src{v}\relatebeta\trg{v}
			&
			\src{H}\relatebeta\trg{H}
		}{
			\src{\cbh{f}{v}{H}}\relatebeta\trg{\cbh{f}{v}{H}}
		}{rel-cb}
		\typerule{Return relation}{
			\src{H}\relatebeta\trg{H}
		}{
			\src{\rth{}{H}}\relatebeta\trg{\rth{}{H}}
		}{rel-rt}
		\typerule{Returnback relation}{
			\src{H}\relatebeta\trg{H}
		}{
			\src{\rbh{}{H}}\relatebeta\trg{\rbh{}{H}}
		}{rel-rb}
		\typerule{Epsilon relation}{
		}{
			\src{\epsilon}\relatebeta\trgb{\epsilon}
		}{rel-ep}
	\end{center}
}{mon-st-act-rel}{Heap, state and action relations.}

With this relation we state a backwards simulation lemma (\Cref{thm:back-sim}) that is necessary for the \rscomp proof (and that can also yield whole program compiler correctness).
Technically, since the semantics is deterministic, this lemma is derived from \emph{forward} simulation, which is the same statement but with the source and target reductions swapped~\cite{leroy2}.
\begin{lemma}[Backward simulation]\label{thm:back-sim}
  \begin{align*}
    \text{if }
    &\
    \trg{C, H \triangleright \compup{s}}\xtot{\trgb{\lambda}}\trg{C, H \triangleright \compup{s'}} 
    \text{ and }
    \src{C, H \triangleright s} \relatebeta \trg{C, H \triangleright \compup{s}}
    \text{ and }
    \src{\lambda}\relatebeta\trgb{\lambda}
    \\
    \text{ then }
    &\
    \src{C, H \triangleright s}\xtos{\lambda}\src{C, H \triangleright s'} 
    \text{ and }
    \exists \beta'\supseteq\beta\ldotp
    \src{C, H \triangleright s'}\relate_{\beta'}\trg{C, H \triangleright \compup{s}}
  \end{align*}
\end{lemma}
The partial bijection $\beta$ grows as we consider successive steps of program execution in our proof.
For example, if executing \src{\letnew{x}{e}{s}} creates some source  location \src{\ell}, then executing its compiled counterpart will create some target location \trg{n} and then protect that location with a fresh capability \trg{k}.
At this point we add $(\src{\ell}, \trg{n}, \trg{k})$ to $\beta$.

\paragraph{Monitor Relation}\label{sec:mon-rel-conc}
In \Cref{sec:rsc-theory}, we left the monitor relation abstract. Here, we define it for our two languages.
Two monitors are related when they can \emph{simulate} each other on related heaps.
Given a monitor-specific relation $\src{\sigma} \relate \trgb{\sigma}$ on monitor states, we say that a relation \mc{R} on source and target monitors is a \emph{bisimulation} if the following hold whenever $\src{M}=\src{(\set{\sigma},\monred,\sigma_0,\ell_{root},\sigma_c)}$ and $\trg{M} = \trg{(\trgb{\set{\sigma}},\monred,\trgb{\sigma}_0,k_{root},\trgb{\sigma}_c)}$ are related by \mc{R}:
\begin{enumerate}
  \item $\src{\sigma_0} \relate \trg{\trgb{\sigma}_0}$, and $\src{\sigma_c} \relate \trg{\trgb{\sigma}_c}$, and
  \item For all $\beta$ containing $(\src{\ell_{root}}, \trg{0}, \trg{k_{root}})$ and all $\src{H}, \trg{H}$ with $\src{H} \relatebeta \trg{H}$: 
  \begin{enumerate}
    \item $(\src{\sigma_c}, \src{H}, \_) \in \src{\monred}$ iff $(\trg{\trgb{\sigma}_c}, \trg{H}, \_) \in \trg{\monred}$, and 
    \item $(\src{\sigma_c}, \src{H}, \src{\sigma'}) \in \src{\monred}$ and $(\trg{\trgb{\sigma}_c}, \trg{H}, \trg{\trgb{\sigma}'}) \in \trg{\monred}$ 
    imply 
    
    \noindent $\src{(\set{\sigma}, \monred, \sigma_0, \ell_{root}, \sigma')} \mc{R} \trg{(\trgb{\set{\sigma}}, \monred, \trgb{\sigma}_0, k_{root}, \trgb{\sigma}')}$.
  \end{enumerate}
\end{enumerate}
In words, \mc{R} is a bisimulation only if $\src{M} \mc{R} \trg{M}$ implies that $\src{M}$ and $\trg{M}$ simulate each other on heaps related by \emph{any} $\beta$ that relates $\src{\ell_{root}}$ to $\trg{0}$.
In particular, this means that neither $\src{M}$ nor $\trg{M}$ can be sensitive to the \emph{specific} addresses allocated during the run of the program. 
However, they can be sensitive to the ``shape'' of the heap or the values stored in the heap. 
Note that the union of any two bisimulations is a bisimulation. 
Hence, there is a largest bisimulation, which we denote as $\relate$.
Intuitively, $\src{M}\relate\trg{M}$ implies that $\src{M}$ and $\trg{M}$ encode the same safety property (up to the relation $\relatebeta$). 
With all the boilerplate for \rscomp in place, we state our main theorem.

\begin{theorem}[\compup{\cdot} attains \rscomp]\label{thm:comp-up-rsc}
  $\vdash\compup{\cdot} : \rscomp$
\end{theorem}
We outline our proof of \Cref{thm:comp-up-rsc}, which relies on a backtranslation we denote \backtrup{\cdot}.
Intuitively, \backtrup{\cdot} takes a target trace \alst and builds a \emph{set} of source contexts such that \emph{one} of them when linked with the source program \src{C}, produces a related trace \alss in the source (\Cref{thm:backtr-corr}).
In prior work, backtranslations return a single context~\cite{scoo-j,mfac,ahmedCPS,nonintfree,max-embed,Ahmed:2008:TCC:1411203.1411227,popl-backtrans-j}. 
This is because they all, explicitly or implicitly, assume that $\relate$ is injective from source to target. 
Under this assumption, the backtranslation is unique: a target value \trg{v} will be related to at most one source value \src{v}.
We do away with this assumption (e.g., the target value \trg{0} is related to both source values \src{0} and \trues) and thus there can be multiple source values related to any given target value. 
This results in a set of backtranslated contexts, of which at least one will reproduce the trace as we need it as presented in \Cref{ex:manyc}.
\begin{example}[Backtranslating a single context into a set]\label{ex:manyc}
  Consider a source component defining a single function \src{succ(x)\mapsto \letin{y}{x+1}{\skips}} and a target context linking against the compilation of that component. 
  Assume the context defines \trg{main(y)\mapsto \call{succ}\ 0}, which means that the trace semantics of the compiled component contains traces of the form \trg{\clh{succ}{0}{\_}}.
  Simply by reasoning at the target level, we cannot know whether \trg{0} will be used as a boolean (e.g., in an \trg{\ifzte{e}{e'}{e''}}) or as a natural number (e.g., in a \trg{x+1}).
  Thus, the backtranslation generates two contexts that call \src{succ} with both values that relate to \trg{0}:
  \begin{align*}
    \{
      \src{main(y)\mapsto \call{succ}\ \trues} ;
      \src{main(y)\mapsto \call{succ}\ 0} ;
    \}
  \end{align*}
  When the first \src{main} is linked with \src{succ} and they execute, the execution gets stuck inside \src{succ}: the \src{x+1} expression is effectively \src{\trues+1}, which does not reduce.
  Since we know that the compiled program emits a !-decorated action, the same must be true in the source too.
  Thus, it is not possible that the execution gets stuck while executing code of the component and there must be another context that does not make the execution get stuck.
  In this case, that context is the one with the \emph{second} implementation of \src{main}.
   
\end{example}

We bypass the lengthy technical setup for this proof and provide an informal description of why the backtranslation achieves what it is supposed to using an example (\Cref{ex:bt}).
We refer the interested reader to the accompanying technical report for full details on the backtranslation~\cite{rsc-techrep}.

\paragraph{Notation}
\Cref{ex:bt} needs to reason about lists of finite length whose elements are pairs, which are not a base type in our language.
However, they can be easily encoded as sequences of pairs, with \units being the empty list.
Thus, list \src{\pair{1,1}::\pair{2,2}::\pair{3,3}} is \src{\pair{\pair{1,1},\pair{\pair{2,2},\pair{\pair{3,3},\units}}}}.
To maintain a lightweight notation, we use some syntactic sugar to encode lists of finite length to our language.
We use metavariable \src{L} to indicate a pointer to a heap-allocated list of that form.
Adding element \src{v} to list \src{L} is denoted as \src{L::v}; given that the content of \src{L} was some list \src{list}, adding an element amounts to making \src{L} point to \src{\pair{v,list}}.
Given an element \src{\pair{v,n}}, we use notation \src{L(n)} to look that element up and return \src{v} (or \units if no element \src{\pair{v,n}} is in \src{L}).
We can easily encode this lookup as a series of projections on the the list and then on each element of the list (note that our language is untyped, so this is possible).

\begin{example}[Trace backtranslation]\label{ex:bt}
\backtrup{\cdot} first generates empty method bodies for all context methods called by the compiled component.
Then it backtranslates each \emph{action} on the given trace, generating code blocks that mimic that action and places that code inside the appropriate method body.
The figure below shows an example trace on the left and the code blocks generated for each action in the trace on the right.
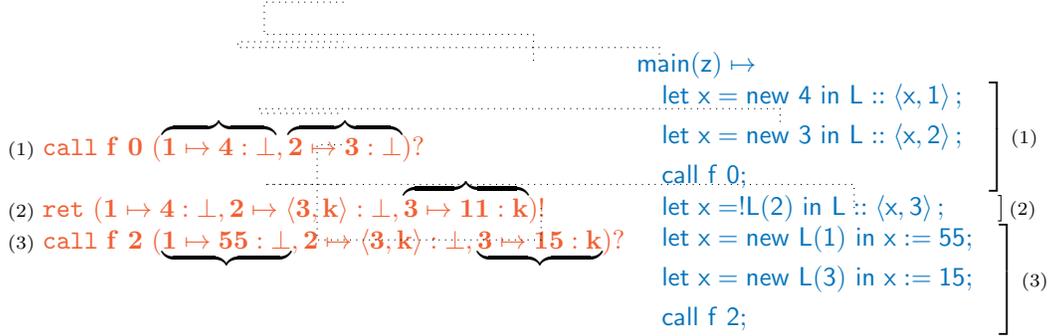
\begin{center}%
    \begin{tikzpicture}[remember picture]
      \node[align=left](trace){ 
        ${\scriptstyle (1)}\left. \trg{\clh{f}{0}{ (\bl{\overbrace{\trg{1\mapsto 4:\bot}}^{\tikz\node(h1){};}}, \bl{\overbrace{\trg{2\mapsto 3:\bot}}^{\tikz\node(h2){};}} )} } \right. $
        \\
        ${\scriptstyle (2)}\left. \trg{\rth{}{( 1\mapsto4:\bot,2\mapsto\pair{3,k}:\bot, \bl{\overbrace{\trg{3\mapsto 11:k}}^{\tikz\node(h3){};}} )}{}} \right. $
        \\
        ${\scriptstyle (3)}\left. \trg{\clh{f}{2}{( \bl{\underbrace{\trg{1\mapsto 55:\bot}}_{\tikz\node(h4){};}}, 2\mapsto\pair{3,k}:\bot, \bl{\underbrace{\trg{3\mapsto 15:k}}_{\tikz\node(h5){};}} )}} \right. $
        };

      \node[align=left, right of=trace, xshift = 17em](code){
        \src{main(z)\mapsto}
        \\
        $\left.
        \begin{aligned}[c]
          &\src{\tikz\node(ch1){}; \letnew{x}{4}{L::\pair{x,1}};}\quad
          \\
          &\src{\tikz\node(ch2){}; \letnew{x}{3}{L::\pair{x,2}};}
          \\
          &\src{\tikz\node(zzz){}; \call{f}~0;}
        \end{aligned}
        \right]{\scriptstyle (1)}$
        \\
        $\left.\src{\tikz\node(ch3){};\, \letin{x}{!L(2)}{L::\pair{x,3}};} \qquad \right]{\scriptstyle (2)}$
        \\
        $\left.
        \begin{aligned}[c]
          &\src{\tikz\node(ch4){}; \letin{x}{L(1)}{x:=55};}\quad
          \\
          &\src{\tikz\node(ch5){}; \letin{x}{L(3)}{x:=15};}
          \\
          &\src{\tikz\node(zz){}; \call{f}~2;}
        \end{aligned}
        \right]{\scriptstyle (3)}$
      };

      \draw[dotted] (h1.-90) |- ([yshift=1em]h1.-90) -| ([xshift=-3em]ch1.0) -- (ch1.0);
      \draw[dotted] (h2.-90) |- ([yshift=.5em]h2.-90) -| ([xshift=-4em]ch2.0) -- (ch2.0);
      \draw[dotted] (h3.-90) |- ([yshift=.5em]h3.-90) -| ([xshift=-3em]ch3.0) -- (ch3.0);
      \draw[dotted] (h4.90) |- ([yshift=-.5em]h4.-90) -| ([xshift=-1em]ch4.0) -- (ch4.0);
      \draw[dotted] (h5.90) |-  ([xshift=-3em]ch5.0) -- (ch5.0);
    \end{tikzpicture}
  \end{center}%
Backtranslated code maintains a support data structure at runtime, a list of locations denoted \src{L} that are known to the target.  Locations are looked up in this list based on their second field \src{n}, which is their target-level address.
Since we have access to the whole trace, we know how many locations we will add to \src{L} so we know its length.
In order to backtranslate the first call, we need to set up the heap with the right values and then perform the call.
In the diagram, dotted lines describe which source statement generates which part of the heap.
The return only generates code that will update the list \src{L} to ensure that the context has access to all the locations it knows in the target too.
In order to backtranslate the last call we look up the locations to be updated in \src{L} so we can ensure that when the \src{\call{f}~2} statement is executed, the heap is in the right state.	
\end{example}

For the backtranslation to be used in the proof we need to prove its correctness, i.e., that \backtrup{\alst} generates a context \src{A} that, together with \src{C}, generates a trace \alss related to the given target trace \alst.
As before, the relatedness of actions (and of states) is stated with respect to a partial bijection $\beta$ between source and target locations (and capabilities) that grows as the execution progresses.
\begin{theorem}[\backtrup{\cdot} is correct]\label{thm:backtr-corr}
  \begin{align*}
    \text{if }
    &
    \trg{A\hole{\compup{C}}} \Xtot{\alst} \trgb{\Omega}
    \text{ then }
    \exists\src{A}\in\backtrup{\alst}\ldotp 
    \src{A\hole{C}}\relatebeta\trg{A\hole{\compup{C}}}
    \text{ and }
    \src{A\hole{C}} \Xtos{\alss} \src{\Omega}
    \\
    \text{ and }
    &
    \exists \beta'\supseteq\beta \ldotp
    \alss\relate_{\beta'}\alst
    \text{ and }
    \src{\Omega}\relate_{\beta'}\trgb{\Omega}.
  \end{align*}
\end{theorem}
This theorem immediately implies that $\vdash\compup{\cdot}
: \pfrscomp$, which, by \Cref{thm:rsc-prf-eq2} below, implies that
$\vdash\compup{\cdot} : \rscomp$.

\begin{theorem}[\pfrscomp and \rscomp are equivalent for \compup{\cdot}]\label{thm:rsc-prf-eq2}
    \begin{align*}
      \vdash\compup{\cdot}: \pfrscomp \iff \vdash\compup{\cdot}:\rscomp
    \end{align*}
\end{theorem}
The intuition behind the proof of \Cref{thm:rsc-prf-eq2} follows the intuition we gave after \Cref{thm:rsc-prf-eq}.
The only missing element in the proof is to demonstrate that related monitor states either both step or both get stuck on related source and target actions. 
We prove this by showing an invariant, namely that the monitor states always remain related. This follows from the rules of 
\Cref{fig:mon-st-act-rel}.
Finally, recall that the function \strip{\alpha} returns just the heap
of the action.  This, with the relatedness of heaps, ensures that
either both monitors step or both get stuck on related actions.

\begin{example}[Compiling a secure program]\label{ex:request}
To illustrate \rscomp at work, consider the following source component \src{C_a}, which manages an account whose balance is security-relevant. The balance is stored in a location (\src{\ell_{root}} that is tracked by the monitor).
\src{C_a} provides functions to deposit to the account as well as to print the account balance.
\begin{align*}
  \src{deposit(x)\mapsto}
    &\ 
    \src{\letins{q}{abs(x)}{\letin{amt}{\,!\ell_{root}}{\ell_{root}:=amt+q}}}
  \\
  \src{balance(x)\mapsto}
    &\
    \src{\letins{tmp}{!\ell_{root}}{\skips}}
\end{align*}

\src{C_a} never leaks the sensitive location \src{\ell_{root}} to an attacker.
Additionally, an attacker has no way to decrement the amount of the balance since deposit only adds the absolute value \src{abs(x)} of its input \src{x} to the existing balance.

By compiling \src{C_a} with \compup{\cdot}, we obtain the following target program.
For simplicity of reading, we provide a simplification of the compiled programs below the output of the compiler.
\begin{align*}
  \trg{deposit(x)\mapsto}
    &\
    \letint{\trg{q}}{\trg{abs(x)}}{
      \\&\ \
      \letint{\trg{amt}}{\trg{!\projone{\pair{0,k_{root}}} \with{\projtwo{\pair{0,k_{root}}}}}}{
      \\&\ \ \ 
        \letint{\trg{l1}}{\trg{\projone{\pair{0,k_{root}}}}}{
        \\&\ \ \ \ 
          \letint{\trg{l2}}{\trg{\projtwo{\pair{0,k_{root}}}}}{
          \\&\ \ \ \ \
            \trg{l1 := amt+q \with{l2}}
          }
        }
      }
    }
  \\
  \trg{balance(x)\mapsto}
    &\
    \trg{\letint{tmp}{!\projone{\pair{0,k_{root}}}\with{\projtwo{\pair{0,k_{root}}}}}{\skipt}}
  \\
  &\ \text{ simplified as }
  \\
  \trg{deposit(x)\mapsto}
    &\
      \letint{\trg{q}}{\trg{abs(x)}}{
      \\
        &\ \ 
        \letint{\trg{amt}}{\trg{!0}\with{k_{root}}}{
        \\&\ \ \ 
        \trg{0:=amt+q}\with{k_{root}}
        }
      }
  \\
  \trg{balance(x)\mapsto}
    &\
    \trg{\letint{tmp}{!0\with{k_{root}}}{\skipt}}
\end{align*}

Recall that location \src{\ell_{root}} is mapped to location \trg{0} and protected by the \trg{k_{root}} capability.
In the compiled code, while location \trg{0} is freely computable by a target attacker, capability \trg{k_{root}} is not.
Since that capability is not leaked to an attacker by the code above, an attacker will not be able to tamper with the balance stored in location \trg{0}, even though it has direct access to \trg{0}.
\end{example}

\paragraph{Failing to Attain \rscomp}
We now provide an insight of how could we detect whether a compiler is not \rscomp.
In fact, if we had made a security-relevant mistake in the compiler, we would like one of the proofs to fall apart.
Two ways to fail at attaining \rscomp come to mind: not using capabilities or leaking them; we explore the former for simplicity.
Let us assume that the compiler does not protect compiled code with capabilities.
Intuitively, compiled code is insecure because location \trg{0} is freely accessible to the attacker, who can alter its content at will.
In terms of proofs, this vulnerability manifests itself \Cref{thm:backtr-corr}, whose proof would fail.
Let us consider a valid target trace emitted by code compiled with the vulnerable compiler we discuss here.
Such a trace would be one where the attacker changes the value of \trg{0} as indicated below.
There, the content of that location is reset to \trg{3} between the return and the call (i.e., when the attacker executes).
\begin{align*}
  \clh{f}{v}{0\mapsto 3:\bot} \cdot \rth{v}{0\mapsto 1:\bot} \cdot \clh{f}{v}{0\mapsto 3:\bot}
\end{align*}
No source attacker can do that, because location \src{\ell_{root}} is not accessible to them.
This would make the proof of \Cref{thm:backtr-corr} fail and reveal a hint that something that was supposed to be secure is now vulnerable.

\section{\rscomp via Bisimulation}\label{sec:comp-effi}
If the source language has a verification system that enforces robust safety, proving that a compiler attains \rscomp \emph{may} be simpler than that of \Cref{sec:rsc-instance} in some cases, as a backtranslation may not be needed at all.
To demonstrate this, we consider a specific class of monitors, namely those that enforce type invariants on a specific set of locations. 
Our source language, \LA, is similar to \LU but it has a type system that accepts only source programs whose traces the source monitor never rejects. 
Our target language is mostly unchanged. Our compiler \compap{\cdot} is directed by typing derivations, and its proof of \rscomp relies on a cross-language invariant of program execution rather than a backtranslation.
A second, independent goal of this section is to show that \rscomp is compatible with concurrency. 
Consequently, our source and target languages include constructs for forking threads.

\subsection{The Source Language \LA}\label{sec:src-a}
\myfig{
  \begin{gather*}
  	\begin{aligned}
    \mi{Components}~\src{C} \bnfdef&\ \src{\Delta ; \OB{F} ; \OB{I}}
    &
    \mi{Heaps}~\src{H} \bnfdef&\ \srce \mid \src{H ; \ell\mapsto v:\tau}
    \\
    \mi{Types}~\src{\tau} \bnfdef&\ \src{\Bools} \mid \src{\Nats} \mid \src{\tau\times\tau} \mid \src{\Refs{\tau}} \mid \UNS
    &
    \mi{Envs.}~\src{\Gamma} \bnfdef&\ \srce \mid \src{\Gamma; (x:\tau)}
    \end{aligned}
    \\
    \begin{aligned}
    \mi{Superf.\ Types}~\src{\varphi} \bnfdef&\ \src{\Bools} \mid \src{\Nats} \mid \src{\UNS\times\UNS} \mid \src{\Refs{\UNS}}
    \\
    \mi{Statements}~\src{s} \bnfdef&\ \cdots
    \mid \src{\fork{s}}
    \mid \src{\myendorse{x}{e}{\varphi}{s}}
    \end{aligned}
    \\
    \begin{aligned}
    \mi{Monitors}~\src{M} \bnfdef&\ \src{(\set{\sigma},\monred,\sigma_0,\Delta,\sigma_c)}
    &
    \mi{Mon.\ Trans.}~\src{\monred} \bnfdef&\ \srce \mid \src{\monred;(\sigma,\sigma)}
    \\
    \mi{Store\ Env.}~\src{\Delta} \bnfdef&\ \srce \mid \src{\Delta; (\ell:\tau)}
    &
    \mi{Processes}~ \src{\pi} \bnfdef&\ \src{\proc{s}{}}
    \\
    \mi{Soups}~\src{\Pi} \bnfdef&\ \srce \mid \src{\Pi \parallel \pi}
    &
    \mi{Prog.\ States}~\src{\Omega}\bnfdef&\ \src{C, H\triangleright \Pi} %
  	\end{aligned}
  \end{gather*}
}{la}{Syntax of \LA. Elided and omitted elements are the same as in \LU (\Cref{fig:lu-syn}).}
\LA extends \LU with concurrency, so it has a fork statement \src{\fork{s}}, processes and process soups~\cite{cham} and an extensive type system (\Cref{fig:la}).
Components define a set of safety-relevant locations \src{\Delta}, so %
and heaps carry type information. %
\src{\Delta} also specifies a type for each safety-relevant location. %

\LA has an unconventional type system that enforces \emph{robust type safety}~\cite{autysec,catalin-rs,sec-typ-prot,tydisa,cca,ot4jc}, which means that no context can cause the static types of sensitive heap locations to be violated at runtime. 
Using a special type $\UNS$ that is described below, a program component statically partitions heap locations it deals with into those it cares about (sensitive or ``trusted'' locations) and those it does not care about (``untrusted'' locations). 
Call a value \emph{shareable} if only untrusted locations can be extracted from it using the language's elimination constructs. 
The type system then ensures that a program component only ever shares shareable values with the context. 
This ensures that the context cannot violate any invariants (including static types, which is what we care about in this section) of the trusted locations, since it can never gets direct access to them. 
\myfig{
  \begin{align*}
  &\vdash \src{C}:\UNS
  &&\text{Component \src{C} is well-typed.}
  &
  &\src{C}\vdash\src{F} : \src{\tau} 
  &&\text{Function \src{F} takes arguments of type \src{\tau}.}
  \\
  &\src{\Delta,\Gamma}\vdash\diamond 
  &&\text{Well-formed environments.}
  &
  &\src{\Delta,\Gamma}\vdash \src{e} : \src{\tau}
  &&\text{Expression \src{e} has type \src{\tau} in \src{\Gamma} and \src{\Delta}.}
  \\
  &\src{\tau}\vdash\circ 
  && \text{Type \src{\tau} is insecure.}
  &
  &\src{C,\Delta,\Gamma}\vdash \src{s} %
  &&\text{Statement \src{s} is well-typed in \src{C}, \src{\Gamma} and \src{\Delta}.}
  \end{align*}

  \hrule
  \begin{center}
    \typerule{T\LA-bool-pub}{
    }{
      \src{\Bools}\vdash\circ
    }{ts-ts-b}
    \typerule{T\LA-nat-pub}{
    }{
      \src{\Nats}\vdash\circ
    }{ts-ts-b}
    \typerule{T\LA-pair-pub}{
      \src{\tau}\vdash\circ
      &
      \src{\tau'}\vdash\circ
    }{
      \src{\tau\times\tau'}\vdash\circ
    }{ts-ts-p}
    \typerule{T\LA-un-pub}{
    }{
      \UNS\vdash\circ
    }{ts-tp-f}
    \typerule{T\LA-references-pub}{
    }{
      \src{\Refs{\UNS}}\vdash\circ
    }{ts-tp-f}
  \end{center}

  \hrule
  \begin{center}
	\typerule{T\LA-true}{
		\src{\Delta,\Gamma}\vdash\diamond
	}{
		\src{\Delta,\Gamma}\vdash\trues:\Bools
	}{ts-true}
	\typerule{T\LA-false}{
		\src{\Delta,\Gamma}\vdash\diamond
	}{
		\src{\Delta,\Gamma}\vdash\falses:\Bools
	}{ts-false}
	\typerule{T\LA-nat}{
		\src{\Delta,\Gamma}\vdash\diamond
	}{
		\src{\Delta,\Gamma}\vdash\src{n}:\Nats
	}{ts-nat}
	\typerule{T\LA-var}{
		\src{x:\tau}\in\src{\Gamma}
	}{
		\src{\Delta,\Gamma}\vdash\src{x}:\src{\tau}
	}{ts-var}
	\typerule{T\LA-loc}{
		\src{l:\tau}\in\src{\Delta}
	}{
		\src{\Delta,\Gamma}\vdash\src{l}:\src{\Refs{\tau}}
	}{ts-loc}
	\typerule{T\LA-pair}{
		\src{\Delta,\Gamma}\vdash \src{e_1} : \src{\tau}
		\\
		\src{\Delta,\Gamma}\vdash \src{e_2} : \src{\tau'} 
	}{
		\src{\Delta,\Gamma}\vdash \src{\pair{e_1,e_2}} : \src{\tau\times\tau'}
	}{ts-pair}
	\typerule{T\LA-proj-1}{
		\src{\Delta,\Gamma}\vdash \src{e} : \src{\tau\times\tau'}
	}{
		\src{\Delta,\Gamma}\vdash \src{\projone{e}} : \src{\tau}
	}{ts-p1}
	\typerule{T\LA-proj-2}{
		\src{\Delta,\Gamma}\vdash \src{e} : \src{\tau\times\tau'}
	}{
		\src{\Delta,\Gamma}\vdash \src{\projtwo{e}} : \src{\tau'}
	}{ts-p2}
	\typerule{T\LA-dereference}{
		\src{\Delta,\Gamma}\vdash \src{e} : \src{\Refs{\tau}} 
	}{
		\src{\Delta,\Gamma}\vdash \src{!e} : \src{\tau}  
	}{ts-deref}
	\typerule{T\LA-op}{
		\src{\Delta,\Gamma}\vdash \src{e} : \Nats
		&
		\src{\Delta,\Gamma}\vdash \src{e'} : \Nats
	}{
		\src{\Delta,\Gamma}\vdash \src{e \op e'} : \src{\Nats}  
	}{ts-op}
	\typerule{T\LA-cmp}{
		\src{\Delta,\Gamma}\vdash \src{e} : \Nats
		&
		\src{\Delta,\Gamma}\vdash \src{e'} : \Nats
	}{
		\src{\Delta,\Gamma}\vdash \src{e \bop e'} : \Bools
	}{ts-bop}
	\typerule{T\LA-coercion}{
		\src{C,\Delta,\Gamma}\vdash \src{e} : \src{\tau} 
		&
		\src{\tau} \vdash \circ
	}{
		\src{C,\Delta,\Gamma}\vdash \src{e} : \UNS 
	}{ts-coe}
\end{center}
\hrule
\begin{center}
	\typerule{T\LA-skip}{
	}{
		\src{C,\Delta,\Gamma}\vdash \skips
	}{ts-skip}	
	\typerule{T\LA-function-call}{
		((\src{f}\in\dom{\src{C}.\mtt{funs}})\vee(\src{f}\in\dom{\src{C}.\mtt{intfs}}))
		\\
		\src{\Delta,\Gamma}\vdash \src{e} : \src{\UNS}
	}{
		\src{\Delta,\Gamma}\vdash \src{\call{f}~e} %
	}{ts-fun}
	\typerule{T\LA-sequence}{
		\src{C,\Delta,\Gamma}\vdash \src{s_u} %
		\\
		\src{C,\Delta,\Gamma}\vdash \src{s} %
	}{
		\src{C,\Delta,\Gamma}\vdash \src{s_u;s} %
	}{ts-seq}
	\typerule{T\LA-letin}{
		\src{\Delta,\Gamma}\vdash \src{e} : \src{\tau} 
		\\
		\src{C,\Gamma;x:\tau}\vdash \src{s} %
	}{
		\src{C,\Delta,\Gamma}\vdash \src{\letin{x:\tau}{e}{s}} %
	}{ts-vardef}
	\typerule{T\LA-assign}{
		\src{\Delta,\Gamma}\vdash \src{x} : \src{\Refs{\tau}}
		\\
		\src{\Delta,\Gamma}\vdash \src{e'} : \src{\tau} 
	}{
		\src{C,\Delta,\Gamma}\vdash \src{x := e'} %
	}{ts-ass}
	\typerule{T\LA-new}{
		\src{\Delta,\Gamma}\vdash \src{e} : \src{\tau} 
		\\
		\src{C,\Gamma;x:\Refs{\tau}}\vdash \src{s} %
	}{
		\src{C,\Delta,\Gamma}\vdash \src{\letnewty{x}{e}{\tau}{s}} %
	}{ts-new}
	\typerule{T\LA-if}{
		\src{\Delta,\Gamma}\vdash \src{e} : \src{\Bool}
		\\
		\src{C,\Delta,\Gamma}\vdash \src{s_t} %
		&
		\src{C,\Delta,\Gamma}\vdash \src{s_e} %
	}{
		\src{C,\Delta,\Gamma}\vdash \src{\ifte{e}{s_t}{s_e}} %
	}{ts-if}
	\typerule{T\LA-fork}{
		\src{C,\Delta,\Gamma}\vdash \src{s} %
	}{
		\src{C,\Delta,\Gamma}\vdash \src{\fork{s}} %
	}{ts-fork}
	\typerule{T\LA-endorse}{
		\src{\Delta,\Gamma}\vdash\src{e}:\UNS
		\\
		\src{C,\Delta,\Gamma;(x:\varphi)} \vdash\src{s} %
	}{
		\src{C,\Delta,\Gamma}\vdash \src{\myendorse{x}{e}{\varphi}{s}} %
	}{ts-end}
\end{center}
}{la-ts}{Typing judgements and rules of \LA.}

Type \UNS stands for ``untrusted'' or ``shareable'' and contains all values that can be passed to the context. 
Every type that is not a subtype of \UNS is implicitly trusted and cannot be passed to the context. 
Untrusted locations are explicitly marked \UNS at their allocation points in the program. 
Other types are deemed shareable via subtyping. 
Intuitively, a type is safe if values in it can only yield locations of type \UNS by the language's elimination constructs. 
For example, $\src{\UNS\times\UNS}$ is a subtype of \UNS. 
We write $\src{\tau}\vdash\src{\circ}$ to mean that $\src{\tau}$ is a subtype of \UNS.

Further, \LA contains an \emph{endorsement} statement (\src{\myendorse{x}{e}{\varphi}{s}}) that dynamically checks the top-level constructor of a value of type \UNS  and gives it a more precise superficial type $\src{\varphi}$~\cite{chong-thesis}. %
This allows a program to safely inspect values coming from the context. 
It is similar to existing type casts~\cite{nonpapa} but it only inspects one structural layer of the value (this simplifies the compilation).

\myfig{
  \begin{center}
    \typerule{E\LA-endorse}{
      \src{H \triangleright e \redtos v}
      &
      \src{\Delta,\srce}\vdash\src{v}:\src{\varphi}
      \\
      \src{\Delta}= \myset{\src{\ell:\tau}}{\src{\ell\mapsto v':\tau}\in\src{H}}
    }{
    		\src{C, H \triangleright \myendorse{x}{e}{\varphi}{s} } \xtos{} 
    		\src{C, H \triangleright s\subs{v}{x}}
    }{es-end}
    \typerule{E\LA-fork}{
      \src{\Pi}=\src{\Pi_1 \parallel \fork{s};s' \parallel \Pi_2}
      \\
      \src{\Pi'}=\src{\Pi_1 \parallel \skips;s' \parallel \Pi_2 \parallel s}
    }{
      \src{C, H \triangleright \Pi} \xtos{} \src{C, H \triangleright \Pi'}
    }{es-fork}
  \end{center}
}{la-sem}{Semantics of \LA. Omitted elements are the same as in \LU (\Cref{fig:lu-sem}).}
The operational semantics of \LA updates that of \LU to deal with concurrency and endorsement (\Cref{fig:la-sem}). 
For concurrency, the program state $\src{\Omega}$ contains a soup (i.e., a multiset) \src{\Pi} of processes, where each process is a statement executing as in the program state for \LU, and a soup takes a step if any process in it does.
The latter performs a runtime check on the endorsed value~\cite{dpdec}, which performs a syntactic check on a value given some superficial type \src{\varphi}.
Superficial types \src{\varphi} only allow checking types ``on the surface'', so pairs and references are not nested; in order to endorse a nested pair, multiple endorse statements must be used.

Monitors \src{M} %
check at runtime that the set of trusted heap locations \src{\Delta} have values of their intended static types. 
Accordingly, the description of the monitor includes a list of trusted locations and their expected types (in the form of an environment $\src{\Delta}$). 
The type $\src{\tau}$ of any location in $\src{\Delta}$ must be trusted, so $\src{\tau} \not\vdash \src{\circ}$. 
To facilitate the monitor's checks, every heap location carries a type at runtime (in addition to a value). 
The monitor transitions should, therefore, be of the form \src{(\sigma,\Delta,\sigma')}, but since \src{\Delta} never changes (it maps trusted locations to \emph{static} types), we write the transitions as pairs of states only. %

A monitor and a component agree if they have the same $\src{\Delta}$:
\begin{align*}
  \src{M\agree C}\isdef\ \src{M} = \src{(\set{\sigma},\monred,\sigma_0,\Delta,\sigma_c)} \text{ and } \src{C}=\src{(\Delta ; \OB{F} ; \OB{I})}
\end{align*}
Other definitions (safety, robust safety and actions) are as in \Cref{sec:rsc}.

Importantly, we show that a well-typed component generates traces that are always accepted by an agreeing monitor, so every component typed at \UNS is robustly safe.

\begin{theorem}[Typability Implies Robust Safety in \LA]\label{thm:src-ty-impl-safe}
  \begin{align*}
    \text{If }
    \vdash\src{C}:\UNS
    \text{ and }
    \src{C\agree M}
    \text{ then }
    \src{M}\vdash\src{C}:\src{rs}
  \end{align*}
\end{theorem}

\paragraph{Richer Source Monitors}
In \LA, source language monitors only enforce the property of type safety on specific memory locations (robustly). 
This can be generalized substantially to enforce arbitrary invariants other than types on locations. 
The only requirement is to find a type system (e.g., based on refinements or Hoare logics) that can enforce robust safety in the source (for example, as in the work of Swasey \etal~\cite{davidcaps}). 
Our compilation and proof strategy should work with little modification. 
Another easy generalization is allowing the set of locations considered by the monitor to grow over time, as in \Cref{sec:rsc-instance}.

\subsection{The Target Language \LC}\label{sec:src-p}
Our target language, \LC, extends the previous target language \LP, with support for concurrency (forking, processes and process soups), atomic co-creation of a protected location and its protecting capability and for examining the top-level construct of a value according to a pattern $\trg{B}$ (\Cref{fig:lc}).
\myfig{
\begin{align*}
  \mi{Statements}~\trg{s} \bnfdef&\ \cdots \mid \trg{\fork{s}} 
  \mid \trg{\letatom{x}{e}{s}}
  \mid \trg{\destruct{x}{e}{B}{s}{s}}
  \\
  \mi{Patterns}~\trg{B} \bnfdef&\ \trg{nat} \mid \trg{pair}
  \qquad\qquad
  \mi{Monitors}~\trg{M} \bnfdef\ \trg{(\trgb{\set{\sigma}},\monred,\trgb{\sigma}_0,H_0,\trgb{\sigma}_c)}
\end{align*}

\hrule
\begin{center}
  \typerule{E\LC-destruct-nat}{
    \trg{H\triangleright e \redtot n}
  }{
      \trg{C, H \triangleright \destruct{x}{e}{nat}{s}{s'} } 
      \xtot{} \trg{C, H \triangleright s\subt{n}{x}}
  }{et-end-n}
  \typerule{E\LC-new}{
    \trg{H} = \trg{H_1;n\mapsto (v',\trgb{\eta})}
    &
    \trg{H \triangleright e\redtot v}
    &
    \trg{k}\notin\dom{\trg{H}}
    &
    \trg{s'}=\trg{s}\subt{\pair{n+1,k}}{x}
  }{
      \trg{C, H \triangleright \letatom{x}{e}{s}} 
      \xtot{} \trg{C,H; n+1\mapsto v:k;k \triangleright s'}
  }{et-atom}
\end{center}
}{lc}{Syntax and semantics of \LC. Elided elements are either the same as \LP (\Cref{fig:lp-syn,fig:lp-sem}) or \LA (\Cref{fig:la}).}

Monitors are also updated to consider a fixed set of locations (a heap part \trg{H_0}).
Atomic co-creation of locations and capabilities is provided to match modern security architectures such as Cheri~\cite{cheri} (which implement capabilities at the hardware level).
This atomicity is not strictly necessary and we prove that \rscomp is attained both by a compiler relying on it and by one that allocates a location and then protects it non-atomically. 
The former compiler (with this atomicity in the target) is a bit easier to describe, so we start with it (\Cref{sec:compap}) before moving to the non-atomic one (\Cref{sec:nonatom-alloc}).

\subsection{Compiler from \LA to \LC}\label{sec:compap}
The high-level structure of the compiler, \compap{\cdot}, is similar to that of our earlier compiler \compup{\cdot} (\Cref{sec:comp-up}). 
However, \compap{\cdot} is defined by induction on the type derivation of the source component to be compiled.
Most cases are a straightforward adaptation of the analogous cases from \Cref{fig:compup}, so we omit them. We show only a few instructional cases in \Cref{fig:compap,fig:compap2}.
When compiling a component we ensure that monitor-sensitive locations from \src{\Delta} are allocated to related locations and initialised to valid values, i.e., values that respect the cross-language relation. 
Such a set up of heaps is denoted with $\src{\Delta}\vdash_{\beta_0}\trg{H_0}$, whose details are in \Cref{tr:ini-heap} (\Cref{fig:details-1}).  
Intuitively, the domain of the safety-relevant heap must be related to the domain of \src{\Delta}.
Additionally, the heap is populated with values \trg{v} whose source-level counterparts have type \src{\tau}; for locations, we do not allow cycles in memory for simplicity (\Cref{tr:ini-val}).
The most interesting cases of the compiler are are allocation and endorsement.
The former explicitly uses type information to achieve security efficiently, protecting only those locations whose type is not \UNS.
The latter performs a 1-level de-structuring of the value to be endorsed according to the expected superficial type \src{\varphi}.
\myfig{
	\begin{align*}
		\compap{
			\typerulenolabel{T\LA-component}{
				\src{C}\equiv\src{\Delta ; \OB{F} ; \OB{I}} 
				&
				\src{C}\vdash\src{\OB{F}}:\UNS
				\\
				\fun{names}{\src{\OB{F}}}\cap\fun{names}{\src{\OB{I}}}=\emptyset
				&
				\src{\Delta}\vdash\src{ok}
			}{
				\vdash \src{C}:\UNS
			}{compap-co}
		} &= \trg{%
		\trg{H_0} ;
		\compap{\OB{F}}; \compap{\OB{I}}}
		\qquad
		\text{if }\src{\Delta}\vdash_{\beta_0}\trg{H_0}
		\\
		\compap{
			\typerulenolabel{T\LA-function}{
				\src{F}\equiv \src{f(x:\UNS)\mapsto s;\ret}
				&
				\src{C},\src{\Delta;x:\UNS}\vdash \src{s} %
				\\
				\forall\src{f}\in\fun{fn}{s}, \src{f}\in\dom{\src{C}.\mtt{funs}} 
				\vee \src{f}\in\dom{\src{C}.\mtt{intfs}}
			}{
				\src{C}\vdash\src{F}:\src{\UNS} 
			}{compap-fu}
		} &= \trg{f(x)\mapsto \compap{\src{C};\src{\Delta;x:\UNS}\vdash \src{s}};\ret}
	\end{align*}
	\hrule
	\begin{gather*}
		\begin{aligned}
		\compap{
			\typerulenolabel{T\LA-nat}{
				\src{\Delta,\Gamma}\vdash\diamond
			}{
				\src{\Delta,\Gamma}\vdash\src{n}:\Nats
			}{compap-nat}
		}
		=&\ 
			\trg{n}
		&
		\compap{
			\typerulenolabel{T\LA-var}{
				\src{x:\tau}\in\src{\Gamma}
			}{
				\src{\Delta,\Gamma}\vdash\src{x}:\src{\tau}
			}{compap-var}
		}
		=&\  \trg{x}
		&
		\compap{
			\typerulenolabel{T\LA-loc}{
				\src{\ell:\tau}\in\src{\Delta}
			}{
				\src{\Delta,\Gamma}\vdash\src{\ell}:\src{\tau}
			}{compap-loc}
		}
		=&\  
				\trg{\pair{n,v}}
		\end{aligned}
		\\
		\begin{aligned}
		\compap{
			\typerulenolabel{T\LA-coercion}{
				\src{\Delta,\Gamma}\vdash \src{e} : \src{\tau} 
				&
				\src{\tau} \vdash \circ
			}{
				\src{\Delta,\Gamma}\vdash \src{e} : \UNS 
			}{compap-coe}
		}
		=&\
				\trg{
					\compap{\src{\Delta,\Gamma}\vdash \src{e} : \src{\tau} }
				}
	\\
		\compap{
			\typerulenolabel{T\LA-op}{
				\src{\Delta,\Gamma}\vdash \src{e} : \Nats
				&
				\src{\Delta,\Gamma}\vdash \src{e'} : \Nats
			}{
					\src{\Delta,\Gamma}\vdash 
					\src{e \op e'} : \src{\Nats}  
			}{compap-op}
		}
		=&\ \trg{
					\compap{\src{\Delta,\Gamma}\vdash \src{e} : \Nats} 
					\op \compap{\src{\Delta,\Gamma}\vdash \src{e'} : \Nats}
			}
		\\
		\compap{
			\typerulenolabel{T\LA-cmp}{
				\src{\Delta,\Gamma}\vdash \src{e} : \Nats
				&
				\src{\Delta,\Gamma}\vdash \src{e'} : \Nats
			}{
					\src{\Delta,\Gamma}\vdash 
					\src{e \bop e'} : \Bools
			}{compap-bop}
		}
		=&\ 
			\trg{
					\compap{\src{\Delta,\Gamma}\vdash \src{e} : \Nats} 
					\bop \compap{\src{\Delta,\Gamma}\vdash \src{e'} : \Nats}	
			}
		\\
		\compap{
			\typerulenolabel{T\LA-dereference}{
				\src{\Delta,\Gamma}\vdash \src{e} : \src{\Refs{\tau}} 
			}{
				\src{\Delta,\Gamma}\vdash \src{!e} : \src{\tau}  
			}{compap-deref}
		}
		=&\ 
		\trg{
			!\projone{\compap{\src{\Delta,\Gamma}\vdash \src{e} : \src{\Refs{\tau}}}}
			\with{\projtwo{\compap{\src{\Delta,\Gamma}\vdash \src{e} : \src{\Refs{\tau}}}}} %
		}
		\end{aligned}
	\end{gather*}
}{compap}{\compap{\cdot}, compilation of components, functions and expressions from \LA to \LC (excerpts, omitted elements are analogous to their counterparts in \Cref{fig:compup}).}
\myfig{
	\begin{gather*}
	\begin{aligned}
    \compap{
      \typerulenolabel{T\LA-new}{
        \src{\Delta,\Gamma}\vdash \src{e} : \src{\tau} 
        &
        \src{C,\Delta,\Gamma;x:\Refs{\tau}}\vdash \src{s}
      }{
          \src{C,\Delta,\Gamma}\vdash 
          \src{\letnewty{x}{e}{\tau}{s}} 
      }{compap-new} 
    }
    =&\
    \begin{cases}
      \trg{
      \begin{aligned}
        &
        \letnewt{\trg{xo}~}{\compap{\src{\Delta,\Gamma}\vdash \src{e} : \src{\tau}} 
        \\
        &\
        }{ 
          \letint{\trg{x}~}{~\trg{\pair{xo,0}}
          \\
          &\ \ 
          }{
            \compap{\src{C,\Delta,\Gamma;x:\Refs{\tau}}\vdash \src{s}} 
          }
        }
      \end{aligned}
      }
      &
      \text{if }\src{\tau}=\UNS %
      \\
      \vspace{-1em}
      \\
      \trg{
      \begin{aligned}
        &
        \letatom{\trg{x}~}{ \compap{\src{\Delta,\Gamma}\vdash \src{e} : \src{\tau}} 
        \\
        &\
        }{ 
          \compap{\src{C,\Delta,\Gamma;x:\Refs{\tau}}\vdash \src{s}} 
        }
      \end{aligned}
      }
      &
      \text{otherwise} 
    \end{cases}
    \\
		\compap{
			\typerulenolabel{T\LA-if}{
				\src{\Delta,\Gamma}\vdash \src{e} : \src{\Bool}
				\\
				\src{C,\Delta,\Gamma}\vdash \src{s_t} 
				&
				\src{C,\Delta,\Gamma}\vdash \src{s_e} 
			}{
				\src{C,\Delta,\Gamma}\vdash \src{\ifte{e}{s_t}{s_e}} 
			}{compap-if}
		}
		=&\ \trg{
			\begin{aligned}
				&\ifztet{\compap{\src{\Delta,\Gamma}\vdash \src{e} : \src{\Bool}}
				\\
				&}
				{\compap{\src{C,\Delta,\Gamma}\vdash \src{s_t} }
				}{\compap{\src{C,\Delta,\Gamma}\vdash \src{s_e} }}
			\end{aligned}
		}
		\\
		\compap{
			\typerulenolabel{T\LA-sequence}{
				\src{C,\Delta,\Gamma}\vdash \src{s_u} 
				&
				\src{C,\Delta,\Gamma}\vdash \src{s} 
			}{
				\src{C,\Delta,\Gamma}\vdash \src{s_u;s} 
			}{compap-seq}
		}
		=&\ \trg{\compap{\src{C,\Delta,\Gamma}\vdash \src{s_u} } ; \compap{\src{C,\Delta,\Gamma;\Gamma'}\vdash \src{s} }}
		\\
		\compap{
			\typerulenolabel{T\LA-letin}{
				\src{\Delta,\Gamma}\vdash \src{e} : \src{\tau} 
				&
				\src{C,\Delta,\Gamma;x:\tau}\vdash \src{s}
			}{
				\src{C,\Delta,\Gamma}\vdash \src{\letin{x:\tau}{e}{s}}
			}{compap-vardef}
		}
		=&\ \trg{
				\letint{\trg{x}}{\compap{\src{\Delta,\Gamma}\vdash \src{e} : \src{\tau}}
				}{\compap{\src{C,\Delta,\Gamma;x:\tau}\vdash \src{s} }}
			}
		\\
		\compap{
			\typerulenolabel{T\LA-assign}{
				\src{\Delta,\Gamma}\vdash \src{x} : \src{\Refs{\tau}}
				&
				\src{\Delta,\Gamma}\vdash \src{e} : \src{\tau} 
			}{
				\src{C,\Delta,\Gamma}\vdash \src{x := e} 
			}{compap-ass}
		}
		=&\ \trg{
			\begin{aligned}
				&
				\letint{\trg{x_l}~}{~\trg{\projone{x}}
				}{
					\letint{\trg{x_c}~}{~\trg{\projtwo{x}}
					\\
					&\ \ 
					}{
						\trg{x_l := \compap{\src{\Delta,\Gamma}\vdash \src{e} : \src{\tau}} \with{x_c}}
					}
				}
			\end{aligned}
		}
		\\
		\compap{
			\typerulenolabel{T\LA-fork}{
				\src{C,\Delta,\Gamma}\vdash \src{s} 
			}{
				\src{C,\Delta,\Gamma}\vdash \src{\fork{s}} 
			}{compap-fork}
		}
		=&\ \trg{\fork{\compap{\src{C,\Delta,\Gamma}\vdash \src{s} }}}
		\end{aligned}
	\\
	\begin{aligned}
		&
		\compap{
			\typerulenolabel{T\LA-endorse}{
				\src{\Delta,\Gamma}\vdash\src{e}:\UNS
				&
				\src{C,\Delta,\Gamma;(x:\varphi)} \vdash\src{s} %
			}{
				\src{C,\Delta,\Gamma}\vdash \src{\myendorse{x}{e}{\varphi}{s}} 
			}{compap-end}
		}
		=
		\\
		&\
			\begin{cases}
				\trg{
					\left|\begin{aligned}
						&
						\destruct{\trg{x}~}{\compap{\src{\Delta,\Gamma}\vdash \src{e} : \UNS}}{\trg{nat}}{
							\\
							&\ \
							\ifztet{\trg{x}}{
							\compap{ \src{C,\Delta,\Gamma;(x:\varphi)} \vdash\src{s} }
							}{
							\\
							&\quad
								\ifztet{\trg{x-1}}{
								\compap{ \src{C,\Delta,\Gamma;(x:\varphi)} \vdash\src{s} }
								}{\wrong}
							}
							\\
							&
						}{\wrong}
					\end{aligned}\right.
				}
				&
				\text{ if }\src{\varphi}=\src{\Bools}
				\\
				\trg{
					\left|\begin{aligned}
						&
						\destruct{\trg{x}~}{\compap{\src{\Delta,\Gamma}\vdash \src{e} : \UNS}}{\trg{nat}}{
							\compap{ \src{C,\Delta,\Gamma;(x:\varphi)} \vdash\src{s} }
							\\
							&
						}{\wrong}
					\end{aligned}\right.
				}
				&
				\text{ if }\src{\varphi}=\src{\Nats}
				\\
				\trg{
					\left|\begin{aligned}
						&
						\destruct{\trg{x}~}{\compap{\src{\Delta,\Gamma}\vdash \src{e} : \UNS}}{\trg{pair}}{
							\compap{ \src{C,\Delta,\Gamma;(x:\varphi)} \vdash\src{s} }
							\\
							&
						}{\wrong}
					\end{aligned}\right.
				}
				&
				\text{ if }\src{\varphi}=\src{\UNS\times\UNS}
				\\
				\trg{
					\left|\begin{aligned}
						&
						\destruct{\trg{x}~}{\compap{\src{\Delta,\Gamma}\vdash \src{e} : \UNS}}{\trg{pair}}{
							\trg{!\projone{x} \with{\projtwo{x}} ;}
							\compap{ \src{C,\Delta,\Gamma;(x:\varphi)} \vdash\src{s} }
							\\
							&
						}{\wrong}
					\end{aligned}\right.
				}
				&
				\text{ if }\src{\varphi}=\src{\Refs{\UNS}}
			\end{cases}
		\end{aligned}
  \end{gather*}
}{compap2}{\compap{\cdot}, compilation of statements from \LA to \LC (excerpts, omitted elements are analogous to their counterparts in \Cref{fig:compup}).}

\myfig{
		\begin{center}
	\typerule{Initial-heap}{
		\src{\Delta}\vdash\trg{H}
		&
		\src{\Delta},\trg{H}\vdash_{\beta}\trg{v}\src{:\tau}
		&
		\src{\ell}\relatebeta\trg{\pair{n,k}}
	}{
		\src{\Delta,\ell:\tau}\vdash_{\beta}\trg{H;n\mapsto v:k}
	}{ini-heap}
	\typerule{Initial-value}{
			(\src{\tau}\equiv\Bools \wedge \trg{v}\equiv\trg{0}) 
			&
			\vee
			&
			(\src{\tau}\equiv\Nats \wedge \trg{v}\equiv\trg{0}) 
			&
			\vee
			\\
			(\src{\tau}\equiv\src{\Refs{\tau}} \wedge \trg{v}\equiv\trg{n'} \wedge \trg{n'\mapsto v':k'}\in\trg{H} \wedge \src{\ell'}\relatebeta\trg{\pair{n',k'}} \wedge \src{\ell:\tau}\in\src{\Delta}, 
			\src{\Delta},\trg{H}\vdash\trg{v'}\src{:\tau}
			) 
			&
			\vee
			\\
			(\src{\tau}\equiv\src{\tau_1\times\tau_2} \wedge \trg{v}\equiv\trg{\pair{v_1,v_2}} \wedge 
			\src{\Delta},\trg{H}\vdash\trg{v_1}\src{:\tau_1}\wedge
			\src{\Delta},\trg{H}\vdash\trg{v_2}\src{:\tau_2}
			)
	}{
		\src{\Delta},\trg{H}\vdash_{\beta}\trg{v}\src{:\tau}
	}{ini-val}
\end{center}
}{details-1}{Initialisation of the safety-relevant target heap based on the source typing environment.}

\paragraph{New Monitor Relation}\label{sec:mon-rel-par}
As monitors have changed, we also need a new monitor relation $\src{M}\relate\trg{M}$.
Informally, a source and a target monitor are related if the target monitor can always step whenever the target heap satisfies the types specified in the source monitor's $\Delta$ (up to renaming by the partial bijection $\beta_0$).

We write $\vdash\src{H}:\src{\Delta}$ to mean that for each location
$\src{\ell} \in \src{\Delta}$, we have that $\src{\Delta};\srce\vdash \src{H}(\src{\ell}):
\src{\Delta}(\src{\ell})$: i.e., the contents of \src{H} are well-typed according to \src{\Delta}.
Given a partial bijection $\beta$ from
source to target locations, we say that a target monitor $\trg{M} =
\trg{(\trgb{\set{\sigma}},\monred,\trgb{\sigma}_0,H_0,\trgb{\sigma}_c)}$ is good, written
$\vdash \trg{M}: \beta, \src{\Delta}$, if for all $\trgb{\sigma} \in
\trgb{\set{\sigma}}$ and all $\src{H} \relatebeta \trg{H}$ such that
$\vdash \src{H}: \src{\Delta}$, there is a $\trgb{\sigma'}$ such that
$(\trgb{\sigma}, \trg{H}, \trgb{\sigma'}) \in \trg{\monred}$.
For a fixed partial bijection $\beta_0$ between the domains of
$\src{\Delta}$ and $\trg{H_0}$, we say that the source monitor
$\src{M}$ and the target monitor $\trg{M}$ are related, written
$\src{M} \relate \trg{M}$, if $\vdash \trg{M}: \beta_0, \src{\Delta}$
for the $\src{\Delta}$ in $\src{M}$. With this setup, we define
\rscomp as in \Cref{sec:rsc}. Our main theorem is that \compap{\cdot}
attains \rscomp under this definition.

\begin{theorem}[Compiler \compap{\cdot} attains \rscomp]\label{thm:comp-ap-rsc}
  $\vdash\compap{\cdot} : \rscomp$
\end{theorem}

To prove that \compap{\cdot} attains \rscomp we do not rely on a backtranslation and we show \rscomp as of \Cref{def:rsc} instead of the property-free version. 
Here, we know statically which locations can be monitor-sensitive: they must all be trusted, i.e., must have a type $\src{\tau}$ satisfying $\src{\tau}\nvdash\src{\circ}$. 
Using this, we set up a simple cross-language relation (later indicated as $\relatetbeta$) and show it to be an invariant on runs of source and compiled target states. 
Like previous relations, this relation is also indexed by a partial bijection $\beta$.
The relation captures the following:
\begin{itemize}
\item Heaps (both \src{source} and \trg{target}) can be partitioned into two parts, a \emph{trusted} part and an \emph{untrusted} part;
\item The trusted \src{source\ heap} contains only locations whose type is trusted ($\src{\tau}\nvdash\src{\circ}$);
\item The \trg{trusted} \trg{target\ heap} contains only locations related to \src{trusted} \src{source} \src{locations} and these point to related values; more importantly, every \trg{trusted} \trg{target} \trg{location} is protected by a capability;
\item In the \trg{target}, any capability protecting a trusted location does not occur in attacker code, nor is it stored in an untrusted heap location.
\end{itemize}

We need to prove that this relation is preserved by reductions both in compiled and in attacker code.
The former follows from the proof of source robust safety (\Cref{thm:src-ty-impl-safe}).

The latter is formalised in \Cref{thm:att-act-pres-relatet} below and it is simple to prove.
Since all trusted locations are protected with capabilities, attackers have no access to trusted locations, and capabilities are unforgeable and unguessable (by the semantics of \LC).
At this point, knowing that the monitors at hand are related, and that source traces are always accepted by the considered source monitors, we can conclude that target traces are always accepted by the considered target monitors too.
Note that this kind of an argument requires all compilable source programs to be robustly safe and is, therefore, impossible for our first compiler \compup{\cdot}. 
Overall, avoiding the backtranslation results in a proof much simpler than that of \Cref{sec:rsc-instance}.

In order to state the lemma discussed above, we rely on notation $\src{C}\vdashatts\src{\Pi}\xtos{\lambda}\src{\Pi'}$, which states that the reduction taking place occurs in attacker code.
This can be easily defined by observing the stack of functions (explained in \Cref{ex:fstack}) and controlling whether the executing function is defined by the attacker or not.
Note that unlike before, we do not need to grow the partial bijection here because it is used to track locations as used by the compiled code and not by the attacker.
\begin{lemma}[Attacker actions preserve relatedness]\label{thm:att-act-pres-relatet}
  \begin{align*}
    \text{ if }
    &\
    \src{C, H\triangleright \Pi} \xtos{\lambda} \src{C, H'\triangleright \Pi'} 
    \text{ and }
    \trg{C,H\triangleright \trgb{\Pi}} \xtot{\trgb{\lambda}} \trg{C,H'\triangleright \trgb{\Pi'}}
    \text{ and }
    \src{C, H \triangleright \Pi} \relatetbeta \trg{C,H\triangleright \trgb{\Pi}}
    \\
    \text{ and }
    &
    \src{C}\vdashatts\src{\Pi}\xtos{\lambda}\src{\Pi'}
    \text{ and }
    \trg{C}\vdashattt\trgb{\Pi}\xtot{\trgb{\lambda}}\trgb{\Pi'}
    \text{ then }
    \src{C, H'\triangleright \Pi'} \relatetbeta \trg{C, H'\triangleright \trgb{\Pi'}}
  \end{align*}
\end{lemma}

\subsection{Non-Atomic Allocation of Capabilities}\label{sec:nonatom-alloc}
The compiler of \Cref{fig:compap} uses a new target language construct, \trg{newhide}, that simultaneously allocates a new location and protects it with a capability. 
This atomic construct is what certain capability machines provide, and it simplifies our proof of security in the concurrent setting at hand.
If allocation and protection were not atomic, then a concurrent adversary thread could protect a location that had just been allocated and acquire the capability to it before the allocating thread could do so. 
This would break the cross-language relation we use in our proof. 
However, note that this is not really an attack since it does not give the adversary any additional power to violate the safety property enforced by the monitor.
The reason is that the thread allocating the location gets stuck when it tries to acquire the capability itself (since the adversary obtained the capability), and, by the design of our compiler, it will not try to use the location before obtaining the capability.
Consequently, it is possible to do away with the \trg{newhide} construct for compiling allocation. \Cref{fig:nonat} shows how compilation can also be done using the $\lethide{\trg{x_k}~}{\trg{x}} \trg{\ldots}$ construct of \Cref{sec:rsc-instance}.
\myfig{
 \begin{align*}
 \compap{
   \typerulenolabel{T\LA-new-nonatomic}{
     \src{\Delta,\Gamma}\vdash \src{e} : \src{\tau} 
     &
     \src{C,\Delta,\Gamma;x:\Refs{\tau}}\vdash \src{s} %
   }{
     \src{C,\Delta,\Gamma}\vdash \src{\letnewty{x}{e}{\tau}{s}} %
   }{compap-new-nonat} 
 }
 &=
 \begin{cases}
   \trg{
     \begin{aligned}
       &
       \letnewt{\trg{x_t}~}{\trg{0}}{
       \\
       &\
         \lethide{\trg{x_k}~}{\trg{x_t}}{
         \\
         &\ \ 
           \letint{\trg{x_c}~}{~ \compap{\src{\Delta,\Gamma}\vdash \src{e} : \src{\tau}} 
           }{
           \\
           &\ \ \ 
             \trg{x_t := x_c \with{x_k};}
           \\
           &\ \ \ 
            \letint{\trg{x}~}{~\trg{\pair{x_t,x_k}}}{
            \\
            &\ \ \ \
             \compap{\src{C,\Delta,\Gamma;x:\Refs{\tau}}\vdash \src{s}}
             }
           }
         }
       }
     \end{aligned}
   }
   &
   \text{if }\src{\tau}\neq\UNS %
 \end{cases}
 \end{align*}
}{nonat}{Non-atomic implementation of capability allocation, only the interesting case for $\src{\tau}\neq\UNS$ is reported, the other is analogous to the one in \protect\Cref{fig:compap}.}

The price to pay is a slightly more involved cross-language relation, which must relate states where either (i) the heaps are partitioned as before, or (ii) the target execution is stuck trying to acquire a capability for a location that should be trusted.

We refer the interested reader to the accompanying technical report for details of the new relation and the proof that this alternative compiler also attains \rscomp~\cite{rsc-techrep}.

\section{\rscomp Relying on Target Memory Isolation}\label{sec:rsc-iso}

Both compilers presented so far used a capability-based target
language. To avoid giving the false impression that \rscomp is only
useful for this kind of a target, we show here how to attain \rscomp
when the protection mechanism in the target is completely
different. We consider a new target language, \LI, which does not have
capabilities, but instead offers \emph{coarse-grained} memory isolation based
on \emph{enclaves} (\Cref{sec:target-iso}). This mechanism is supported (in hardware) in
mainstream x86-64 and ARM CPUs (Intel calls this SGX~\cite{intel}; ARM
calls it TrustZone~\cite{arm}). It is also straightforward to
implement purely in software using any VM-based,
process-based, or in-process isolation technique. We present a
compiler \compai{\cdot} from our last source language \LA to \LI and
prove that it attains \rscomp (\Cref{sec:comp-ai}).

\subsection{\LI, a Target Language with Coarse-Grained Memory Isolation}
\label{sec:target-iso}

Language \LI replaces \LC's capabilities with a simple security
abstraction called an enclave. An enclave is a collection of code and
memory locations, with the properties that: (a) only code within the
enclave can access the memory locations of the enclave, and (b) code
from outside can transfer control only to designated entry points in
the enclave's code. For simplicity, \LI supports only one
enclave. Generalizing this to many enclaves is straightforward, but not necessary for our purposes.

To model the enclave, \LI components carry additional information
\oth{\OB{E}}, the list of functions that reside in the enclave. 
Only functions that are listed in \oth{\OB{E}}
can create%
, read %
and write
locations in the enclave, using statements and expressions. Locations in \LI are
\emph{integers} (not natural numbers). By convention, non-negative
locations are outside the enclave (accessible from any function),
while negative locations are inside the enclave (accessible only from
functions in \oth{\OB{E}}).
\myfig{
	\begin{align*}
		\mi{Components}~\oth{C} \bnfdef&\ \oth{H_0 ; \OB{F} ; \OB{I} ; \OB{E}}
		\qquad\qquad
		\mi{Enclave\ funcs.}~\oth{E} \bnfdef\ \oth{f}
		\qquad\qquad
		\mi{Heaps}~\oth{H} \bnfdef\ \othe \mid \oth{H ; n\mapsto v}
		\\
		\mi{Values}~\oth{v} \bnfdef&\ \oth{n}\in\mb{Z} \mid \oth{\pair{v,v}} \mid \oth{k} %
		\qquad\qquad
		\mi{Expressions}~\oth{e} \bnfdef\ \cdots \mid \oth{!e}
		\\
		\mi{Statements}~\oth{s} \bnfdef&\ \cdots \mid \oth{x := e} \mid \oth{\letnew{x}{e}{s}} \mid \oth{\letiso{x}{e}{s}}			
	\end{align*}	
}{li-syn}{Syntax of \LI. Elided and omitted elements are the same as in \LP (\Cref{fig:lp-syn}) or \LC (\Cref{fig:lc}).}

The semantics (\Cref{fig:li-sem}) are almost those of \LC, but
the expression semantics change to $\oth{C;H;f\triangleright e \redtoo v}$, recording which function \oth{f} is currently
executing. The operational rule for any memory operation checks that
either the access is to a location outside the enclave or that
$\oth{f} \in \oth{\OB{E}}$ (formalised by
$\oth{C}\vdash\oth{f}:\oth{prog}$).
Allocating a protected location (\Cref{tr:eo-hi}) is done with respect to location \oth{n}, which is the smallest allocated location.
This way, the new location \oth{n-1} will be for sure in the domain of the enclave.
Monitors of \LI are the same as those of \LC.

\myfig{
	\begin{center}
	\typerule{E\LI-deref}{
	    \oth{n\mapsto v}\in\oth{H}
	    &
	    (\oth{n}\geq \oth{0}) \text{ or } (\oth{n}< \oth{0} \text{ and } \oth{C}\vdash\oth{f}:\oth{prog})
	  }{
	    \oth{C; H ; f \triangleright !n } \redtoo \oth{v}
	  }{eo-de-t}
  \end{center}
	\hrule
  \begin{center}
	  \typerule{E\LI-isolate}{
	    \oth{H} = \oth{n\mapsto \_;H_1}
	    &
	    \oth{C; H ; f \triangleright e\redtoo v}
	    &
	    \oth{C}\vdash\oth{f}:\oth{prog}
	  }{
	    \oth{C, H \triangleright \proc{\letiso{x}{e}{s}}{\OB{f};f}} \xtoo{\epsilon} \oth{C, n-1\mapsto v;H \triangleright \proc{s\subo{n-1}{x}}{\OB{f};f}}
	  }{eo-hi}
	  \typerule{E\LI-assign}{
	    \oth{C; H ; f \triangleright e\redtoo v}
	    &
	    \oth{H} = \oth{H_1;n\mapsto \_;H_2}
	    &
	    \oth{H'} = \oth{H_1;n\mapsto v;H_2}
	    \\
	    (\oth{n}\geq\oth{0}) \text{ or } (\oth{n}<\oth{0} \text{ and } \oth{C}\vdash\oth{f}:\oth{prog})
	  }{
	    \oth{C, H\triangleright \proc{n:=e}{\OB{f};f}} \xtoo{\epsilon} \oth{C, H'\triangleright \proc{\skipo}{\OB{f};f}}
	  }{eo-ac-t}\hfill
	\end{center}
}{li-sem}{Semantics of \LI. Omitted rules are as in \LP (\Cref{fig:lp-sem}) and \LA (\Cref{fig:la-sem}).}

\subsection{Compiler from \LA to \LI}\label{sec:comp-ai}
\myfig{
  \begin{align*}
    \compai{
      \typerulenolabel{T\LA-component}{
        \src{C}\equiv\src{\Delta ; \OB{F} ; \OB{I}} 
        &
        \src{C}\vdash\src{\OB{F}}:\UNS
        &
        \src{\Delta}\vdash\src{ok}
        \\
        \fun{names}{\src{\OB{F}}}\cap\fun{names}{\src{\OB{I}}}=\emptyset
      }{
        \vdash \src{C}:\UNS
      }{compai-co}
    } &= \oth{%
    \oth{H_0} ;
    \compai{\OB{F}}; \compai{\OB{I}}} ; 
    \dom{\OB{\src{F}}}
    \qquad
    \text{if }\src{\Delta}\vdash_{\varphi}\oth{H_0}
  \end{align*}
  \begin{align*}
    \compai{
      \typerulenolabel{T\LA-new}{
        \src{C,\Delta,\Gamma}\vdash \src{e} : \src{\tau} 
        \\
        \src{C,\Delta,\Gamma;x:\Refs{\tau}}\vdash \src{s} %
      }{
        \begin{multlined}
          \src{C,\Delta,\Gamma}\vdash 
          \\
          \src{\letnewty{x}{e}{\tau}{s}} %
        \end{multlined}
      }{compai-new} 
    }
    =&\
    \begin{cases}
      \oth{
      \begin{aligned}
        &
        \letnewo{\oth{x}~}{\compai{\src{\Delta,\Gamma}\vdash \src{e} : \src{\tau}} 
        \\
        &\
        }{ 
          \compai{\src{C,\Delta,\Gamma;x:\Refs{\tau}}\vdash \src{s}} 
        }
      \end{aligned}
      }
      &
      \text{if }\src{\tau}=\UNS %
      \\
      \vspace{-1em}
      \\
      \oth{
      \begin{aligned}
        &
        \letiso{\oth{x}~}{ \compai{\src{\Delta,\Gamma}\vdash \src{e} : \src{\tau}} 
        \\
        &\
        }{ 
          \compai{\src{C,\Delta,\Gamma;x:\Refs{\tau}}\vdash \src{s}} 
        }
      \end{aligned}
      }
      &
      \text{else}
    \end{cases}
  \end{align*}
}{compai}{\compai{\cdot}, compilation of components and statements from \LA to \LI (excerpts, missing bits are adaptations of those bits from \Cref{fig:compup,fig:compap}).}
The high-level structure of the compiler \compai{\cdot} is similar to that of \compap{\cdot} from \Cref{sec:compap} (\Cref{fig:compai}). 
Compiler \compai{\cdot} ensures that all the (and only the) functions of the (trusted) component we write are part of the enclave, i.e., constitute \oth{\OB{E}}.  
Additionally, the compiler populates the safety-relevant heap \oth{H_0} based on the information in \src{\Delta} according to bijection $\varphi$.
This is captured by the judgement $\src{\Delta}\vdash_\varphi\oth{H_0}$, whose details are in \Cref{fig:details} and they follow the same intuition of \Cref{fig:details-1}.
Compiler, \compai{\cdot} also ensures that trusted locations are stored in the enclave.  
As before, the compiler relies on typing information for this.  
Locations whose types are shareable (subtypes of \UNS) are placed outside the enclave while those that are trusted (not subtypes of \UNS) are placed inside (second rule in \Cref{fig:compai}).
\myfig{
		\begin{center}
	\typerule{Initial-heap}{
		\src{\Delta}\vdash\oth{H}
		&
		\src{\Delta},\oth{H}\vdash\oth{v}\src{:\tau}
		&
		\src{\ell}\relatephi\oth{n}
	}{
		\src{\Delta,\ell:\tau}\vdash_{\varphi}\oth{H;n\mapsto v}
	}{ini-heap-o}
	\typerule{Initial-value}{
			(\src{\tau}\equiv\Bools \wedge \oth{v}\equiv\oth{0}) 
			&
			\vee
			&
			(\src{\tau}\equiv\Nats \wedge \oth{v}\equiv\oth{0}) 
			&
			\vee
			\\
			(\src{\tau}\equiv\src{\Refs{\tau}} \wedge \oth{v}\equiv\oth{n'} \wedge \oth{n'\mapsto v'}\in\oth{H} \wedge \src{\ell'}\relatephi\oth{n'} \wedge \src{\ell:\tau}\in\src{\Delta}, 
			\src{\Delta},\oth{H}\vdash\oth{v'}\src{:\tau}
			) 
			&
			\vee
			\\
			(\src{\tau}\equiv\src{\tau_1\times\tau_2} \wedge \oth{v}\equiv\oth{\pair{v_1,v_2}} \wedge 
			\src{\Delta},\oth{H}\vdash\oth{v_1}\src{:\tau_1}\wedge
			\src{\Delta},\oth{H}\vdash\oth{v_2}\src{:\tau_2}
			)
	}{
		\src{\Delta},\oth{H}\vdash_\varphi\oth{v}\src{:\tau}
	}{ini-val-o}
\end{center}
}{details}{Initialisation of the safety-relevant target heap based on the source typing environment.}

For this compiler we need a different partial bijection that drops capabilities and considers integers instead of natural numbers.
We indicate such a bijection with $\varphi$. Its type is $\src{\ell}\times\oth{n}$, but it has the same properties as $\beta$ in \Cref{sec:compup-proof}.

The cross-language relation $\relate$ is mostly unchanged. The only change is for relating locations, as defined below:
\begin{itemize}
	\item $\src{\ell}\relatephi\oth{n}$ if $(\src{\ell},\oth{n})\in\varphi$
\end{itemize}

Our main theorem is that \compai{\cdot} attains \rscomp.
\begin{theorem}[Compiler \compai{\cdot} attains \rscomp]\label{thm:compai-rsc}
  $\vdash\compai{\cdot}:\rscomp$  
\end{theorem}
The intuition behind the proof is
simple: all trusted locations (including safety-relevant locations)
are in the enclave and adversarial code cannot tamper with them.  The
proof follows the idea of the proof of \Cref{thm:comp-ap-rsc}: we build a
cross-language relation, which we show to be an invariant on
executions of source and corresponding compiled programs.  The only
change is that every location in the \oth{trusted} \oth{target}
\oth{heap} is isolated in the enclave.

\section{Fully Abstract Compilation}\label{sec:facomp}
Our next goal is to compare \rscomp to fully abstract compilation (or \facomp) at an intuitive
level. We first define \facomp
(\Cref{sec:facomp-def}).  
Then, we present a series of examples of how \facomp
may result in inefficiencies in compiled code (\Cref{sec:fac-short}).
Relying on these examples, we present what is needed to write a fully abstract compiler from \LU to \LP, the languages of our first compiler (\Cref{sec:facomp-instance}). 
We use this compiler to compare \rscomp and \facomp concretely, showing that, at least on this example, \rscomp permits more efficient code and affords simpler proofs than {\facomp} (\Cref{sec:fac-inst-proof}).

\paragraph{Remark}
This does not imply that one should always prefer \rscomp to \facomp blindly. 
In some cases, one may want to establish full abstraction for reasons other than security, so there \facomp is preferable. 
Also,  when the target language is typed~\cite{ahmedCPS,Ahmed:2008:TCC:1411203.1411227,nonintfree,max-embed} or has abstractions similar to those of the source, full abstraction may have no downsides (in terms of efficiency of compiled code and simplicity of proofs) relative to \rscomp.
However, in many settings, including those we consider, target languages are not typed, and often differ significantly from the source in their abstractions. 
In such cases, \rscomp is a worthy alternative.

\subsection{Formalising Fully Abstract Compilation}\label{sec:facomp-def}
As stated in \Cref{sec:intro}, \facomp requires the preservation and reflection of observational equivalence, and most existing work instantiates observational equivalence with contextual equivalence ($\ceq$).
Contextual equivalence and \facomp are defined below.
Informally, two components \P{_1} and \P{_2} are contextually equivalent if no context \ctxc{} interacting with them can tell them apart, i.e., they are \emph{indistinguishable}.
Contextual equivalence can encode security properties such as confidentiality, integrity, invariant maintenance and non-interference~\cite{scoo-j,schp,KULeuven-358154,Abadi:2012}.
We do not explain this well-known observation here, but refer the interested reader to the survey of Patrignani \etal~\cite{scsurvey}.
Informally, a compiler \compgen{\cdot} is fully abstract if it translates (only) contextually-equivalent source components into contextually-equivalent target ones.
\begin{definition}[Contextual equivalence and fully abstract compilation]\label{def:ceq}
\begin{align*}
  \P{_1}\ceq\P{_2} \isdef&\ \forall\ctxc{}\ldotp \ctxc{}\hole{\P{_1}}\divrc \iff \ctxc{}\hole{\P{_2}}\divrc,
  \text{ where $\divrc$ means execution divergence}
  \\
    \vdash\compgen{\cdot}:\facomp \isdef&\
        \forall\Cs{_1}, \Cs{_2}\ldotp \Cs{_1}\ceqs\Cs{_2}\iff\compgen{\Cs{_1}}\ceqt\compgen{\Cs{_2}}
  \end{align*}
\end{definition}

The security-relevant direction of \facomp is $\Rightarrow$~\cite{popl-backtrans}.
This direction is security-relevant because the proof thesis concerns target contextual equivalence ($\ceqt$).
Unfolding the definition of $\ceqt$ on the right of the implication yields a universal quantification over all possible target contexts \ctxt{}, which captures malicious attackers.
In fact, there may be target contexts \ctxt{} that can interact with compiled code in ways that are impossible in the source language.
Compilers that attain \facomp with untyped target languages often insert checks in compiled code that detect such interactions and respond to them securely~\cite{schp}, often by halting the execution~\cite{KULeuven-358154,scoo-j,popl-backtrans,catalin,Jagadeesan:2011:LMV:2056311.2056556,Abadi:2012,adriaanaplas,scsurvey}. 
These checks are often inefficient, but must be performed even if the interactions are not security-relevant. We now present examples of this.

\subsection{\facomp and Inefficient Compiled Code}\label{sec:fac-short}
We illustrate various ways in which \facomp forces inefficiencies in compiled code via a running example.
Consider a password manager written in an object-oriented language that is compiled to an assembly-like language.
We elide most code details and focus only on the relevant aspects.

\begin{lstlisting}[mathescape]
private db: Database;

public testPwd( user: Char[8], pwd: BitString): Bool{
  if( db.contains( user )){ return db.get( user ).getPassword() == pwd; }
}
...
private class Database{ ... }
\end{lstlisting}

The source program exports the function \lst{testPwd} to check whether a \lst{user}'s stored password matches a given password \lst{pwd}. The stored password is in a local database, which is represented by a piece of \emph{local state} in the variable \lst{db}. The details of \lst{db} are not important here, but the database is marked private, so it is not directly accessible to the context of this program in the source language.

\begin{example}[Extensive checks]\label{ex:checks}
A fully-abstract compiler for the program above must generate code that checks that the arguments passed to \lst{testPwd} by the context are of the right type~\cite{KULeuven-358154,scoo-j,fstar2js,catalin,mfac}. 
The code expects an array of characters of length 8. A parameter of a different type (e.g., an array of objects) cannot be passed in the source, so it must also be prevented in the target. Since the target is untyped, code must be inserted to check the argument.
Specifically, a fully abstract compiler will generate code similar to the following (we assume that arrays are passed as pointers into the heap).
\lstset{language=Asm,numbersep=5pt,frame=single}
\begin{lstlisting}[mathescape]
label testpwd
  for i = 0; i<8; i++  // 8 is the legth of the $\lst{user}$ field in the previous snippet
    load the memory word stored at address r0+i into r1
    test that r1 is a valid char encoding
  ...
\end{lstlisting} 
\lstset{language=Java,numbersep=5pt,frame=single}
Basically, this code dynamically checks that the first argument is a character array of length 8 because a type mismatch could lead to a violation of \facomp.
Such a check can be very inefficient when the length is very long.
\end{example}
The problem here is that \facomp forces these checks on all arguments, even those that have no security relevance.
In contrast, \rscomp does not need these checks. 
Indeed, none of our earlier compilers (\compup{\cdot}, \compap{\cdot} and \compai{\cdot}) insert them. 
Note that any robustly safe source program will already have programmer-inserted checks for all parameters that are relevant to the safety property of interest, and these checks will be compiled to the target. 
For other parameters, the checks are irrelevant, both in the source and the target, so there is no need to insert them.

\begin{example}[Component size in memory]\label{ex:size}
Let us now consider two different ways to implement the \lst{Database} class: as a \lst{List} and as a \lst{RedBlackTree}.
As the class is \lst{private}, its internal behaviour and representation of the database is invisible to the outside.
Let  \Cs{_{list}} be the program with the \lst{List} implementation and \Cs{_{tree}} be the program with the \lst{RedBlackTree} implementation; in the source language, these are equivalent.

However, a subtlety arises when considering the assembly-level, compiled counterparts of \Cs{_{list}} and \Cs{_{tree}}: the \emph{code} of a \lst{RedBlackTree} implementation consumes more memory than the code of a \lst{List} implementation.
Thus, a target-level context can distinguish \Cs{_{list}} from \Cs{_{tree}} by just inspecting the sizes of the code segments.
So, in order for the compiler to be fully abstract, it must produce code of a fixed size~\cite{KULeuven-358154,scoo-j}. This wastes memory and makes it impossible to compile some components.
An alternative would be to spread the components in an overly-large memory at random places i.e., use address-space layout randomization or ASLR, so that detecting different code sizes has a negligible chance of success~\cite{Abadi:2012,Jagadeesan:2011:LMV:2056311.2056556}. However, ASLR is now known to be broken~\cite{break-aslr2,break-aslr1}.
\end{example}
Again, we see that \facomp introduces an inefficiency in compiled code (pointless code memory consumption) even though this has no security implication here.
In contrast, \rscomp does not require this unless the safety property(ies) of interest care about the size of the code (which is very unlikely in a security context, since security by code obscurity is a strongly discouraged practice). 

\begin{example}[Wrappers for heap resources]\label{ex:wrap}
Assume that the \lst{Database} class is implemented as a \lst{List}.
Shown below are two implementations of the \lst{newList} method inside \lst{List} which we call \Cs{_{one}} and \Cs{_{two}}.
The only difference between \Cs{_{one}} and \Cs{_{two}} is that \Cs{_{two}} allocates two lists internally; one of these (\lst{shadow}) is used for internal purposes only.

\noindent\begin{minipage}{.45\textwidth}
\begin{lstlisting}[mathescape]
public newList(): List{

  ell = new List();
  return ell;
} 
\end{lstlisting} 
\end{minipage}\hfill
\begin{minipage}{.45\textwidth}
\begin{lstlisting}[mathescape]
public newList(): List{
  shadow = new List();
  ell = new List();
  return ell;
}
\end{lstlisting} 
\end{minipage}

Again, \Cs{_{one}} and \Cs{_{two}} are equivalent in a source language that does not allow pointer comparison. To attain {\facomp} when the target allows pointer comparisons, the pointers returned by \lst{newList} in the two implementations must be the same, but this is very difficult to ensure since the second implementation does more allocations.
A simple solution to this problem is to wrap \lst{ell} in a proxy object and return the proxy~\cite{scoo-j,mfac,KULeuven-358154,Morris:1973:PPL:361932.361937}.
Compiled code needs to maintain a lookup table mapping the proxy to the original object. 
Proxies must have allocation-independent addresses.
Proxies work but they are inefficient due to the need to look up the table on every object access.

Another way to attain \facomp is to weaken the source language, introducing an operation to distinguish object identities in the source~\cite{gcFA}.
However, this is a widely discouraged practice, as it changes the source language from what it really is and the implication of such a change may be difficult to fathom for programmers and verifiers alike.
\end{example}
In this example, \facomp forces all privately allocated locations to be wrapped in proxies, but \rscomp does not require this. 
Our target languages \LP, \LC and \LI support address comparison (addresses are natural numbers or integers in their heaps) but \compup{\cdot} and \compap{\cdot} just use capabilities to attain security efficiently, while \compai{\cdot} just relies on enclaves.
On the other hand, for attaining \facomp, capabilities or enclaves would be insufficient since they do not hide addresses; proxies would still be required (this point is concretely demonstrated in \Cref{sec:facomp-instance}).

\begin{example}[Strict termination vs divergence]\label{ex:termina}
Consider a  source language that is strictly terminating while a target language that is not.
Below is an extension of the password manager to allow database encryption via an externally-defined function.
As the database is not directly accessible from external code, the two implementations below \Cs{_{enc}} (which does the encryption) and \Cs{_{skip}} which skips the encryption are equivalent in the source.

\noindent\begin{minipage}{.45\textwidth}
\begin{lstlisting}[mathescape]
public encryptDB( func : Database -> Bitstring) : void {
  func( this.db );
  return;
}
\end{lstlisting} 
\end{minipage}\hfill
\begin{minipage}{.45\textwidth}
\begin{lstlisting}[mathescape]
public encryptDB( func : Database -> Bitstring) : void {

  return;
}
\end{lstlisting} 
\end{minipage}

If we compile \Cs{_{enc}} and \Cs{_{skip}} to an assembly language, the compiled counterparts \emph{cannot} be equivalent, since the target-level context can detect which function is compiled by passing a \lst{func} that diverges. Calling the compilation of \Cs{_{enc}} with such a \lst{func} will cause divergence, while calling the compilation of \Cs{_{skip}} will immediately return.
\end{example}
This case presents a situation where \facomp is outright \emph{impossible}. 
The only way to get \facomp is to make the source language artificially non-terminating. (See the work of Devriese \etal~\citet{lambda-seal} for more details of this particular problem.)
On the other hand, \rscomp can be easily attained even in such settings since it is completely independent of termination in the languages (unless the safety properties of interest are termination-sensitive, which is usually not the case). 
For the specific examples we have considered, even if our source languages \LU and \LA were restricted to terminating programs only, the same compilers and the same proofs of \rscomp would still work.

\paragraph{Remark}
It is worth noting that many of the inefficiencies above might be resolved by just replacing contextual equivalence with a different equivalence in the statement of \facomp.
However, it is not known how to do this generally for arbitrary sources of inefficiency and, further, it is unclear what the security consequences of such instantiations of \facomp would be. 
On the other hand, {\rscomp} is \emph{uniform} and it does address all these inefficiencies.

An issue that can normally not be addressed just by tweaking equivalences is side-channel leaks, as they are, by definition, not expressible in the language.
Neither \facomp nor \rscomp deals with side channels, but recent results describe how to account for side channels in security-preserving compilers~\cite{BartheGL18}.

\subsection{Towards a Fully Abstract Compiler from \LU to \LP}\label{sec:facomp-instance}

In this section, we describe what it would take to build a fully abstract compiler from \LU to \LP. Along the way, we note how this compiler would be less efficient than the \rscomp compiler we described earlier. In fact, to get a fully abstract compiler, we need to adjust the languages. We describe these language changes first.

\subsubsection{Language Extensions to \LU and \LP}\label{sec:fac-langs}
This section lists the language extensions required for a fully-abstract compiler from \LU to \LP. It is not possible to motivate all the language changes before explaining the details of the compiler, so some of the justification is postponed to \Cref{sec:fac-comp}.

A first concern for full abstraction is that a target context can always determine the memory consumption of two compiled components, analogously to \Cref{ex:size}. 
To ensure that this does not break full abstraction, we add a source expression \src{size} that returns the number of locations \src{\ell} allocated in the heap.

In the target language \LP, we need to know whether an expression is a pair, whether it is a location, and we need to be able to compare two capabilities. 
Accordingly, we add the operations \trg{isloc(e)}, \trg{ispair(e)} and \trg{eqcap(e,e)}, respectively.

Finally, compiled code needs private functions for its runtime checks that must not be visible to the context.
\LP does not have this functionality: all functions defined by a component can be called by the context.
Accordingly, we modify \LP so that all functions \trg{\OB{F}} defined in a component are private to it by default.
Each component explicitly includes the list of functions it exports; only these functions can be called by the context.

\subsubsection{The \compfac{\src{\cdot}} Compiler}\label{sec:fac-comp}
The fully abstract compiler \compfac{\src{\cdot}} is similar to the \rscomp attaining compiler \compup{\cdot}, but with critical differences.
We know that fully abstract compilation preserves all source abstractions in the target language. Here, the only abstraction that distinguishes \LP from \LU is that locations are abstract in \LU, but concrete natural numbers in \LP.
Thus, locations allocated by compiled code must not be passed directly to the context as this would reveal the allocation order (as seen in \Cref{ex:wrap}).
Instead of passing the location \trg{\pair{n,k}} to the context, the compiler arranges for an opaque handle \trg{\pair{n',k_{com}}} (that cannot be used to access any location directly) to be passed.
Such an opaque handle is often  called a \emph{mask} or \emph{seal} in the literature and this technique is often called dynamic sealing~\cite{sumii}.

To ensure that masking is done properly, \compfac{\src{\cdot}} inserts code at entry points and at exit points to compiled code (i.e., at function calls, before returning, before and after callbacks), \emph{wrapping} the compiled code in a way that enforces masking.
This notion of wrapping is standard in literature on fully abstract compilation~\cite{fstar2js,mfac}.
The wrapper keeps a list \trg{\OB{L}} of component-allocated locations that are shared with the context in order to know their masks.
When a component-allocated location is shared, it is added to the list \trg{\OB{L}}. The mask of a location is its index in this list. If the same location is shared again it is not added again but its previous index is used.
So if \trg{\pair{n,k}} is the 4th element of \trg{\OB{L}}, its mask is \trg{\pair{4,k_{com}}}.
To implement lookup in \trg{\OB{L}} we must compare capabilities too, for we  rely on the newly added operation \trg{eqcap}.
To ensure capabilities do not leak to the context, the second field of the pair is a constant capability \trg{k_{com}}, which protects a dummy location the compiled code does not actually use.
Technically speaking, this is exactly how existing fully abstract compilers operate (e.g., as in the work of Patrignani \etal~\cite{scoo-j}).

As should be clear, this kind of masking is very inefficient at runtime. 
However, even this masking is not sufficient for full abstraction. Next, we explain additional things the compiler must do.

\paragraph{Determining when a Location is Passed to the Context.}
A component-allocated location can be passed to the context not just as a function argument but on the heap.
So before passing control to the context the compiled code needs to scan the whole heap where a location can be passed and mask any component-allocated locations it finds.
Dually, when receiving control the compiled code must scan the heap to unmask all masked locations.
The problem now is determining what parts of the heap to scan and how.
Specifically, the compiled code needs to keep track of all the locations (and related capabilities) that are shared, i.e., (i) passed from the context to the component and (ii) passed from the component to the context.
These are the locations through which possible communication of locations can happen.
Compiled code keeps track of these shared locations in a list \trg{\OB{S}}.
Intuitively, on the first function call from the context to the compiled component, assuming the parameter is a location, the compiled code will register that location and all other locations reachable from it in \trg{\OB{S}}.
On subsequent ? (incoming) actions, the compiled code will register all new locations available as parameters or reachable from \trg{\OB{S}}.
Then, on any ! (outgoing) action, the compiled code must scan whatever locations (that the compiled code has created) are now reachable from \trg{\OB{S}} and add them to \trg{\OB{S}}. We need the new instructions \trg{isloc} and \trg{ispair} in \LP to compute these reachable locations.
Of course, this kind of scanning of locations reachable from \trg{\OB{S}} at every call/return between components can be extremely costly.

\paragraph{Enforcing the Masking of Locations}
The functions \trg{mask} and \trg{unmask} are added by the compiler to the compiled code. 
The first function takes a location (which intuitively contains a value \trg{v}) and replaces (in \trg{v}) any pair \trg{\pair{n,k}} of a location protected with a component-created capability \trg{k} with its index in the masking list \trg{\OB{L}}.
The second function replaces any pair \trg{\pair{n,k_{com}}} with the $n$th element of the masking list \trg{\OB{L}}.
These functions should not be directly accessible to the context (else it can \trg{unmask} any \trg{mask}ed location and break full abstraction). 
This is why \LP needs private functions.

\paragraph{Letting the Context use Masked Locations}
Masked locations cannot be used directly by the context for reading and writing.
Thus, compiled code must provide a \trg{read} and a \trg{write} function (both of which are public) that implement reading and writing for masked locations.

\smallskip

As should be clear, code compiled through \compfac{\src{\cdot}} has a lot of runtime overhead in calculating the heap reachable from \trg{\OB{S}} and in \trg{mask}ing and \trg{unmask}ing locations.
Additionally, it also has code memory overhead: the functions \trg{read}, \trg{write}, \trg{mask}, \trg{unmask} and list manipulation code must be included.
Finally, there is data overhead in maintaining  \trg{\OB{S}}, \trg{\OB{L}} and other supporting data structures to implement the runtime checks described above. 
In contrast, the code compiled through \compup{\cdot} (which is just robustly safe and not fully abstract) has none of these overheads.

\subsection{Proving that \compfac{\src{\cdot}} is a Fully Abstract Compiler}\label{sec:fac-inst-proof}
Using \compfac{\src{\cdot}} as a concrete example, we now discuss why \emph{proving}  \facomp can be harder than proving \rscomp.
Consider the hard part of \facomp, the forward implication,
  $\Cs{_1}\ceqs\Cs{_2}\Rightarrow\compgen{\Cs{_1}}\ceqt\compgen{\Cs{_2}}$. The contrapositive of this statement is
  $\compgen{\Cs{_1}}\nceqt\compgen{\Cs{_2}}\Rightarrow\Cs{_1}\nceqs\Cs{_2}$.
By unfolding the definition of $\nceqs$ we see that, given a target context \ctxt{} that distinguishes \compgen{\Cs{_1}} from \compgen{\Cs{_2}}, it is necessary to show that there exists a source context \ctxs{} that distinguishes \Cs{_1} from \Cs{_2}.
That source context \ctxs{} must be built (backtranslated) starting from the already given target context~\ctxt{} that differentiates \compgen{\Cs{_1}} from \compgen{\Cs{_2}}.

A backtranslation directed by the syntax of the target context \ctxt{} is hopeless here since  the target expressions \trg{iscap} and \trg{isloc} cannot be directly backtranslated to valid source expressions.
Hence, we resort to another well-known technique~\cite{KULeuven-358154,mfac}. 
First, we define a \emph{fully abstract (labeled) trace semantics} for the target language. 
A trace semantics is fully abstract when  two components are contextually inequivalent iff their trace semantics differ in at least one trace. 
So if we write \trt{\trg{C}} to denote the traces of the component \trg{C}, we can formally state full abstraction of the trace semantics as: $\trt{\compfac{\src{C_1}}}=\trt{\compfac{\src{C_2}}}\iff\compfac{\src{C_1}}\ceqt\compfac{\src{C_2}}$.
Given this trace semantics, the statement of the forward implication of full abstraction reduces to: 
\begin{align*}
  \trt{\compfac{\src{C_1}}}\neq\trt{\compfac{\src{C_2}}}\Rightarrow\Cs{_1}\nceqs\Cs{_2}.
\end{align*}
The advantage of this formulation over the original one is that now we can construct a distinguishing source context for \src{C_1} and \src{C_2} using the \emph{trace} on which \trt{\compfac{\src{C_1}}} and \trt{\compfac{\src{C_2}}} disagree. 
While this proof strategy of constructing a source context from a trace is similar to our proof of \rscomp, %
it is fundamentally much harder and much more involved. There are two reasons for this.

First, fully abstract trace semantics are much more complex than our simple trace semantics of \LP from earlier sections. 
The reason is that our earlier trace semantics include the entire heap in every action, but this breaks full abstraction of the trace semantics: such trace semantics also distinguish contextually equivalent components that differ in their internal private state. 
In a fully abstract trace semantics, the trace actions must record \emph{only} those heap locations that are shared between the component and the context. Consequently, the definition of the trace semantics must inductively track what has been shared in the past. In particular, the definition must account for locations reachable indirectly from explicitly shared locations. This complicates both the definition of traces and the proofs that build on the definition.

Second, the source context that the backtranslation constructs from a target trace must simulate the shared part of the heap at every context switch. Since locations in the target may be masked now, the source context must maintain a map with the source locations corresponding to the target masked ones, which complicates the source context substantially.
Call this map \src{B}.
Now, this affects two patterns of target traces that need to be handled in a special way: \trg{\clh{read}{v}{H}\cdot\rth{}{H'}{}} and \trg{\clh{write}{v}{H}\cdot\rth{}{H'}{}}.
Normally, these patterns would be translated to source-level calls to the same functions (\src{read} and \src{write}), but this is not possible.
In fact, the source code has no \src{read} or \src{write} function, and the target-level calls to these functions need to be backtranslated to the corresponding source constructs (\src{!} and \src{:=}, respectively).
The locations used by these constructs must be looked up from \src{B} as these are reads and writes to masked locations.
Moreover, calls and returns to \trg{read}  can be simply ignored since the effects of reads are already captured by later actions in traces.
Calls and returns to \trg{write} cannot be ignored as they set up a component location (albeit masked) in a certain way and that affects the behaviour of the component.
We show in \Cref{ex:backtr-fatr} how to backtranslate calls and returns to \trg{write}.
\begin{example}[Backtranslation of traces]\label{ex:backtr-fatr}
Consider the trace below and its backtranslation.
\begin{center}
	\begin{tikzpicture}[remember picture]
		\node[align=left](trace){ 
			${\scriptstyle (1)}\left. \trg{~~\clh{f}{0}{ \trg{1\mapsto4} }} \right. $
			\\
			${\scriptstyle (2)}\left. \trg{~~\rth{}{1\mapsto\pair{1,k_{com}} }{}} \right. $
			\\
			${\scriptstyle (3)}\left[
				\begin{aligned}
					&\trg{\clh{write}{\pair{\pair{1,k_{com}},5}}{ 1\mapsto\pair{1,k_{com}} }}
					\\
					&\trg{\rth{}{ 1\mapsto\pair{1,k_{com}} }{} }
				\end{aligned}
			\right.$
			};

		\node[align=left, right of=trace, xshift = 17em](code){
			\src{main(x)\mapsto}
			\\
			$\left.
			\begin{aligned}[c]
				&\src{\tikz\node(ch1){}; \letnew{x}{4}{L::\pair{x,1}}}\quad
				\\
				&\src{\tikz\node(zzz){}; \call{f}~0}
			\end{aligned}
			\right]{\scriptstyle (1)}$
			\\
			$\left.\src{\tikz\node(ch3){};\, \letin{x}{!L(1)}{B::\pair{x,1}}} \qquad \right]{\scriptstyle (2)}$
			\\
			$\left.
			\begin{aligned}[c]
				&\src{\tikz\node(ch4){}; !B(1) := 5 }\quad
			\end{aligned}
			\right]{\scriptstyle (3)}$
		};

	\end{tikzpicture}
\end{center}

	The first action, where the context registers the first location in the list \src{L}, is as before. %
	Then in the second action the compiled component passes to the context (in location \trg{1}) a masked location with index \trg{1} and, later, the context writes \trg{5} to it.
	The backtranslated code must recognise this pattern and store the location that, in the source, corresponds to the mask \trg{1} in the list \src{B} (action 2).
	In action 3, when it is time to write \src{5} to that location, the code looks up the location to write to from \src{B}.
\end{example}
It should be clear that this proof of \facomp is substantially harder than our corresponding proof of \rscomp, which needed neither fully abstract traces, nor tracking any mapping in the backtranslated source contexts.

\section{Related Work}\label{sec:rw}
Recent work~\citet{rhpc-arx,rhc} presents new criteria for secure compilation that ensure preservation of subclasses of hyperproperties. 
Hyperproperties~\cite{hyperproperties} are a formal representation of predicates on programs, i.e., they are predicates on sets of traces.
Hyperproperties capture many security-relevant properties including not just conventional safety and liveness, which are predicates on traces, but also properties like non-interference, which is a predicate on pairs of traces. 
Modulo technical differences, our definition of \rscomp coincides with the criterion of ``robust safety property preservation'' in~\citet{rhpc-arx,rhc}. 
We show, through concrete instances, that this criterion can be easily realized by compilers, and develop two proof techniques for establishing it. 
We further show that the criterion leads to more efficient compiled code than does \facomp.
Additionally, the criteria in~\cite{rhpc-arx,rhc} assume that behaviours in the source and target are represented using the same alphabet. 
Hence, the definitions (somewhat unrealistically or ideally) do not require a translation of source properties to target properties. 
That line of work has been extended to consider criteria that preserve hyperproperties between languages with different trace models which are connected by a trace relation similar to ours~\cite{rc-rel}.
Like this last work, we consider differences in the representation of behaviour in the source and in the target and this is accounted for in our monitor relation $\src{M} \relate \trg{M}$.
Unlike this last work, we provide different instances where the relation is instantiated in order to show how the theory scales to different protection mechanisms.
A slightly different account of the difference between traces across languages is presented by Patrignani and Garg~\citet{schp} in the context of reactive black-box programs.
Abate \etal~\citet{catalinRSC} define a variant of robustly-safe compilation called RSCC specifically tailored to the case where (source) components can perform undefined behaviour.
RSCC does not consider attacks from arbitrary target contexts but from compiled components that can become compromised and behave in arbitrary ways.
To demonstrate RSCC, Abate \etal~\citet{catalinRSC} rely on two backends for their compiler: software fault isolation and tag-based monitors. On the other hand, we rely on capability machines and memory isolation.
RSCC also preserves (a form of) safety properties and can be achieved by relying on a trace-based backtranslation; it is unclear whether proofs can be simplified when the source is verified and concurrent, as in our second compiler.

ASLR~\cite{Abadi:2012,Jagadeesan:2011:LMV:2056311.2056556}, protected module architectures~\cite{KULeuven-358154,scoo-j,mfac,adriaanaplas}, tagged architectures~\cite{catalin}, capability machines~\cite{cheri-c-fcs} and cryptographic primitives~\cite{Abadi:2000:APC:325694.325734,Abadi:2002:SIC:570966.570969,Bugliesi:2007:SIT:1190215.1190253,Corin:2008:SCS:1454415.1454419} have been used as targets for \facomp.
We believe all of these can also be used as targets of \rscomp-attaining compilers. In fact, some targets such as capability machines seem to be better suited to \rscomp than \facomp, as we demonstrated.

Ahmed \etal prove full abstraction for several compilers between typed languages~\cite{ahmedCPS,Ahmed:2008:TCC:1411203.1411227,max-embed}. As compiler intermediate languages are often typed, and as these types often serve as the basis for complex static analyses, full abstraction seems like a reasonable goal for (fully typed) intermediate compilation steps.
In the last few steps of compilation, where the target languages are unlikely to be typed, one could establish robust safety preservation and combine the two properties (vertically) to get an end-to-end security guarantee.

There are  three other criteria for secure compilation that we would like to mention: securely compartmentalised compilation (SCC)~\cite{catalin}, trace-preserving compilation (TPC)~\cite{schp} and non-interference-preserving compilation (NIPC)~\cite{Barthe:2007:STP:1223678.1223704,shao,jasmin,BartheGL18,kedar}.
SCC is a re-statement of the ``hard'' part of full abstraction (the forward implication), but adapted to languages with undefined behaviour and a static notion of components. 
Thus, SCC suffers from much of the same efficiency drawbacks as \facomp.
TPC is a stronger criterion than \facomp, that most existing fully abstract compilers also attain.
Again, compilers attaining TPC also suffer from the drawbacks of compilers attaining \facomp.

NIPC preserves a single property: noninterference (NI). 
However, this line of work does not consider active target-level adversaries yet. Instead, the focus is on compiling whole programs.
Since noninterference is not a safety property, it is difficult to compare NIPC to \rscomp directly. However, noninterference can also be approximated as a safety property~\cite{seinfl}. So, in principle, \rscomp (with adequate massaging of observations) can be applied to stronger end-goals than NIPC. 

Swamy \etal~\cite{embed} embed an F$^*$ model of a gradually and robustly typed variant of JavaScript into an F$^*$ model of JavaScript. Gradual typing supports constructs similar to our endorsement construct in \LA. Their type-directed compiler is proven to attain memory isolation as well as static and dynamic memory safety. However, they do not consider general safety properties, nor a general criterion for compiler security.

Two of our target languages rely on capabilities for restricting access to sensitive locations from the context. Although capabilities are not mainstream in any processor, fully functional research prototypes such as Cheri exist~\cite{cheri}. Capability machines have previously been advocated as a target for efficient secure compilation~\cite{akram} and preliminary work on compiling C-like languages to them exists, but the criterion applied is \facomp~\cite{cheri-c-fcs,SkorstengaardDB19,SkorstengaardDB18}.

On the other hand, one of our target languages relies on coase-grained isolation, a feature that is being increasingly supported in hardware (Intel calls this SGX~\cite{intel}; ARM
calls it TrustZone~\cite{arm}).
Coarse-grained isolation has also been advocated as a target for secure compilation~\cite{tome-secure-compilation,adriaanaplas,KULeuven-358154}. 
The criterion applied in these works is \facomp, which is what lets us draw a starker comparison in \Cref{sec:facomp}.

\section{Conclusion}\label{sec:conc}
This paper has examined robustly safe compilation (\rscomp), a soundness criterion for compilers with direct relevance to security.
We have shown that the criterion is easily realizable and may lead to more efficient code than does fully abstract compilation. 
We have also presented two techniques for establishing that a compiler attains \rscomp. 
One is an adaptation of an existing technique, backtranslation, and the other is based on inductive invariants. 

\begin{acks}
The authors would like to thank Dominique Devriese, Akram El-Korashy, C\u{a}t\u{a}lin Hri\c{t}cu, Frank Piessens, David Swasey and the anonymous reviewers of ESOP'19 and of TOPLAS for useful feedback and discussions on an earlier draft. 

This work was partially supported by the German Federal Ministry of Education and Research (BMBF) through funding for the CISPA-Stanford Center for Cybersecurity (FKZ: 13N1S0762) 
and 
    by the Office of Naval Research for support through grant N00014-18-1-2620, Accountable Protocol Customization.

\end{acks}

\appendix 
\section{The Untyped Source Language: \LU}\label{sec:src}
This is a sequential while language with monitors.

\subsection{Syntax}
\begin{align*}
	\mi{Whole\ Programs}~\src{P} \bnfdef&\ \src{\src{\ell_{root}} ; H ; \OB{F} ; \OB{I}}
	\\
	\mi{Components}~\src{C} \bnfdef&\ \src{\src{\ell_{root}} ; \OB{F} ; \OB{I}}
	\\
	\mi{Contexts}~\src{A} \bnfdef&\ \src{ H ; \OB{F}\hole{\cdot}}
	\\
	\mi{Interfaces}~\src{I} \bnfdef&\ \src{f}
	\\
	\mi{Functions}~\src{F} \bnfdef&\ \src{f(x)\mapsto s;\ret}
	\\
	\mi{Operations}~\src{\op} \bnfdef&\ \src{+} \mid \src{-}
	\\
	\mi{Comparison}~\src{\bop} \bnfdef&\ \src{==} \mid \src{<} \mid \src{>}
	\\
	\mi{Values}~\src{v} \bnfdef&\ %
	\src{b}\in\{\src{\trues},\src{\falses}\} \mid \src{n}\in\mb{N} \mid \src{\pair{v,v}} \mid \src{\ell}  
	\\
	\mi{Expressions}~\src{e} \bnfdef&\ \src{x} \mid \src{v} \mid \src{e \op e} \mid \src{e \bop e} \mid \src{\pair{e,e}} \mid \src{\projone{e}} \mid \src{\projtwo{e}} \mid \src{!e} 
	\\
	\mi{Statements}~\src{s} \bnfdef&\ \skips \mid \src{s;s} \mid \src{\letin{x}{e}{s}} \mid \src{\ifte{e}{s}{s}} 
	\\
	\mid&\ \src{\call{f}~e} \mid \src{\letnew{x}{e}{s}} \mid \src{x := e} 
	\\
	\mi{Eval.\ Ctxs.}~\src{E} \bnfdef&\ \src{\hole{\cdot}} \mid \src{e \op E} \mid \src{E \op n} \mid \src{e \bop E} \mid \src{E \bop n} 
	\\
	\mid&\ \src{\pair{e,E}} \mid \src{\pair{E,v}} \mid \src{\projone{E}} \mid \src{\projtwo{E}} \mid \src{!E}
	\\
	\mi{Heaps}~\src{H} \bnfdef&\ \srce \mid \src{H ; \ell\mapsto v}
	\\
	\mi{Monitors}~\src{M} \bnfdef&\ \src{(\set{\sigma},\monred,\sigma_0,\ell_{root},\sigma_c)}
	\\
	\mi{Mon.\ States}~\src{\sigma} \in&\ \src{\mc{S}}
	\\
	\mi{Mon.\ Reds.}~\src{\monred} \bnfdef&\ \srce \mid \src{\monred;(s,H,s)}
	\\
	\mi{Substitutions}~\src{\rho} \bnfdef&\ \srce \mid \src{\rho}\subs{v}{x}
	\\
	\mi{Prog.\ States}~\src{\Omega}\bnfdef&\ \src{C, H\triangleright \proc{s}{\OB{f}} } 
	\\
	\mi{Labels}~\src{\lambda} \bnfdef&\ \src{\epsilon} \mid \src{\alpha}
	\\
	\mi{Actions}~\src{\alpha} \bnfdef&\ \src{\clh{f}{v}{H}} \mid \src{\cbh{f}{v}{H}} \mid \src{\rth{v}{H}} \mid \src{\rbh{v}{H}}
	\\
	\mi{Traces}~\src{\OB{\alpha}} \bnfdef&\ \srce \mid \src{\OB{\alpha}\cdot\alpha}
\end{align*}

\subsection{Dynamic Semantics}
\Cref{tr:us-aux-intern,tr:us-aux-in,tr:us-aux-out} dictate the kind of a jump between two functions: if internal to the component/attacker, in(from the attacker to the component) or out(from the component to the attacker).
\Cref{tr:plug-us} tells how to obtain a whole program from a component and an attacker.
\Cref{tr:whole-us} tells when a program is whole.
\Cref{tr:ini-us} tells the initial state of a whole program.
\Cref{tr:ms-us} tells when a monitor makes a single step given a heap.

\mytoprule{\text{Helpers}}
\begin{center}
	\typerule{\LU-Jump-Internal}{
		((\src{f'}\in\src{\OB{I}} \wedge \src{f}\in\src{\OB{I}}) \vee
				\\
		(\src{f'}\notin\src{\OB{I}} \wedge \src{f}\notin\src{\OB{I}}))
	}{
		\src{\OB{I}}\vdash\src{f,f'}:\src{internal}
	}{us-aux-intern}
	\typerule{\LU-Jump-IN}{
		\src{f}\in\src{\OB{I}} \wedge \src{f'}\notin\src{\OB{I}}
	}{
		\src{\OB{I}}\vdash\src{f,f'}:\src{in}
	}{us-aux-in}
	\typerule{\LU-Jump-OUT}{
		\src{f}\notin\src{\OB{I}} \wedge \src{f'}\in\src{\OB{I}}
	}{
		\src{\OB{I}}\vdash\src{f,f'}:\src{out}
	}{us-aux-out}
	\typerule{\LU-Plug}{
		\src{A} \equiv \src{H ; \OB{F}\hole{\cdot}}
		&
		\src{C}\equiv\src{\src{\ell_{root}} ; \OB{F'} ; \OB{I}} 
		&
		\vdash\src{C,\OB{F}}:\src{whole}
		\\
		\src{main}\in\fun{names}{\src{\OB{F}}}
	}{
		\src{A\hole{C}} = \src{\src{\ell_{root}} ; H;\ell_{root}\mapsto0 ; \OB{F;F'}; \OB{I}}
	}{plug-us}
	\typerule{\LU-Whole}{
		\src{C}\equiv\src{\src{\ell_{root}} ; \OB{F'} ; \OB{I}} 
		&
		\fun{names}{\src{\OB{F}}}\cap\fun{names}{\src{\OB{F'}}}=\emptyset
		\\
		\fun{names}{\src{\OB{I}}}\subseteq \fun{names}{\src{\OB{F}}}\cup\fun{names}{\src{\OB{F'}}}
		&
		\fun{fv}{\src{\OB{F}}}\cup\fun{fv}{\src{\OB{F'}}}=\srce
	}{
		\vdash\src{C,\OB{F}}:\src{whole}
	}{whole-us}
	\typerule{\LU-Initial State}{
		\src{P}\equiv\src{\src{\ell_{root}} ; H ; \OB{F} ; \OB{I}}
		&
		\src{C}\equiv\src{\src{\ell_{root}} ; \OB{F} ; \OB{I}}
		&
	    \src{main(x)\mapsto s;\ret}\in\src{\OB{F}}
  	}{
    	\SInits{P} = \src{C ; H,\src{\ell_{root}\mapsto 0} \triangleright \proc{s\subs{0}{x}}{main}}
	}{ini-us}
\end{center}
\botrule

\subsubsection{Component Semantics}\label{src:src-sem-com}
\begin{align*}
	&\src{H\triangleright e \redtos e'} 
	&&\text{Expression \src{e} reduces to \src{e'}.}
	\\
	&\src{C, H \triangleright s} \xtos{\epsilon} \src{C', H' \triangleright s'} 
	&&\text{Statement \src{s} reduces to \src{s'} and evolves the rest accordingly,}
	\\
	&&&\text{emitting label \src{\lambda}.}
	\\
	&\src{\Omega} \Xtos{\OB{\alpha}} \src{\Omega'}
	&& \text{Program state \src{\Omega} steps to \src{\Omega'} emitting trace \src{\OB{\alpha}}.}
\end{align*}

\mytoprule{\src{H}\triangleright \src{e} \redtos \src{e'} }
\begin{center}
	\typerule{E\LU-ctx}{
		\src{H \triangleright e \redtos e'}
	}{
		\src{H \triangleright E\hole{e}} \redtos \src{E\hole{e'}}
	}{eus-cth}
	\typerule{E\LU-val}{
	}{
		\src{H \triangleright v} \redtos \src{v}
	}{eus-val}
	\typerule{E\LU-p1}{
	}{
		\src{H \triangleright \projone{\pair{v,v'}} \redtos v}
	}{eus-p1}
	\typerule{E\LU-p2}{
	}{
		\src{H \triangleright \projtwo{\pair{v,v'}} \redtos v'}
	}{eus-p2}
	\typerule{E\LU-op}{
		n\op n'=n''
	}{
		\src{H \triangleright n \op n' \redtos n''}
	}{eus-op}
	\typerule{E\LU-comp}{
		n\bop n'=b
	}{
		\src{H \triangleright n \bop n' \redtos b}
	}{eus-bop}
	\typerule{E\LU-dereference}{
		\src{\ell\mapsto v } \in \src{H}
	}{
		\src{H \triangleright !\ell \redtos v }
	}{eus-de}
\end{center}
\botrule

\mytoprule{\src{C, H \triangleright s} \xtos{\lambda} \src{C', H' \triangleright s'} }
\begin{center}
	\typerule{E\LU-sequence}{
	}{
		\src{C, H \triangleright \skips;s} \xtos{\epsilon} \src{C, H \triangleright s}
	}{eus-seq}
	\typerule{E\LU-step}{
		\src{C, H \triangleright s} \xtos{\lambda} \src{C, H' \triangleright s'}
	}{
		\src{C, H \triangleright s;s''} \xtos{\lambda} \src{C, H' \triangleright s';s''}
	}{eus-step}
	\typerule{E\LU-if-true}{
		\src{H \triangleright e\redtos\trues}
	}{
		\src{C, H \triangleright \ifte{e}{s}{s'}} \xtos{\epsilon} \src{C, H \triangleright s}
	}{eus-ift}
	\typerule{E\LU-if-false}{
		\src{H \triangleright e\redtos\falses}
	}{
		\src{C, H \triangleright \ifte{e}{s}{s'}} \xtos{\epsilon} \src{C, H \triangleright s}
	}{eus-iff}
	\typerule{E\LU-letin}{
		\src{H \triangleright e\redtos v}
	}{
		\src{C, H \triangleright \letin{x}{e}{s}} \xtos{\epsilon} \src{C, H \triangleright s\subs{v}{x}}
	}{eus-letin}
	\typerule{E\LU-alloc}{
		\src{H \triangleright e\redtos v}
		&
		\src{\ell}\notin\dom{\src{H}}
	}{
		\src{C, H \triangleright \letnew{x}{e}{s}} \xtos{\epsilon} \src{C, H; \ell\mapsto v \triangleright s\subs{\ell}{x} }
	}{eus-al}
	\typerule{E\LU-update}{
		\src{H \triangleright e\redtos v}
		\\
		\src{H}=\src{H_1; \ell\mapsto v' ; H_2}
		&
		\src{H'}=\src{H_1; \ell\mapsto v ; H_2}
	}{
		\src{C, H \triangleright \ell:= e} \xtos{\epsilon} \src{C, H' \triangleright \skips }
	}{eus-up}
	\typerule{E\LU-call-internal}{
		\src{\OB{C}.\mtt{intfs}}\vdash\src{f,f'}:\src{internal}
		&
		\src{\OB{f'}} = \src{\OB{f''};f'}
		\\
		\src{f(x)\mapsto s;\ret}\in\src{C}.\mtt{funs}
		&
		\src{H \triangleright e\redtos v}
	}{
		\src{C, H \triangleright \proc{{\call{f}~e}}{\OB{f'}}} \xtos{\epsilon} \src{C, H \triangleright \proc{{s;\ret\subs{v}{x}}}{\OB{f'};f}}
	}{eus-call-i}
	\typerule{E\LU-callback}{
		\src{\OB{f'}} = \src{\OB{f''};f'}
		&
		\src{f(x)\mapsto s;\ret}\in\src{\OB{F}}
		\\
		\src{\OB{C}.\mtt{intfs}}\vdash\src{f',f}:\src{out}
		&
		\src{H \triangleright e\redtos v}
	}{
		\src{C, H \triangleright \proc{{\call{f}~e}}{\OB{f'}}} \xtos{\cbh{f}{v}{H}} \src{C, H \triangleright \proc{{s;\ret\subs{v}{x}}}{\OB{f'};f}}
	}{eus-callback}
	\typerule{E\LU-call}{
		\src{\OB{f'}} = \src{\OB{f''};f'}
		&
		\src{f(x)\mapsto s;\ret}\in\src{C}.\mtt{funs}
		\\
		\src{\OB{C}.\mtt{intfs}}\vdash\src{f',f}:\src{in}
		&
		\src{H \triangleright e\redtos v}
	}{
		\src{C, H \triangleright \proc{{\call{f}~e}}{\OB{f'}}} \xtos{\clh{f}{v}{H}} \src{C, H \triangleright \proc{{s;\ret\subs{v}{x}}}{\OB{f'};f}}
	}{eus-call}
	\typerule{E\LU-ret-internal}{
		\src{\OB{f'}} = \src{\OB{f''};f'}
		&
		\src{\OB{C}.\mtt{intfs}}\vdash\src{f,f'}:\src{internal}
	}{
		\src{C, H \triangleright \proc{{\ret}}{\OB{f'};f}} \xtos{\epsilon} \src{C, H \triangleright \proc{\skips}{\OB{f'}}}
	}{eus-ret-i}
	\typerule{E\LU-retback}{
		\src{\OB{f'}} = \src{\OB{f''};f'}
		&
		\src{\OB{C}.\mtt{intfs}}\vdash\src{f,f'}:\src{in}
	}{
		\src{C, H \triangleright \proc{{\ret}}{\OB{f'};f}} \xtos{\rbh{v}{H}} \src{C, H \triangleright \proc{\skips}{\OB{f'}}}
	}{eus-retb}
	\typerule{E\LU-return}{
		\src{\OB{f'}} = \src{\OB{f''};f'}
		&
		\src{\OB{C}.\mtt{intfs}}\vdash\src{f,f'}:\src{out}
	}{
		\src{C, H \triangleright \proc{{\ret}}{\OB{f'};f}} \xtos{\rth{v}{H}} \src{C, H \triangleright \proc{\skips}{\OB{f'}}}
	}{eus-ret}
	\end{center}
\botrule

\mytoprule{ \src{\Omega} \Xtos{\OB{\alpha}} \src{\Omega'} }
\begin{center}
	\typerule{E\LU-single}{
		\src{\Omega}\xtos{\alpha}\src{\Omega'}
	}{
		\src{\Omega}\Xtos{\alpha}\src{\Omega'}
	}{eus-tr-sin}
	\typerule{E\LU-silent}{
		\src{\Omega}\xtos{\epsilon}\src{\Omega'}
	}{
		\src{\Omega}\Xtos{}\src{\Omega'}
	}{eus-tr-silent}
	\typerule{E\LU-trans}{
		\src{\Omega}\Xtos{\OB{\alpha}}\src{\Omega''}
		\\
		\src{\Omega''}\Xtos{\OB{\alpha'}}\src{\Omega'}
	}{
		\src{\Omega}\Xtos{\OB{\alpha}\cdot\OB{\alpha'}}\src{\Omega'}
	}{eus-tr-trans}
\end{center}
\botrule

\subsection{Monitor Semantics}
Let \fun{reach}{\src{\ell_o},\src{H}} return a set of locations \src{\set{\ell}} in \src{H} such that it is possible to reach any $\src{\ell}\in\src{\set{\ell}}$ from \src{\ell_o} just by expression evaluation.
\begin{align*}
	\fun{reach}{\src{\ell},\src{H}} =&\ \myset{\src{\ell}}{\exists\src{e}.~ \src{H\triangleright e\redtos ~\ell} \wedge \src{\ell}\in\dom{\src{H}} }
\end{align*}

To ensure monitor transitions have a meaning, they are assumed to be closed under bijective renaming of locations.

\mytoprule{\src{M;H\monred M'}}
\begin{center}
	\typerule{\LU-Monitor Step}{
		\src{M}= \src{(\set{\sigma},\monred,\sigma_0,\ell_{root},\sigma_c)}
		&
		\src{M'}= \src{(\set{\sigma},\monred,\sigma_0,\ell_{root},\sigma_f)}
		\\
		\src{(\sigma_c,H',\sigma_f)}\in\src{\monred}
		\src{H'}\subseteq\src{H}
		&
		\dom{\src{H'}}=\fun{reach}{\src{\ell_{root}},\src{H}}
	}{
		\src{M;H\monred M'}
	}{ms-us}
	\typerule{\LU-Monitor Step Trace Base}{
	}{
		\src{M;\srce\monred M}
	}{ms-t-s-b}
	\typerule{\LU-Monitor Step Trace}{
		\src{M;\OB{H}\monred M''}
		&
		\src{M'';H\monred M'}
	}{
		\src{M;\OB{H}\cdot H\monred M'}
	}{ms-t-s}
	\typerule{\LU-valid trace}{
		\src{M;\OB{H}\monred M'}
		&
		\strip{\src{\OB{\alpha}}}=\src{\OB{H}}
	}{
		\src{M}\vdash\src{\OB{\alpha}}
	}{ms-valid}
\end{center}
\botrule

Monitor actions are the only part of traces that matter for safety, so we define function \strip{\cdot} that takes a general trace and elides all but the heap of actions.
This function is used by both languages so we typeset it in black.
\begin{align*}
	\strip{\come} =&\ \come
	\\
	\strip{\clh{f}{v}{H}\cdot\OB{\alpha}} =&\ H\cdot\strip{\OB{\alpha}}
	\\
	\strip{\cbh{f}{v}{H}\cdot\OB{\alpha}} =&\ H\cdot\strip{\OB{\alpha}}
	\\
	\strip{\rth{}{H}{}\cdot\OB{\alpha}} =&\ H\cdot\strip{\OB{\alpha}}
	\\
	\strip{\rbh{}{H}{}\cdot\OB{\alpha}} =&\ H\cdot\strip{\OB{\alpha}}
\end{align*}
 \newpage
\section{The Target Language: \LP}\label{sec:trg}
\subsection{Syntax}\label{sec:trg-syn}
\begin{align*}
	\mi{Whole\ Programs}~\trg{P} \bnfdef&\ \trg{k_{root} ; \OB{F} ;\OB{I}}
	\\
	\mi{Components}~\trg{C} \bnfdef&\ \trg{k_{root} ; \OB{F} ; \OB{I}}
	\\
	\mi{Contexts}~\trg{A} \bnfdef&\ \trg{\OB{F}\hole{\cdot}}
	\\
	\mi{Interfaces}~\trg{I} \bnfdef&\ \trg{f}
	\\
	\mi{Functions}~\trg{F} \bnfdef&\ \trg{f(x)\mapsto s;\ret}
	\\
	\mi{Operations}~\trg{\op} \bnfdef&\ \trg{+} \mid \trg{-}
	\\
	\mi{Comparison}~\trg{\bop} \bnfdef&\ \trg{==} \mid \trg{<} \mid \trg{>}
	\\
	\mi{Values}~\trg{v} \bnfdef&\ \trg{n}\in\mb{N} \mid \trg{\pair{v,v}} \mid \trg{k} %
	\\
	\mi{Expressions}~\trg{e} \bnfdef&\ \trg{x} \mid \trg{v} \mid \trg{e \op e} \mid \trg{e \bop e} \mid \trg{\pair{e,e}} 
	\mid \trg{\projone{e}} \mid \trg{\projtwo{e}} \mid \trg{!e \with{e}}
	\\
	\mi{Statements}~\trg{s} \bnfdef&\ \skipt \mid \trg{s;s} \mid \trg{\letin{x}{e}{s}} \mid \trg{\ifzte{e}{s}{s}} 
	\mid \trg{\call{f}~e} 
	\\
	\mid&\ \trg{x := e \with{e}} \mid \trg{\letnew{x}{e}{s}}  \mid \trg{\lethide{x}{e}{s}}
	\\
	\mi{Eval.\ Ctxs.}~\trg{E} \bnfdef&\ \trg{\hole{\cdot}} \mid \trg{e \op E} \mid \trg{E \op n} \mid \trg{e \bop E} \mid \trg{E \bop n} 
	\mid \trg{!E \with{v}} \mid \trg{!e \with{E}} 
	\\
	\mid&\ \trg{\pair{e,E}} \mid \trg{\pair{E,v}} \mid \trg{\projone{E}} \mid \trg{\projtwo{E}} 
	\\
	\mi{Heaps}~\trg{H} \bnfdef&\ \trge \mid \trg{H ; n\mapsto v:\eta} \mid \trg{H ; k}
	\\
	\mi{Tag}~\trg{\eta} \bnfdef&\ \trg{\bot} \mid \trg{k}
	\\
	\mi{Monitors}~\trg{M} \bnfdef&\ \trg{(\set{\sigma},\monred,\sigma_0,k_{root},\sigma_c)}
	\\
	\mi{Mon.\ States}~\trg{\sigma} \in&\ \trg{\mc{S}}
	\\
	\mi{Mon.\ Reds.}~\trg{\monred} \bnfdef&\ \trge \mid \trg{\monred;(s,H,s)}
	\\
	\mi{Substitutions}~\trg{\rho} \bnfdef&\ \trge \mid \trg{\rho}\subt{v}{x}
	\\
	\mi{Prog.\ States}~\trg{\Omega}\bnfdef&\ \trg{C, H\triangleright \proc{s}{\OB{f}} }
	\\
	\mi{Labels}~\trg{\lambda} \bnfdef&\ \trg{\epsilon} \mid \trg{\alpha}
	\\
	\mi{Actions}~\trg{\alpha} \bnfdef&\ \trg{\clh{f}{v}{H}} \mid \trg{\cbh{f}{v}{H}} \mid \trg{\rth{v}{H}} \mid \trg{\rbh{v}{H}}
	\\
	\mi{Traces}~\trg{\OB{\alpha}} \bnfdef&\ \trge \mid \trg{\OB{\alpha}\cdot\alpha}
\end{align*}

\subsection{Operational Semantics of \LP}\label{sec:trg-sem}
\mytoprule{\text{Helpers}}
\begin{center}
	\typerule{\LP-Jump-Internal}{
		((\trg{f'}\in\trg{\OB{I}} \wedge \trg{f}\in\trg{\OB{I}}) \vee
				\\
		(\trg{f'}\notin\trg{\OB{I}} \wedge \trg{f}\notin\trg{\OB{I}}))
	}{
		\trg{\OB{I}}\vdash\trg{f,f'}:\trg{internal}
	}{t-aux-intern}
	\typerule{\LP-Jump-IN}{
		\trg{f}\in\trg{\OB{I}} \wedge \trg{f'}\notin\trg{\OB{I}}
	}{
		\trg{\OB{I}}\vdash\trg{f,f'}:\trg{in}
	}{t-aux-in}
	\typerule{\LP-Jump-OUT}{
		\trg{f}\notin\trg{\OB{I}} \wedge \trg{f'}\in\trg{\OB{I}}
	}{
		\trg{\OB{I}}\vdash\trg{f,f'}:\trg{out}
	}{t-aux-out}
	\typerule{\LP-Plug}{
		\trg{A} \equiv \trg{\OB{F}\hole{\cdot}}
		&
		\trg{C}\equiv\trg{k_{root} ; \OB{F'} ; \OB{I} } 
		\\
		\vdash\trg{C,\OB{F}}:\trg{whole}
		&
		\trg{main(x)\mapsto s;\ret}\in\trg{\OB{F}}
	}{
		\trg{A\hole{C}} = \trg{k_{root}; \OB{F;F'}; \OB{I}}
	}{plug-t}
	\typerule{\LP-Whole}{
		\trg{C}\equiv\trg{k_{root} ; \OB{F'} ; \OB{I}} 
		\\
		\fun{names}{\trg{\OB{F}}}\cap\fun{names}{\trg{\OB{F'}}}=\emptyset
		&
		\fun{names}{\trg{\OB{I}}}\subseteq \fun{names}{\trg{\OB{F}}}
	}{
		\vdash\trg{C,\OB{F}}:\trg{whole}
	}{whole-t}
	\typerule{\LP-Initial State}{
		\trg{P}\equiv\trg{k_{root}; \OB{F} ; \OB{I}}
		&
		\trg{C}\equiv\trg{k_{root} ; \OB{F} ; \OB{I}} 
		&
		\trg{main(x)\mapsto s;\ret}\in\trg{\OB{F}}
  	}{
		\SInitt{P} = \trg{C, k_{root}; 0\mapsto 0:k_{root} \triangleright \proc{s\subt{0}{x}}{main}}
	}{ini-t}
\end{center}
\botrule

\subsubsection{Component Semantics}\label{src:trg-sem-com}
\begin{align*}
	&\trg{H \triangleright e \redtot e'} 
	&&\text{Expression \trg{e} reduces to \trg{e'}.}
	\\
	&\trg{C, H \triangleright s} \xtot{\epsilon} \trg{C', H' \triangleright s'} 
	&&\text{Statement \trg{s} reduces to \trg{s'} and evolves the rest accordingly,}
	\\
	&&&\text{emitting label \trg{\lambda}.}
	\\
	&\trg{\Omega} \Xtot{\OB{\alpha}} \trg{\Omega'}
	&& \text{Program state \trg{\Omega} steps to \trg{\Omega'} emitting trace \trg{\OB{\alpha}}.}
\end{align*}

\mytoprule{\trg{H \triangleright e \redtot e'} }
\begin{center}
	\typerule{E\LP-val}{
	}{
		\trg{H \triangleright v} \redtot \trg{v}
	}{et-val}
	\typerule{E\LP-p1}{
	}{
		\trg{H\triangleright \projone{\pair{v,v'}} \redtot v}
	}{et-p1}
	\typerule{E\LP-p2}{
	}{
		\trg{H\triangleright \projone{\pair{v,v'}} \redtot v'}
	}{et-p2}
	\typerule{E\LP-op}{
		n\op n'=n''
	}{
		\trg{H\triangleright n \op n' \redtot n''}
	}{et-op}
	\typerule{E\LP-comp}{
		\text{if } n\bop n'= \text{true} \text{ then } \trg{n''}=\trg{0} \text{ else } \trg{n''}=\trg{1}
	}{
		\trg{H \triangleright n \bop n' \redtot n''}
	}{et-bop}
	\typerule{E\LP-deref-top}{
		\trg{n\mapsto v:\bm{\bot}}\in\trg{H}
	}{
		\trg{H\triangleright !n \with{\_}} \redtot \trg{ v}
	}{et-de-t}
	\typerule{E\LP-deref-k}{
		\trg{n\mapsto (v,k)}\in\trg{H}
	}{
		\trg{H\triangleright !n \with{k}} \redtot \trg{ v}
	}{et-de-k}
	\typerule{E\LP-ctx}{
		\trg{H\triangleright e \redtot e'}
	}{
		\trg{H \triangleright E\hole{e}} \redtot \trg{ E\hole{e'}}
	}{et-cth}
\end{center}
\botrule

\mytoprule{\trg{C; H \triangleright s} \xtot{\lambda} \trg{C'; H' \triangleright s'} }
\begin{center}
	\typerule{E\LP-sequence}{
	}{
		\trg{C, H \triangleright \skipt;s} \xtot{\epsilon} \trg{C, H \triangleright s}
	}{et-seq}
	\typerule{E\LP-step}{
		\trg{C, H \triangleright s} \xtot{\lambda} \trg{C, H \triangleright s'}
	}{
		\trg{C, H \triangleright s;s''} \xtot{\lambda} \trg{C, H \triangleright s';s}
	}{et-step}
	\typerule{E\LP-if-true}{
		\trg{H \triangleright e\redtot 0}
	}{
		\trg{C, H \triangleright \ifzte{e}{s}{s'}} \xtot{\epsilon} \trg{C, H \triangleright s}
	}{et-ift}
	\typerule{E\LP-if-false}{
		\trg{H \triangleright e\redtot n}
		&
		\trg{n}\not\equiv\trg{0}
	}{
		\trg{C, H \triangleright \ifzte{e}{s}{s'}} \xtot{\epsilon} \trg{C, H \triangleright s'}
	}{et-iff}
	\typerule{E\LP-letin}{
		\trg{H \triangleright e\redtot v}
	}{
		\trg{C, H \triangleright \letin{x}{e}{s}} \xtot{\epsilon} \trg{C, H \triangleright s\subt{v}{x}}
	}{et-letin}
	\typerule{E\LP-new}{
    \trg{H} = \trg{H_1;n\mapsto (v',\trgb{\eta})}
    &
    \trg{H \triangleright e\redtot v}
    &
    \trg{H'}=\trg{H; n+1\mapsto v:\trgb{\bot} }
  }{
      \trg{C, H \triangleright \letnew{x}{e}{s}} \xtot{} 
      \trg{C,H'\triangleright s\subt{n+1}{x}}
  }{et-nu}
  \typerule{E\LP-hide}{
		\trg{H \triangleright e\redtot n}
		&
		\trg{k}\notin\dom{\trg{H}}
		\\
		\trg{H} = \trg{H_1;n\mapsto v:\bm{\bot};H_2}
		&
		\trg{H'}=\trg{H_1;n\mapsto v:k;H_2;k}
	}{
		\trg{C, H \triangleright \lethide{x}{e}{s}} \xtot{\epsilon} \trg{C, H' \triangleright s\subt{k}{x}}
	}{et-hi}
	\typerule{E\LP-assign-top}{
		\trg{H \triangleright e\redtot v}
		\\
		\trg{H} = \trg{H_1;n\mapsto \_:\bm{\bot};H_2}
		&
		\trg{H'} = \trg{H_1;n\mapsto v:\bm{\bot};H_2}
	}{
		\trg{C, H\triangleright n:=e \with{\_}} \xtot{\epsilon} \trg{C, H'\triangleright \skipt}
	}{et-ac-t}
	\typerule{E\LP-assign-k}{
		\trg{H \triangleright e\redtot v}
		&
		\trg{H \triangleright e'\redtot k}
		\\
		\trg{H} = \trg{H_1;n\mapsto \_:k;H_2}
		&
		\trg{H'} = \trg{H_1;n\mapsto v:k;H_2}
	}{
		\trg{C, H\triangleright n:=e \with{e'}} \xtot{\epsilon} \trg{C, H'\triangleright \skipt}
	}{et-ac-k}
	\typerule{E\LP-call-internal}{
		\trg{\OB{C}.\mtt{intfs}}\vdash\trg{f,f'}:\trg{internal}
		&
		\trg{\OB{f'}} = \trg{\OB{f''};f'}
		\\
		\trg{f(x)\mapsto s;\ret}\in\trg{C}.\mtt{funs}
		&
		\trg{H \triangleright e\redtot v}
	}{
		\trg{C, H \triangleright \proc{{\call{f}~e}}{\OB{f'}}} \xtot{\epsilon} \trg{C, H \triangleright \proc{{s;\ret\subt{v}{x}}}{\OB{f'};f}}
	}{et-call-i}
	\typerule{E\LP-callback}{
		\trg{\OB{f'}} = \trg{\OB{f''};f'}
		&
		\trg{f(x)\mapsto s;\ret}\in\trg{\OB{F}}
		\\
		\trg{\OB{C}.\mtt{intfs}}\vdash\trg{f',f}:\trg{out}
		\trg{H \triangleright e\redtot v}
	}{
		\trg{C, H \triangleright \proc{{\call{f}~e}}{\OB{f'}}} \xtot{\cbh{f}{v}{H}} \trg{C, H \triangleright \proc{{s;\ret\subt{v}{x}}}{\OB{f'};f}}
	}{et-callback}
	\typerule{E\LP-call}{
		\trg{\OB{f'}} = \trg{\OB{f''};f'}
		&
		\trg{f(x)\mapsto s;\ret}\in\trg{C}.\mtt{funs}
		\\
		\trg{\OB{C}.\mtt{intfs}}\vdash\trg{f',f}:\trg{in}
		&
		\trg{H \triangleright e\redtot v}
	}{
		\trg{C, H \triangleright \proc{{\call{f}~e}}{\OB{f'}}} \xtot{\clh{f}{v}{H}} \trg{C, H \triangleright \proc{{s;\ret\subt{v}{x}}}{\OB{f'};f}}
	}{et-call}
	\typerule{E\LP-ret-internal}{
		\trg{\OB{C}.\mtt{intfs}}\vdash\trg{f,f'}:\trg{internal}
		&
		\trg{\OB{f'}} = \trg{\OB{f''};f'}
	}{
		\trg{C, H \triangleright \proc{{\ret}}{\OB{f'};f}} \xtot{\epsilon} \trg{C, H \triangleright \proc{\skipt}{\OB{f'}}}
	}{et-ret-i}
	\typerule{E\LP-retback}{
		\trg{\OB{C}.\mtt{intfs}}\vdash\trg{f,f'}:\trg{in}
		&
		\trg{\OB{f'}} = \trg{\OB{f''};f'}
	}{
		\trg{C, H \triangleright \proc{{\ret}}{\OB{f'};f}} \xtot{\rbh{}{H}} \trg{C, H \triangleright \proc{\skipt}{\OB{f'}}}
	}{et-retb}
	\typerule{E\LP-return}{
		\trg{\OB{C}.\mtt{intfs}}\vdash\trg{f,f'}:\trg{out}
		&
		\trg{\OB{f'}} = \trg{\OB{f''};f'}
	}{
		\trg{C, H \triangleright \proc{{\ret}}{\OB{f'};f}} \xtot{\rth{}{H}} \trg{C, H \triangleright \proc{\skipt}{\OB{f'}}}
	}{et-ret}
	\end{center}
\botrule
\mytoprule{ \trg{\Omega} \Xtot{\OB{\alpha}} \trg{\Omega'} }
\begin{center}
	\typerule{E\LP-single}{
		\trg{\Omega}\xtot{\alpha}\trg{\Omega'}
	}{
		\trg{\Omega}\Xtot{\alpha}\trg{\Omega'}
	}{et-tr-sin}
	\typerule{E\LP-silent}{
		\trg{\Omega}\xtot{\epsilon}\trg{\Omega'}
	}{
		\trg{\Omega}\Xtot{}\trg{\Omega'}
	}{et-tr-silent}
	\typerule{E\LP-trans}{
		\trg{\Omega}\Xtot{\OB{\alpha}}\trg{\Omega''}
		\\
		\trg{\Omega''}\Xtot{\OB{\alpha'}}\trg{\Omega'}
	}{
		\trg{\Omega}\Xtot{\OB{\alpha}\cdot\OB{\alpha'}}\trg{\Omega'}
	}{et-tr-trans}
\end{center}
\botrule

\subsection{Monitor Semantics}
Define \fun{reach}{\trg{n_r},\trg{k_r},\trg{H}} as the set of locations \trg{\set{n}} such that it is possible to reach any $\trg{n}\in\trg{\set{n}}$ from \trg{n_r} using any expression and relying on capability \trg{k_r} as well as any capability reachable from \trg{n_r}.
Formally:
\begin{align*}
	\fun{reach}{\trg{n_r},\trg{k_r},\trg{H}} =&\ \myset{ \trg{n} }{ 
		\begin{aligned}
			&
			\trg{H \triangleright e \redtot !n \with{v} \redtot v'}
			\\
			&
			\fun{fv}{\trg{e}}= \trg{n_r}\cup\trg{k_r}
		\end{aligned}
	}
\end{align*}

\mytoprule{\trg{M;H\monred M'}}
\begin{center}
	\typerule{\LP-Monitor Step}{
		\trg{M}= \trg{(\set{\sigma},\monred,\sigma_0,k_{root},\sigma_c)}
		&
		\trg{M'}= \trg{(\set{\sigma},\monred,\sigma_0,k_{root},\sigma_f)}
		\\
		\trg{(s_c,H',s_f)}\in\trg{\monred}
		\trg{H'}\subseteq\trg{H}
		&
		\dom{\trg{H'}}=\fun{reach}{\trg{0},\trg{k_{root}},\trg{H}}
	}{
		\trg{M;H\monred M'}
	}{ms-t}
	\typerule{\LP-Monitor Step Trace Base}{
	}{
		\trg{M;\trge\monred M}
	}{mt-t-s-b}
	\typerule{\LP-Monitor Step Trace}{
		\trg{M;\OB{H}\monred M''}
		&
		\trg{M'';H\monred M'}
	}{
		\trg{M;\OB{H}\cdot H\monred M'}
	}{mt-t-s}
	\typerule{\LP-valid trace}{
		\trg{M;\OB{H}\monred M'}
		&
		\strip{\trg{\OB{\alpha}}}=\trg{\OB{H}}
	}{
		\trg{M}\vdash\trg{\OB{\alpha}}
	}{mt-valid}
\end{center}
\botrule

 \newpage
\section{Language and Compiler Properties}\label{sec:lang-prop}

\subsection{Safety, Attackers and Robust Safety}
These properties hold for both languages are written in black and only once.

\begin{definition}[Safety]\label{def:safe}
	\begin{align*}
		M\vdash C : \com{safe} \isdef&\ 
		\ldotp 
		\\
		\text{if }
		&
		\vdash \com{C}:\com{whole}
		\\
		\text{then }
		&
		\begin{aligned}[t]
			\text{if }
			&
			\SInit{\com{C}} \Xtoc{\OB{\alpha}} \com{\_} 
			\\
			\text{then }
			&
			M\vdash{\OB{\alpha}}
		\end{aligned}
	\end{align*}
\end{definition}
A program is safe for a monitor if the monitor accepts any trace the program generates.

For now, we give an informal definition for function \fun{locs}{\com{A}}, which returns all the locations statically bound in the heap and code of attacker \com{A}.

\begin{definition}[(Informal) Attacker]\label{def:att}
	\begin{align*}
		\com{C}\vdash\com{A} : \com{attacker} \isdef&\ \text{ no location the component cares about }\in\fun{locs}{\com{A}}
	\end{align*}
\end{definition}
An attacker is valid if it does not refer to the locations the component cares about.
We leave the notion of \emph{location the component cares about} abstract and instantiate it on a per-language basis later on.

\begin{definition}[Robust Safety]\label{def:rs}
	\begin{align*}
		\com{M}\vdash \com{C} : \com{rs} \isdef&\  \forall \com{A}.
		\\
		\text{ if }
		& 
		M\agree C
		\\
		&
		\com{C}\vdash \com{A}: \com{attacker}
		\\
		\text{ then }
		&
		\com{M}\vdash \com{A\hole{C}} : \com{safe}
	\end{align*}
\end{definition}
A program is robustly safe if it is safe for any attacker it is composed with.

The definition of $M\agree C$ is to be specified on a language-specific basis, as the next section does for \LU and \LP.

\subsection{Monitor Agreement and Attacker for \LP and \LU}\label{sec:mon-agr}
\begin{definition}[\LU: \src{M\agree C}]\label{def:lu-agree}
	\begin{align*}
		\src{(\set{\sigma},\monred,\sigma_0,\ell_{root},\sigma_c) \agree (\ell_{root} ; \OB{F} ; \OB{I})}
	\end{align*}
\end{definition}
A monitor and a component agree if they focus on the same initial location \src{\ell_{root}}.

\begin{definition}[\LP: \trg{M\agree C}]\label{def:lp-agree}
	\begin{align*}
		\trg{(\set{\sigma},\monred,\sigma_0,k_{root},\sigma_c) \agree (k_{root} ; \OB{F} ; \OB{I})}
	\end{align*}
\end{definition}
A monitor and a component agree if they use the same capabilty \trg{k_{root}} to protect the initial location \trg{0}.

To define attackers, we define functions \fun{locs}{\src{A}}, which returns all the locations bound in the code and heap of \src{A} and function \fun{caps}{\trg{A}}, which returns all the capabilities bound in the code of \trg{A}.
Intuitively, the former is defined inductively on the structure of functions, then statements, then expressions, where it collects all \src{\ell} in an expression \src{e}, and also on the structure of heaps, where it collects all \src{\ell} in a binding of the form \src{\ell'\mapsto\ell}.
The latter is defined inductively on the structure of functions, then statements, then expressions, where it collects all \trg{k} in an expression \trg{e}.

\begin{definition}[\LU attacker]\label{def:lu-att}
	\begin{align*}
		\src{C}\vdash\src{A}:\src{attacker} \isdef&\ \src{C}=\src{(\ell_{root} ; \OB{F} ; \OB{I})}, \src{A}=\src{ H ; \OB{F'}}
		\\
		&
		\src{\ell_{root}}\notin\fun{locs}{\src{A}}
	\end{align*}
\end{definition}

\begin{definition}[\LP attacker]\label{def:lp-att}
	\begin{align*}
		\trg{C}\vdash\trg{A}:\trg{attacker} \isdef&\ \trg{C}=\trg{(k_{root} ; \OB{F} ; \OB{I})}, \trg{A}=\trg{ \OB{F'}}
		\\
		&
		\trg{k_{root}}\notin\fun{caps}{\trg{\OB{F'}}}
	\end{align*}
\end{definition}

\subsection{Cross-language Relations}\label{sec:cr-rel}
Assume a partial bijection $\beta : \src{\ell}\times\trg{n}\times\trg{\eta}$ from source to target heap locations such that
\begin{itemize}
	\item if $(\src{\ell_1},\trg{n},\trg{\eta})\in\beta$ and $(\src{\ell_2},\trg{n},\trg{\eta})$ then $\src{\ell_1}=\src{\ell_2}$;
	\item if $(\src{\ell},\trg{n_1},\trg{\eta_1})\in\beta$ and $(\src{\ell},\trg{n_2},\trg{\eta_2})$ then $\trg{n_1}=\trg{n_2}$ and $\trg{\eta_1}=\trg{\eta_2}$.
\end{itemize}
we use this bijection to parametrise the relation so that we can relate meaningful locations.

For compiler correctness we rely on a $\beta_0$ which relates initial locations of monitors.

Assume a relation $\relatebeta : \src{v}\times\beta\times\trg{v}$ that is total so it maps any source value to a target value \trg{v}.
\begin{itemize}
	\item $\forall \src{v}. \exists \trg{v}. \src{v}\relatebeta\trg{v}$.
\end{itemize}

This relation is used for defining compiler correctness.
By inspecting the semantics of \LU, \Cref{tr:et-seq,tr:et-ift} let us derive that 
\begin{itemize}
	\item $\trues\relatebeta\trg{0}$;
	\item $\falses\relatebeta\trg{n}$ where $\trg{n}\neq\trg{0}$;
	\item $\src{\ell}\relatebeta\trg{\pair{n,v}}$ where 
		$\begin{cases}
			\trg{v} = \trg{k}	&\text{if } (\src{\ell},\trg{n},\trg{k})\in\beta
			\\
			\trg{v} \neq \trg{k} &\text{otherwise, so } (\src{\ell},\trg{n},\trg{\bot})\in\beta
		\end{cases}$
	\item $\src{\pair{v_1,v_2}}\relatebeta\trg{\pair{v_1,v_2}}$ iff $\src{v_1}\relatebeta\trg{v_1}$ and $\src{v_2}\relatebeta\trg{v_2}$.
\end{itemize}

We overload the notation and use the same notation to indicate the (assumed) relation between monitor states: $\src{\sigma}\relate\trg{\sigma}$.

We lift this relation to sets of states point-wise and indicate it as follows: $\src{\set{\sigma}}\relate\trg{\set{\sigma}}$.
In these cases the bijection $\beta$ is not needed as states do not have locations inside.

Function names are related when they are the same: $\src{f}\relatebeta\trg{f}$.

Variables names are related when they are the same: $\src{x}\relatebeta\trg{x}$.

Substitutions are related when they replace related values for related variables: $\subs{v}{x}\relatebeta\subt{v}{x}$ iff $\src{v}\relatebeta\trg{v}$ and $\src{x}\relatebeta\trg{x}$.

\mytoprule{\src{\alpha}\relatebeta\trg{\alpha}}
\begin{center}
	\typerule{Call relation}{
		\src{f}\relate\trg{f}
		&
		\src{v}\relatebeta\trg{v}
		&
		\src{H}\relatebeta\trg{H}
	}{
		\src{\clh{f}{v}{H}}\relatebeta\trg{\clh{f}{v}{H}}
	}{rel-cl}
	\typerule{Callback relation}{
		\src{f}\relate\trg{f}
		&
		\src{v}\relatebeta\trg{v}
		&
		\src{H}\relatebeta\trg{H}
	}{
		\src{\cbh{f}{v}{H}}\relatebeta\trg{\cbh{f}{v}{H}}
	}{rel-cb}
	\typerule{Return relation}{
		\src{H}\relatebeta\trg{H}
	}{
		\src{\rth{}{H}}\relatebeta\trg{\rth{}{H}}
	}{rel-rt}
	\typerule{Returnback relation}{
		\src{H}\relatebeta\trg{H}
	}{
		\src{\rbh{}{H}}\relatebeta\trg{\rbh{}{H}}
	}{rel-rb}
	\typerule{Epsilon relation}{
	}{
		\src{\epsilon}\relatebeta\trg{\epsilon}
	}{rel-ep}
\end{center}
\botrule

\begin{definition}[$\src{M}\mc{R}\trg{M}$]\label{tr:mon-rel-sing}\label{def:mon-rel}
Given a monitor-specific relation $\src{\sigma} \relate \trg{\sigma}$ on monitor states, we say that a relation $\mc{R}$ on source and target monitors is a \emph{bisimulation} if the following hold whenever $\src{M}=\src{(\set{\sigma},\monred,\sigma_0,\ell_{root},\sigma_c)}$ and $\trg{M} = \trg{(\set{\sigma},\monred,\sigma_0,k_{root},\sigma_c)}$ are related by $\mc{R}$:
\begin{enumerate}
	\item $\src{\sigma_0} \relate \trg{\sigma_0}$, and
	\item $\src{\sigma_c} \relate \trg{\sigma_c}$, and
	\item For all $\beta$ containing $(\src{\ell_{root}}, \trg{0}, \trg{k_{root}})$ and all $\src{H}, \trg{H}$ with $\src{H} \relatebeta \trg{H}$ the following hold: 
	\begin{enumerate}
	\item $(\src{\sigma_c}, \src{H}, \_) \in \src{\monred}$ iff $(\trg{\sigma_c}, \trg{H}, \_) \in \trg{\monred}$, and 
    \item $(\src{\sigma_c}, \src{H}, \src{\sigma'}) \in \src{\monred}$ and $(\trg{\sigma_c}, \trg{H}, \trg{\sigma'}) \in \trg{\monred}$ imply $\src{(\set{\sigma}, \monred, \sigma_0, \ell_{root}, \sigma')} \mc{R} \trg{(\set{\sigma}, \monred, \sigma_0, k_{root}, \sigma')}$.
	\end{enumerate}
\end{enumerate}
\end{definition}

\begin{definition}[$\src{M}\relate\trg{M}$]\label{tr:mon-rel}
	$\src{M}\relate\trg{M}$ is the union of all bisimulations $\src{M}\mc{R}\trg{M}$, which is also a bisimulation.
\end{definition}

\mytoprule{\src{H}\relatebeta\trg{H}}
\begin{center}
	\typerule{Heap relation}{
		\src{H}\relatebeta\trg{H_1;H_2}
		&
		\src{\ell}\relatebeta\trg{\pair{n,\eta}}
		&
		\src{v}\relatebeta\trg{v}
		\\
		\trg{H} = \trg{H_1;n\mapsto v:\eta;H_2}
	}{
		\src{H;\ell\mapsto v} \relatebeta \trg{H}
	}{hrel-i}
	\typerule{Empty relation}{
	}{
		\srce \relatebeta \trg{\OB{k}}
	}{hrel-b}
\end{center}
\botrule

The heap relation is crucial.
A source heap \src{H} is related to a target heap \trg{H} if for any location pointing to a value in the former, a related location points to a related value in the target (\Cref{tr:hrel-i}).
The base case (\Cref{tr:hrel-b}) considers that in the target heap we may have keys, which are not related to source elements.

As additional notation for states, we define when a state is stuck as follows
\begin{center}
	\typerule{Stuck state}{
	\com{\Omega}=\com{M ; \OB{F} ; \OB{I} ; H \triangleright s}
	&
	\com{s}\not\equiv \com{\lskip}
	&
	\nexists \com{\Omega'},\com{\lambda}. \com{\Omega}\xtol{\lambda}\com{\Omega'}
}{
	\com{\Omega}\stu
}{state-stuck}
\end{center}
A state that terminated is defined as follows; this definition is given for a concurrent version of the language too (this is relevant for languages defined later):
\begin{center}
	\typerule{Terminated state}{
		\com{\Omega}=\com{M ; \OB{F} ; \OB{I} ; H \triangleright \lskip}
	}{
		\com{\Omega}\termc
	}{state-term}
	\typerule{Terminated soup}{
		\com{\Omega}=\com{M ; \OB{F} ; \OB{I} ; H \triangleright \Pi}
		&
		\forall \com{\pi}\in\com{\Pi}\ldotp \com{M ; \OB{F} ; \OB{I} ; H \triangleright \pi}\term
	}{
		\com{\Omega}\termc
	}{state-term}
\end{center}

To define compiler correctness, we rely on a cross-language relation for program states.
Two states are related if their monitors are related and if their whole heap is related (\Cref{tr:state-rel-whole}). %

\mytoprule{\src{\Omega}\relatebeta\trg{\Omega}}
\begin{center}
	\typerule{Related states -- Whole}{
		\src{\Omega}=\src{M ; \OB{F},\OB{F'} ; \OB{I} ; H \triangleright s}
		\\
		\trg{\Omega}=\trg{M ; \OB{F},\compup{\OB{F'}} ; \OB{I} ; H \triangleright s}
		\\
		\src{M}\relatebeta\trg{M}
		&
		\src{H}\relatebeta\trg{H}
	}{
		\src{\Omega}\relatebeta\trg{\Omega}
	}{state-rel-whole}
\end{center}
\botrule

\subsection{Correct and Robustly-safe Compilation}
Consider a compiler to be a function of this form: $\compgen{\cdot} : \src{C}\to \trg{C}$, taking a source component and producing a target component.

\begin{definition}[Correct Compilation]\label{def:comp-corr}
	\begin{align*}
		\vdash\compgen{\cdot} : \ccomp \isdef&\
		\forall \src{C},
		\exists \beta.
		\\
		\text{ if }&\
		\SInitt{\compgen{\src{C}}} \Xtot{\OB{\alpha}} \trg{\Omega}
		\\
		&\
		\trg{\Omega}\termt%
		\\
		&\
		\SInits{C} \relate_{\beta_0}\SInitt{\compgen{\src{C}}}
		\\
		\text{ then }&\
		\SInits{C}\Xtos{\OB{\alpha}} \src{\Omega} 
		\\
		&\
		\beta_0\subseteq\beta
		\\
		&\
		\src{\Omega}\relatebeta\trg{\Omega}
		\\
		&\
		\src{\OB{\alpha}}\relatebeta\trg{\OB{\alpha}}
		\\
		&\
		\src{\Omega}\termsl%
	\end{align*}
\end{definition}
Technically, any sequence {\OB{\alpha}} above is empty, as \src{\OB{I}} is empty (the program is whole).

\begin{definition}[Robust Safety Preserving Compilation]\label{def:rsc-pres}
	\begin{align*}
		\vdash\compgen{\cdot}: \rscomp \isdef&\ 
		\forall \src{C},\src{M},\trg{M}\ldotp 
		\\
		\text{if }&\
		\src{M}\vdash \src{C} : \src{rs}
		\\
		&\
		\src{M}\relate\trg{M}
		\\
		\text{ then }&\
		\trg{M}\vdash \compgen{\src{C}} : \trg{rs}
	\end{align*}
\end{definition}

\subsubsection{Alternative definition for RSC}\label{sec:other-def-rsc}
\begin{definition}[Property-Free RSC]\label{def:rsc-eq}
	\begin{align*}
		\vdash\compgen{\cdot}: \pfrscomp \isdef&\ 
		\forall \src{C}\ldotp 
		\\
		\text{if }&\
		\forall\trg{A},\trg{\OB{\alpha}}.~
		\\
		&\
		\compgen{C}\vdash\trg{A}:\trg{attacker}
		\\
		&\
		\vdash\trg{A\hole{\compgen{C}}}:\trg{whole} 
		\\
		&\
		\SInitt{\compgen{\src{C}}} \Xtot{\OB{\alpha}} \trg{\_}
		\\
		\text{then }&\
		\exists\src{A},\src{\OB{\alpha}}.~
		\\
		&\
		\src{C}\vdash\src{A}:\src{attacker}
		\\
		&\
		\vdash\src{A\hole{C}}:\src{whole} 
		\\
		&\
		\SInits{C}\Xtos{\OB{\alpha}} \src{\_} 
		\\
		&\
		\strip{\src{\OB{\alpha}}}\relatebeta\strip{\trg{\OB{\alpha}}}
	\end{align*}
\end{definition}

The property-free characterisation of RSC is equivalent to its original characterisation.
\begin{theorem}[\pfrscomp and \rscomp are equivalent]\label{thm:rsc-prf-eq}
	\begin{align*}
		\forall\compgen{\cdot}, \vdash\compgen{\cdot}: \pfrscomp \iff \vdash\compgen{\cdot}:\rscomp
	\end{align*}
\end{theorem}

\subsubsection{Compiling Monitors}\label{sec:rsc-comp-mon}
We can change the definition of compiler to also compile the monitor so we are not given a target monitor related to the source one, but the compiler gives us that monitor.
Consider this compiler to have this type and this notation: $\compgenm{\src{\cdot}} : \src{C}\to \trg{C}$.

\begin{definition}[Robustly-safe Compilation with Monitors]\label{def:rsc-m}
	\begin{align*}
		\vdash\compgenm{\src{\cdot}}: \lst{rs-pres(M)} \isdef&\ 
		\forall \src{C},\src{M}\ldotp 
		\\
		\text{if }&\
		\src{M}\vdash \src{C} : \src{rs}
		\\
		\text{ then }&\
		\compgenm{\src{M}}\vdash \compgenm{\src{C}} : \trg{rs}
	\end{align*}
\end{definition}

 \newpage

\section{Compiler from \LU to \LP}\label{sec:compup}

\begin{definition}[Compiler \LU to \LP]\label{def:comp-lu-lp}
	$\compup{\cdot} : \src{C}\to \trg{C}$

	\compup{\src{C}} is defined as follows:
	\begin{align*}
		\tag{\compup{\cdot}-Comp}
		\compup{
			\src{\ell_{root} ; \OB{F} ; \OB{I}}
			,
			\trg{M}
		} &= \trg{
			k_{root} ;
			\compup{\OB{F}}; \compup{\OB{I}}
		}
		\\
		\tag{\compup{\cdot}-Function}
		\compup{
			\src{f(x) \mapsto s;\ret}
		} &= \trg{
			f(x)\mapsto\compup{\src{s}};\ret
		}
		\\
		\tag{\compup{\cdot}-Interfaces}
		\compup{f} &= \trg{f}
		\\
	\end{align*}

	\mytoprule{Expressions}
	\begin{align*}
		\tag{\compup{\cdot}-True} \label{tr:compup-true}
		\compup{
			\trues
		}
		=&\ \trg{0}  
		\\
		\tag{\compup{\cdot}-False} \label{tr:compup-false}
		\compup{
			\falses
		}
		=&\ \trg{1} 
		\\
		\tag{\compup{\cdot}-nat} \label{tr:compup-nat}
		\compup{
			\src{n}
		}
		=&\ \trg{n}
		\\
		\tag{\compup{\cdot}-Var}  \label{tr:compup-var}
		\compup{
			\src{x}
		}
		=&\  \trg{x}
		\\
		\tag{\compup{\cdot}-Loc} \label{tr:compup-loc}
		\compup{
			\src{\ell}
		}
		=&\  
			\begin{aligned}
				&
				\trg{\pair{n,v}}
			\end{aligned}
		\\
		\tag{\compup{\cdot}-Pair}  \label{tr:compup-pair}
		\compup{
			\src{\pair{e_1,e_2}}
		}
		=&\ \trg{\pair{\compup{\src{e_1}},\compup{\src{e_2}}}}
		\\
		\tag{\compup{\cdot}-P1}  \label{tr:compup-p1}
		\compup{
			\src{\projone{e}}
		}
		=&\ \trg{\projone{\compup{\src{e}}}}
		\\
		\tag{\compup{\cdot}-P2}  \label{tr:compup-p2}
		\compup{
			\src{\projtwo{e}}
		}
		=&\ \trg{\projtwo{\compup{\src{e}}}}
		\\
		\tag{\compup{\cdot}-Deref}  \label{tr:compup-deref}
		\compup{
			\src{!e}
		}
		=&\ \trg{
			!\projone{\compup{e}} \with{\projtwo{\compup{e}}}
		}
		\\
		\tag{\compup{\cdot}-op} \label{tr:compup-op}
		\compup{
			\src{e \op e'}
		}
		=&\ \trg{\compup{\src{e}} \op \compup{\src{e'}}}
		\\
		\tag{\compup{\cdot}-cmp} \label{tr:compup-bop}
		\compup{
			\src{e \bop e'}
		}
		=&\ 
				\trg{
					\compup{\src{e}} \bop \compup{\src{e'}}
				}
	\end{align*}

	\mytoprule{Statements}
	\begin{align*}
		\tag{\compup{\cdot}-Skip}\label{tr:compup-skip}
		\compup{
			\skips
		}
		=&\ \skipt  
		\\
		\tag{\compup{\cdot}-Seq}  \label{tr:compup-seq}
		\compup{
			\src{s_u;s}
		}
		=&\ \trg{\compup{\src{s_u}} ; \compup{\src{s}}}
		\\
		\tag{\compup{\cdot}-Letin}  \label{tr:compup-letin}
		\compup{
			\src{\letin{x}{e}{s}}
		}
		=&\ \trg{
				\letin{\trg{x}}{\compup{ \src{e}} }{\compup{ \src{s}}}
			}
		\\
		\tag{\compup{\cdot}-If}  \label{tr:compup-if}
		\compup{
			\src{\ifte{e}{s_t}{s_e}}
		}
		=&\ \trg{
				\ifzte{\compup{\src{e}}
				}
				{\compup{\src{s_t}}
				}{\compup{\src{s_e}}}
		}
		\\
		\tag{\compup{\cdot}-New}  \label{tr:compup-new}
		\compup{
			\src{\letnew{x}{e}{s}}
		}
		=&\
		\trg{
			\begin{aligned}[t]
				&
				\letnewt{\trg{x_{loc}}}{\compup{\src{e}}}{
				\\
				&\
					\lethide{\trg{x_{cap}}}{x_{loc}}{
					\\
					&\ \ 
						\letint{\trg{x} ~}{~ \trg{\pair{x_{loc},x_{cap}}}}{\compup{s}}
					}
				}
			\end{aligned}
		}
		\\
		\tag{\compup{\cdot}-Assign} \label{tr:compup-ass}
		\compup{
			\src{x := e'}
		}
		=&\ \trg{
			\begin{aligned}[t]
				&
				\letint{\trg{x1}~}{~\trg{\projone{x}}}{
				\\
				&\
					\letint{\trg{x2}~}{~\trg{\projtwo{x}}}{
					\\
					&\ \ 
						\trg{x1 :=}\ \compup{e} \with{x2}
					}
				}
			\end{aligned}
		}
		\\
		\tag{\compup{\cdot}-call} \label{tr:compup-call}
		\compup{
			\src{\call{f}~e}
		}
		=&\ \trg{\call{f}~\compup{\src{e}}}
	\end{align*}
\end{definition}
Note that the case for \Cref{tr:compup-new} only works because we are in a sequential setting.
In a concurrent setting an adversary could access \trg{x_{loc}} before it is hidden, so the definition would change.
See \Cref{tr:compap-new} for a concurrent correct implementation.

\begin{align*}
	\compup{\subs{v}{x}} =&\ \subt{\compup{v}}{x}
\end{align*}

\paragraph{Optimisation}\label{sec:compup-opt-deref}
We could optimise \Cref{tr:compup-deref} as follows:
\begin{itemize}
	\item rename the current expressions except dereferencing to \com{b};
	\item reform expressions both in \LA and \LP as $\com{e} ::= \com{b} \mid \com{\letin{x}{b}{e}} \mid \com{!b}$. 
	In the case of \LP it would be $\cdots \mid \trg{!b \with{b}}$.

	This allows expressions to compute e.g., pairs and projections.
	\item rewrite the \Cref{tr:compup-deref} case for compiling \src{!b} into: 

	$\trg{\letin{x}{\compup{b}}{\letin{x1}{\projone{x}}{\letin{x2}{\projtwo{x}}{ !x1 \with{x2} }}}}$.
	\item as expressions execute atomically, this would also scale to the compiler for concurrent languages defined in later sections.
\end{itemize}
We do not use this approach to avoid nonstandard constructs.

\subsection{Properties of the \compup{\cdot} Compiler}

\begin{theorem}[Compiler \compup{\cdot} is \ccomp]\label{thm:comp-up-cc}
	$\vdash\compup{\cdot} : \ccomp$
\end{theorem}

\begin{theorem}[Compiler \compup{\cdot} is \rscomp]\label{thm:comp-up-rsc}
	$\vdash\compup{\cdot} : \rscomp$
\end{theorem}

\subsection{Back-translation from \LP to \LU}\label{sec:general-backtr-trace-based}

\subsubsection{Values Backtranslation}\label{sec:val-bt}

Here is how values are back translated.

\mytoprule{\backtrup{\cdot}:\trg{v}\to\src{v}}
\begin{align*}
	\backtrup{\trg{0}} 
	=&\
	\trues
	\\
	\backtrup{\trg{n}} 
	=&\
	\falses
	&
	\text{if }\trg{n}\neq\trg{0}
	\\
	\backtrup{\trg{n}} 
	=&\
	\src{n}
	&
	\text{where }\src{n}\relatebeta\trg{n}
	\\
	\backtrup{\trg{k}} 
	=&\
	\src{0}
	\\
	\backtrup{\trg{\pair{v,v'}}}
	=&\ 
	\src{\pair{\backtrup{\trg{v}},\backtrup{\trg{v'}}}}
	\\
	\backtrup{\trg{\pair{n,v}}}
	=&\
	\src{\ell}
	&
	\text{where }\src{\ell}\relatebeta\trg{\pair{n,v}}
\end{align*}
\botrule

The backtranslation is nondeterministic, as $\relatebeta$ is not injective.
In this case we cannot make it injective (in the next compiler we can index it by types and make it so but here we do not have them).
This is the reason why the backtranslation algorithm returns a set of contexts, as backtranslating an action that performs \trg{\cl{f}{v~H}} could result in either \src{\cl{f}{\trues~H}} or \src{\cl{f}{0~H}}.
Now depending on \src{f}'s body, which is the component to be compiled, supplying \trues or \src{0} may have different outcomes.
Let us assume that the compilation of \src{f}, when receiving \trg{\cl{f}{v~H}} does not get stuck.
If \src{f} contains \src{\ifte{x}{s}{s'}}, supplying \src{0} will make it stuck.
However, because we generate all possible contexts, we know that we generate also the context that will not cause \src{f} to be stuck.
This is captured in \Cref{thm:action-det} below.

\begin{lemma}[Compiled code steps imply existence of source steps]\label{thm:action-det}
	\begin{align*}
		&\
		\forall
		\\
		\text{if }
		&\
		\src{\Omega''}\relatebeta\trg{\Omega''}
		\\
		&\
		\trg{\Omega''}\Xtot{\alpha?}\trg{C, H \triangleright \compup{s};s'\rho}
		\\
		&\
		\trg{C, H \triangleright \compup{s};s'\rho}\Xtot{\alpha!}\trg{\Omega'} 
		\\	
		&\
		\set{\src{\alpha?}}=\myset{\src{\alpha?}}{\src{\alpha?}\relatebeta\trg{\alpha?}}
		\\
		&\
		\set{\src{\rho}}=\myset{\src{\rho}}{\src{\rho}\relatebeta\trg{\rho}}
		\\
		\text{then }
		&\
		\exists \src{\alpha_j?}\in\set{\src{\alpha?}}, \src{\rho_y}\in\set{\src{\rho}}, \src{C_j, H_j, s_j;s_j'\rho'}.
		\\
		\text{if }
		&\
		\src{\Omega''}\Xtos{\alpha_j?}\src{C_j, H_j \triangleright s_j;s_j'\rho'}
		\\
		\text{then }
		&\
		\src{C_j, H_j \triangleright s_j;s_j'\rho_y} \relatebeta \trg{C, H \triangleright \compup{s};s'\rho}
		\\
		&\
		\src{C_j, H_j \triangleright s_j;s_j'\rho_y}\Xtos{\alpha!}\src{\Omega'} 
		\\
		&\
		\src{\alpha!}\relatebeta\trg{\alpha!}
		\\
		&\
		\src{\Omega'}\relatebeta\trg{\Omega'}
	\end{align*}
\end{lemma}

\subsubsection{Skeleton}\label{sec:backtr-skel}

\mytoprule{\backtrup{\cdot}:\trg{I}\to\src{F}}
\begin{align*}
	\tag{\backtrup{\cdot}-fun}
	\label{tr:backtr-fun}
	\backtrup{\trg{f}}
	=&\
	\src{f(x)\mapsto incrementCounter(); \ret}
\end{align*}
\botrule
Functions call \src{incrementCounter()} before returning to ensure that when a returnback is modelled, the counter is incremented right before returning and not beforehand, as doing so would cause the possible execution of other bactranslated code blocks.
Its implementation is described below.

\mytoprule{\backtrup{\cdot}:\trg{\OB{I}}\to\src{A}}
\begin{align*}
	\tag{\backtrup{\cdot}-skel}
	\label{tr:backtr-skel}
	\backtrup{\trg{\OB{I}}}
	=&\
	\src{
		\begin{aligned}[t]
			&
			\ell_{i} \mapsto 1 ;
			\\
			&
			\ell_{glob} \mapsto 0
			\\
			&
			\src{
				main(x)\mapsto incrementCounter();\ret
			}
			\\
			&
			\src{
				incrementCounter()\mapsto \bl{\text{see below}}
			}
			\\
			&
			\src{
				register(x)\mapsto \bl{\text{see below}}
			}
			\\
			&
			\src{
				update(x)\mapsto \bl{\text{see below}}
			}
			\\
			&
			\backtrup{\trg{f}}
			& 
			\forall \trg{f}\in\trg{\OB{I}}
		\end{aligned}
	}
\end{align*}
\botrule
We assume compiled code does not implement functions \src{incrementCounter}, \src{register} and \src{update}, they could be renamed to not generate conflicts if they were.

The skeleton sets up the infrastructure.
It allocates global locations \src{\ell_i}, which is used as a counter to count steps in actions, and \src{\ell_{glob}}, which is used to keep track of attacker knowledge, as described below.
Then it creates a dummy for all functions expected in the interfaces \trg{\OB{I}} as well as a dummy for the \src{main}.
Dummy functions return their parameter variable and they increment the global counter before that for reasons explained later.

\subsubsection{Single Action Translation}\label{sec:backtr-single-action}
We use the shortcut \trg{ak} to indicate a list of pairs of locations and tag to access them \trg{\OB{\pair{n,\eta}}} that is what the context has access to.
We use functions $.\mtt{loc}$ to access obtain all locations of such a list and $.\mtt{cap}$ to obtain all the capabilities (or \trg{0} when $\trg{\eta}=\trg{\bot}$) of the list.

We use function \src{incrementCounter} to increment the contents of \src{\ell_i} by one.
\begin{align*}
	\src{
		\begin{aligned}
			&
			\src{incrementCounter( ~)} \mapsto
			\\
			&\
			\letins{\src{c}~}{~!\src{\ell_i}}{
				\letins{\src{l}~}{~\ell_{i}}{\src{l:=c+1}}
			}
		\end{aligned}
	}
\end{align*}

Starting from location \src{\ell_{g}} we keep a list whose elements are pairs locations-numbers, we indicate this list as \src{L_{glob}}.

We use function \src{register(\pair{\ell, n})} which adds the pair \src{\pair{\ell,n}} to the list \src{L_{glob}}.
Any time we use this we are sure we are adding a pair for which no other existing pair in \src{L_{glob}} has a second projection equal to \src{n}.
This function can be defined as follows:
\begin{align*}
	\src{
		\begin{aligned}
			&
			\src{register(x)} \mapsto
			\\
			&\
			\letins{\src{xl}~}{~\projone{x}
			\\
			&\
			}{
				\letins{\src{xn}~}{~\projtwo{x}
				\\
				&\ \ 
				}{
					\src{L_{glob} :: \pair{xl,xn}}
				}
			}
		\end{aligned}
	}
\end{align*}
\src{L_{glob}} is a list of pair elements, so it is implemented as a pair whose first projection is an element (a pair) and its second projection is another list; the empty list being \src{0}.
Where \src{::} is a recursive function that starts from \src{\ell_{glob}} and looks for its last element (i.e., it performs second projections until it hits a \src{0}), then replaces that second projection with \src{\pair{\pair{xl,xn},0}}
\begin{lemma}[\src{register(\ell,n)} does not add duplicates for \src{n}]\label{thm:reg-no-dup}
	For \src{n} supplied as parameter by \backtrup{\cdot},
	$\src{C;H\triangleright register(\ell,n)}\xtos{\epsilon}\src{C;H'\triangleright \skips}$ and $\src{\pair{\_,n}}\notin\src{L_{glob}}$
\end{lemma}
\begin{proof}
	Simple analysis of \Cref{tr:backtr-call,tr:backtr-callback-loc,tr:backtr-retback,tr:backtr-ret-loc}.
\end{proof}

We use function \src{update(n, v)} which accesses the elements in the \src{L_{glob}} list, then takes the second projection of the element: if it is \src{n} it updates the first projection to \src{v}, otherwise it continues its recursive call.
If it does not find an element for \src{n}, it gets stuck
\begin{lemma}[\src{update(n,v)} never gets stuck]\label{thm:update-no-stuck}
	$\src{C;H\triangleright update(n,v)}\xtos{\epsilon}\src{C;H'\triangleright \skips}$ for \src{n} and \src{v} supplied as parameters by \backtrup{\cdot} and \src{H'}=\src{H}\subs{\ell\mapsto v}{\ell\mapsto\_} for $\src{\ell}\relatebeta\trg{\pair{n,\_}}$.
\end{lemma}
\begin{proof}
	Simple analysis of \Cref{tr:backtr-call,tr:backtr-retback}.
\end{proof}

We use the meta-level function \fun{reachable}{\trg{H}, \trg{v}, \trg{ak}} that returns a set of pairs $\pair{\trg{n\mapsto v:\eta},\src{e}}$ such that all locations in are reachable from \trg{H} starting from any location in $\trg{ak}\cup\trg{v}$ and that are not already in \trg{ak} and such that \src{e} is a sequence of source-level instructions that evaluate to \src{\ell} such that $\src{\ell}\relatebeta\trg{\pair{n,\_}}$.
\begin{definition}[Reachable]\label{def:reachable}
	\begin{align*}
		\fun{reachable}{\trg{H}, \trg{v}, \trg{ak}} =&\ 
			\myset{ \pair{\trg{n\mapsto v:k},\src{e}} }{
				\begin{aligned}
					&
					\trg{n} \in \fun{reach}{ \trg{n_{st}}, \trg{k_{st}}, \trg{H} } 
					\\
					&
					\text{where } \trg{n_{st}}\in\trg{v}\cup\trg{ak}.\mtt{loc} 
					\\
					&
					\text{and } \trg{k_{st}}\in\trg{k_{root}}\cup\trg{ak}.\mtt{cap}
					\\
					&
					\text{and } \trg{n\mapsto v:k}\in\trg{H}
					\\
					&
					\text{and } \trg{H \triangleright !e \redtot !n \with{k}}
					\\
					&
					\text{and } \forall\src{H}\ldotp \src{H}\relatebeta\trg{H}
					\\
					&
					\src{H \triangleright \backtrup{\trg{e}} \redtos \ell}
					\\
					&
					\text{and } (\src{\ell},\trg{n,k})\in\beta
				\end{aligned}
			}
	\end{align*}
\end{definition}

Intuitively, \fun{reachable}{\cdot} finds out which new locations have been allocated by the compiled component and that are now reachable by the attacker (the first projection of the pair, \trg{n\mapsto v:\eta}).
Additionally, it tells how to reach those locations in the source so that we can \src{register(\cdot)} them for the source attacker (the backtranslated context) to access.

In this case we know by definition that \trg{e} can only contain one \trg{!} and several \trg{\projone{\cdot}} or \trg{\projtwo{\cdot}}.
The base case for values is as before.

\mytoprule{\backtrup{\cdot}:\trg{e}\to\src{e}}
\begin{align*}
	\backtrup{\trg{!e}} &= \src{!\backtrup{\trg{e}}}
	\\
	\backtrup{\trg{\projone{e}}} &= \src{\projone{\backtrup{\trg{e}}}}
	\\
	\backtrup{\trg{\projtwo{e}}} &= \src{\projtwo{\backtrup{\trg{e}}}}
\end{align*}
\botrule

The next function takes the following inputs: an action, its index, the previous function's heap, the previous attacker knowledge and the stack of functions called so far.
It returns a set of: code, the new attacker knowledge, its heap, the stack of functions called and the function where the code must be put.
In the returned parameters, the attacker knowledge, the heap and the stack of called functions serve as input to the next call.

\mytoprule{ \backtrup{\cdot}:\trg{\alpha}\times n\in\mb{N} \times \trg{H} \times \trg{\OB{n\times\eta}} \times \src{\OB{f}} \to \set{\src{s} \times \trg{\OB{n\times\eta}} \times \trg{H} \times \src{\OB{f}} \times \src{f}} }
\begin{align*}
	\tag{\backtrup{\cdot}-call}
	\label{tr:backtr-call}
	\backtrup{
		\begin{aligned}
			&
			\trg{\cl{f}{v~ H}},
			\\
			&
			n, \trg{H_{pre}}, \trg{ak}, \src{\OB{f}}	
		\end{aligned}
	}
	=&\
	\myset{
		\begin{aligned}
			&
			\src{
				\left(\begin{aligned}
					&
					\iftes{!\ell_i == n
					}{
						\\
						&\ \ 
						\src{incrementCounter()}
						\\
						&\ \ 
						\letnews{\src{x1}~}{~\src{v_1}}{\src{register(\pair{x1,n_1})}}
						\\
						&\ \ 
						\cdots
						\\
						&\ \ 
						\letnews{\src{xj}~}{~\src{v_j}}{\src{register(\pair{xj,n_j})}}
						\\
						&\ \ 
						\src{update(m_1, u_1) }
						\\
						&\ \ 
						\cdots
						\\
						&\ \ 
						\src{update(m_l, u_l) }
						\\
						&\ \ 
						\src{\calls{f}{~v}}
						\\
						&
					}{
					\skips
					}
				\end{aligned}\right)
			} 
			\\
			&
			, \trg{ak'} , \trg{H}, \src{f;\OB{f}}, \src{f'}
		\end{aligned}
	}{
		\begin{aligned}
			&
			\forall 
			\\
			&
			\src{v_1}=\backtrup{\trg{v_1}}
			\\
			&
			\cdots
			\\
			&
			\src{v_j}=\backtrup{\trg{v_j}}
			\\
			&
			\src{u_1}=\backtrup{\trg{u_1}}
			\\
			&
			\cdots
			\\
			&
			\src{u_l}=\backtrup{\trg{u_l}}
			\\
			&
			\src{v}=\backtrup{\trg{v}}
		\end{aligned}
	}
	\\
	\text{where }
	&
	\trg{H}\setminus\trg{H_{pre}} = \trg{H_{n}}
	\\
	& 
	\trg{H_n} = \trg{n_1\mapsto v_1:\eta_1, \cdots, n_j\mapsto v_j:\eta_j}
	\\
	\text{and }
	&
	\trg{H}\cap\trg{H_{pre}} = \trg{H_c}
	\\
	&
	\trg{H_c} = \trg{m_1\mapsto u_1:\eta_1', \cdots, m_l\mapsto u_l:\eta_l'}
	\\
	\text{and }
	&
	\trg{ak'} = \trg{ak} , \trg{\pair{n_1,\eta_1}, \cdots, \pair{n_j,\eta_j}}
	\\
	\text{and }
	&
	\src{\OB{f}} = \src{f'\OB{f'}}
	\\
	\\
	\tag{\backtrup{\cdot}-callback-loc}
	\label{tr:backtr-callback-loc}
	\backtrup{
		\begin{aligned}
			&
			\trg{\cb{f}{v~ H}},
			\\
			&
			n, \trg{H_{pre}}, \trg{ak}, \src{\OB{f}}
		\end{aligned}
	}
	=&\
	\set{
		\src{
			\left(\begin{aligned}
				&
				\iftes{!\ell_i == n
				}{
					\\
					&\ \ 
					\src{incrementCounter()}
					\\
					&\ \ 
					\letins{\src{l1}~}{~\src{e_1}}{\src{register(\pair{l1, n_1})}}
					\\
					&\ \ 
					\cdots
					\\
					&\ \ 
					\letins{\src{lj}~}{~\src{e_j}}{\src{register(\pair{lj, n_j})}}
					\\
					&
				}{
				\skips
				}
			\end{aligned}\right)}
		, \trg{ak'} , \trg{H}, \src{f;\OB{f}}, \src{f}
	}
	\\
	\text{if }
	&
	\fun{reachable}{\trg{H},\trg{v},\trg{ak}} = \pair{\trg{n_1\mapsto v_1:\eta_1},\src{e_1}},\cdots,\pair{\trg{n_j\mapsto v_j:\eta_j},\src{e_j}}
	\\
	\text{and }
	&
	\trg{ak'} = \trg{ak} , \trg{\pair{n_1,\eta_1}, \cdots, \pair{n_j,\eta_j}}
	\\
	\\
	\tag{\backtrup{\cdot}-retback}
	\label{tr:backtr-retback}
	\backtrup{
		\begin{aligned}
			&
			\trg{\rb{H}},
			\\
			&
			n, \trg{H_{pre}}, \trg{ak}, \src{f;\OB{f}}
		\end{aligned}}
	=&\
	\myset{
		\begin{aligned}
			&
			\src{
				\left(\begin{aligned}
					&
					\iftes{!\ell_i == n
					}{
						\\
						&\ \ 
						\bl{\small// \text{\textit{no \src{incrementCounter()} as explained}}}
						\\
						&\ \ 
						\letnews{\src{x1}~}{~\src{v_1}}{\src{register(\pair{x1,n_1})}}
						\\
						&\ \ 
						\cdots
						\\
						&\ \ 
						\letnews{\src{xj}~}{~\src{v_j}}{\src{register(\pair{xj,n_j})}}
						\\
						&\ \ 
						\src{update(m_1, u_1) }
						\\
						&\ \ 
						\cdots
						\\
						&\ \ 
						\src{update(m_l, u_l) }
						\\
						&
					}{
					\skips
					}
				\end{aligned}\right)
			} 
			\\
			&
			, \trg{ak'} , \trg{H}, \src{\OB{f}}, \src{f}
		\end{aligned}
	}{
		\begin{aligned}
			&
			\forall 
			\\
			&
			\src{v_1}=\backtrup{\trg{v_1}}
			\\
			&
			\cdots
			\\
			&
			\src{v_j}=\backtrup{\trg{v_j}}
			\\
			&
			\src{u_1}=\backtrup{\trg{u_1}}
			\\
			&
			\cdots
			\\
			&
			\src{u_l}=\backtrup{\trg{u_l}}
		\end{aligned}
	}
	\\
	\text{where }
	&
	\trg{H}\setminus\trg{H_{pre}} = \trg{H_{n}}
	\\
	& 
	\trg{H_n} = \trg{n_1\mapsto v_1:\eta_1, \cdots, n_j\mapsto v_j:\eta_j}
	\\
	\text{and }
	&
	\trg{H}\cap\trg{H_{pre}} = \trg{H_c}
	\\
	&
	\trg{H_c} = \trg{m_1\mapsto u_1:\eta_1', \cdots, m_l\mapsto u_l:\eta_l'}
	\\
	\text{and }
	&
	\trg{ak'} = \trg{ak} , \trg{\pair{n_1,\eta_1}, \cdots, \pair{n_j,\eta_j}}
	\\
	\\
	\tag{\backtrup{\cdot}-ret-loc}
	\label{tr:backtr-ret-loc}
	\backtrup{
		\begin{aligned}
			&
			\trg{\rt{H}},
			\\
			&
			n, \trg{H_{pre}}, \trg{ak}, \src{f;\OB{f}}		
		\end{aligned}
	}
	=&\
	\set{
		\src{
			\left(\begin{aligned}
				&
				\iftes{!\ell_i == n
				}{
					\\
					&\ \ 
					\src{incrementCounter()}
					\\
					&\ \ 
					\letins{\src{l1}~}{~\src{e_1}}{\src{register(\pair{l1, n_1})}}
					\\
					&\ \ 
					\cdots
					\\
					&\ \ 
					\letins{\src{lj}~}{~\src{e_j}}{\src{register(\pair{lj, n_j})}}
					\\
					&
				}{
				\skips
				}
			\end{aligned}\right)}
		, \trg{ak'} , \trg{H}, \src{\OB{f}}, \src{f'}
	}
	\\
	\text{if }
	&
	\fun{reachable}{\trg{H},\trg{0},\trg{ak}} = \pair{\trg{n_1\mapsto v_1:\eta_1},\src{e_1}},\cdots,\pair{\trg{n_j\mapsto v_j:\eta_j},\src{e_j}}
	\\
	\text{and }
	&
	\trg{ak'} = \trg{ak} , \trg{\pair{n_1,\eta_1}, \cdots, \pair{n_j,\eta_j}}
	\\
	\text{and }
	&
	\src{\OB{f}} = \src{f'\OB{f'}}
\end{align*}
\botrule

This is the back-translation of functions.
Each action is wrapped in an if statement checking that the action to be mimicked is that one (the same function may behave differently if called twice and we need to ensure this).
After the if, the counter checking for the action index \src{\ell_i} is incremented.
This is not done in case of a return immediately, but only just before the return itself, so the increment is added in the skeleton already.
(there could be a callback to the same function after the return and then we wouldn't return but execute the callback code instead)

When back-translating a ?-decorated, we need to set up the heap correctly before the call itself.
That means calculating the new locations that this action allocated (\trg{H_n}), allocating them and registering them in the \src{L_{glob}} list via the \src{register(\cdot)} function.
These locations are also added to the attacker knowledge \trg{ak'}.
Then we need to update the heap locations we already know of.
These locations are \trg{H_c} and as we know them already, we use the \src{update(\cdot)} function.

When back-translating a !-decorated action we need to calculate what part of the heap we can reach from there, and so we rely on the \fun{reachable}{\cdot} function to return a list of pairs of locations \trg{n} and expressions \src{e}.
We use \trg{n} to expand the attacker knowledge \trg{ak'} as these locations are now reachable.
We use \trg{e} to reach these locations in the source heap so that we can \src{register} them and ensure they are accessible through \src{L_{glob}}.

Finally, we use parameter \src{\OB{f}} to keep track of the call stack, so making a call to \src{f} pushes \src{f} on the stack (\src{f;\OB{f}}) and making a return pops a stack \src{f;\OB{f}} to \src{\OB{f}}.
That stack carries the information to instantiate the \src{f} in the return parameters, which is the location where the code needs to be allocated.

\mytoprule{ \backtrup{ \cdot} : \trg{\OB{\alpha}} \times n\in\mb{N} \times \trg{H} \times \trg{\OB{n\times\eta}} \times \src{\OB{f}} \to \set{\src{\OB{s,f}}} }
\begin{align*}
	\tag{\backtrup{\cdot}-listact-b}
	\label{tr:backtr-listact-b}
	\backtrup{\trge}
	=&\
	\srce
	\\
	\tag{\backtrup{\cdot}-listact-i}
	\label{tr:backtr-listact}
	\backtrup{\trg{\alpha \OB{\alpha}}, n, \trg{H_{pre}}, \trg{ak}, \src{\OB{f}} }
	=&\ 
	\myset{\src{s,f};\src{\OB{s,f}}}{
		\begin{aligned}
			&
			\src{s},\trg{ak'},\trg{H'},\src{\OB{f'}},\src{f}=%
			\backtrup{\trg{\alpha}, n, \trg{H_{pre}}, \trg{ak}, \src{\OB{f}} } 
			\\
			&\
			\src{\OB{s,f}} \in %
			\backtrup{\trg{\OB{\alpha}}, n+1, \trg{H'}, \trg{ak'}, \src{\OB{f'}} }		
		\end{aligned}
	}
\end{align*}
\botrule
This recursive call ensures the parameters are passed around correctly.
Note that each element in a set returned by the single-action back-translation has the same \trg{ak'}, \trg{H} and \src{f'}, the only elements that change are in the code \src{s} due to the backtranslation of values.
Thus the recursive call can pass those parameters taken from any element of the set.

\subsubsection{The Back-translation Algorithm \backtrup{\cdot}}
\mytoprule{ \backtrup{\cdot}: \trg{\OB{I}} \times \trg{\OB{\alpha}} \to \set{\src{A}} }
\begin{align*}
	\tag{\backtrup{\cdot}-main}
	\label{tr:backtr-main}
	\backtrup{ \trg{\OB{I}}, \trg{\OB{\alpha}} }
	=&\
	\myset{
		\src{A}
	}{ 
		\begin{aligned}
			&
			\src{A} = \src{A_{skel}}\Join\src{\OB{s,f}} 
			\\
			&
			\text{for all } \src{\OB{s,f}} \in \set{\src{\OB{s,f}}} 
			\\
			&
			\text{where } \set{\src{\OB{s,f}}} = \backtrup{ \trg{\OB{\alpha}}, 1, \trg{H_0}, \trge, \src{main} }
			\\
			&
			\trg{H_0} = \trg{0\mapsto0:k_{root}}
			\\
			&
			\src{A_{skel}} = \backtrup{ \trg{\OB{I}}}
		\end{aligned}
	}
\end{align*}
\botrule
This is the real back-translation algorithm: it calls the skeleton and joins it with each element of the set returned by the trace back-translation.

\mytoprule{ \Join : \src{A} \times \src{\OB{s,f}} \to \src{A}}
\begin{align*}
	\tag{\backtrup{\cdot}-join}
	\label{tr:backtr-join}
	\src{A} \Join \srce
	=&\ 
	\src{A}
	\\
	\src{H ; F_1;\cdots;F;\cdots;F_n} \Join \src{\OB{s,f};s,f }
	=&\
	\src{H ; F_1;\cdots;F';\cdots;F_n} \Join \src{\OB{s,f}}
	\\
	\text{where }
	&
	\src{F} = \src{f(x)\mapsto s';\ret}
	\\
	&
	\src{F'} = \src{f(x)\mapsto s;s';\ret}
\end{align*}
\botrule
When joining we add from the last element of the list so that the functions we create have the concatenation of if statements (those guarded by the counter on \src{\ell_i}) that are sorted (guards with a test for \src{\ell_i = 4} are before those with a test \src{\ell_i = 5}).

\subsubsection{Correctness of the Back-translation}\label{sec:corr-backtrans}
\begin{theorem}[\backtrup{\cdot} is correct]\label{thm:backtr-corr}
	\begin{align*}
		&
		\forall
		\\
		\text{if }
		&
		\trg{\SInitt{A\hole{\compup{C}}}} \Xtot{\OB{\alpha}} \trg{\Omega}
		\\
		&\ 
		\trg{\Omega} \Xtot{\epsilon}\trg{\Omega'}
		\\
		&
		\trg{\OB{I}} = \fun{names}{\trg{A}}
		\\
		&
		\trg{\OB{\alpha}}\equiv\trg{\OB{\alpha'}\cdot\alpha?}
		\\
		&
		\src{\ell_i;\ell_{glob}}\notin\beta
		\\
		\text{ then }
		&
		\exists \src{A}\in \backtrup{\trg{I},\trg{\OB{\alpha}}}
		\\
		\text{ such that }
		&\
		\src{\SInits{A\hole{C}}}\Xtos{\OB{\alpha}} \src{\Omega}
		\\
		\text{ and }
		&
		\src{\OB{\alpha}}\relatebeta\trg{\OB{\alpha}}
		\\
		&
		\src{\Omega}\relatebeta\trg{\Omega}
		\\
		&
		\src{\Omega}.\src{H}.\src{\ell_i}=\card{\trg{\OB{\alpha}}}+1
	\end{align*}
\end{theorem}

The back-translation is correct if it takes a target attacker that will reduce to a state together with a compiled component and it produces a set of source attackers such that one of them, that together with the source component will reduce to a related state performing related actions.
Also it needs to ensure the step is incremented correctly.

\subsubsection{Remark on the Backtranslation}\label{sec:remark-bt}
Some readers may wonder whether the hassle of setting up a source-level representation of the whole target heap is necessary.
Indeed for those locations that are allocated by the context, this is not.
If we changed the source semantics to have an oracle that predicts what a \src{\letnew{x}{e}{s}} statement will return as the new location, we could simplify this.
In fact, currently the backtranslation stores target locations in the list \src{L_{glob}} and looks them up based on their target name, as it does not know what source name will be given to them.
The oracle would obviate this problem, so we could hard code the name of these locations, knowing exactly the identifier that will be returned by the allocator.
For the functions to be correct in terms of syntax, we would need to pre-emptively allocate all the locations with that identifier so that their names are in scope and they can be referred to.

However, the problem still persists for locations created by the component, as their names cannot be hard coded, as they are not in scope.
Thus we would still require \src{reach} to reach these locations, \src{register} to add them to the list and \src{update} to update their values in case the attacker does so.

Thus we simplify the scenario and stick to a more standard, oracle-less semantics and to a generalised approach to location management in the backtranslation.

 \newpage
\section{The Source Language: \LA}\label{sec:src}
This is an imperative, concurrent while language with monitors.

\begin{align*}
	\mi{Whole\ Programs}~\src{P} \bnfdef&\ \src{\Delta ; H ; \OB{F} ; \OB{I}}
	\\
	\mi{Components}~\src{C} \bnfdef&\ \src{\Delta ; \OB{F} ; \OB{I}}
	\\
	\mi{Contexts}~\src{A} \bnfdef&\ \src{H ; \OB{F}\hole{\cdot}}
	\\
	\mi{Interfaces}~\src{I} \bnfdef&\ \src{f}
	\\
	\mi{Functions}~\src{F} \bnfdef&\ \src{f(x:\tau)\mapsto s;\ret}
	\\
	\mi{Operations}~\src{\op} \bnfdef&\ \src{+} \mid \src{-}
	\\
	\mi{Comparison}~\src{\bop} \bnfdef&\ \src{==} \mid \src{<} \mid \src{>}
	\\
	\mi{Values}~\src{v} \bnfdef&\ %
	\src{b}\in\{\src{\trues},\src{\falses}\} \mid \src{n}\in\mb{N} \mid \src{\pair{v,v}} \mid \src{\ell}
	\\
	\mi{Expressions}~\src{e} \bnfdef&\ \src{x} \mid \src{v} \mid \src{e \op e} \mid \src{e \bop e} \mid \src{!e} \mid \src{\pair{e,e}} \mid \src{\projone{e}} \mid \src{\projtwo{e}}
	\\
	\mi{Statements}~\src{s} \bnfdef&\ \skips \mid \src{s;s} \mid \src{\letin{x:\tau}{e}{s}}  \mid \src{\ifte{e}{s}{s}} 
	\\
	\mid&\
	\src{x := e} \mid \src{\letnewty{x}{e}{\tau}{s}}\mid \src{\call{f}~e}
	\\
	\mid&\ \hl{ \src{\fork{s}} }
	\mid \hl{ \src{\myendorse{x}{e}{\varphi}{s}} }
	\\
	\mi{Types}~\hl{\src{\tau}} \bnfdef&\ %
	\src{\Bools} \mid \src{\Nats} \mid \src{\tau\times\tau} \mid \src{\Refs{\tau}} \mid \UNS
	\\
	\mi{Superficial\ Types}~\hl{\src{\varphi}} \bnfdef&\ %
	\src{\Bools} \mid \src{\Nats} \mid \src{\UNS\times\UNS} \mid \src{\Refs{\UNS}}
	\\
	\mi{Eval.\ Ctxs.}~\src{E} \bnfdef&\ \src{\hole{\cdot}} \mid \src{e \op E} \mid \src{E \op n} \mid \src{e \bop E} \mid \src{E \bop n} 
	\\
	\mid&\ \src{!E} \mid \src{\pair{e,E}} \mid \src{\pair{E,v}} \mid \src{\projone{E}} \mid \src{\projtwo{E}}
	\\
	\mi{Heaps}~\src{H} \bnfdef&\ \srce \mid \hl{\src{H ; \ell\mapsto v:\tau}}
	\\
	\mi{Monitors}~\hl{\src{M}} \bnfdef&\ \src{(\set{\sigma},\monred,\sigma_0,\Delta,\sigma_c)}
	\\
	\mi{Mon.\ States}~\src{\sigma} \in&\ \src{\mc{S}}
	\\
	\mi{Mon.\ Reds.}~\hl{\src{\monred}} \bnfdef&\ \srce \mid \src{\monred;(s,s)} %
	\\
	\mi{Environments}~\hl{\src{\Gamma}},\src{\Delta} \bnfdef&\ \srce \mid \src{\Gamma; (x:\tau)} %
	\\
	\mi{Store\ Env.}~\hl{\src{\Delta}} \bnfdef&\ \srce \mid \src{\Delta; (\ell:\tau)}
	\\
	\mi{Substitutions}~\src{\rho} \bnfdef&\ \srce \mid \src{\rho}\subs{v}{x}
	\\
	\mi{Processes}~ \hl{\src{\pi}} \bnfdef&\ \src{\proc{s}{\OB{f}}}
	\\
	\mi{Soups}~\hl{\src{\Pi}} \bnfdef&\ \srce \mid \src{\Pi \parallel \pi} 
	\\
	\mi{Prog.\ States}~\src{\Omega}\bnfdef&\ \src{C, H\triangleright \Pi} 
	\\
	\mi{Labels}~\src{\lambda} \bnfdef&\ \src{\epsilon} \mid \src{\alpha}%
	\\
	\mi{Actions}~\hl{\src{\alpha}} \bnfdef&\ \src{\cl{f}{v}} \mid \src{\cb{f}{v}} \mid \src{\rt{}} \mid \src{\rb{}}
	\\
	\mi{Traces}~\src{\OB{\alpha}} \bnfdef&\ \srce \mid \src{\OB{\alpha}\cdot\alpha}
\end{align*}
We \hl{highlight} elements that have changed from \LU.

\subsection{Static Semantics of \LA}\label{sec:src-types}
The static semantics follows these typing judgements.
\begin{align*}
	&\vdash \src{C}:\UNS
	&&\text{Component \src{C} is well-typed.}
	\\
	&\src{C}\vdash\src{F} : \src{\tau} 
	&&\text{Function \src{F} takes arguments of type \src{\tau} under component \src{C}.}
	\\
	&\src{\Delta,\Gamma}\vdash\diamond 
	&&\text{Environments \src{\Gamma} and \src{\Delta} are well-formed.}
	\\
	&\src{\Delta}\vdash\src{ok} 
	&&\text{Environment \src{\Delta} is safe.}
	\\
	&\src{\tau}\vdash\circ 
	&& \text{Type \src{\tau} is insecure.}
	\\
	&\src{\Delta,\Gamma}\vdash \src{e} : \src{\tau}
	&&\text{Expression \src{e} has type \src{\tau} in\src{\Gamma}.}
	\\
	&\src{C,\Delta,\Gamma}\vdash \src{s} %
	&&\text{Statement \src{s} is well-typed in \src{C} and \src{\Gamma}.}
	\\
	&\src{C,\Delta,\Gamma}\vdash \src{\pi} %
	&&\text{Single process \src{\pi} is well-typed in \src{C} and \src{\Gamma}.}
	\\
	&\src{C,\Delta,\Gamma}\vdash \src{\Pi} 
	&&\text{Soup \src{\Pi} is well-typed in \src{C} and \src{\Gamma}.}
	\\
	&\vdash\src{H}:\src{\Delta} 
	&& \text{Heap \src{H} respects the typing of \src{\Delta}.}
	\\
	&\vdash\src{M}
	&& \text{Monitor \src{M} is valid.}
\end{align*}

\subsubsection{Auxiliary Functions}
We rely on these standard auxiliary functions: \fun{names}{\cdot} extracts the defined names (e.g., function and interface names).
\fun{fv}{\cdot} returns free variables while \fun{fn}{\cdot} returns free names (i.e., a call to a defined function).
\fun{dom}{\cdot} returns the domain of a particular element (e.g., all the allocated locations in a heap).
We denote access to the parts of \src{C} and \src{P} via functions $\mtt{.funs}$, $\mtt{.intfs}$ and $\mtt{.mon}$.
We denote access to parts of \src{M} with a dot notation, so $\src{M}.\src{\Delta}$ means \src{\Delta} where $\src{M}=\src{(\set{\sigma},\monred,\sigma_0,\Delta,\sigma_c)}$.

\subsubsection{Typing Rules}
\mytoprule{\vdash \src{C}}
\begin{center}
	\typerule{T\LA-component}{
		\src{C}\equiv\src{\Delta ; \OB{F} ; \OB{I}} 
		&
		\src{C}\vdash\src{\OB{F}}:\src{\UNS}
		&
		\fun{names}{\src{\OB{F}}}\cap\fun{names}{\src{\OB{I}}}=\srce
		&
		\src{\Delta}\vdash\src{ok}
	}{
		\vdash \src{C}:\UNS
	}{ts-co}
\end{center}
\botrule
\mytoprule{\src{C}\vdash\src{F} : \src{\UNS}}
\begin{center}
	\typerule{T\LA-function}{
		\src{F}\equiv \src{f(x:\UNS)\mapsto s;\ret}
		&
		\src{C},\src{\Delta;x:\UNS}\vdash \src{s} %
		\\
		\src{C}\equiv\src{\Delta ; \OB{F} ; \OB{I}} 
		&
		\forall\src{f}\in\fun{fn}{s}, \src{f}\in\dom{\src{C}.\mtt{funs}} \vee \src{f}\in\dom{\src{C}.\mtt{intfs}}
	}{
		\src{C}\vdash\src{F}:\src{\UNS} 
	}{ts-fu}
\end{center}
\botrule
\mytoprule{\src{\Delta,\Gamma}\vdash\diamond}
\begin{center}
	\typerule{T\LA-env-e}{
	}{
		\srce;\srce\vdash\diamond
	}{ts-e-e}
	\typerule{T\LA-env-var}{
		\src{\Delta,\Gamma}\vdash\diamond
		&
		\src{x}\notin \dom{\src{\Gamma}}
	}{
		\src{\Delta,\Gamma;(x:\tau)}\vdash\diamond
	}{ts-e-v}
	\typerule{T\LA-env-loc}{
		\src{\Delta,\Gamma}\vdash\diamond
		&
		\src{l}\notin \dom{\src{\Delta}}
	}{
		\src{\Delta;(l:\tau);\Gamma}\vdash\diamond
	}{ts-e-l}
\end{center}
\botrule
\mytoprule{\src{\Delta,\Gamma}\vdash\src{ok}}
\begin{center}
	\typerule{T\LA-safe-e}{
	}{
		\srce\vdash\src{ok}
	}{ts-s-e}
	\typerule{T\LA-safe-loc}{
		\src{\Delta}\vdash\src{ok}
		&
		\src{l}\notin \dom{\src{\Gamma}}
		&
		\UNS\notin\src{\tau}
	}{
		\src{\Gamma;(l:\tau)}\vdash\src{ok}
	}{ts-s-l}
\end{center}
\botrule
\mytoprule{\src{\Delta,\Gamma}\vdash\UNS}
\begin{center}
	\typerule{T\LA-env-e}{
	}{
		\srce\vdash\UNS
	}{ts-e-e}
	\typerule{T\LA-env-var}{
		\src{\Delta,\Gamma}\vdash\UNS
		&
		\src{x}\notin \dom{\src{\Gamma}}
	}{
		\src{\Gamma,(x:\UNS)}\vdash\UNS
	}{ts-e-v}
	\typerule{T\LA-env-loc}{
		\src{\Delta,\Gamma}\vdash\UNS
		&
		\src{l}\notin \dom{\src{\Gamma}}
	}{
		\src{\Gamma,(l:\UNS)}\vdash\UNS
	}{ts-e-l}
\end{center}
\botrule
\mytoprule{\src{\tau}\vdash\circ}
\begin{center}
	\typerule{T\LA-bool-pub}{
	}{
		\src{\Bools}\vdash\circ
	}{ts-ts-b}
	\typerule{T\LA-nat-pub}{
	}{
		\src{\Nats}\vdash\circ
	}{ts-ts-b}
	\typerule{T\LA-pair-pub}{
		\src{\tau}\vdash\circ
		&
		\src{\tau'}\vdash\circ
	}{
		\src{\tau\times\tau'}\vdash\circ
	}{ts-ts-p}
	\typerule{T\LA-un-pub}{
	}{
		\UNS\vdash\circ
	}{ts-tp-f}
	\typerule{T\LA-references-pub}{
	}{
		\src{\Refs{\UNS}}\vdash\circ
	}{ts-tp-f}
\end{center}
\botrule
\mytoprule{\src{\Delta,\Gamma}\vdash \src{e} : \src{\tau}}
\begin{center}
	\typerule{T\LA-true}{
		\src{\Delta,\Gamma}\vdash\diamond
	}{
		\src{\Delta,\Gamma}\vdash\trues:\Bools
	}{ts-true}
	\typerule{T\LA-false}{
		\src{\Delta,\Gamma}\vdash\diamond
	}{
		\src{\Delta,\Gamma}\vdash\falses:\Bools
	}{ts-false}
	\typerule{T\LA-nat}{
		\src{\Delta,\Gamma}\vdash\diamond
	}{
		\src{\Delta,\Gamma}\vdash\src{n}:\Nats
	}{ts-nat}
	\typerule{T\LA-var}{
		\src{x:\tau}\in\src{\Gamma}
	}{
		\src{\Delta,\Gamma}\vdash\src{x}:\src{\tau}
	}{ts-var}
	\typerule{T\LA-loc}{
		\src{l:\tau}\in\src{\Delta}
	}{
		\src{\Delta,\Gamma}\vdash\src{l}:\src{\Refs{\tau}}
	}{ts-loc}
	\typerule{T\LA-pair}{
		\src{\Delta,\Gamma}\vdash \src{e_1} : \src{\tau}
		\\
		\src{\Delta,\Gamma}\vdash \src{e_2} : \src{\tau'} 
	}{
		\src{\Delta,\Gamma}\vdash \src{\pair{e_1,e_2}} : \src{\tau\times\tau'}
	}{ts-pair}
	\typerule{T\LA-proj-1}{
		\src{\Delta,\Gamma}\vdash \src{e} : \src{\tau\times\tau'}
	}{
		\src{\Delta,\Gamma}\vdash \src{\projone{e}} : \src{\tau}
	}{ts-p1}
	\typerule{T\LA-proj-2}{
		\src{\Delta,\Gamma}\vdash \src{e} : \src{\tau\times\tau'}
	}{
		\src{\Delta,\Gamma}\vdash \src{\projtwo{e}} : \src{\tau'}
	}{ts-p2}
	\typerule{T\LA-dereference}{
		\src{\Delta,\Gamma}\vdash \src{e} : \src{\Refs{\tau}} 
	}{
		\src{\Delta,\Gamma}\vdash \src{!e} : \src{\tau}  
	}{ts-deref}
	\typerule{T\LA-op}{
		\src{\Delta,\Gamma}\vdash \src{e} : \Nats
		&
		\src{\Delta,\Gamma}\vdash \src{e'} : \Nats
	}{
		\src{\Delta,\Gamma}\vdash \src{e \op e'} : \src{\Nats}  
	}{ts-op}
	\typerule{T\LA-cmp}{
		\src{\Delta,\Gamma}\vdash \src{e} : \Nats
		&
		\src{\Delta,\Gamma}\vdash \src{e'} : \Nats
	}{
		\src{\Delta,\Gamma}\vdash \src{e \bop e'} : \Bools
	}{ts-bop}
	\typerule{T\LA-coercion}{
		\src{C,\Delta,\Gamma}\vdash \src{e} : \src{\tau} 
		&
		\src{\tau} \vdash \circ
	}{
		\src{C,\Delta,\Gamma}\vdash \src{e} : \UNS 
	}{ts-coe}
\end{center}
\botrule
\mytoprule{\src{C,\Delta,\Gamma}\vdash \src{s} }
\begin{center}
	\typerule{T\LA-skip}{
	}{
		\src{C,\Delta,\Gamma}\vdash \skips
	}{ts-skip}	
	\typerule{T\LA-function-call}{
		((\src{f}\in\dom{\src{C}.\mtt{funs}})\vee(\src{f}\in\dom{\src{C}.\mtt{intfs}}))
		\\
		\src{\Delta,\Gamma}\vdash \src{e} : \src{\UNS}
	}{
		\src{\Delta,\Gamma}\vdash \src{\call{f}~e} %
	}{ts-fun}
	\typerule{T\LA-sequence}{
		\src{C,\Delta,\Gamma}\vdash \src{s_u} %
		\\
		\src{C,\Delta,\Gamma}\vdash \src{s} %
	}{
		\src{C,\Delta,\Gamma}\vdash \src{s_u;s} %
	}{ts-seq}
	\typerule{T\LA-letin}{
		\src{\Delta,\Gamma}\vdash \src{e} : \src{\tau} 
		\\
		\src{C,\Gamma;x:\tau}\vdash \src{s} %
	}{
		\src{C,\Delta,\Gamma}\vdash \src{\letin{x:\tau}{e}{s}} %
	}{ts-vardef}
	\typerule{T\LA-assign}{
		\src{\Delta,\Gamma}\vdash \src{x} : \src{\Refs{\tau}}
		\\
		\src{\Delta,\Gamma}\vdash \src{e'} : \src{\tau} 
	}{
		\src{C,\Delta,\Gamma}\vdash \src{x := e'} %
	}{ts-ass}
	\typerule{T\LA-new}{
		\src{\Delta,\Gamma}\vdash \src{e} : \src{\tau} 
		\\
		\src{C,\Gamma;x:\Refs{\tau}}\vdash \src{s} %
	}{
		\src{C,\Delta,\Gamma}\vdash \src{\letnewty{x}{e}{\tau}{s}} %
	}{ts-new}
	\typerule{T\LA-if}{
		\src{\Delta,\Gamma}\vdash \src{e} : \src{\Bool}
		\\
		\src{C,\Delta,\Gamma}\vdash \src{s_t} %
		&
		\src{C,\Delta,\Gamma}\vdash \src{s_e} %
	}{
		\src{C,\Delta,\Gamma}\vdash \src{\ifte{e}{s_t}{s_e}} %
	}{ts-if}
	\typerule{T\LA-fork}{
		\src{C,\Delta,\Gamma}\vdash \src{s} %
	}{
		\src{C,\Delta,\Gamma}\vdash \src{\fork{s}} %
	}{ts-fork}
	\typerule{T\LA-endorse}{
		\src{\Delta,\Gamma}\vdash\src{e}:\UNS
		&
		\src{C,\Delta,\Gamma;(x:\varphi)} \vdash\src{s} %
	}{
		\src{C,\Delta,\Gamma}\vdash \src{\myendorse{x}{e}{\varphi}{s}} %
	}{ts-end}
\end{center}
\botrule

\mytoprule{\src{C,\Delta,\Gamma}\vdash \src{\pi} }
\begin{center}
	\typerule{T\LA-process}{
		\src{C,\Delta,\Gamma}\vdash \src{s} %
	}{
		\src{C,\Delta,\Gamma}\vdash \src{\proc{s}{\OB{f}}} %
	}{ts-proc}
\end{center}
\botrule
\mytoprule{\src{C,\Delta,\Gamma}\vdash \src{\Pi} }
\begin{center}
	\typerule{T\LA-soup}{
		\src{C,\Delta,\Gamma}\vdash \src{\pi} %
		&
		\src{C,\Delta,\Gamma}\vdash \src{\Pi}
	}{
		\src{C,\Delta,\Gamma}\vdash \src{\pi\parallel\Pi}
	}{ts-soup}
\end{center}
\botrule
\mytoprule{\vdash\src{H}:\src{\Delta}}
	\begin{center}
	\typerule{\LA-Heap-ok-i}{
		\src{\ell:\mapsto v:\tau}\in\src{H}
		&
		\src{\ell:\tau}\in\src{\Delta}
		\\
		\vdash \src{H}:\src{\Delta}
		&
		\src{\Delta}\srce\vdash\src{v}:\src{\tau}
	}{
		\vdash\src{H}:\src{\Delta;\ell:\tau}
	}{heap-ok}
	\typerule{\LA-Heap-ok-b}{
	}{
		\vdash\src{H}:\srce
	}{heap-ok-b}
	\end{center}
\botrule
\mytoprule{\vdash\src{M}}
	\begin{center}
	\typerule{\LA-Monitor}{
		\src{M}\equiv\src{(\set{s},\monred,s_0,\Delta,s_c)}
		&
		\forall \src{s} \exists {\src{s'}}. \src{(s,s')}\in\src{\monred}
	}{
		\vdash\src{M}
	}{mon-ok-s}
	\end{center}
\botrule

\paragraph{Notes}

Monitor typing just ensures that the monitor is coherent and that it can't get stuck for no good reason.

\subsubsection{\UNS Typing}\label{sec:un-ty}
Attackers cannot have \src{\newty{t}{\tau}} terms where \src{\tau} is different from \UNS. %

\mytoprule{\src{\Delta,\Gamma}\vdashun \src{e} : \UNS}
\begin{center}
	\typerule{TU\LA-base}{
		\src{A}= \src{H ; \OB{F}\hole{\cdot}}
		&
		\src{\Delta}\vdashun\src{\OB{F}}	
		\\
		\dom{\src{H}}\cap \dom{\src{\Delta}} = \srce
		&
		\dom{\src{\Delta}}\cap(\fun{fv}{\src{\OB{F}}} \cup \fun{fv}{\src{H}})=\srce
	}{
		\src{\Delta}\vdashun\src{A}
	}{tsu-base}

	\typerule{TU\LA-true}{
		\src{\Delta,\Gamma}\vdash\diamond
	}{
		\src{\Delta,\Gamma}\vdashun\trues:\UNS
	}{tsu-true}
	\typerule{TU\LA-false}{
		\src{\Delta,\Gamma}\vdash\diamond
	}{
		\src{\Delta,\Gamma}\vdashun\falses:\UNS
	}{tsu-false}
	\typerule{TU\LA-nat}{
		\src{\Delta,\Gamma}\vdash\diamond
	}{
		\src{\Delta,\Gamma}\vdashun\src{n}:\UNS
	}{tsu-nat}
	\typerule{TU\LA-var}{
		\src{x:\tau}\in\src{\Gamma}
	}{
		\src{\Delta,\Gamma}\vdashun\src{x}:\UNS
	}{tsu-var}
	\typerule{TU\LA-loc}{
		\src{l:\UNS}\notin\src{\Delta}
	}{
		\src{\Delta,\Gamma}\vdashun\src{l}:\UNS
	}{tsu-loc}
	\typerule{TU\LA-pair}{
		\src{\Delta,\Gamma}\vdashun \src{e_1} : \UNS
		\\
		\src{\Delta,\Gamma}\vdashun \src{e_2} : \UNS
	}{
		\src{\Delta,\Gamma}\vdashun \src{\pair{e_1,e_2}} : \UNS
	}{tsu-pair}
	\typerule{TU\LA-proj-1}{
		\src{\Delta,\Gamma}\vdashun \src{e} : \UNS
	}{
		\src{\Delta,\Gamma}\vdashun \src{\projone{e}} : \UNS
	}{tsu-p1}
	\typerule{TU\LA-proj-2}{
		\src{\Delta,\Gamma}\vdashun \src{e} : \UNS
	}{
		\src{\Delta,\Gamma}\vdashun \src{\projtwo{e}} : \UNS
	}{tsu-p2}
	\typerule{TU\LA-dereference}{
		\src{\Delta,\Gamma}\vdashun \src{e} : \UNS
	}{
		\src{\Delta,\Gamma}\vdashun \src{!e} : \UNS
	}{tsu-deref}
	\typerule{TU\LA-op}{
		\src{\Delta,\Gamma}\vdashun \src{e} :\UNS
		&
		\src{\Delta,\Gamma}\vdashun \src{e'} : \UNS
	}{
		\src{\Delta,\Gamma}\vdashun \src{e \op e'} : \UNS
	}{tsu-op}
	\typerule{TU\LA-cmp}{
		\src{\Delta,\Gamma}\vdashun \src{e} : \UNS
		&
		\src{\Delta,\Gamma}\vdashun \src{e'} : \UNS
	}{
		\src{\Delta,\Gamma}\vdashun \src{e \bop e'} : \UNS
	}{tsU-bop}
\end{center}
\botrule
\mytoprule{\src{C,\Delta,\Gamma}\vdashun \src{s} }
\begin{center}
	\typerule{TU\LA-skip}{
	}{
		\src{C,\Delta,\Gamma}\vdashun \skips
	}{tsu-skip}	
	\typerule{TU\LA-function-call}{
		((\src{f}\in\dom{\src{C}.\mtt{funs}})\vee(\src{f}\in\dom{\src{C}.\mtt{intfs}}))
		\\
		\src{\Delta,\Gamma}\vdashun \src{e} : \src{\UNS}
	}{
		\src{\Delta,\Gamma}\vdashun \src{\call{f}~e} %
	}{tsu-fun}
	\typerule{TU\LA-sequence}{
		\src{C,\Delta,\Gamma}\vdashun \src{s_u} %
		\\
		\src{C,\Delta,\Gamma}\vdashun \src{s} %
	}{
		\src{C,\Delta,\Gamma}\vdashun \src{s_u;s} %
	}{tsu-seq}
	\typerule{TU\LA-letin}{
		\src{\Delta,\Gamma}\vdashun \src{e} : \UNS
		\\
		\src{C,\Gamma;x:\UNS}\vdashun \src{s} %
	}{
		\src{C,\Delta,\Gamma}\vdashun \src{\letin{x:\UNS}{e}{s}} %
	}{tsu-vardef}
	\typerule{TU\LA-assign}{
		\src{\Delta,\Gamma}\vdashun \src{x} : \UNS
		\\
		\src{\Delta,\Gamma}\vdashun \src{e'} : \UNS
	}{
		\src{C,\Delta,\Gamma}\vdashun \src{x := e'} %
	}{tsu-ass}
	\typerule{TU\LA-new}{
		\src{\Delta,\Gamma}\vdashun \src{e} : \UNS
		\\
		\src{C,\Gamma;x:\UNS}\vdashun \src{s} %
	}{
		\src{C,\Delta,\Gamma}\vdashun \src{\letnewty{x}{e}{\UNS}{s}} %
	}{tsu-new}
	\typerule{TU\LA-if}{
		\src{\Delta,\Gamma}\vdashun \src{e} : \src{\Bool}
		\\
		\src{C,\Delta,\Gamma}\vdashun \src{s_t} %
		&
		\src{C,\Delta,\Gamma}\vdashun \src{s_e} %
	}{
		\src{C,\Delta,\Gamma}\vdashun \src{\ifte{e}{s_t}{s_e}} %
	}{tsu-if}
	\typerule{TU\LA-fork}{
		\src{C,\Delta,\Gamma}\vdashun \src{s} %
	}{
		\src{C,\Delta,\Gamma}\vdashun \src{\fork{s}} %
	}{tsu-fork}
\end{center}
\botrule

\subsection{Dynamic Semantics of \LA}\label{sec:src-sem}
Function \monh{\cdot} returns the part of a heap the monitor cares for (\Cref{tr:h-mon-s}).
\Cref{tr:s-aux-intern,tr:s-aux-in,tr:s-aux-out} dictate the kind of a jump between two functions: if internal to the component/attacker, in(from the attacker to the component) or out(from the component to the attacker).
\Cref{tr:plug-s} tells how to obtain a whole program from a component and an attacker.
\Cref{tr:ini-s} tells the initial state of a whole program.
\Cref{tr:ini-heap-s} produces a heap that satisfies a \src{\Delta}, initialised with base values. %
\Cref{tr:mts-s} tells when a monitor makes a single step given a heap.

\mytoprule{\monh{\cdot}}
\begin{center}
	\typerule{\LA-Monitor-related heap}{
		\src{H'}=\myset{\src{\ell\mapsto v:\tau}}{\src{\ell\mapsto v:\tau}\in\src{H}}
	}{
	\vdash\monh{\src{H},\src{\Delta}} = \src{H'}
	}{h-mon-s}
	\end{center}
\botrule
\mytoprule{\text{Helpers}}
\begin{center}
	\typerule{\LA-Jump-Internal}{
		((\src{f'}\in\src{\OB{I}} \wedge \src{f}\in\src{\OB{I}}) \vee
				\\
		(\src{f'}\notin\src{\OB{I}} \wedge \src{f}\notin\src{\OB{I}}))
	}{
		\src{\OB{I}}\vdash\src{f,f'}:\src{internal}
	}{s-aux-intern}
	\typerule{\LA-Jump-IN}{
		\src{f}\in\src{\OB{I}} \wedge \src{f'}\notin\src{\OB{I}}
	}{
		\src{\OB{I}}\vdash\src{f,f'}:\src{in}
	}{s-aux-in}
	\typerule{\LA-Jump-OUT}{
		\src{f}\notin\src{\OB{I}} \wedge \src{f'}\in\src{\OB{I}}
	}{
		\src{\OB{I}}\vdash\src{f,f'}:\src{out}
	}{s-aux-out}
	\typerule{\LA-Plug}{
		\src{A} \equiv \src{H ; \OB{F}\hole{\cdot}}
		&
		\src{C}\equiv\src{\Delta ; \OB{F'} ; \OB{I}} 
		\\
		\vdash\src{C,\OB{F}}:\src{whole}
		&
		\src{\Delta}\vdash\src{H_0}
		&
		\src{main(x:\UNS)\mapsto s;\ret}\in\src{\OB{F}}
	}{
		\src{A\hole{C}} = \src{\Delta ; H\cup H_0; \OB{F;F'}; \OB{I}}
	}{plug-s}
	\typerule{\LA-Whole}{
		\src{C}\equiv\src{\Delta ; \OB{F'} ; \OB{I}} 
		\\
		\fun{names}{\src{\OB{F}}}\cap\fun{names}{\src{\OB{F'}}}=\emptyset
		\\
		\fun{names}{\src{\OB{I}}}\subseteq \fun{names}{\src{\OB{F}}}\cup\fun{names}{\src{\OB{F'}}}
	}{
		\vdash\src{C,\OB{F}}:\src{whole}
	}{whole-s}
	\typerule{\LA-Initial State}{
		\src{P}\equiv\src{\Delta ; H ; \OB{F} ; \OB{I}}
		\\
		\src{C}\equiv\src{\Delta ; \OB{F} ; \OB{I}}
		&
		\src{main(x)\mapsto s;\ret}\in\src{\OB{F}}
	}{
		\SInits{P} = \src{C, H \triangleright \proc{s\subs{0}{x}}{main}}
	}{ini-s}
\end{center}
\botrule
\mytoprule{\src{\Delta}\vdash\src{H_0}}
\begin{center}
	\typerule{\LA-Initial-heap}{
		\src{\Delta}\vdash\src{H}
		&
		\srce\vdash\src{v}:\src{\tau}
	}{
		\src{\Delta,\ell:\tau}\vdash\src{H;\ell\mapsto v:\tau}
	}{ini-heap-s}
\end{center}
\botrule
\mytoprule{\src{M;H\monred M'}}
\begin{center}
	\typerule{\LA-Monitor Step}{
		\src{M}= \src{(\set{\sigma},\monred,\sigma_0,\Delta,\sigma_c)}
		&
		\src{M'}= \src{(\set{\sigma},\monred,\sigma_0,\Delta,\sigma_f)}
		\\
		\src{(\sigma_c,\sigma_f)}\in\src{\monred}
		&
		\vdash\src{H}:\src{\Delta}
	}{
		\src{M;H\monred M'}
	}{mts-s}
	\typerule{\LA-Monitor Step Trace Base}{
	}{
		\src{M;\srce\monred M}
	}{mts-t-s-b}
	\typerule{\LA-Monitor Step Trace}{
		\src{M;\OB{H}\monred M''}
		&
		\src{M'';H\monred M'}
	}{
		\src{M;\OB{H}\cdot H\monred M'}
	}{mts-t-s}
	\typerule{\LA-valid trace}{
		\src{M;\OB{H}\monred M'}
		&
		\strip{\src{\OB{\alpha}}}=\src{\OB{H}}
	}{
		\src{M}\vdash\src{\OB{\alpha}}
	}{mts-valid}
\end{center}
\botrule

\subsubsection{Component Semantics}\label{src:src-sem-com}
\begin{align*}
	&\src{H\triangleright e \redtos e'} 
	&&\text{Expression \src{e} reduces to \src{e'}.}
	\\
	&\src{C, H \triangleright \pi} \xtos{\lambda} \src{C', H' \triangleright \pi} 
	&&\text{Process \src{\pi} reduces to \src{\pi'} and evolves the rest accordingly.}
	\\
	&\src{C, H \triangleright \Pi} \xtos{\lambda} \src{C', H' \triangleright \Pi'} 
	&&\text{Soup \src{\Pi} reduce to \src{\Pi'} and evolve the rest accordingly.}
	\\
	&\src{\Omega} \Xtos{\OB{\alpha}} \src{\Omega'}
	&& \text{Program state \src{\Omega} steps to \src{\Omega'} emitting trace \src{\OB{\alpha}}.}
\end{align*}

\mytoprule{\src{H\triangleright e \redtos e'} }
\begin{center}
	\typerule{E\LA-val}{
	}{
		\src{H \triangleright v} \redtos \src{v}
	}{es-val}
	\typerule{E\LA-p1}{
	}{
		\src{H\triangleright \projone{\pair{v,v'}}} \redtos \src{v}
	}{es-p1}
	\typerule{E\LA-p2}{
	}{
		\src{H\triangleright \projone{\pair{v,v'}}} \redtos \src{v'}
	}{es-p2}
	\typerule{E\LA-op}{
		n\op n'=n''
	}{
		\src{H\triangleright n \op n'} \redtos \src{n''}
	}{es-op}
	\typerule{E\LA-comp}{
		n\bop n'=b
	}{
		\src{H\triangleright n \bop n'} \redtos \src{b}
	}{es-op}
	\typerule{E\LA-dereference}{
		\src{H \triangleright e \redtos \ell}
		&
		\src{\ell\mapsto v:\tau } \in \src{H}
	}{
		\src{H \triangleright !\ell} \redtos \src{ v }
	}{es-de}
	\typerule{E\LA-ctx}{
		\src{H\triangleright e} \redtos \src{e'}
	}{
		\src{H \triangleright E\hole{e}} \redtos \src{E\hole{e'}}
	}{es-cth}
\end{center}
\botrule

\mytoprule{\src{C, H \triangleright \pi} \xtos{\epsilon} \src{C', H' \triangleright \pi'} }
\begin{center}
	\typerule{E\LA-sequence}{
	}{
		\src{C, H \triangleright \skips;s} \xtos{\epsilon} \src{C, H \triangleright s}
	}{es-seq}
	\typerule{E\LA-step}{
		\src{C, H \triangleright s} \xtos{\lambda} \src{C, H \triangleright s'}
	}{
		\src{C, H \triangleright s;s''} \xtos{\lambda} \src{C, H \triangleright s';s}
	}{es-step}
	\typerule{E\LA-if-true}{
		\src{H \triangleright e \redtos \trues}
	}{
		\src{C, H \triangleright \ifte{e}{s}{s'}} \xtos{\epsilon} \src{C, H \triangleright s}
	}{es-ift}
	\typerule{E\LA-if-false}{
		\src{H \triangleright e \redtos \falses}
	}{
		\src{C, H \triangleright\ifte{e}{s}{s'}} \xtos{\epsilon} \src{C, H \triangleright s'}
	}{es-iff}
	\typerule{E\LA-letin}{
		\src{H \triangleright e \redtos v}
	}{
		\src{C, H \triangleright \letin{x:\tau}{e}{s}} \xtos{\epsilon} \src{C, H \triangleright s\subs{v}{x}}
	}{es-letin}
	\typerule{E\LA-alloc}{
		\src{\ell}\notin\dom{\src{H}}
		&
		\src{H \triangleright e \redtos v}
	}{
		\src{C, H \triangleright \letnewty{x}{e}{\tau}{s}} \xtos{\epsilon} \src{C, H; \ell\mapsto v:\tau \triangleright s\subs{\ell}{x} }
	}{es-al}
	\typerule{E\LA-update}{
		\src{H}=\src{H_1; \ell\mapsto v':\tau ; H_2}
		\\
		\src{H'}=\src{H_1; \ell\mapsto v:\tau ; H_2}
	}{
		\src{C, H \triangleright \ell:=v} \xtos{\epsilon} \src{C, H' \triangleright \skips }
	}{es-up}
	\typerule{E\LA-endorse}{
		\src{H \triangleright e \redtos v}
		&
		\src{\Delta,\srce}\vdash\src{v}:\src{\varphi}
		&
		\src{\Delta}= \myset{\src{\ell:\tau}}{\src{\ell\mapsto v:\tau}\in\src{H}}
	}{
		\src{C, H \triangleright \myendorse{x}{e}{\varphi}{s} } \xtos{\epsilon} \src{C, H \triangleright s\subs{v}{x}}
	}{es-end}
	\typerule{E\LA-call-internal}{
		\src{\OB{C}.\mtt{intfs}}\vdash\src{f,f'}:\src{internal}
		&
		\src{\OB{f'}} = \src{\OB{f''};f'}
		\\
		\src{f(x:\tau):\tau'\mapsto s;\ret}\in\src{C}.\mtt{funs}
		&
		\src{H \triangleright e \redtos v}
	}{
		\src{C, H \triangleright \proc{{\call{f}~e}}{\OB{f'}}} \xtos{\epsilon} \src{C, H \triangleright \proc{{s;\ret\subs{v}{x}}}{\OB{f'};f}}
	}{es-call-i}
	\typerule{E\LA-callback}{
		\src{\OB{f'}} = \src{\OB{f''};f'}
		&
		\src{f(x:\tau):\tau'\mapsto s;\ret}\in\src{\OB{F}}
		\\
		\src{\OB{C}.\mtt{intfs}}\vdash\src{f',f}:\src{out}
		&
		\src{H \triangleright e \redtos v}
	}{
		\src{C, H \triangleright \proc{{\call{f}~e}}{\OB{f'}}} \xtos{\cb{f}{v}} \src{C, H \triangleright \proc{{s;\ret\subs{v}{x}}}{\OB{f'};f}}
	}{es-callback}
	\typerule{E\LA-call}{
		\src{\OB{f'}} = \src{\OB{f''};f'}
		&
		\src{f(x:\tau):\tau'\mapsto s;\ret}\in\src{C}.\mtt{funs}
		\\
		\src{\OB{C}.\mtt{intfs}}\vdash\src{f',f}:\src{in}
		&
		\src{H \triangleright e \redtos v}
	}{
		\src{C, H \triangleright \proc{{\call{f}~e}}{\OB{f'}}} \xtos{\cl{f}{v}} \src{C, H \triangleright \proc{{s;\ret\subs{v}{x}}}{\OB{f'};f}}
	}{es-call}
	\typerule{E\LA-ret-internal}{
		\src{\OB{C}.\mtt{intfs}}\vdash\src{f,f'}:\src{internal}
		&
		\src{\OB{f'}} = \src{\OB{f''};f'}
	}{
		\src{C, H \triangleright \proc{{\ret}}{\OB{f'};f}} \xtos{\epsilon} \src{C, H \triangleright \proc{{\skips}}{\OB{f'}}}
	}{es-ret-i}
	\typerule{E\LA-retback}{
		\src{\OB{C}.\mtt{intfs}}\vdash\src{f,f'}:\src{in}
		&
		\src{\OB{f'}} = \src{\OB{f''};f'}
	}{
		\src{C, H \triangleright \proc{{\ret}}{\OB{f'};f}} \xtos{\rb{}} \src{C, H \triangleright \proc{\skips}{\OB{f'}}}
	}{es-retb}
	\typerule{E\LA-return}{
		\src{\OB{C}.\mtt{intfs}}\vdash\src{f,f'}:\src{out}
		&
		\src{\OB{f'}} = \src{\OB{f''};f'}
		&
		\src{H \triangleright e \redtos v}
	}{
		\src{C, H \triangleright \proc{{\ret}}{\OB{f'};f}} \xtos{\rt{}} \src{C, H \triangleright \proc{\skips}{\OB{f'}}}
	}{es-ret}
\end{center}
\botrule
\mytoprule{\src{C, H \triangleright \Pi} \xtos{\lambda} \src{C', H' \triangleright \Pi'}}
\begin{center}
	\typerule{E\LA-par}{
		\src{\Pi}=\src{\Pi_1 \parallel \proc{s}{\OB{f}}\parallel \Pi_2}
		\\
		\src{\Pi'}=\src{\Pi_1 \parallel \proc{s'}{\OB{f'}} \parallel \Pi_2}
		\\
		\src{C, H \triangleright \proc{s}{\OB{f}}} \xtos{\lambda} \src{C', H'\triangleright \proc{s'}{\OB{f'}}}
	}{
		\src{C, H \triangleright \Pi} \xtos{\lambda} \src{C', H' \triangleright \Pi'}
	}{es-par}
	\typerule{E\LA-fail}{
		\src{\Pi}=\src{\Pi_1 \parallel \proc{s}{\OB{f}}\parallel \Pi_2}
		\\
		\src{C, H \triangleright \proc{s}{\OB{f}}} \xtos{\epsilon} \fails
	}{
		\src{C, H \triangleright \Pi} \xtos{\epsilon} \fails
	}{es-fail}
	\typerule{E\LA-fork}{
		\src{\Pi}=\src{\Pi_1 \parallel \proc{\fork{s};s'}{\OB{f}} \parallel \Pi_2}
		\\
		\src{\Pi'}=\src{\Pi_1 \parallel \proc{\skips;s'}{\OB{f}} \parallel \Pi_2 \parallel \proc{s}{\srce}}
	}{
		\src{C, H \triangleright \Pi} \xtos{\epsilon} \src{C, H \triangleright \Pi'}
	}{es-fork}
\end{center}
\botrule 
\mytoprule{ \src{\Omega} \Xtos{\OB{\alpha}} \src{\Omega'} }
\begin{center}
	\typerule{E\LA-single}{
		\src{\Omega}\xtos{\alpha}\src{\Omega'}
	}{
		\src{\Omega}\Xtos{\alpha}\src{\Omega'}
	}{es-tr-sin}
	\typerule{E\LA-silent}{
		\src{\Omega}\xtos{\epsilon}\src{\Omega'}
	}{
		\src{\Omega}\Xtos{}\src{\Omega'}
	}{es-tr-silent}
	\typerule{E\LA-trans}{
		\src{\Omega}\Xtos{\OB{\alpha}}\src{\Omega''}
		\\
		\src{\Omega''}\Xtos{\OB{\alpha'}}\src{\Omega'}
	}{
		\src{\Omega}\Xtos{\OB{\alpha}\cdot\OB{\alpha'}}\src{\Omega'}
	}{es-tr-trans}
\end{center}
\botrule

 \newpage

\section{\LC: Extending \LP with Concurrency and Informed Monitors}

\subsection{Syntax}
This extends the syntax of \Cref{sec:trg-syn} with concurrency and a memory allocation instruction that atomically hides the new location.
\begin{align*}
	\mi{Whole\ Programs}~\trg{P} \bnfdef&\ \trg{H_0 ; \OB{F} ;\OB{I}}
	\\
	\mi{Components}~\trg{C} \bnfdef&\ \trg{H_0 ; \OB{F} ; \OB{I}}
	\\
	\mi{Statements}~\trg{s} \bnfdef&\ \cdots \mid \trg{\fork{s}} 
	\mid \trg{\destruct{x}{e}{B}{s}{s}}
	\\
	\mid&\ \trg{\letatom{x}{e}{s}}
	\\
	\mi{Patterns}~\trg{B} \bnfdef&\ \trg{nat} \mid \trg{pair}
	\\
	\mi{Monitors}~\trg{M} \bnfdef&\ \trg{(\set{\sigma},\monred,\sigma_0,H_0,\sigma_c)}
	\\
	\mi{Single\ Process}~ \trg{\pi} \bnfdef&\ \trg{\proc{s}{\OB{f}}}
	\\
	\mi{Processes}~\trg{\Pi} \bnfdef&\ \trge \mid \trg{\Pi \parallel \pi}
	\\
	\mi{Prog.\ States}~\trg{\Omega}\bnfdef&\ \trg{C, H\triangleright \Pi} 
\end{align*}

\subsection{Dynamic Semantics}

Following is the definition of the \monh{\cdot} function for \LC.

\mytoprule{\monh{\cdot}}
\begin{center}
	\typerule{\LC-Monitor-related heap}{
          \trg{H'} = \{\trg{n\mapsto v:\eta} ~|~ \trg{n} \in \dom{\trg{H_0}} \text{ and } \trg{n\mapsto v:\eta}\in\trg{H}\}
	}{
		\monh{\trg{H},\trg{H_0}} = \trg{H'}
	}{th-pub-t}
\end{center}
\botrule

\mytoprule{\text{Helpers}}
\begin{center}
	\typerule{\LC-Plug}{
		\trg{A} \equiv \trg{\OB{F}\hole{\cdot}}
		&
		\trg{C}\equiv\trg{H_0 ; \OB{F'} ; \OB{I}} 
		\\
		\vdash\trg{C,\OB{F}}:\trg{whole}
		&
		\trg{main(x)\mapsto \ret{s}}\in\trg{\OB{F}}
		\\
		\trg{C}\vdashattt\trg{A}
		&
		\forall \trg{n\mapsto v:k}\in\trg{H_0}, \trg{k}\in\trg{H_0}
	}{
		\trg{A\hole{C}} = \trg{H_0; \OB{F;F'}; \OB{I}}
	}{plug-t-2}
	\typerule{\LC-Initial State}{
		\trg{P}\equiv\trg{H_0; \OB{F} ; \OB{I}}
		&
		\trg{main(x)\mapsto s;\ret}\in\trg{\OB{F}}
  	}{
		\SInitt{P} = \trg{P, H_0 \triangleright \proc{s\subt{0}{x}}{main}}
	}{ini-t-2}
\end{center}
\botrule
\mytoprule{\trg{M;H\monred M'}}
\begin{center}
	\typerule{\LC-Monitor Step}{
		\trg{M}= \trg{(\set{\sigma},\monred,\sigma_0,H_0,\sigma_c)}
		&
		\trg{M'}= \trg{(\set{\sigma},\monred,\sigma_0,H_0,\sigma_f)}
		\\
		\trg{(s_c,\monh{\trg{H},\trg{H_0}},s_f)}\in\trg{\monred}
	}{
		\trg{M;H\monred M'}
	}{ms-t-2}
	\typerule{\LC-Monitor Step Trace Base}{
	}{
		\trg{M;\trge\monred M}
	}{mtt-t-s-b}
	\typerule{\LC-Monitor Step Trace}{
		\trg{M;\OB{H}\monred M''}
		&
		\trg{M'';H\monred M'}
	}{
		\trg{M;\OB{H}\cdot H\monred M'}
	}{mtt-t-s}
	\typerule{\LC-valid trace}{
		\trg{M;\OB{H}\monred M'}
	}{
		\trg{M}\vdash\trg{\OB{\mnh{H}}}
	}{mtt-valid}
\end{center}
\botrule

\subsubsection{Component Semantics}
\begin{align*}
	&\trg{C, H \triangleright \Pi} \xtot{\epsilon} \trg{C', H' \triangleright \Pi'} 
	&&\text{Processes \trg{\Pi} reduce to \trg{\Pi'} and evolve the rest accordingly.}
\end{align*}

\mytoprule{\trg{C, H \triangleright s} \xtot{\epsilon} \trg{C', H' \triangleright s'} }
\begin{center}
	\typerule{E\LC-destruct-nat}{
		\trg{H\triangleright e \redtot n}
	}{
		\trg{C, H \triangleright \destruct{x}{e}{nat}{s}{s'} } \xtot{\epsilon} \trg{C, H \triangleright s\subt{n}{x}}
	}{et-end-n}
	\typerule{E\LC-destruct-pair}{
		\trg{H\triangleright e \redtot \pair{v,v'}}
	}{
		\trg{C, H \triangleright \destruct{x}{e}{pair}{s}{s'} } \xtot{\epsilon} \trg{C, H \triangleright s\subt{\pair{v,v'}}{x}}
	}{et-end-p}
	\typerule{E\LC-destruct-not}{
		\text{ otherwise }
	}{
		\trg{C, H \triangleright \destruct{x}{e}{B}{s}{s'} } \xtot{\epsilon} \trg{C, H \triangleright s'}
	}{et-end-nope}
	\typerule{E\LC-new}{
		\trg{H} = \trg{H_1;n\mapsto (v,\eta)}
		&
		\trg{H \triangleright e\redtot v}
		&
		\trg{k}\notin\dom{\trg{H}}
	}{
		\trg{C, H \triangleright \letatom{x}{e}{s}} \xtot{\epsilon} \trg{C,H; n+1\mapsto v:k;k \triangleright s\subt{\pair{n+1,k}}{x}}
	}{et-atom}
\end{center}
\botrule

\mytoprule{\trg{C, H \triangleright \Pi} \redtot \trg{C', H' \triangleright \Pi'}}
\begin{center}
	\typerule{E\LC-par}{
		\trg{\Pi}=\trg{\Pi_1 \parallel \proc{s}{\OB{f}}\parallel \Pi_2}
		\\
		\trg{\Pi'}=\trg{\Pi_1 \parallel \proc{s'}{\OB{f'}} \parallel \Pi_2}
		\\
		\trg{C, H \triangleright \proc{s}{\OB{f}}} \redtot \trg{C', H'\triangleright \proc{s'}{\OB{f'}}}
	}{
		\trg{C, H \triangleright \Pi} \redtot \trg{C', H' \triangleright \Pi'}
	}{et-par}
	\typerule{E\LC-fork}{
		\trg{\Pi}=\trg{\Pi_1 \parallel \proc{{\fork{s}}}{\OB{f}} \parallel \Pi_2}
		\\
		\trg{\Pi'}=\trg{\Pi_1 \parallel \proc{{0}}{\OB{f}} \parallel \Pi_2 \parallel \proc{s}{\trge}}
	}{
		\trg{C, H \triangleright \Pi} \redtot \trg{C, H \triangleright \Pi'}
	}{et-fork}
\end{center}
\botrule

\newpage

\section{Extended Language Properties and Necessities}

\subsection{Monitor Agreement for \LA and \LC}
\begin{definition}[\LA: \src{M\agree C}]\label{def:la-agree}
	\begin{align*}
		\src{(\set{\sigma},\monred,\sigma_0,\Delta,\sigma_c) \agree (\Delta ; \OB{F} ; \OB{I})}
	\end{align*}
\end{definition}
A monitor and a component agree if they focus on the same set of locations \src{\Delta}.

\begin{definition}[\LC: \trg{M\agree C}]\label{def:lc-agree}
	\begin{align*}
		\trg{(\set{\sigma},\monred,\sigma_0,H_0,\sigma_c) \agree (H_0 ; \OB{F} ;\OB{I})}
	\end{align*}
\end{definition}
A monitor and a component agree if they focus on the same set of locations, protected with the same capabilities \trg{H_0}

\subsection{Properties of \LA}\label{sec:src-prop}

\begin{definition}[\LA Semantics Attacker]\label{def:src-att-original}
	\begin{align*}
		\src{C}\vdash_{\src{attacker}}\src{A} \isdef&\ 
			\begin{cases}
				\forall \src{\ell}\in\dom{\src{C}.\src{\Delta}}, \src{\ell}\notin\fun{locs}{\src{A}}
				\\
				\text{ no } \src{\letnewty{x}{e}{\tau}{s}} \text{ in } \src{A} \text{ such that } \src{\tau}\neq\UNS
			\end{cases}
	\end{align*}
\end{definition}
This semantic definition of an attacker is captured by typing below, which allows for simpler reasoning.

\begin{definition}[\LA Attacker]\label{def:src-att}
	\begin{align*}
		\src{C}\vdashatts\src{A} \isdef&\ \src{C}= \src{\Delta ; \OB{F} ; \OB{I}},
		\src{\Delta}\vdash_\UNS \src{A}
		\\
		\src{C}\vdashatts\src{\pi} \isdef&\ \src{\pi}=\src{\proc{s}{\OB{f};f}} \text{ and } \src{f}\in\src{C}.\mtt{itfs}
		\\
		\src{C}\vdashatts\src{\Pi}\xtos{}\src{\Pi'} \isdef&\ \src{\Pi}=\src{\Pi_1 \parallel \pi \parallel \Pi_2} \text { and } \src{\Pi'}=\src{\Pi_1 \parallel \pi' \parallel \Pi_2}
		\\
		&\ \text{and } \src{C}\vdashatts\src{\pi} \text{ and } \src{C}\vdashatts\src{\pi'}
	\end{align*}
\end{definition}

The two notions of attackers coincide.
\begin{lemma}[Semantics and typed attackers coincide]\label{thm:atk-src-coincide}
	\begin{align*}
		\src{C}\vdash_{\src{attacker}}\src{A} \iff (\src{C}\vdashatts\src{A} )
	\end{align*}
\end{lemma}

\begin{theorem}[Typability Implies Robust Safety in \LA]\label{thm:src-ty-impl-safe}
	\begin{align*}
		&
		\forall\src{C},\src{M}
		\\
		\text{if }
		&
		\vdash\src{C}:\UNS
		\\
		&
		\src{C\agree M}
		\\
		\text{ then }
		&
		\src{M}\vdash\src{C}:\src{rs}
	\end{align*}
\end{theorem}

\subsection{Properties of \LC}
\begin{definition}[\LC Attacker]\label{def:trg-att}
	\begin{align*}
		\trg{C}\vdashattt\trg{A} \isdef&\ \trg{C}= \trg{H_0 ; \OB{F} ; \OB{I}}, \forall\trg{k}\in\trg{H_0}. \trg{k}\notin\fun{fv}{\trg{A}} %
		\\
		\trg{C}\vdashattt\trg{\pi}  \isdef&\ \trg{\pi}=\trg{\proc{s}{\OB{f};f}} \text{ and } \trg{f}\in\trg{C}.\mtt{itfs}
		\\
		\trg{C}\vdashattt\trg{\Pi}\xtot{}\trg{\Pi'} \isdef&\ \trg{\Pi}=\trg{\Pi_1 \parallel \pi \parallel \Pi_2} \text { and } \trg{\Pi'}=\trg{\Pi_1 \parallel \pi' \parallel \Pi_2}
		\\
		&\ \text{and } \trg{C}\vdashattt\trg{\pi}  \text{ and } \trg{C}\vdashattt\trg{\pi'} 
	\end{align*}
\end{definition}
 \newpage
\section{Compiler from \LA to \LC}\label{sec:comp2}

\subsection{Assumed Relation between \LA and \LC Elements}\label{sec:relate-elems}

We can scale the $\relatebeta$ relation to monitors, heaps, actions and processes as follows.

\mytoprule{\hl{$\src{M}\relate\trg{M} $} }
\begin{center}
	\typerule{Ok Mon}{
		\trg{M} = \trg{(\set{\sigma},\monred,\sigma_0,H_0,\sigma_c)}
		\\
		\forall\trg{\sigma} \in \trg{\set{\sigma}},
		\monh{\src{H};\src{\Delta}} \relatebeta \monh{\trg{H},\trg{H_0}}.
		\\
		\text{ if } \vdash \src{H}: \src{\Delta}
		\text{ then } \exists \trg{\sigma'}. (\trg{\sigma}, \monh{\trg{H},\trg{H_0}}, \trg{\sigma'}) \in \trg{\monred}
	}{
		\beta, \src{\Delta} \vdash \trg{M}
	}{ok-mon}
	\typerule{
		Monitor relation
	}{
		\src{M} = \src{(\set{\sigma},\monred,\sigma_0,\Delta,\sigma_c)}
		&
		\trg{M} = \trg{(\set{\sigma},\monred,\sigma_0,H_0,\sigma_c)}
		\\
		\beta_0, \src{\Delta}\vdash \trg{M}
		&
		\beta_0=(\dom{\src{\Delta}},\dom{\trg{H_0}},\trg{H_0}.\trg{\eta})
	}{
		\src{M} \relate \trg{M}
	}{mon-rel-ap}
\end{center}
\botrule

\mytoprule{\src{\Delta}\vdash_{\beta}\trg{H_0} ~~ \src{\Delta, \trg{H} }\vdash\trg{v}: \src{\tau} }
\begin{center}
	\typerule{Initial-heap}{
		\src{\Delta}\vdash\trg{H}
		&
		\src{\Delta},\trg{H}\vdash_{\beta}\trg{v}\src{:\tau}
		\\
		\src{\ell}\relatebeta\trg{\pair{n,k}}
	}{
		\src{\Delta,\ell:\tau}\vdash_{\beta}\trg{H;n\mapsto v:k}
	}{ini-heap}
	\typerule{Initial-value}{
			(\src{\tau}\equiv\Bools \wedge \trg{v}\equiv\trg{0}) 
			&
			\vee
			&
			(\src{\tau}\equiv\Nats \wedge \trg{v}\equiv\trg{0}) 
			&
			\vee
			\\
			(\src{\tau}\equiv\src{\Refs{\tau}} \wedge \trg{v}\equiv\trg{n'} \wedge \trg{n'\mapsto v':k'}\in\trg{H} \wedge \src{\ell'}\relatebeta\trg{\pair{n',k'}} \wedge \src{\ell:\tau}\in\src{\Delta}, 
			\src{\Delta},\trg{H}\vdash\trg{v'}\src{:\tau}
			) 
			&
			\vee
			\\
			(\src{\tau}\equiv\src{\tau_1\times\tau_2} \wedge \trg{v}\equiv\trg{\pair{v_1,v_2}} \wedge 
			\src{\Delta},\trg{H}\vdash\trg{v_1}\src{:\tau_1}\wedge
			\src{\Delta},\trg{H}\vdash\trg{v_2}\src{:\tau_2}
			)
	}{
		\src{\Delta},\trg{H}\vdash_{\beta}\trg{v}\src{:\tau}
	}{ini-val}
\end{center}
\botrule

\mytoprule{\src{\Pi}\relatebeta\trg{\Pi}}
\begin{center}
	\typerule{Single process relation}{
		\src{\OB{f}}\relate\trg{\OB{f}}
	}{
		\src{\proc{\skips}{\OB{f}}}\relatebeta\trg{\proc{\skipt}{\OB{f}}}
	}{sing-proc-rel}
	\typerule{Process relation}{
		\src{\Pi}\relatebeta\trg{\Pi}
		&
		\src{\pi}\relatebeta\trg{\pi}
	}{
		\src{\Pi\parallel\pi}\relatebeta\trg{\Pi\parallel\pi}
	}{proc-rel}
\end{center}
\botrule

\subsection{Compiler Definition}\label{sec:comp-def}

\begin{definition}[Compiler \LA to \LC]\label{def:comp-la-lp}
	$\compap{\cdot} : \src{C}\to \trg{C}$%

	Given that $\src{C}=\src{\Delta ; \OB{F} ; \OB{I}}$ %
	if $\vdash \src{C :\UNS}$ %
	then \compap{\src{C}} is defined as follows:
	\begin{align*}
		\tag{\compap{\cdot}-Component}
		\compap{
			\typerule{T\LA-component}{
				\src{C}\equiv\src{\Delta ; \OB{F} ; \OB{I}} 
				\\
				\src{C}\vdash\src{\OB{F}}:\UNS
				\\
				\fun{names}{\src{\OB{F}}}\cap\fun{names}{\src{\OB{I}}}=\emptyset
				\\
				\src{\Delta}\vdash\src{ok}
			}{
				\vdash \src{C}:\UNS
			}{compap-co}
		} &= \trg{%
		\trg{H_0} ;
		\compap{\OB{F}}; \compap{\OB{I}}}
		\qquad\qquad
		\text{if }\src{\Delta}\vdash_{\beta_0}\trg{H_0}
		\\
		\tag{\compap{\cdot}-Function}
		\compap{
			\typerule{T\LA-function}{
				\src{F}\equiv \src{f(x:\UNS)\mapsto s;\ret}
				\\
				\src{C},\src{\Delta;x:\UNS}\vdash \src{s} %
				\\
				\forall\src{f}\in\fun{fn}{s}, \src{f}\in\dom{\src{C}.\mtt{funs}} 
				\\
				\vee \src{f}\in\dom{\src{C}.\mtt{intfs}}
			}{
				\src{C}\vdash\src{F}:\src{\UNS} 
			}{compap-fu}
		} &= \trg{f(x)\mapsto \compap{\src{C};\src{\Delta;x:\UNS}\vdash \src{s}};\ret}
		\\
		\tag{\compap{\cdot}-Interfaces}
		\compap{f} &= \trg{f}
	\end{align*}

	\mytoprule{Expressions}
	\begin{align*}
		\\
		\tag{\compap{\cdot}-True}
		\compap{
			\typerule{T\LA-true}{
				\src{\Delta,\Gamma}\vdash\diamond
			}{
				\src{\Delta,\Gamma}\vdash\trues:\Bools
			}{compap-true}
		}
		=&\ \trg{0}  
		\\
		\tag{\compap{\cdot}-False}
		\compap{
			\typerule{T\LA-false}{
				\src{\Delta,\Gamma}\vdash\diamond
			}{
				\src{\Delta,\Gamma}\vdash\falses:\Bools
			}{compap-false}
		}
		=&\ \trg{1} 
		\\
		\tag{\compap{\cdot}-Nat}
		\compap{
			\typerule{T\LA-nat}{
				\src{\Delta,\Gamma}\vdash\diamond
			}{
				\src{\Delta,\Gamma}\vdash\src{n}:\Nats
			}{compap-nat}
		}
		=&\ \trg{n}
		\\
		\tag{\compap{\cdot}-Var} 
		\compap{
			\typerule{T\LA-var}{
				\src{x:\tau}\in\src{\Gamma}
			}{
				\src{\Delta,\Gamma}\vdash\src{x}:\src{\tau}
			}{compap-var}
		}
		=&\  \trg{x}
		\\
		\tag{\compap{\cdot}-Loc} 
		\compap{
			\typerule{T\LA-loc}{
				\src{\ell:\tau}\in\src{\Delta}
			}{
				\src{\Delta,\Gamma}\vdash\src{\ell}:\src{\tau}
			}{compap-loc}
		}
		=&\  
				\trg{\pair{n,v}}
		\\
		\tag{\compap{\cdot}-Pair} 
		\compap{
			\typerule{T\LA-pair}{
				\src{\Delta,\Gamma}\vdash \src{e_1} : \src{\tau}
				\\
				\src{\Delta,\Gamma}\vdash \src{e_2} : \src{\tau'} 
			}{
				\src{\Delta,\Gamma}\vdash \src{\pair{e_1,e_2}} : \src{\tau\times\tau'}
			}{compap-pair}
		}
		=&\ \trg{\pair{\compap{\src{\Delta,\Gamma}\vdash \src{e_1} : \src{\tau}},\compap{\src{\Delta,\Gamma}\vdash \src{e_2} : \src{\tau'}}}}
		\\
		\tag{\compap{\cdot}-P1} 
		\compap{
			\typerule{T\LA-proj-1}{
				\src{\Delta,\Gamma}\vdash \src{e} : \src{\tau\times\tau'}
			}{
				\src{\Delta,\Gamma}\vdash \src{\projone{e}} : \src{\tau}
			}{compap-p1}
		}
		=&\ \trg{\projone{\compap{\src{\Delta,\Gamma}\vdash \src{e} : \src{\tau\times\tau'}}}}
		\\
		\tag{\compap{\cdot}-P2} 
		\compap{
			\typerule{T\LA-proj-2}{
				\src{\Delta,\Gamma}\vdash \src{e} : \src{\tau\times\tau'}
			}{
				\src{\Delta,\Gamma}\vdash \src{\projtwo{e}} : \src{\tau'}
			}{compap-p2}
		}
		=&\ \trg{\projtwo{\compap{\src{\Delta,\Gamma}\vdash \src{e} : \src{\tau\times\tau'}}}}
		\\
		\tag{\compap{\cdot}-Deref} 
		\compap{
			\typerule{T\LA-dereference}{
				\src{\Delta,\Gamma}\vdash \src{e} : \src{\Refs{\tau}} 
			}{
				\src{\Delta,\Gamma}\vdash \src{!e} : \src{\tau}  
			}{compap-deref}
		}
		=&\ \trg{
			!\projone{\compap{\src{\Delta,\Gamma}\vdash \src{e} : \src{\Refs{\tau}}}} \with{\projtwo{\compap{\src{\Delta,\Gamma}\vdash \src{e} : \src{\Refs{\tau}}}}} %
		}
		\\
		\tag{\compap{\cdot}-op}
		\compap{
			\typerule{T\LA-op}{
				\src{\Delta,\Gamma}\vdash \src{e} : \Nats
				\\
				\src{\Delta,\Gamma}\vdash \src{e'} : \Nats
			}{
				\src{\Delta,\Gamma}\vdash \src{e \op e'} : \src{\Nats}  
			}{compap-op}
		}
		=&\ \trg{\compap{\src{\Delta,\Gamma}\vdash \src{e} : \Nats} \op\compap{\src{\Delta,\Gamma}\vdash \src{e'} : \Nats}}
		\\
		\tag{\compap{\cdot}-cmp}
		\compap{
			\typerule{T\LA-cmp}{
				\src{\Delta,\Gamma}\vdash \src{e} : \Nats
				\\
				\src{\Delta,\Gamma}\vdash \src{e'} : \Nats
			}{
				\src{\Delta,\Gamma}\vdash \src{e \bop e'} : \Bools
			}{compap-bop}
		}
		=&\ 
				\trg{
					\compap{\src{\Delta,\Gamma}\vdash \src{e} : \Nats} \bop \compap{\src{\Delta,\Gamma}\vdash \src{e'} : \Nats}
				}
		\\
		\tag{\compap{\cdot}-Coerce} 
		\compap{
			\typerule{T\LA-coercion}{
				\src{\Delta,\Gamma}\vdash \src{e} : \src{\tau} 
				&
				\src{\tau} \vdash \circ
			}{
				\src{\Delta,\Gamma}\vdash \src{e} : \UNS 
			}{compap-coe}
		}
		=&\
				\trg{
					\compap{\src{\Delta,\Gamma}\vdash \src{e} : \src{\tau} }
				}
	\end{align*}

	\mytoprule{Statements}
	\begin{align*}
		\tag{\compap{\cdot}-Skip} 
		\compap{
			\typerule{T\LA-skip}{
			}{
				\src{C,\Delta,\Gamma}\vdash \skips
			}{compap-skip}	
		}
		=&\
		\skipt
		\\
		\tag{\compap{\cdot}-New} 
		\compap{
			\typerule{T\LA-new}{
				\src{C,\Delta,\Gamma}\vdash \src{e} : \src{\tau} 
				\\
				\src{C,\Delta,\Gamma;x:\Refs{\tau}}\vdash \src{s} %
			}{
				\src{C,\Delta,\Gamma}\vdash \src{\letnewty{x}{e}{\tau}{s}} %
			}{compap-new}	
		}
		=&\
		\begin{cases}
			\trg{
			\begin{aligned}
				&
				\letnewt{\trg{xo}~}{\compap{\src{\Delta,\Gamma}\vdash \src{e} : \src{\tau}} 
				\\
				&\
				}{ 
					\letint{\trg{x}~}{~\trg{\pair{xo,0}}
					\\
					&\ \ 
					}{
						\compap{\src{C,\Delta,\Gamma;x:\Refs{\tau}}\vdash \src{s}} 
					}
				}
			\end{aligned}
			}
			\\
			\text{if }\src{\tau}=\UNS %
			\\
			\\
			\trg{
			\begin{aligned}
				&
				\letatom{\trg{x}~}{ \compap{\src{\Delta,\Gamma}\vdash \src{e} : \src{\tau}} 
				\\
				&\
				}{ 
					\compap{\src{C,\Delta,\Gamma;x:\Refs{\tau}}\vdash \src{s}} 
				}
			\end{aligned}
			}
			\\ 
			\text{else} 
		\end{cases}
		\\
		\tag{\compap{\cdot}-call}
		\compap{
			\typerule{T\LA-function-call}{
				((\src{f}\in\dom{\src{C}.\mtt{funs}})
				\\
				\vee(\src{f}\in\dom{\src{C}.\mtt{intfs}}))
				\\
				\src{\Delta,\Gamma}\vdash \src{e} : \src{\UNS}
			}{
				\src{\Delta,\Gamma}\vdash \src{\call{f}~e}
			}{compap-fun}
		}
		=&\ \trg{\call{f}~\compap{\src{\Delta,\Gamma}\vdash \src{e}: \UNS}}
		\\
		\tag{\compap{\cdot}-If} 
		\compap{
			\typerule{T\LA-if}{
				\src{\Delta,\Gamma}\vdash \src{e} : \src{\Bool}
				\\
				\src{C,\Delta,\Gamma}\vdash \src{s_t} 
				\\
				\src{C,\Delta,\Gamma}\vdash \src{s_e} 
			}{
				\src{C,\Delta,\Gamma}\vdash \src{\ifte{e}{s_t}{s_e}} 
			}{compap-if}
		}
		=&\ \trg{
			\begin{aligned}
				&\ifztet{\compap{\src{\Delta,\Gamma}\vdash \src{e} : \src{\Bool}}
				\\
				&}
				{\compap{\src{C,\Delta,\Gamma}\vdash \src{s_t} }
				\\
				&
				}{\compap{\src{C,\Delta,\Gamma}\vdash \src{s_e} }}
			\end{aligned}
		}
		\\
		\tag{\compap{\cdot}-Seq} 
		\compap{
			\typerule{T\LA-sequence}{
				\src{C,\Delta,\Gamma}\vdash \src{s_u} 
				\\
				\src{C,\Delta,\Gamma}\vdash \src{s} 
			}{
				\src{C,\Delta,\Gamma}\vdash \src{s_u;s} 
			}{compap-seq}
		}
		=&\ \trg{\compap{\src{C,\Delta,\Gamma}\vdash \src{s_u} } ; \compap{\src{C,\Delta,\Gamma;\Gamma'}\vdash \src{s} }}
		\\
		\tag{\compap{\cdot}-Letin} 
		\compap{
			\typerule{T\LA-letin}{
				\src{\Delta,\Gamma}\vdash \src{e} : \src{\tau} 
				\\
				\src{C,\Delta,\Gamma;x:\tau}\vdash \src{s}
			}{
				\src{C,\Delta,\Gamma}\vdash \src{\letin{x:\tau}{e}{s}}
			}{compap-vardef}
		}
		=&\ \trg{
			\begin{aligned}
				&
				\letint{\trg{x}}{\compap{\src{\Delta,\Gamma}\vdash \src{e} : \src{\tau}}
				\\
				&
				}{\compap{\src{C,\Delta,\Gamma;x:\tau}\vdash \src{s} }}
			\end{aligned}
			}
		\\
		\tag{\compap{\cdot}-Assign}
		\compap{
			\typerule{T\LA-assign}{
				\src{\Delta,\Gamma}\vdash \src{x} : \src{\Refs{\tau}}
				\\
				\src{\Delta,\Gamma}\vdash \src{e} : \src{\tau} 
			}{
				\src{C,\Delta,\Gamma}\vdash \src{x := e} 
			}{compap-ass}
		}
		=&\ \trg{
			\begin{aligned}
				&
				\letint{\trg{x1}~}{~\trg{\projone{x}}
				\\
				&\
				}{
					\letint{\trg{x2}~}{~\trg{\projtwo{x}}
					\\
					&\ \ 
					}{
						\trg{x1 := \compap{\src{\Delta,\Gamma}\vdash \src{e} : \src{\tau}} \with{x2}}
					}
				}
			\end{aligned}
		}
		\\
		\tag{\compap{\cdot}-Fork} 
		\compap{
			\typerule{T\LA-fork}{
				\src{C,\Delta,\Gamma}\vdash \src{s} 
			}{
				\src{C,\Delta,\Gamma}\vdash \src{\fork{s}} 
			}{compap-fork}
		}
		=&\ \trg{\fork{\compap{\src{C,\Delta,\Gamma}\vdash \src{s} }}}
		\\
		\tag{\compap{\cdot}-Proc} 
		\compap{
			\typerule{T\LA-process}{
				\src{C,\Delta,\Gamma}\vdash \src{s} 
			}{
				\src{C,\Delta,\Gamma}\vdash \src{\proc{s}{\OB{f}}}
			}{compap-proc}
		}
		=&\ 
			\begin{aligned}
				&
				\trg{\proc{\compap{\src{C,\Delta,\Gamma}\vdash \src{s} }}{\compap{\OB{f}}}}
			\end{aligned}
		\\
		\tag{\compap{\cdot}-Soup} 
		\compap{
			\typerule{T\LA-soup}{
				\src{C,\Delta,\Gamma}\vdash \src{\pi} 
				\\
				\src{C,\Delta,\Gamma}\vdash \src{\Pi}
			}{
				\src{C,\Delta,\Gamma}\vdash \src{\pi\parallel\Pi}
			}{compap-soup}
		}
		=&\
			\trg{\compap{\src{C,\Delta,\Gamma}\vdash \src{\pi} }\parallel\compap{\src{C,\Delta,\Gamma}\vdash \src{\Pi} }}
		\\
		\tag{\compap{\cdot}-Endorse} 
		\compap{
			\typerule{T\LA-endorse}{
				\src{\Delta,\Gamma}\vdash\src{e}:\UNS
				\\
				\src{C,\Delta,\Gamma;(x:\varphi)} \vdash\src{s} %
			}{
				\src{C,\Delta,\Gamma}\vdash \src{\myendorse{x}{e}{\varphi}{s}} 
			}{compap-end}
		}
		=&\
			\begin{cases}
				\trg{
					\begin{aligned}
						&
						\destruct{\trg{x}~}{\compap{\src{\Delta,\Gamma}\vdash \src{e} : \UNS}}{\trg{nat}}{
							\\
							&\ 
							\ifztet{\trg{x}}{
							\\
							&\ \ 
							\compap{ \src{C,\Delta,\Gamma;(x:\varphi)} \vdash\src{s} }
							\\
							&\ \ 
							}{
								\ifztet{\trg{x-1}}{
								\\
								&\ \ \ 
								\compap{ \src{C,\Delta,\Gamma;(x:\varphi)} \vdash\src{s} }
								\\
								& \ \ \
								}{\wrong}
							}
							\\
							&
						}{\wrong}
					\end{aligned}
				}
				\\
				\text{ if }\src{\varphi}=\src{\Bools}
				\\
				\\
				\trg{
					\begin{aligned}
						&
						\destruct{\trg{x}~}{\compap{\src{\Delta,\Gamma}\vdash \src{e} : \UNS}}{\trg{nat}}{
							\\
							&\ 
							\compap{ \src{C,\Delta,\Gamma;(x:\varphi)} \vdash\src{s} }
							\\
							&
						}{\wrong}
					\end{aligned}
				}
				\\
				\text{ if }\src{\varphi}=\src{\Nats}
				\\
				\\
				\trg{
					\begin{aligned}
						&
						\destruct{\trg{x}~}{\compap{\src{\Delta,\Gamma}\vdash \src{e} : \UNS}}{\trg{pair}}{
							\\
							&\ 
							\compap{ \src{C,\Delta,\Gamma;(x:\varphi)} \vdash\src{s} }
							\\
							&
						}{\wrong}
					\end{aligned}
				}
				\\
				\text{ if }\src{\varphi}=\src{\UNS\times\UNS}
				\\
				\\
				\trg{
					\begin{aligned}
						&
						\destruct{\trg{x}~}{\compap{\src{\Delta,\Gamma}\vdash \src{e} : \UNS}}{\trg{pair}}{
							\\
							&\ 
							\trg{!\projone{x} \with{\projtwo{x}} ;}
							\\
							&\ \ 
							\compap{ \src{C,\Delta,\Gamma;(x:\varphi)} \vdash\src{s} }
							\\
							&
						}{\wrong}
					\end{aligned}
				}
				\\
				\text{ if }\src{\varphi}=\src{\Refs{\UNS}}
			\end{cases}
	\end{align*}
\end{definition}
We write \wrong~ as a shortcut for a failign expression like \trg{3+\truet}.

The remark about optimisation for \compup{\cdot} in \Cref{sec:compup-opt-deref} is also valid for the \Cref{tr:compap-deref} case above.
As expressions are executed atomically, we are sure that albeit inefficient, dereferencing will correctly succeed.

We can add reference to superficial types and check this dynamically in the source, as we have the heap there.
But how do we check this in the target? %
We only assume that reference must be passed as a pair: location- key from the attacker.
Thus the last case of \Cref{tr:compap-end}, where we check that we can access the location, otherwise we'd get stuck.

\paragraph{NonAtomic Implementation of New-Hide}\label{sec:nonatom-new-hide}
We can also implement \Cref{tr:compap-new} using non-atomic instructions are defined in \Cref{tr:compap-new-nonat} below.
\begin{align*}
	\tag{\compap{\cdot}-New-nonat} 
	\compap{
		\typerule{T\LA-new}{
			\src{\Delta,\Gamma}\vdash \src{e} : \src{\tau} 
			\\
			\src{C,\Delta,\Gamma;x:\Refs{\tau}}\vdash \src{s} %
		}{
			\src{C,\Delta,\Gamma}\vdash \src{\letnewty{x}{e}{\tau}{s}} %
		}{compap-new-nonat}	
	}
	=&\
	\begin{cases}
		\trg{
		\begin{aligned}
			&
			\letnewt{\trg{xo}~}{\compap{\src{\Delta,\Gamma}\vdash \src{e} : \src{\tau}} 
			\\
			&\
			}{ 
				\letint{\trg{x}~}{~\trg{\pair{xo,0}}
				\\
				&\ \ 
				}{
					\compap{\src{C,\Delta,\Gamma;x:\Refs{\tau}}\vdash \src{s}} 
				}
			}
		\end{aligned}
		}
		\\
		\text{if }\src{\tau}=\UNS %
		\\
		\\
		\trg{
			\begin{aligned}
       &
       \letnewt{\trg{x_t}~}{\trg{0}}{
       \\
       &\
         \lethide{\trg{x_k}~}{\trg{x_t}}{
         \\
         &\ \ 
           \letint{\trg{x_c}~}{~ \compap{\src{\Delta,\Gamma}\vdash \src{e} : \src{\tau}} 
           }{
           \\
           &\ \ \ 
             \trg{x_t := x_c \with{x_k};}
           \\
           &\ \ \ 
            \letint{\trg{x}~}{~\trg{\pair{x_t,x_k}}}{
            \\
            &\ \ \ \
             \compap{\src{C,\Delta,\Gamma;x:\Refs{\tau}}\vdash \src{s}}
             }
           }
         }
       }
     \end{aligned}
		}
		\\
		\text{otherwise}
	\end{cases}
\end{align*}

\subsection{Properties of the \LA-\LC Compiler}\label{sec:prop-comp}

\begin{theorem}[Compiler \compap{\cdot} is \ccomp]\label{thm:comp-ap-cc}
	$\vdash\compap{\cdot} : \ccomp$
\end{theorem}

\begin{theorem}[Compiler \compap{\cdot} is \rscomp]\label{thm:comp-ap-rsc}
	$\vdash\compap{\cdot} : \rscomp$
\end{theorem}

\subsection{Cross-language Relation $\relatetbeta$}\label{sec:relatet-def}
We define a more lenient relation on states $\relatetbeta$ analogous to $\relatebeta$ (\Cref{tr:state-rel-whole}) but that ensures that all target locations that are related to secure source ones only vary accordingly: i.e., the attacker cannot change them.

\mytoprule{\src{\Omega}\relatetbeta\trg{\Omega}}
\begin{center}
	\typerule{\LA-Secure heap}{
		\src{H'}=\myset{\src{\ell\mapsto v:\tau}}{\src{\ell\mapsto v:\tau}\in\src{H} \text{ and } \src{\tau}\nvdash\circ}
	}{
	\vdash\sech{\src{H}} = \src{H'}
	}{h-sec-s}
	\typerule{\LC-Low Location}{
		\nexists\src{\ell}\in\sech{\src{H}}
		&
		\src{\ell}\relatebeta\trg{\pair{n,\_}}
		&
		\trg{n}\in\dom{\trg{H}}
	}{
		\src{H},\trg{H}\vdash\lowloc{\trg{n}}
	}{lowloc}
	\typerule{\LC-High Location}{
		\src{\ell}\in\sech{\src{H}}
		&
		\src{\ell}\relatebeta\trg{\pair{n,k}}
		&
		\trg{n\mapsto \_:k}\in\trg{H}
	}{
		\src{H},\trg{H}\vdash\highloc{\trg{n}} = \src{\ell},\trg{k}
	}{highloc}
	\typerule{\LC-High Capability}{
		\src{\ell}\in\sech{\src{H}}\ldotp
		\src{\ell}\relatebeta\trg{\pair{n,k}}
		&
		\trg{n\mapsto \_:k}\in\trg{H}
	}{
		\src{H},\trg{H}\vdash\highcap{\trg{k}}
	}{highcap}
	\typerule{Related states -- Secure}{
		\src{\Omega}=\src{\Delta ; \OB{F},\OB{F'} ; \OB{I} ; H \triangleright \Pi}
		&
		\trg{\Omega}=\trg{H_0 ; \OB{F},\compap{\OB{F'}} ; \OB{I} ; H \triangleright \Pi}
		&
		\src{\Delta}\vdash_{\beta}\trg{H_0}
		\\
		\forall \trg{k}, \trg{n}, \src{\ell} \ldotp \text{ if } \src{H},\trg{H}\vdash\highloc{\trg{n}} = \src{\ell},\trg{k} \text{ then }
			\\
			(1)~ \forall \trg{\pi}\in\trg{\Pi} \text{ if } \src{C}\vdash\trg{\pi}:\trg{attacker} \text{ then } \trg{k}\notin\fun{fv}{\trg{\pi}}
			\\
			(2)~ \forall \trg{n'\mapsto v:\eta} \in \trg{H},
			\\
			(2a) \text{ if } 
				\trg{\eta}=\trg{k}
			\text{ then } 
				\trg{n}=\trg{n'} \text{ and }
				\src{\ell}\relatebeta\trg{\pair{n,k}} \text{ and } \src{\ell\mapsto v:\tau}\in\src{H} \text{ and } \src{v}\relatebeta\trg{v} 
			\\
			(2b) \text{ if }
				\trg{\eta}\neq\trg{k}
			\text{ then } 
				\src{H},\trg{H}\vdash\lowloc{\trg{n'}} \text{ and } \forall \trg{k'}\ldotp\src{H},\trg{H}\vdash\highcap{\trg{k'}}, \trg{v}\neq\trg{k'} 
	}{
		\src{\Omega}\relatetbeta\trg{\Omega}
	}{state-rel-proof}
\end{center}
\botrule
There is no \sech{\cdot} function for the target because they would be all locations that are related to a source location that itself is secure in the source.
An alternative is to define \sech{\cdot} as all locations protected by a key \trg{k} but the point of \sech{\cdot} is to setup the invariant to ensure the proof hold, so this alternative would be misleading.

\Cref{tr:lowloc} tells when a target location is not secure.
That is, when there is no secure source location that is related to it.
This can be because the source location is not secure or because the relation does not exist, as in order for it to exist the triple must be added to $\beta$ and we only add the triple for secure locations.

The intuition behind \Cref{tr:state-rel-proof} is that two states are related if the set of locations they monitor is related and then:
for any target location \trg{n} that is high (i.e., it has a related source counterpart \src{\ell} whose type is secure and that is protected with a capability \trg{k} that we call a high capability), then we have:
(1) the capability \trg{k} used to lock it is not in in any attacker code;
(2) for any target level location \trg{n'}: 
(2a) either it is locked with a high capability \trg{k} (i.e., a capability used to hide a high location) thus \trg{n'} is also high, in which case it is related to a source location \src{\ell} and the values \src{v}, \trg{v} they point to are related;
or (2b) it is not locked with a high capability, so we can derive that \trg{n'} is a low location and its content \trg{v} is not any high capability \trg{k'}.

\BREAK

\begin{lemma}[A target location is either high or low]\label{thm:targ-loc-h-or-l}
	\begin{align*}
		&
		\forall 
		\\
		\text{if }
		&\
		\src{H}\relatebeta\trg{H}
		\\
		&\
		\trg{n\mapsto v:\eta}\in\trg{H}
		\\
		\text{ then either }
		&\
		\src{H},\trg{H}\vdash\lowloc{\trg{n}}
		\\
		\text{ or }
		&\
		\exists \src{\ell}\in\dom{\src{H}}.
		\\
		&\
		\src{H},\trg{H}\vdash\highloc{\trg{n}} = \src{\ell},\trg{\eta}
	\end{align*}
\end{lemma}
\begin{proof}
	Trivial, as \Cref{tr:lowloc} and \Cref{tr:highloc} are duals.
\end{proof}

\newpage

 \newpage
\section{\rscomp: Third Instance with Target Memory Isolation}\label{sec:rsc-iso}

Both compilers presented so far used a capability-based target
language. To avoid giving the false impression that \rscomp is only
useful for this kind of a target, we show here how to attain \rscomp
when the protection mechanism in the target is completely
different. We consider a new target language, \LI, which does not have
capabilities, but instead offers coarse-grained memory isolation based
on \emph{enclaves}. This mechanism is supported (in hardware) in
mainstream x86-64 and ARM CPUs (Intel calls this SGX~\cite{intel}; ARM
calls it TrustZone~\cite{arm}). This is also straightforward to
implement purely in software using any physical, VM-based,
process-based, or in-process isolation technique. 
This section provides a high-level discussion on how to devise
compiler \compai{\cdot} from our source language \LA to \LI and why it attains \rscomp.
Full formal details are presented in subsequent sections.

\subsection{\LI, a Target Language with Memory Isolation}\label{sec:target-iso}

Language \LI replaces \LC's capabilities with a simple security
abstraction called an enclave. An enclave is a collection of code and
memory locations, with the properties that: (a) only code within the
enclave can access the memory locations of the enclave, and (b) Code
from outside can transfer control only to designated entry points in
the enclave's code. For simplicity, \LI supports only one
enclave. Generalizing this to many enclaves is straightforward.

To model the enclave, a \LI program has an additional component
\oth{\OB{E}}, the list of functions that reside in the enclave. A
component thus has the form $\oth{C} \bnfdef\ \oth{H_0 ; \OB{F} ;
  \OB{I} ; \OB{E}}$. Only functions that are listed in \oth{\OB{E}}
can create (\oth{\letiso{x}{e}{s}}), read (\oth{!e}) and write
(\oth{x:=e}) locations in the enclave. Locations in \LI are
\emph{integers} (not natural numbers). By convention, non-negative
locations are outside the enclave (accessible from any function),
while negative locations are inside the enclave (accessible only from
functions in \oth{\OB{E}}). The semantics are almost those of \LC, but
the expression semantics change to $\trg{C;H;f\triangleright e
  \redtoo v}$, recording which function \oth{f} is currently
executing. The operational rule for any memory operation checks that
either the access is to a location outside the enclave or that
$\oth{f} \in \oth{\OB{E}}$ (formalized by
$\oth{C}\vdash\oth{f}:\oth{prog}$).
Monitors of \LI are the same as those of \LC.

\subsection{Compiler from \LA to \LI}\label{sec:comp-ai}
The high-level structure of the compiler \compai{\cdot} is similar to
that of \compap{\cdot}.  \compai{\cdot} ensures
that all the (and only the) functions of the (trusted) component we
write are part of the enclave, i.e., constitute \oth{\OB{E}} (first
rule below).  Additionally, the compiler populates the safety-relevant
heap \oth{H_0} based on the information in \src{\Delta} (captured by
the judgement $\src{\Delta}\vdash\oth{H_0}$, whose details we elide
here).  Importantly, \compai{\cdot} also ensures that trusted
locations are stored in the enclave.  As before, the compiler relies
on typing information for this.  Locations whose types are shareable
(subtypes of \UNS) are placed outside the enclave while those that
trusted (not subtypes of \UNS) are placed inside.

As mentioned, \compai{\cdot} also attains \rscomp. The intuition is
simple: all trusted locations (including safety-relevant locations)
are in the enclave and adversarial code cannot tamper with them.  The
proof follows the proof of the previous compiler: We build a
cross-language relation, which we show to be an invariant on
executions of source and corresponding compiled programs.  The only
change is that every location in the \oth{trusted} \oth{target}
\oth{heap} is isolated in the enclave.

\section{The Second Target Language: \LI}\label{sec:trg-iso}
For clarity, we use a \oth{pink, italics} font for \LI.

\subsection{Syntax}\label{sec:iso-syn}
\begin{align*}
	\mi{Whole\ Programs}~\oth{P} \bnfdef&\ \oth{H_0 ; \OB{F} ; \OB{I} ; \OB{E}}
	\\
	\mi{Components}~\oth{C} \bnfdef&\ \oth{H_0 ; \OB{F} ; \OB{I} ; \OB{E}}
	\\
	\mi{Contexts}~\oth{A} \bnfdef&\ \oth{\OB{F}\hole{\cdot}}
	\\
	\mi{Interfaces}~\oth{I} \bnfdef&\ \oth{f}
	\\
	\mi{Enclave\ functions}~\oth{E} \bnfdef&\ \oth{f}
	\\
	\mi{Functions}~\oth{F} \bnfdef&\ \oth{f(x)\mapsto s;\ret}
	\\
	\mi{Operations}~\oth{\op} \bnfdef&\ \oth{+} \mid \oth{-}
	\\
	\mi{Comparison}~\oth{\bop} \bnfdef&\ \oth{==} \mid \oth{<} \mid \oth{>}
	\\
	\mi{Values}~\oth{v} \bnfdef&\ \oth{n}\in\mb{Z} \mid \oth{\pair{v,v}} \mid \oth{k} %
	\\
	\mi{Expressions}~\oth{e} \bnfdef&\ \oth{x} \mid \oth{v} \mid \oth{e \op e} \mid \oth{e \bop e} \mid \oth{\pair{e,e}} 
	\mid \oth{\projone{e}} \mid \oth{\projtwo{e}} \mid \oth{!e}
	\\
	\mi{Statements}~\oth{s} \bnfdef&\ \skipo \mid \oth{s;s} \mid \oth{\letin{x}{e}{s}} \mid \oth{\ifzte{e}{s}{s}} 
	\mid \oth{\call{f}~e} 
	\\
	\mid&\
	\mid \oth{\fork{s}} \mid \oth{\destructo{x}{e}{B}{s}{s}}
	\\
	\mid&\ \oth{x := e} \mid \oth{\letnew{x}{e}{s}} \mid \oth{\letiso{x}{e}{s}}
	\\
	\mi{Patterns}~\oth{B} \bnfdef&\ \oth{nat} \mid \oth{pair}
	\\
	\mi{Eval.\ Ctxs.}~\oth{E} \bnfdef&\ \oth{\hole{\cdot}} \mid \oth{e \op E} \mid \oth{E \op n} \mid \oth{e \bop E} \mid \oth{E \bop n} 
	\mid \oth{!E}
	\\
	\mid&\ \oth{\pair{e,E}} \mid \oth{\pair{E,v}} \mid \oth{\projone{E}} \mid \oth{\projtwo{E}} 
	\\
	\mi{Heaps}~\oth{H} \bnfdef&\ \othe \mid \oth{H ; n\mapsto v}
	\\
	\mi{Monitors}~\oth{M} \bnfdef&\ \oth{(\set{\sigma},\monred,\sigma_0,H_0,\sigma_c)}
	\\
	\mi{Mon.\ States}~\oth{\sigma} \in&\ \oth{\mc{S}}
	\\
	\mi{Mon.\ Reds.}~\oth{\monred} \bnfdef&\ \othe \mid \oth{\monred;(s,H,s)}
	\\
	\mi{Substitutions}~\oth{\rho} \bnfdef&\ \othe \mid \oth{\rho}\subo{v}{x}
	\\
	\mi{Single\ Process}~ \oth{\pi} \bnfdef&\ \oth{\proc{s}{\OB{f}}}
	\\
	\mi{Processes}~\oth{\Pi} \bnfdef&\ \othe \mid \oth{\Pi \parallel \pi}
	\\
	\mi{Prog.\ States}~\oth{\Omega}\bnfdef&\ \oth{C, H\triangleright \Pi}  %
	\\
	\mi{Labels}~\oth{\lambda} \bnfdef&\ \oth{\epsilon} \mid \oth{\alpha}
	\\
	\mi{Actions}~\oth{\alpha} \bnfdef&\ \oth{\clh{f}{v}{H}} \mid \oth{\cbh{f}{v}{H}} \mid \oth{\rth{v}{H}} \mid \oth{\rbh{v}{H}}
	\\
	\mi{Traces}~\oth{\OB{\alpha}} \bnfdef&\ \othe \mid \oth{\OB{\alpha}\cdot\alpha}
\end{align*}

\subsection{Operational Semantics of \LI}\label{sec:iso-sem}
\mytoprule{\text{Helpers}}
\begin{center}
	\typerule{\LI-Jump-Internal}{
		((\oth{f'}\in\oth{\OB{I}} \wedge \oth{f}\in\oth{\OB{I}}) \vee
				\\
		(\oth{f'}\notin\oth{\OB{I}} \wedge \oth{f}\notin\oth{\OB{I}}))
	}{
		\oth{\OB{I}}\vdash\oth{f,f'}:\oth{internal}
	}{o-aux-intern}
	\typerule{\LI-Jump-IN}{
		\oth{f}\in\oth{\OB{I}} \wedge \oth{f'}\notin\oth{\OB{I}}
	}{
		\oth{\OB{I}}\vdash\oth{f,f'}:\oth{in}
	}{o-aux-in}
	\typerule{\LI-Jump-OUT}{
		\oth{f}\notin\oth{\OB{I}} \wedge \oth{f'}\in\oth{\OB{I}}
	}{
		\oth{\OB{I}}\vdash\oth{f,f'}:\oth{out}
	}{o-aux-out}
	\typerule{\LI-prog-execs}{
		 \oth{C} = \oth{H_0 ; \OB{F} ; \OB{I} ; \OB{E}}
		 &
		 \oth{f}\in\oth{\OB{E}}
	}{
		\oth{C}\vdash\oth{f}:\oth{prog}
	}{o-aux-prog}
	\typerule{\LI-Plug}{
		\oth{A} \equiv \oth{\OB{F}\hole{\cdot}}
		&
		\oth{C}\equiv\oth{H_0 ; \OB{F'} ; \OB{I} ; \OB{E} } 
		\\
		\vdash\oth{C,\OB{F}}:\oth{whole}
		&
		\oth{main(x)\mapsto s;\ret}\in\oth{\OB{F}}
	}{
		\oth{A\hole{C}} = \oth{H_0; \OB{F;F'}; \OB{I} ; \OB{E}}
	}{plug-o}
	\typerule{\LI-Whole}{
		\oth{C}\equiv\oth{H_0 ; \OB{F'} ; \OB{I} ; \OB{E}} 
		\\
		\fun{names}{\oth{\OB{F}}}\cap\fun{names}{\oth{\OB{F'}}}=\emptyset
		&
		\fun{names}{\oth{\OB{I}}}\subseteq \fun{names}{\oth{\OB{F}}}
		&
		\forall \oth{n\mapsto v}\in\oth{H_0}, \oth{n}<\oth{0}
	}{
		\vdash\oth{C,\OB{F}}:\oth{whole}
	}{whole-o}
	\typerule{\LI-Initial State}{
		\oth{P}\equiv\oth{H_0 ; \OB{F} ; \OB{I} ; \OB{E}}
		&
		\oth{main(x)\mapsto s;\ret}\in\oth{\OB{F}}
	}{
		\SInito{P} = \oth{P; H_0 \triangleright \proc{s\subo{0}{x}}{main}}
	}{ini-o}
\end{center}
\botrule

\subsubsection{Component Semantics}\label{src:trg-sem-com}
\begin{align*}
	&\oth{C; H ; \OB{f} \triangleright e \redtoo e'} 
	&&\text{Expression \oth{e} reduces to \oth{e'}.}
	\\
	&\oth{C, H \triangleright \Pi} \xtoo{\epsilon} \oth{C', H' \triangleright \Pi'} 
	&&\text{Processes \oth{\Pi} reduce to \oth{\Pi'} and evolve the rest accordingly.}
	\\
	&&&\text{emitting label \oth{\lambda}.}
	\\
	&\oth{\Omega} \Xtoo{\OB{\alpha}} \oth{\Omega'}
	&& \text{Program state \oth{\Omega} steps to \oth{\Omega'} emitting trace \oth{\OB{\alpha}}.}
	\\
\end{align*}

\mytoprule{\oth{C; H ; \OB{f} \triangleright e \redtoo e'} }
\begin{center}
	\typerule{E\LI-val}{
	}{
		\oth{C; H ; \OB{f} \triangleright v} \redtoo \oth{v}
	}{eo-val}
	\typerule{E\LI-p1}{
	}{
		\oth{C; H ; \OB{f} \triangleright \projone{\pair{v,v'}} \redtoo v}
	}{eo-p1}
	\typerule{E\LI-p2}{
	}{
		\oth{C; H ; \OB{f} \triangleright \projone{\pair{v,v'}} \redtoo v'}
	}{eo-p2}
	\typerule{E\LI-op}{
		n\op n'=n''
	}{
		\oth{C; H ; \OB{f} \triangleright n \op n' \redtoo n''}
	}{eo-op}
	\typerule{E\LI-comp}{
		\text{if } n\bop n'= \text{true} \text{ then } \oth{n''}=\oth{0} \text{ else } \oth{n''}=\oth{1}
	}{
		\oth{C; H ; \OB{f} \triangleright n \bop n' \redtoo n''}
	}{eo-bop}
	\typerule{E\LI-deref}{
		\oth{n\mapsto v}\in\oth{H}
		&
		\oth{n}\geq \oth{0}
	}{
		\oth{C; H ; \OB{f} \triangleright !n } \redtoo \oth{v}
	}{eo-de-t}
	\typerule{E\LI-deref-iso}{
		\oth{n\mapsto v}\in\oth{H}
		&
		\oth{n}< \oth{0}
		&
		\oth{C}\vdash\oth{f}:\oth{prog}
	}{
		\oth{C; H ; \OB{f;f} \triangleright !n } \redtoo \oth{v}
	}{eo-de-k}
	\typerule{E\LI-ctx}{
		\oth{C; H ; \OB{f} \triangleright e \redtoo e'}
	}{
		\oth{C; H ; \OB{f} \triangleright E\hole{e} \redtoo E\hole{e'}}
	}{eo-cth}
\end{center}
\botrule

\mytoprule{\oth{C; H \triangleright s} \xtoo{\lambda} \oth{C'; H' \triangleright s'} }

We elide the suffix with the stack of functions when obvious.
\begin{center}
	\typerule{E\LI-sequence}{
	}{
		\oth{C, H \triangleright \skipo;s} \xtoo{\epsilon} \oth{C, H \triangleright s}
	}{eo-seq}
	\typerule{E\LI-step}{
		\oth{C, H \triangleright s} \xtoo{\lambda} \oth{C, H \triangleright s'}
	}{
		\oth{C, H \triangleright s;s''} \xtoo{\lambda} \oth{C, H \triangleright s';s}
	}{eo-step}
	\typerule{E\LI-if-true}{
		\oth{C; H ; \OB{f} \triangleright e\redtoo 0}
	}{
		\oth{C, H \triangleright \ifzte{e}{s}{s'}} \xtoo{\epsilon} \oth{C, H \triangleright s}
	}{eo-ift}
	\typerule{E\LI-if-false}{
		\oth{C; H ; \OB{f} \triangleright e\redtoo n}
		&
		\oth{n}\not\equiv\oth{0}
	}{
		\oth{C, H \triangleright \ifzte{e}{s}{s'}} \xtoo{\epsilon} \oth{C, H \triangleright s'}
	}{eo-iff}
	\typerule{E\LI-letin}{
		\oth{C; H ; \OB{f} \triangleright e\redtoo v}
	}{
		\oth{C, H \triangleright \letin{x}{e}{s}} \xtoo{\epsilon} \oth{C, H \triangleright s\subo{v}{x}}
	}{eo-letin}
	\typerule{E\LI-new}{
		\oth{H} = \oth{H_1;n\mapsto \_}
		&
		\oth{C; H ; \OB{f} \triangleright e\redtoo v}
	}{
		\oth{C, H \triangleright \letnew{x}{e}{s}} \xtoo{\epsilon} \oth{C,H; n+1\mapsto v \triangleright s\subo{n+1}{x}}
	}{eo-nu}
	\typerule{E\LI-isolate}{
		\oth{H} = \oth{n\mapsto \_;H_1}
		&
		\oth{C; H ; \OB{f} \triangleright e\redtoo v}
		\\	
		\oth{\OB{f}} = \oth{\OB{f'}\cdot f}
		&
		\oth{C}\vdash\oth{f}:\oth{prog}
	}{
		\oth{C, H \triangleright \proc{\letiso{x}{e}{s}}{\OB{f}}} \xtoo{\epsilon} \oth{C, n-1\mapsto v;H \triangleright \proc{s\subo{n-1}{x}}{\OB{f}}}
	}{eo-hi}
	\typerule{E\LI-assign}{
		\oth{C; H ; \OB{f} \triangleright e\redtoo v}
		\\
		\oth{H} = \oth{H_1;n\mapsto \_;H_2}
		&
		\oth{H'} = \oth{H_1;n\mapsto v;H_2}
		&
		\oth{n}\geq\oth{0}
	}{
		\oth{C, H\triangleright n:=e} \xtoo{\epsilon} \oth{C, H'\triangleright \skipo}
	}{eo-ac-t}
	\typerule{E\LI-assign-iso}{
		\oth{C; H ; \OB{f} \triangleright e\redtoo v}
		&
		\oth{\OB{f}} = \oth{\OB{f'}\cdot f}
		&
		\oth{C}\vdash\oth{f}:\oth{prog}
		\\
		\oth{H} = \oth{H_1;n\mapsto \_;H_2}
		&
		\oth{H'} = \oth{H_1;n\mapsto v;H_2}
		&
		\oth{n}<\oth{0}
	}{
		\oth{C, H\triangleright \proc{n:=e}{\OB{f}} } \xtoo{\epsilon} \oth{C, H'\triangleright \proc{\skipo}{\OB{f}}}
	}{eo-ac-k}
	\typerule{E\LI-call-internal}{
		\oth{\OB{C}.\mtt{intfs}}\vdash\oth{f,f'}:\oth{internal}
		&
		\oth{\OB{f'}} = \oth{\OB{f''};f'}
		\\
		\oth{f(x)\mapsto s;\ret}\in\oth{C}.\mtt{funs}
		&
		\oth{C; H ; \OB{f} \triangleright e\redtoo v}
	}{
		\oth{C, H \triangleright \proc{{\call{f}~e}}{\OB{f'}}} \xtoo{\epsilon} \oth{C, H \triangleright \proc{{s;\ret\subo{v}{x}}}{\OB{f'};f}}
	}{eo-call-i}
	\typerule{E\LI-callback}{
		\oth{\OB{f'}} = \oth{\OB{f''};f'}
		&
		\oth{f(x)\mapsto s;\ret}\in\oth{\OB{F}}
		\\
		\oth{\OB{C}.\mtt{intfs}}\vdash\oth{f',f}:\oth{out}
		&
		\oth{C; H ; \OB{f} \triangleright e\redtoo v}
	}{
		\oth{C, H \triangleright \proc{{\call{f}~e}}{\OB{f'}}} \xtoo{\cbh{f}{v}{H}} \oth{C, H \triangleright \proc{{s;\ret\subo{v}{x}}}{\OB{f'};f}}
	}{eo-callback}
	\typerule{E\LI-call}{
		\oth{\OB{f'}} = \oth{\OB{f''};f'}
		&
		\oth{f(x)\mapsto s;\ret}\in\oth{C}.\mtt{funs}
		\\
		\oth{\OB{C}.\mtt{intfs}}\vdash\oth{f',f}:\oth{in}
		&
		\oth{C; H ; \OB{f} \triangleright e\redtoo v}
	}{
		\oth{C, H \triangleright \proc{{\call{f}~e}}{\OB{f'}}} \xtoo{\clh{f}{v}{H}} \oth{C, H \triangleright \proc{{s;\ret\subo{v}{x}}}{\OB{f'};f}}
	}{eo-call}
	\typerule{E\LI-reo-internal}{
		\oth{\OB{C}.\mtt{intfs}}\vdash\oth{f,f'}:\oth{internal}
		&
		\oth{\OB{f'}} = \oth{\OB{f''};f'}
	}{
		\oth{C, H \triangleright \proc{{\ret}}{\OB{f'};f}} \xtoo{\epsilon} \oth{C, H \triangleright \proc{\skipo}{\OB{f'}}}
	}{eo-reo-i}
	\typerule{E\LI-retback}{
		\oth{\OB{C}.\mtt{intfs}}\vdash\oth{f,f'}:\oth{in}
		&
		\oth{\OB{f'}} = \oth{\OB{f''};f'}
	}{
		\oth{C, H \triangleright \proc{{\ret}}{\OB{f'};f}} \xtoo{\rbh{}{H}} \oth{C, H \triangleright \proc{\skipo}{\OB{f'}}}
	}{eo-retb}
	\typerule{E\LI-return}{
		\oth{\OB{C}.\mtt{intfs}}\vdash\oth{f,f'}:\oth{out}
		&
		\oth{\OB{f'}} = \oth{\OB{f''};f'}
	}{
		\oth{C, H \triangleright \proc{{\ret}}{\OB{f'};f}} \xtoo{\rth{}{H}} \oth{C, H \triangleright \proc{\skipo}{\OB{f'}}}
	}{eo-ret}
	\typerule{E\LI-destruct-nat}{
		\oth{C; H ; \OB{f} \triangleright e \redtot n}
	}{
		\oth{C, H \triangleright \destructo{x}{e}{nat}{s}{s'} } \xtoo{\epsilon} \oth{C, H \triangleright s\subo{n}{x}}
	}{eo-end-n}
	\typerule{E\LI-destruct-pair}{
		\oth{C; H ; \OB{f} \triangleright e \redtot \pair{v,v'}}
	}{
		\oth{C, H \triangleright \destructo{x}{e}{pair}{s}{s'} } \xtoo{\epsilon} \oth{C, H \triangleright s\subo{\pair{v,v'}}{x}}
	}{eo-end-p}
	\typerule{E\LI-destruct-not}{
		\text{ otherwise }
	}{
		\oth{C, H \triangleright \destructo{x}{e}{B}{s}{s'} } \xtoo{\epsilon} \oth{C, H \triangleright s'}
	}{eo-end-nope}
	\end{center}
\botrule

\mytoprule{\oth{C, H \triangleright \Pi} \redtoo \oth{C', H' \triangleright \Pi'}}
\begin{center}
	\typerule{E\LI-par}{
		\oth{\Pi}=\oth{\Pi_1 \parallel \proc{s}{\OB{f}}\parallel \Pi_2}
		\\
		\oth{\Pi'}=\oth{\Pi_1 \parallel \proc{s'}{\OB{f'}} \parallel \Pi_2}
		\\
		\oth{C, H \triangleright \proc{s}{\OB{f}}} \redtoo \oth{C', H'\triangleright \proc{s'}{\OB{f'}}}
	}{
		\oth{C, H \triangleright \Pi} \redtoo \oth{C', H' \triangleright \Pi'}
	}{eo-par}
	\typerule{E\LI-fork}{
		\oth{\Pi}=\oth{\Pi_1 \parallel \proc{{\fork{s}}}{\OB{f;f}} \parallel \Pi_2}
		\\
		\oth{\Pi'}=\oth{\Pi_1 \parallel \proc{{0}}{\OB{f;f}} \parallel \Pi_2 \parallel \proc{s}{f}}
	}{
		\oth{C, H \triangleright \Pi} \redtoo \oth{C, H \triangleright \Pi'}
	}{eo-fork}
\end{center}
\botrule

\mytoprule{ \oth{\Omega} \Xtoo{\OB{\alpha}} \oth{\Omega'} }
\begin{center}
	\typerule{E\LI-single}{
		\oth{\Omega}\xtoo{\alpha}\oth{\Omega'}
	}{
		\oth{\Omega}\Xtoo{\alpha}\oth{\Omega'}
	}{eo-tr-sin}
	\typerule{E\LI-silent}{
		\oth{\Omega}\xtoo{\epsilon}\oth{\Omega'}
	}{
		\oth{\Omega}\Xtoo{}\oth{\Omega'}
	}{eo-tr-silent}
	\typerule{E\LI-trans}{
		\oth{\Omega}\Xtoo{\OB{\alpha}}\oth{\Omega''}
		\\
		\oth{\Omega''}\Xtoo{\OB{\alpha'}}\oth{\Omega'}
	}{
		\oth{\Omega}\Xtoo{\OB{\alpha}\cdot\OB{\alpha'}}\oth{\Omega'}
	}{eo-tr-trans}
\end{center}
\botrule

\subsection{Monitor Semantics}
\mytoprule{\monh{\cdot}}
\begin{center}
	\typerule{\LI-Monitor-related heap}{
          \oth{H'} = \{\oth{n\mapsto v:\eta} ~|~ \oth{n} \in \dom{\oth{H_0}} \text{ and } \oth{n\mapsto v:\eta}\in\oth{H}\}
	}{
		\monh{\oth{H},\oth{H_0}} = \oth{H'}
	}{oh-pub-t}
\end{center}
\botrule

\mytoprule{\oth{M;H\monred M'}}
\begin{center}
	\typerule{\LI-Monitor Step}{
		\oth{M}= \oth{(\set{\sigma},\monred,\sigma_0,H_0,\sigma_c)}
		&
		\oth{M'}= \oth{(\set{\sigma},\monred,\sigma_0,H_0,\sigma_f)}
		\\
		\oth{(s_c,\monh{\oth{H},\oth{H_0}},s_f)}\in\oth{\monred}
	}{
		\oth{M;H\monred M'}
	}{mo-t-2}
	\typerule{\LI-Monitor Step Trace Base}{
	}{
		\oth{M;\othe\monred M}
	}{mo-t-s-b}
	\typerule{\LI-Monitor Step Trace}{
		\oth{M;\OB{H}\monred M''}
		&
		\oth{M'';H\monred M'}
	}{
		\oth{M;\OB{H}\cdot H\monred M'}
	}{mo-t-s}
	\typerule{\LI-valid trace}{
		\oth{M;\OB{H}\monred M'}
	}{
		\oth{M}\vdash\oth{\OB{\mnh{H}}}
	}{mo-valid}
\end{center}
\botrule

\subsection{Monitor Agreement for \LI}
\begin{definition}[\LI: \oth{M\agree C}]\label{def:li-agree}
	\begin{align*}
		\oth{(\set{\sigma},\monred,\sigma_0,H_0,\sigma_c) \agree (H_0 ; \OB{F} ;\OB{I} ; \OB{E})}
	\end{align*}
\end{definition}
A monitor and a component agree if they focus on the same set of locations \oth{H_0}.

\subsection{Properties of \LI}
\begin{definition}[\LI Attacker]\label{def:oth-att}
	\begin{align*}
		\oth{C}\vdashatto\oth{A} \isdef&\ \oth{C}= \oth{H_0 ; \OB{F} ; \OB{I} ; \OB{E}}, \oth{A} = \oth{\OB{F'}}, \fun{names}{\oth{\OB{F}}}\cap\fun{names}{\oth{\OB{F'}}}=\emptyset
		\\
		\oth{C}\vdashatto\oth{\pi}  \isdef&\ \oth{\pi}=\oth{\proc{s}{\OB{f};f}} \text{ and } \oth{f}\in\oth{C}.\mtt{itfs}
		\\
		\oth{C}\vdashatto\oth{\Pi}\xtoo{}\oth{\Pi'} \isdef&\ \oth{\Pi}=\oth{\Pi_1 \parallel \pi \parallel \Pi_2} \text { and } \oth{\Pi'}=\oth{\Pi_1 \parallel \pi' \parallel \Pi_2}
		\\
		&\ \text{and } \oth{C}\vdashatto\oth{\pi}  \text{ and } \oth{C}\vdashatto\oth{\pi'} 
	\end{align*}
\end{definition}

\section{Second Compiler from \LA to \LI}\label{sec:comp3}

For this compiler we need a different partial bijection, which we indicate with $\varphi$ and its type is $\src{\ell}\times\oth{n}$.
It has the same properties of $\beta$ listed in \Cref{sec:cr-rel}.

The cross-language relation $\relate$ is unchanged but for the relation of locations, as they are no longer compiled as pairs:
\begin{itemize}
	\item $\src{\ell}\relatephi\oth{n}$ if $(\src{\ell},\oth{n})\in\varphi$
\end{itemize}

Actions relation is unchanged from \Cref{tr:rel-cl} etc.

Heaps relation is unchanged (modulo the elision of capabilities) from \Cref{tr:hrel-i}.

Process relation is unchanged from \Cref{tr:sing-proc-rel} etc.

State relation is unchanged from \Cref{tr:state-rel-whole}.

The monitor relation $\src{M}\relate\oth{M}$ is defined as in \Cref{tr:mon-rel-ap}.

Some auxiliary functions are changed:

\mytoprule{\src{\Delta}\vdash_{\varphi}\oth{H_0} ~~ \src{\Delta, \oth{H} }\vdash\oth{v}: \src{\tau} }
\begin{center}
	\typerule{Initial-heap}{
		\src{\Delta}\vdash\oth{H}
		&
		\src{\Delta},\oth{H}\vdash\oth{v}\src{:\tau}
		\\
		\src{\ell}\relatephi\oth{n}
	}{
		\src{\Delta,\ell:\tau}\vdash_{\varphi}\oth{H;n\mapsto v}
	}{ini-heap-o}
	\typerule{Initial-value}{
			(\src{\tau}\equiv\Bools \wedge \oth{v}\equiv\oth{0}) 
			&
			\vee
			&
			(\src{\tau}\equiv\Nats \wedge \oth{v}\equiv\oth{0}) 
			&
			\vee
			\\
			(\src{\tau}\equiv\src{\Refs{\tau}} \wedge \oth{v}\equiv\oth{n'} \wedge \oth{n'\mapsto v'}\in\oth{H} \wedge \src{\ell'}\relatephi\oth{n'} \wedge \src{\ell:\tau}\in\src{\Delta}, 
			\src{\Delta},\oth{H}\vdash\oth{v'}\src{:\tau}
			) 
			&
			\vee
			\\
			(\src{\tau}\equiv\src{\tau_1\times\tau_2} \wedge \oth{v}\equiv\oth{\pair{v_1,v_2}} \wedge 
			\src{\Delta},\oth{H}\vdash\oth{v_1}\src{:\tau_1}\wedge
			\src{\Delta},\oth{H}\vdash\oth{v_2}\src{:\tau_2}
			)
	}{
		\src{\Delta},\oth{H}\vdash_\varphi\oth{v}\src{:\tau}
	}{ini-val-o}
\end{center}
\botrule

\begin{definition}[Compiler \LA to \LI]\label{def:comp-la-li}
	$\compai{\cdot} : \src{C}\to \oth{C}$

	Given that $\src{C}=\src{\Delta ; \OB{F} ; \OB{I}}$
	if $\vdash \src{C :\UNS}$ 
	then \compai{\src{C}} is defined as follows:
	\begin{align*}
		\tag{\compai{\cdot}-Component}
		\compai{
			\typerule{T\LA-component}{
				\src{C}\equiv\src{\Delta ; \OB{F} ; \OB{I}} 
				\\
				\src{C}\vdash\src{\OB{F}}:\UNS
				\\
				\fun{names}{\src{\OB{F}}}\cap\fun{names}{\src{\OB{I}}}=\emptyset
				\\
				\src{\Delta}\vdash\src{ok}
			}{
				\vdash \src{C}:\UNS
			}{compai-co}
		} &= \oth{%
		\oth{H_0} ;
		\compai{\OB{F}}; \compai{\OB{I}}} ; 
		\dom{\OB{\src{F}}}
		\qquad\qquad
		\text{if }\src{\Delta}\vdash_{\varphi_0}\oth{H_0}
		\\
		\tag{\compai{\cdot}-Function}
		\compai{
			\typerule{T\LA-function}{
				\src{F}\equiv \src{f(x:\UNS)\mapsto s;\ret}
				\\
				\src{C},\src{\Delta;x:\UNS}\vdash \src{s} %
				\\
				\forall\src{f}\in\fun{fn}{s}, \src{f}\in\dom{\src{C}.\mtt{funs}} 
				\\
				\vee \src{f}\in\dom{\src{C}.\mtt{intfs}}
			}{
				\src{C}\vdash\src{F}:\src{\UNS} 
			}{compai-fu}
		} &= \oth{f(x)\mapsto \compai{\src{C};\src{\Delta;x:\UNS}\vdash \src{s}};\ret}
		\\
		\tag{\compai{\cdot}-Interfaces}
		\compai{f} &= \oth{f}
	\end{align*}

	\mytoprule{Expressions}
	\begin{align*}
		\\
		\tag{\compai{\cdot}-True}
		\compai{
			\typerule{T\LA-true}{
				\src{\Delta,\Gamma}\vdash\diamond
			}{
				\src{\Delta,\Gamma}\vdash\trues:\Bools
			}{compai-true}
		}
		=&\ \oth{0}  
		\\
		\tag{\compai{\cdot}-False}
		\compai{
			\typerule{T\LA-false}{
				\src{\Delta,\Gamma}\vdash\diamond
			}{
				\src{\Delta,\Gamma}\vdash\falses:\Bools
			}{compai-false}
		}
		=&\ \oth{1} 
		\\
		\tag{\compai{\cdot}-Nat}
		\compai{
			\typerule{T\LA-nat}{
				\src{\Delta,\Gamma}\vdash\diamond
			}{
				\src{\Delta,\Gamma}\vdash\src{n}:\Nats
			}{compai-nat}
		}
		=&\ \oth{n}
		\\
		\tag{\compai{\cdot}-Var} 
		\compai{
			\typerule{T\LA-var}{
				\src{x:\tau}\in\src{\Gamma}
			}{
				\src{\Delta,\Gamma}\vdash\src{x}:\src{\tau}
			}{compai-var}
		}
		=&\  \oth{x}
		\\
		\tag{\compai{\cdot}-Loc} 
		\compai{
			\typerule{T\LA-loc}{
				\src{\ell:\tau}\in\src{\Delta}
			}{
				\src{\Delta,\Gamma}\vdash\src{\ell}:\src{\tau}
			}{compai-loc}
		}
		=&\  
			\oth{n}
		\\
		\tag{\compai{\cdot}-Pair} 
		\compai{
			\typerule{T\LA-pair}{
				\src{\Delta,\Gamma}\vdash \src{e_1} : \src{\tau}
				\\
				\src{\Delta,\Gamma}\vdash \src{e_2} : \src{\tau'} 
			}{
				\src{\Delta,\Gamma}\vdash \src{\pair{e_1,e_2}} : \src{\tau\times\tau'}
			}{compai-pair}
		}
		=&\ \oth{\pair{\compai{\src{\Delta,\Gamma}\vdash \src{e_1} : \src{\tau}},\compai{\src{\Delta,\Gamma}\vdash \src{e_2} : \src{\tau'}}}}
		\\
		\tag{\compai{\cdot}-P1} 
		\compai{
			\typerule{T\LA-proj-1}{
				\src{\Delta,\Gamma}\vdash \src{e} : \src{\tau\times\tau'}
			}{
				\src{\Delta,\Gamma}\vdash \src{\projone{e}} : \src{\tau}
			}{compai-p1}
		}
		=&\ \oth{\projone{\compai{\src{\Delta,\Gamma}\vdash \src{e} : \src{\tau\times\tau'}}}}
		\\
		\tag{\compai{\cdot}-P2} 
		\compai{
			\typerule{T\LA-proj-2}{
				\src{\Delta,\Gamma}\vdash \src{e} : \src{\tau\times\tau'}
			}{
				\src{\Delta,\Gamma}\vdash \src{\projtwo{e}} : \src{\tau'}
			}{compai-p2}
		}
		=&\ \oth{\projtwo{\compai{\src{\Delta,\Gamma}\vdash \src{e} : \src{\tau\times\tau'}}}}
		\\
		\tag{\compai{\cdot}-Deref} 
		\compai{
			\typerule{T\LA-dereference}{
				\src{\Delta,\Gamma}\vdash \src{e} : \src{\Refs{\tau}} 
			}{
				\src{\Delta,\Gamma}\vdash \src{!e} : \src{\tau}  
			}{compai-deref}
		}
		=&\ \oth{
			!\projone{\compai{\src{\Delta,\Gamma}\vdash \src{e} : \src{\Refs{\tau}}}}
		}
		\\
		\tag{\compai{\cdot}-op}
		\compai{
			\typerule{T\LA-op}{
				\src{\Delta,\Gamma}\vdash \src{e} : \Nats
				\\
				\src{\Delta,\Gamma}\vdash \src{e'} : \Nats
			}{
				\src{\Delta,\Gamma}\vdash \src{e \op e'} : \src{\Nats}  
			}{compai-op}
		}
		=&\ \oth{\compai{\src{\Delta,\Gamma}\vdash \src{e} : \Nats} \op\compai{\src{\Delta,\Gamma}\vdash \src{e'} : \Nats}}
		\\
		\tag{\compai{\cdot}-cmp}
		\compai{
			\typerule{T\LA-cmp}{
				\src{\Delta,\Gamma}\vdash \src{e} : \Nats
				\\
				\src{\Delta,\Gamma}\vdash \src{e'} : \Nats
			}{
				\src{\Delta,\Gamma}\vdash \src{e \bop e'} : \Bools
			}{compai-bop}
		}
		=&\ 
			\oth{
					\compai{\src{\Delta,\Gamma}\vdash \src{e} : \Nats} \bop \compai{\src{\Delta,\Gamma}\vdash \src{e'} : \Nats}
				}
		\\
		\tag{\compai{\cdot}-Coerce} 
		\compai{
			\typerule{T\LA-coercion}{
				\src{\Delta,\Gamma}\vdash \src{e} : \src{\tau} 
				&
				\src{\tau} \vdash \circ
			}{
				\src{\Delta,\Gamma}\vdash \src{e} : \UNS 
			}{compai-coe}
		}
		=&\
				\oth{
					\compai{\src{\Delta,\Gamma}\vdash \src{e} : \src{\tau} }
				}
	\end{align*}

	\mytoprule{Statements}
	\begin{align*}
		\tag{\compai{\cdot}-Skip} 
		\compai{
			\typerule{T\LA-skip}{
			}{
				\src{C,\Delta,\Gamma}\vdash \skips
			}{compai-skip}	
		}
		=&\
		\skipo
		\\
		\tag{\compai{\cdot}-New} 
		\compai{
			\typerule{T\LA-new}{
				\src{C,\Delta,\Gamma}\vdash \src{e} : \src{\tau} 
				\\
				\src{C,\Delta,\Gamma;x:\Refs{\tau}}\vdash \src{s} %
			}{
				\src{C,\Delta,\Gamma}\vdash \src{\letnewty{x}{e}{\tau}{s}} %
			}{compai-new}	
		}
		=&\
		\begin{cases}
			\oth{
			\begin{aligned}
				&
				\letnewo{\oth{x}~}{\compai{\src{\Delta,\Gamma}\vdash \src{e} : \src{\tau}} 
				\\
				&\
				}{ 
					\compai{\src{C,\Delta,\Gamma;x:\Refs{\tau}}\vdash \src{s}} 
				}
			\end{aligned}
			}
			\\
			\text{if }\src{\tau}=\UNS %
			\\
			\\
			\oth{
			\begin{aligned}
				&
				\letiso{\oth{x}~}{ \compai{\src{\Delta,\Gamma}\vdash \src{e} : \src{\tau}} 
				\\
				&\
				}{ 
					\compai{\src{C,\Delta,\Gamma;x:\Refs{\tau}}\vdash \src{s}} 
				}
			\end{aligned}
			}
			\\ 
			\text{else}
		\end{cases}
		\\
		\tag{\compai{\cdot}-call}
		\compai{
			\typerule{T\LA-function-call}{
				((\src{f}\in\dom{\src{C}.\mtt{funs}})
				\\
				\vee(\src{f}\in\dom{\src{C}.\mtt{intfs}}))
				\\
				\src{\Delta,\Gamma}\vdash \src{e} : \src{\UNS}
			}{
				\src{\Delta,\Gamma}\vdash \src{\call{f}~e}
			}{compai-fun}
		}
		=&\ \oth{\call{f}~\compai{\src{\Delta,\Gamma}\vdash \src{e}: \UNS}}
		\\
		\tag{\compai{\cdot}-If} 
		\compai{
			\typerule{T\LA-if}{
				\src{\Delta,\Gamma}\vdash \src{e} : \src{\Bool}
				\\
				\src{C,\Delta,\Gamma}\vdash \src{s_t} 
				\\
				\src{C,\Delta,\Gamma}\vdash \src{s_e} 
			}{
				\src{C,\Delta,\Gamma}\vdash \src{\ifte{e}{s_t}{s_e}} 
			}{compai-if}
		}
		=&\ \oth{
			\begin{aligned}
				&\ifzteo{\compai{\src{\Delta,\Gamma}\vdash \src{e} : \src{\Bool}}
				\\
				&}
				{\compai{\src{C,\Delta,\Gamma}\vdash \src{s_t} }
				\\
				&
				}{\compai{\src{C,\Delta,\Gamma}\vdash \src{s_e} }}
			\end{aligned}
		}
		\\
		\tag{\compai{\cdot}-Seq} 
		\compai{
			\typerule{T\LA-sequence}{
				\src{C,\Delta,\Gamma}\vdash \src{s_u} 
				\\
				\src{C,\Delta,\Gamma}\vdash \src{s} 
			}{
				\src{C,\Delta,\Gamma}\vdash \src{s_u;s} 
			}{compai-seq}
		}
		=&\ \oth{\compai{\src{C,\Delta,\Gamma}\vdash \src{s_u} } ; \compai{\src{C,\Delta,\Gamma;\Gamma'}\vdash \src{s} }}
		\\
		\tag{\compai{\cdot}-Letin} 
		\compai{
			\typerule{T\LA-letin}{
				\src{\Delta,\Gamma}\vdash \src{e} : \src{\tau} 
				\\
				\src{C,\Delta,\Gamma;x:\tau}\vdash \src{s}
			}{
				\src{C,\Delta,\Gamma}\vdash \src{\letin{x:\tau}{e}{s}}
			}{compai-vardef}
		}
		=&\ \oth{
			\begin{aligned}
				&
				\letino{\oth{x}}{\compai{\src{\Delta,\Gamma}\vdash \src{e} : \src{\tau}}
				\\
				&
				}{\compai{\src{C,\Delta,\Gamma;x:\tau}\vdash \src{s} }}
			\end{aligned}
			}
		\\
		\tag{\compai{\cdot}-Assign}
		\compai{
			\typerule{T\LA-assign}{
				\src{\Delta,\Gamma}\vdash \src{x} : \src{\Refs{\tau}}
				\\
				\src{\Delta,\Gamma}\vdash \src{e} : \src{\tau} 
			}{
				\src{C,\Delta,\Gamma}\vdash \src{x := e} 
			}{compai-ass}
		}
		=&\ \oth{
			\compai{\src{\Delta,\Gamma}\vdash \src{x} : \Refs{\tau}} := \compai{\src{\Delta,\Gamma}\vdash \src{e} : \src{\tau}}
		}
		\\
		\tag{\compai{\cdot}-Fork} 
		\compai{
			\typerule{T\LA-fork}{
				\src{C,\Delta,\Gamma}\vdash \src{s} 
			}{
				\src{C,\Delta,\Gamma}\vdash \src{\fork{s}} 
			}{compai-fork}
		}
		=&\ \oth{\fork{\compai{\src{C,\Delta,\Gamma}\vdash \src{s} }}}
		\\
		\tag{\compai{\cdot}-Proc} 
		\compai{
			\typerule{T\LA-process}{
				\src{C,\Delta,\Gamma}\vdash \src{s} 
			}{
				\src{C,\Delta,\Gamma}\vdash \src{\proc{s}{\OB{f}}}
			}{compai-proc}
		}
		=&\ 
			\begin{aligned}
				&
				\oth{\proc{\compai{\src{C,\Delta,\Gamma}\vdash \src{s} }}{\compai{\OB{f}}}}
			\end{aligned}
		\\
		\tag{\compai{\cdot}-Soup} 
		\compai{
			\typerule{T\LA-soup}{
				\src{C,\Delta,\Gamma}\vdash \src{\pi} 
				\\
				\src{C,\Delta,\Gamma}\vdash \src{\Pi}
			}{
				\src{C,\Delta,\Gamma}\vdash \src{\pi\parallel\Pi}
			}{compai-soup}
		}
		=&\
			\oth{\compai{\src{C,\Delta,\Gamma}\vdash \src{\pi} }\parallel\compai{\src{C,\Delta,\Gamma}\vdash \src{\Pi} }}
		\\
		\tag{\compai{\cdot}-Endorse} 
		\compai{
			\typerule{T\LA-endorse}{
				\src{\Delta,\Gamma}\vdash\src{e}:\UNS
				\\
				\src{C,\Delta,\Gamma;(x:\varphi)} \vdash\src{s} %
			}{
				\src{C,\Delta,\Gamma}\vdash \src{\myendorse{x}{e}{\varphi}{s}} 
			}{compai-end}
		}
		=&\
			\begin{cases}
				\oth{
					\begin{aligned}
						&
						\destructo{\oth{x}~}{\compai{\src{\Delta,\Gamma}\vdash \src{e} : \UNS}}{\oth{nat}}{
							\\
							&\ 
							\ifzteo{\oth{x}}{
							\\
							&\ \ 
							\compai{ \src{C,\Delta,\Gamma;(x:\varphi)} \vdash\src{s} }
							\\
							&\ \ 
							}{
								\ifzteo{\oth{x-1}}{
								\\
								&\ \ \ 
								\compai{ \src{C,\Delta,\Gamma;(x:\varphi)} \vdash\src{s} }
								\\
								& \ \ \
								}{\wrongo}
							}
							\\
							&
						}{\wrongo}
					\end{aligned}
				}
				\\
				\text{ if }\src{\varphi}=\src{\Bools}
				\\
				\\
				\oth{
					\begin{aligned}
						&
						\destructo{\oth{x}~}{\compai{\src{\Delta,\Gamma}\vdash \src{e} : \UNS}}{\oth{nat}}{
							\\
							&\ 
							\compai{ \src{C,\Delta,\Gamma;(x:\varphi)} \vdash\src{s} }
							\\
							&
						}{\wrongo}
					\end{aligned}
				}
				\\
				\text{ if }\src{\varphi}=\src{\Nats}
				\\
				\\
				\oth{
					\begin{aligned}
						&
						\destructo{\oth{x}~}{\compai{\src{\Delta,\Gamma}\vdash \src{e} : \UNS}}{\oth{pair}}{
							\\
							&\ 
							\compai{ \src{C,\Delta,\Gamma;(x:\varphi)} \vdash\src{s} }
							\\
							&
						}{\wrongo}
					\end{aligned}
				}
				\\
				\text{ if }\src{\varphi}=\src{\UNS\times\UNS}
				\\
				\\
				\oth{
					\begin{aligned}
						&
						\destructo{\oth{x}~}{\compai{\src{\Delta,\Gamma}\vdash \src{e} : \UNS}}{\oth{nat}}{
							\\
							&\ 
							\oth{!x ;}
							\\
							&\ \ 
							\compai{ \src{C,\Delta,\Gamma;(x:\varphi)} \vdash\src{s} }
							\\
							&
						}{\wrongo}
					\end{aligned}
				}
				\\
				\text{ if }\src{\varphi}=\src{\Refs{\UNS}}
			\end{cases}
	\end{align*}
\end{definition}
We use \wrongo\ as before for \wrong.

\subsection{Properties of the \LA-\LI Compiler}\label{sec:prop-comp-o}

\begin{theorem}[Compiler \compai{\cdot} is \ccomp]\label{thm:comp-ai-cc}
	$\vdash\compai{\cdot} : \ccomp$
\end{theorem}

\begin{theorem}[Compiler \compai{\cdot} is \rscomp]\label{thm:comp-ai-rsc}
	$\vdash\compai{\cdot} : \rscomp$
\end{theorem}

\subsection{Cross-language Relation $\relatetphi$}\label{sec:relatet-def}
As before, we define a more lenient relation on states $\relatetphi$ 

\mytoprule{\src{\Omega}\relatetphi\oth{\Omega}}
\begin{center}
	\typerule{\LI-Low Location}{
		\nexists\src{\ell}\in\sech{\src{H}}
		&
		\src{\ell}\relatephi\oth{n}
		&
		\oth{n}\geq\oth{0}
	}{
		\src{H},\oth{H}\vdash\lowloc{\oth{n}}
	}{lowloc-o}
	\typerule{\LI-High Location}{
		\src{\ell}\in\sech{\src{H}}
		&
		\src{\ell}\relatephi\oth{n}
		&
		\oth{n}<\oth{0}
	}{
		\src{H},\oth{H}\vdash\highloc{\oth{n}} = \src{\ell}
	}{highloc-o}
	\typerule{Related states -- Secure}{
		\src{\Omega}=\src{\Delta ; \OB{F},\OB{F'} ; \OB{I} ; H \triangleright \Pi}
		&
		\oth{\Omega}=\oth{H_0 ; \OB{F},\compap{\OB{F'}} ; \OB{I} ; \OB{E} ; H \triangleright \Pi}
		&
		\src{\Delta}\vdash_{\varphi}\oth{H_0}
		\\
		\forall \oth{n}, \src{\ell} \ldotp \text{ if } \src{H},\oth{H}\vdash\highloc{\oth{n}} = \src{\ell} \text{ then }
			\\
			\oth{n\mapsto v}\in\oth{H} \text{ and } \src{\ell\mapsto v:\tau}\in\src{H} \text{ and } \src{v}\relatephi\oth{v}
	}{
		\src{\Omega}\relatetphi\oth{\Omega}
	}{state-rel-proof-o}
\end{center}
\botrule

We change the definition of a ``high location'' to be one that is in the enclave, i.e., whose address is less than 0.

The intuition behind \Cref{tr:state-rel-proof-o} is that high locations only need to be in sync, nothing is enforced on low locations.
Compared to \Cref{tr:state-rel-proof}, we have less conditions because we don't have to track fine-grained capabilities but just if an address is part of the enclave or not.

\BREAK

\begin{lemma}[A \LI target location is either high or low]\label{thm:oth-loc-h-or-l}
	\begin{align*}
		&
		\forall 
		\\
		\text{if }
		&\
		\src{H}\relatephi\oth{H}
		\\
		&\
		\oth{n\mapsto v}\in\oth{H}
		\\
		\text{ then either }
		&\
		\src{H},\oth{H}\vdash\lowloc{\oth{n}}
		\\
		\text{ or }
		&\
		\exists \src{\ell}\in\dom{\src{H}}.
		\\
		&\
		\src{H},\oth{H}\vdash\highloc{\oth{n}} = \src{\ell}
	\end{align*}
\end{lemma}
\begin{proof}
	Trivial, as \Cref{tr:lowloc-o} and \Cref{tr:highloc-o} are duals.
\end{proof}

\newpage

 \newpage
\section{Proofs}\label{sec:proofs}

\subsection{Proof of \Thmref{thm:rsc-prf-eq}}
\begin{proof}
\begin{description}
	\item[$\Rightarrow$]
	\begin{description}
		\item[HP]
			\begin{align*}
				\text{if }
				&\
				\forall\trg{A},\trg{\OB{\alpha}}.~
				\compgen{C}\vdash\trg{A}:\trg{attacker}
				\\
				&\
				\vdash\trg{A\hole{\compgen{C}}}:\trg{whole} 
				~
				\SInitt{\compgen{\src{C}}} \Xtot{\OB{\alpha}} \trg{\Omega}
				\\
				\text{then }
				&\
				\exists\src{A},\src{\OB{\alpha}}
				\src{C}\vdash\src{A}:\src{attacker}
				\\
				&\
				\vdash\src{A\hole{{C}}}:\src{whole} 
				~
				\SInits{A\hole{C}}\Xtos{\OB{\alpha}} \src{\Omega} 
				\\
				&\
				\strip{\src{\OB{\alpha}}}\relatebeta\strip{\trg{\OB{\alpha}}}
			\end{align*}
		\item[TH]
			\begin{align*}
				\text{ if }
				&\
				\src{M}\relatebeta\trg{M}
				\\
				&\
				\forall\src{A},\src{\OB{\alpha}}.~
				\vdash\src{A\hole{C}}:\src{whole} 
				\begin{aligned}[t]
					\text{if }
					&
					\SInits{C} \Xtos{\OB{\alpha}} \src{\Omega}
					\\
					\text{then }
					&
					\src{M}\vdash{\src{\OB{\alpha}}}
				\end{aligned}
				\\
				\text{ then }
				&\
				\forall\trg{A},\trg{\OB{\alpha}}.~
				\vdash\trg{A\hole{\compgen{C}}}:\trg{whole} 
				\begin{aligned}[t]
					\text{if }
					&
					\SInitt{\compgen{\src{C},\trg{M}}} \Xtot{\OB{\alpha}} \trg{\Omega}
					\\
					\text{then }
					&
					\trg{M}\vdash{\trg{\OB{\alpha}}}
				\end{aligned}
			\end{align*}
	\end{description}
	We proceed by contradiction and assume that $\trg{M}\nvdash{\trg{\OB{\alpha}}}$ while $\src{C}\vdash{\src{\OB{\alpha}}}$.

	By the relatedness of the traces, by \Cref{tr:rel-cl,tr:rel-cb,tr:rel-rt,tr:rel-rb} we have $\src{H}\relatebeta\trg{H}$ for all heaps in the traces.

	But if the heaps are related and the source steps (by unfolding $\src{M}\vdash{\src{\OB{\alpha}}}$), then by point 3.b in \Thmref{def:mon-rel} we have that the target monitor also steps, so $\trg{M}\vdash{\trg{\OB{\alpha}}}$.

	We have reached a contradiction, so this case holds.

	\item[$\Leftarrow$]
	Switch HP and TH from the point above.

	Analgously, we proceed by contradiction:
	\begin{itemize}
		\item $\forall\src{A},\src{\OB{\alpha}}.~
				\vdash\src{A\hole{{C}}}:\src{whole} 
				\text{ and }
				\SInits{A\hole{C}}\Xtos{\OB{\alpha}} \src{\Omega} 
				\text{ and }
				\strip{\src{\OB{\alpha}}}\not\relatebeta\strip{\trg{\OB{\alpha}}}$
	\end{itemize}
	By the same reasoning as above, with the HP we have we obtain $\src{M}\vdash{\src{\OB{\alpha}}}$ and $\trg{M}\vdash{\trg{\OB{\alpha}}}$.

	Again by 3.b in \Thmref{def:mon-rel} we know that the heaps of all actions in the traces are related.

	Therefore, $\strip{\src{\OB{\alpha}}}\relatebeta\strip{\trg{\OB{\alpha}}}$, which gives us a contradiction.
\end{description}
\end{proof}

\subsection{Proof of \Thmref{thm:comp-up-cc}}
\begin{proof}
	The proof proceeds for $\beta_0=(\src{\ell},\trg{0},\trg{k_{root}})$ and then, given that the languages are deterministic, by \Thmref{thm:gen-cc-up} as initial states are related by definition.
\end{proof}

\BREAK

\begin{lemma}[Expressions compiled with \compup{\cdot} are related]\label{thm:expr-rel-up}
	\begin{align*}
		&
		\forall
		\\
		\text{if }
		&\ 
		\src{H}\relatebeta\trg{H}
		\\
		&\
		\src{H\triangleright e\rho \redtos v}
		\\
		\text{then }
		&\
		\trg{H\triangleright \compup{e}\compup{\src{\rho}} \redtot \compup{v}}
	\end{align*}
\end{lemma}
\begin{proof}
	This proof proceeds by structural induction on \src{e}.
	\begin{description}
		\item[Base case: Values]\hfill
		\begin{description}
			\item[\src{\trues}] 
				By  \Cref{tr:compup-true}, \compup{\trues} = \trg{0}.

				As $\trues\relatebeta\trg{0}$, this case holds.
			\item[\src{\falses}] Analogous to the first case by \Cref{tr:compup-false}.
			\item[\src{n}$\in\mb{N}$] Analogous to the first case by \Cref{tr:compup-nat}.
			\item[\src{x}] Analogous to the first case, by \Cref{tr:compup-var} and by the relatedness of the substitutions.
			\item[\src{\ell}] Analogous to the first case by \Cref{tr:compup-loc}.
			\item[\src{\pair{v,v}}] By induction on \src{v} by \Cref{tr:compup-pair} and then it is analogous to the first case.
		\end{description}
		\item[Inductive case: Expressions]\hfill
		\begin{description}
				\item[\src{e \op e'}] 

				By \Cref{tr:compup-op} we have that 

				$\compup{ \src{e \op e'} } = \trg{\compup{\src{e}} \op \compup{\src{e'}}}$

				By HP we have that $\src{H\triangleright e\rho}\redtos\src{n}$ and $\src{H\triangleright e'\rho}\redtos\src{n'}$.

				By \Cref{tr:eus-op} we have that $\src{H\triangleright n \op n'}\redtos\src{n''}$.

				By IH we have that $\trg{H\triangleright \compup{\src{e}}\compup{\src{\rho}}}\redtot\compup{\src{n}}$ and $\trg{H\triangleright \compup{\src{e'}}\compup{\src{\rho}}}\redtot\compup{\src{n'}}$.

				By \Cref{tr:et-op} we have that $\trg{H\triangleright \compup{n} \op \compup{n'}}\redtot\trg{\compup{n''}}$.

				So this case holds.

				\item[\src{e \bop e'}] Analogous to the case above by IH, \Cref{tr:compup-bop}, \Cref{tr:eus-bop} and \Cref{tr:et-ift,tr:et-op}.
				\item[\src{!e}] Analogous to the case above by IH twice, \Cref{tr:compup-deref}, \Cref{tr:eus-de} and \Cref{tr:et-letin,tr:et-p1,tr:et-p2} and a case analysis by \Cref{tr:et-de-k,tr:et-de-t}.
				\item[\src{\pair{e,e}}] Analogous to the case above by IH and \Cref{tr:compup-pair}.
				\item[\src{\projone{e}}] 

				By \Cref{tr:compup-p1} $\compup{ \src{\projone{e}} } = \trg{\projone{\compup{\src{e}}}}$.

				By HP $\src{H\triangleright \projone{e}\rho}\redtos\src{\pair{v_1,v_2}}\redtos\src{v_1}$.

				By IH we have that $\trg{H\triangleright \projone{\compup{\src{e}}}\compup{\src{\rho}}} \redtot \trg{\projone{\compup{\src{\pair{v_1,v_2}}}}}$.

				By \Cref{tr:compup-pair} we have that $\trg{\projone{\compup{\src{\pair{v_1,v_2}}}}} = \trg{\projone{ \pair{ \compup{\src{v_1}},\compup{\src{v_2}} } }}$.

				Now $\trg{H\triangleright \projone{ \pair{ \compup{\src{v_1}},\compup{\src{v_2}} } }} \redtot \compup{\src{v_1}}$.

				So this case holds.

				\item[\src{\projtwo{e}}] Analogous to the case above by \Cref{tr:compup-p2}, \Cref{tr:eus-p2} and \Cref{tr:et-p2}.
			\end{description}
		\end{description}
\end{proof}

\BREAK

\begin{lemma}[Generalised compiler correctness for \compup{\cdot}]\label{thm:gen-cc-up}
\end{lemma}
\begin{proof}
		\begin{align*}
		&
		\forall ... \exists \beta'
		\\
		\text{ if }&\
		\vdash\src{C}:\src{whole}
		\\
		&\
		\src{C}=\src{\Delta ; \OB{F} ; \OB{I}} 
		\\
		&
		\compap{\src{C}} = \trg{k_{root} ; \OB{F} ; \OB{I} } =\trg{C}
		\\
		&\
		\src{C, H \triangleright s} \relatebeta \trg{C, H \triangleright \compap{\src{s}}}
		\\
		&\
		\src{C,H\triangleright s\rho}\xtos{\lambda} \src{C', H'\triangleright s'\rho'}
		\\
		\text{ then }&\
		\trg{C,H\triangleright \compap{\src{s}}\compup{\src{\rho}}} \xtot{\lambda} \trg{C', H' \triangleright \compap{\src{s'}}\compup{\src{\rho'}}}
		\\
		&\
		\trg{C'}= \trg{ k_{root} ; \OB{F} ; \OB{I} }
		\\
		&\
		\src{C, H \triangleright s'\rho'} \relate_{\beta'} \trg{C, H \triangleright \compap{\src{s'}}\compup{\src{\rho'}}}
		\\
		&\
		\beta\subseteq\beta'
	\end{align*}

	The proof proceeds by induction on \src{C} and the on the reduction steps.

	\begin{description}
		\item[Base case]\hfill
			\begin{description}
				\item[\skips] 

				By \Cref{tr:compup-skip} this case follows trivially.
			\end{description}
		\item[Inductive]\hfill
		\begin{description}
			\item[\src{\letnew{x}{e}{s}}] \hfill

				By \Cref{tr:compup-new} $\compup{ \src{\letnew{x}{e}{s}} } = $
				\begin{align*}
					&
					\letnewt{\trg{x_{loc}}}{\compup{\src{e}}}{
					\\
					&\
						\lethide{\trg{x_{cap}}}{x_{loc}}{
						\\
						&\ \ 
							\letint{\trg{x} ~}{~ \trg{\pair{x_{loc},x_{cap}}}}{\compup{s}}
						}
					}
				\end{align*}

				By HP $\src{H\triangleright e\rho}\redtos\src{v}$

				So by \Cref{thm:expr-rel-up} we have HPE: $\trg{H\triangleright \compup{e}\compup{\src{\rho}}}\redtot\src{\compup{v}}$ and HPV $\src{v}\relatebeta\compup{\src{v}}$.

				By \Cref{tr:eus-al}: $\src{C;H \triangleright \letnew{x}{e}{s}} \xtos{\epsilon} \src{C;H; \ell\mapsto v \triangleright s\subs{\ell}{x} }$.

				So by HPE:
				\begin{align*}
					\trg{C;H\triangleright } 
					&
					\letnewt{\trg{x_{loc}}}{\compup{\src{e}}}{
					\\
					&\
						\lethide{\trg{x_{cap}}}{x_{loc}}{
						\\
						&\ \ 
							\letint{\trg{x} ~}{~ \trg{\pair{x_{loc},x_{cap}}}}{\compup{s}}
						}
					}\trg{\rho}
					\\
					\text{\Cref{tr:et-nu} }
					\\
					\xtot{\epsilon}
					\trg{C;H;n\mapsto\compup{v}:\bot \triangleright } 
					&\
						\lethide{\trg{x_{cap}}}{x_{loc}}{
						\\
						&\ \ 
							\letint{\trg{x} ~}{~ \trg{\pair{x_{loc},x_{cap}}}}{\compup{s}}
						}\compup{\src{\rho}}\subt{n}{x_{loc}}
					\\
					\equiv
					\trg{C;H;n\mapsto\compup{v}:\bot \triangleright } 
					&\
						\lethide{\trg{x_{cap}}}{n}{
						\\
						&\ \ 
							\letint{\trg{x} ~}{~ \trg{\pair{n,x_{cap}}}}{\compup{s}}
						}\compup{\src{\rho}}
					\\
					\text{\Cref{tr:et-hi}}
					\\
					\xtot{\epsilon}
					\trg{C;H;n\mapsto\compup{v}:k;k \triangleright } 
					&\
						\letint{\trg{x} ~}{~ \trg{\pair{n,x_{cap}}}}{\compup{s}}
						\compup{\src{\rho}}\subt{k}{x_{cap}}
					\\
					\equiv
					\trg{C;H;n\mapsto\compup{v}:k;k \triangleright } 
					&\
						\letint{\trg{x} ~}{~ \trg{\pair{n,k}}}{\compup{s}}
						\trg{\rho}
					\\
					\text{\Cref{tr:et-letin} }
					\\
					\xtot{\epsilon}
					\trg{C;H;n\mapsto\compup{v}:k;k \triangleright } 
					&\
						\compup{s}
						\compup{\src{\rho}}\subt{\pair{n,k}}{x}
					\\
				\end{align*}
				
				Let $\beta'=\beta\cup(\src{\ell},\trg{n,k})$.

				By definition of $\relate_{\beta'}$ and by $\beta'$ we get HPL $\src{\ell}\relate_{\beta'}\trg{\pair{n,k}}$.

				By a simple weakening lemma for $\beta$ for substitutions and values applied to HP and HPV we can get HPVB $\src{v}\relate_{\beta'}\compup{\src{v}}$.

				As $\src{H}\relatebeta\trg{H}$ by HP, by a simple weakening lemma get that $\src{H}\relate_{\beta'}\trg{H}$ too and by \Cref{tr:hrel-i} with HPL and HPVB we get $\src{H'}\relate_{\beta'}\trg{H'}$.

				We have that $\src{\rho'}=\src{\rho}\subs{\ell}{x}$ and $\trg{\rho'}=\compup{\src{\rho}}\subt{\pair{n,k}}{x}$.
				
				So by HPL we get that $\src{\rho'}\relate_{\beta'}\trg{\rho'}$.

			\item[\src{s;s'}] Analogous to the case above by IH, \Cref{tr:compup-seq} and a case analysis on what \src{s} reduces to, either with \Cref{tr:eus-seq} and \Cref{tr:et-seq} or with \Cref{tr:eus-step} and \Cref{tr:et-step}.
			\item[\src{\letin{x}{e}{s}}] Analogous to the case above by IH, \Cref{tr:compup-letin}, \Cref{tr:eus-letin} and \Cref{tr:et-letin}.
			\item[\src{x := e'}] Analogous to the case above by \Cref{tr:compup-ass}, \Cref{tr:eus-up} and \Cref{tr:et-letin} (twice), \Cref{tr:et-p1,tr:et-p2} and then a case analysis by \Cref{tr:et-ac-k,tr:et-ac-t}.
			\item[\src{\ifte{e}{s}{s'}}] Analogous to the case above by IH, \Cref{tr:compup-if} and then either \Cref{tr:eus-ift} and \Cref{tr:et-ift} or \Cref{tr:eus-iff} and \Cref{tr:et-iff}.
			\item[\src{\call{f}~e}]

			By \Cref{tr:compup-call} $\compup{ \src{\call{f}~e} } = \trg{\call{f}~\compup{\src{e}}}$

			By HP $\src{H\triangleright e\rho}\redtos\src{v}$

			So by \Cref{thm:expr-rel-up} we have HPE: $\trg{H\triangleright \compup{e}\rho}\redtot\src{\compup{v}}$ and HPR $\src{v}\relatebeta\compup{\src{v}}$.

			So as \src{C} is whole, we apply \Cref{tr:eus-call-i}
				\begin{align*}
					&
					\src{C, H \triangleright \proc{\call{f}~e\rho}{\OB{f'}}} \xtos{\epsilon} 
					\\
					&
					\src{C, H \triangleright \proc{s;\ret\rho\subs{v}{x}}{\OB{f'};f}}
				\end{align*}

			By \Cref{tr:et-call-i}
				\begin{align*}
					&
					\trg{C, H \triangleright \proc{{\call{f}~\compup{\src{e}}}\compup{\src{\rho}}}{\OB{f'}}} \xtot{\epsilon} 
					\\
					&
					\trg{C, H \triangleright \proc{{s;\ret\compup{\src{\rho}}\subt{\compup{\src{v}}}{x}}}{\OB{f'};f}}
				\end{align*}

			By the first induction on \src{C} we get 

			IH1 $\src{C, H \triangleright \proc{{s;\ret\rho'}}{\OB{f'};f}} \relatebeta \trg{C, H \triangleright \proc{{s;\ret\rho'}}{\OB{f'};f}}$ 

			We instantiate \src{\rho'} with $\src{\rho}\subs{v}{x}$ and \trg{\rho'} with $\compup{\src{\rho}}\subt{\compup{\src{v}}}{x}$.

			So by HP and HPR we have that $\src{\rho}\subs{v}{x}\relatebeta\compup{\src{\rho}}\subt{\compup{\src{v}}}{x}$

			We we can use IH1 to conclude 

			$\src{C, H \triangleright \proc{{s;\ret\rho\subs{v}{x}}}{\OB{f'};f}} \relatebeta \trg{C, H \triangleright \proc{{s;\ret\compup{\src{\rho}}\subt{\compup{\src{v}}}{x}}}{\OB{f'};f}}$ 

			As $\beta'=\beta$, this case holds.
		\end{description}
	\end{description}
\end{proof}

\BREAK

\subsection{Proof of \Thmref{thm:comp-up-rsc}}

\begin{proof}
	HPM: $\src{M}\relatebeta\trg{M}$

	HP1: $\src{M}\vdash \src{C} : \src{rs}$
	
	TH1: $\trg{M}\vdash \compup{\src{C}} : \trg{rs}$

	We can state it in contrapositive form as:

	HP2: $\trg{M}\nvdash \compup{\src{C}} : \trg{rs}$	

	TH2: $\src{M}\nvdash \src{C} : \src{rs}$

	By expanding the definition of \com{rs} in HP2 and TH2, we get

	HP21 $\exists \trg{A},\trg{\OB{\alpha}}. \trg{M}\vdash\trg{A}:\trg{attacker} \text{ and either } \nvdash\trg{A\hole{\compup{C}}}:\trg{whole} $ or  

	\noindent HPRT1 $(\trg{\SInitt{\trg{A\hole{\compup{C}}}}\Xtot{\OB{\alpha}} \_ }$ and HPRMT1 $\trg{M}\nvdash{\trg{\OB{\alpha}}})$

	TH21 $\exists \src{A},\src{\OB{\alpha}}. \src{M}\vdash\src{A}:\src{attacker} \text{ and either } \nvdash\src{A\hole{{C}}}:\src{whole} \text{ or }$ TH2 $(\src{\SInits{\src{A\hole{{C}}}}\Xtos{\OB{\alpha}} \_ }$ and TH4 $\src{M}\nvdash{\src{\OB{\alpha}}})$

	We consider the case of a whole \trg{A}, the other is trivial.

	We can apply \Thmref{thm:backtr-corr} with HPRT1 and instantiate $\src{A}$ with a \trg{A} from $\backtr{\trg{A}}$ and we get the following unfolded HPs

	HPRS $\SInits{\src{A\hole{C}}}\Xtos{\OB{\alpha}}\src{\Omega}$

	HPRel $\src{\OB{\alpha}}\relatebeta\trg{\OB{\alpha}}$.

	So TH3 holds by HPRS.

	We need to show TH4

	Assume by contradiction HPBOT: the monitor in the source does not fail: $\src{M}\vdash{\src{\OB{\alpha}}})$

	By \Cref{tr:ms-valid} we know that forall $\src{\alpha}\in\src{\OB{\alpha}}$ such that $\strip{\src{\alpha}}=\src{H}$, this holds: HPHR $\src{M;H\monred M'}$.

	We can expand HPHR by \Cref{tr:ms-us} and get:

	HPMR: $\src{(\sigma_c,H_h',\sigma_f)}\in\src{\monred}$

	for a heap %
	$\src{H_h'}\subseteq\src{H_h}$

	By HPM $\src{M}\relatebeta\trg{M}$ for initial states.

	By \Thmref{tr:mon-rel} and the second clause of \Thmref{tr:mon-rel-sing} with HPMR we know that $\src{M}\relatebeta\trg{M}$ for the current states.

	By the first clause of \Thmref{tr:mon-rel-sing} we know that

	HPMRBI: $(\src{\sigma_c}, \src{H}, \_) \in \src{\monred} \iff (\trg{\sigma_c}, \trg{H}, \_) \in \trg{\monred}$

	By HPMRBI with HPMR we know that

	HPMRTC: $\trg{(\sigma_c,H_h',\sigma_f)}\in\trg{\monred}$

	However, by HPRMT1 and \Cref{tr:mt-valid} we know that

	HPNR: $\trg{M;H \not\monred}$

	so we get

	HPCON: $\nexists \trg{(\sigma_c,H_h',\sigma_f)}\in\trg{\monred}$

	By HPCON and HPMRTC we get the contradiction, so the proof holds.
\end{proof}

\BREAK

\subsection{Proof of \Thmref{thm:action-det}}
\begin{proof}
	The proof proceeds by induction on $\Xtot{\alpha!}$.
	\begin{description}
		\item[Base case:] $\Xtot{\alpha!}$

		By \Cref{tr:et-tr-sin} we need to prove the silent steps and the \trg{\alpha!} action.

		\begin{description}
			\item[\trg{\epsilon}] \hfill

				The proof proceeds by analysis of the target reductions.

				\begin{description}
					\item[\Cref{tr:et-seq}] 

					In this case we do not need to pick and the thesis holds by \Cref{tr:eus-seq}.

					\item[\Cref{tr:et-step}] 

					In this case we do not need to pick and the thesis holds by \Cref{tr:eus-step}.

					\item[\Cref{tr:et-ift}] 

					We have: \trg{H \triangleright \compup{e}\rho\redtot 0}

					We apply \Thmref{thm:action-det-expr} and obtain a $\src{v}\relatebeta\trg{0}$

					By definition we have $\src{0}\relatebeta\trg{0}$ and $\trues\relatebeta\trg{0}$, we pick the second.

					So we have \src{H \triangleright e\rho\redtos \trues}

					We can now apply \Cref{tr:eus-ift} and this case follows.

					\item[\Cref{tr:et-iff}] 

					This is analogous to the case above.

					\item[\Cref{tr:et-ac-t}] 

					Analogous to the case above.

					\item[\Cref{tr:et-ac-k}] 

					This is analogous to the case above but for $\src{v}=\src{\ell}\relatebeta\trg{\pair{n,k}}$.

					\item[\Cref{tr:et-letin}] 

					This follows by \Cref{thm:action-det-expr} and by \Cref{tr:eus-letin}.

					\item[\Cref{tr:et-nu}] 

					This follows by \Cref{thm:action-det-expr} and by \Cref{tr:eus-al}.

					\item[\Cref{tr:et-hi}] 

					By analisis of compiled code we know this only happens after a \trg{new} is executed.

					In this case we do not need to perform a step in the source and the thesis holds.

					\item[\Cref{tr:et-call-i}] 

					This follows by \Cref{thm:action-det-expr} and by \Cref{tr:eus-call-i}.

					\item[\Cref{tr:et-ret-i}] 

					In this case we do not need to pick and the thesis holds by \Cref{tr:eus-ret-i}.
				\end{description}

			\item[\trg{\alpha!}]\hfill
				
				The proof proceeds by case analysis on \trg{\alpha!}
				\begin{description}
					\item[\trg{\cb{f}{v~H}}] 

					This follows by \Thmref{thm:action-det-expr} and by \Cref{tr:eus-callback}.

					\item[\trg{\rt{H}}] 

					In this case we do not need to pick and the thesis holds by \Cref{tr:eus-ret}.
				\end{description}
		\end{description}

		\item[Inductive case:]

		This follows from IH and the same reasoning as for the single action above.

	\end{description}
\end{proof}

\BREAK

\begin{lemma}[Compiled code expression steps implies existence of source expression steps]\label{thm:action-det-expr}
	\begin{align*}
		&\
		\forall
		\\
		\text{if }
		&\
		\trg{H \triangleright \compup{e}\rho \redtot v}
		\\	
		\text{and if }
		&\
		\set{\src{\rho}}=\myset{\src{\rho}}{\src{\rho}\relatebeta\trg{\rho}}
		\\
		&\
		\src{v}\relatebeta\trg{v}
		\\
		&\
		\src{H}\relatebeta\trg{H}
		\\
		\text{then }
		&\
		\exists \src{\rho_j}\in\set{\src{\rho}}\ldotp \src{H \triangleright e\rho_j\redtos v}
	\end{align*}
\end{lemma}
\begin{proof}
	This proceeds by structural induction on \src{e}.
	\begin{description}
		\item[Base case:]
			\begin{description}
				\item[\trues] This follows from \Cref{tr:compup-true}.
				\item[\falses] This follows from \Cref{tr:compup-false}.
				\item[$\src{n}\in\mb{N}$] This follows from \Cref{tr:compup-nat}.
				\item[\src{x}] This follows from the relation of the substitutions and the totality of $\relatebeta$ and \Cref{tr:compup-var}.
				\item[\src{\pair{v,v'}}] This follows from induction on \src{v} and \src{v'}.
			\end{description}
		\item[Inductive case:]  
			\begin{description}
				\item[\src{e\op e'}]

				By definition of $\relatebeta$ we know that \src{v} and \src{v'} could be either natural numbers or booleans.

				We apply the IH with:

				IHV1 $\src{n}\relatebeta\trg{n}$

				IHV2 $\src{n'}\relatebeta\trg{n'}$

				By IH we get 

				IHTE1 $\trg{H \triangleright \compup{e}\rho \redtot n}$

				IHTE2 $\trg{H \triangleright \compup{e'}\rho \redtot n'}$

				IHSE1 $\src{H \triangleright e\rho_j\redtos n}$

				IHSE2 $\src{H \triangleright e'\rho_j\redtos n'}$
				
				By \Cref{tr:compup-op} we have that \compup{\src{e\op e'}}=\trg{\compup{e}\op\compup{e'}}.

				By \Cref{tr:et-op} with IHTE1 and IHTE2 we have that $\trg{H\triangleright \compup{e}\op\compup{e'}}\redtot\trg{n''}$ where IHVT $n''=n\op n'$

				By \Cref{tr:eus-op} with IHSE1 and IHSE2 we have that $\src{H\triangleright e\op e'}\redtos\src{n''}$ if $n''=n\op n'$

				This follows from IHVT and IHV1 and IHV2.

				\item[\src{e\bop e'}]
				As above, this follows from IH and \Cref{tr:compup-bop} and \Cref{tr:eus-bop}.

				\item[\src{\pair{e,e'}}] 
				As above, this follows from IH and \Cref{tr:compup-pair}.

				\item[\src{\projone{e}}] 
				As above, this follows from IH and \Cref{tr:compup-p1} and \Cref{tr:eus-p1}.

				\item[\src{\projtwo{e}}] 
				Analogous to the case above.

				\item[\src{!e}] 
				As above, this follows from IH and \Cref{tr:compup-deref} and \Cref{tr:eus-de} but with the hypothesis that \src{e} evaluates to a \src{v} related to a \trg{\pair{n,v}}.
			\end{description}
	\end{description}
\end{proof}

\BREAK

\subsection{Proof of \Thmref{thm:backtr-corr}}
\begin{proof}
	HP1 $\SInitt{A\hole{\compup{C}}} \Xtot{\OB{\alpha}} \trg{\Omega}$

	HPF $\trg{\Omega} \Xtot{\epsilon}\trg{\Omega'}$
	
	HPN	$\trg{\OB{I}} = \fun{names}{\trg{A}}$
		
	HPT $\trg{\OB{\alpha}}\equiv\trg{\OB{\alpha'}\cdot\alpha?}$
		
	HPL $\src{\ell_i;\ell_{glob}}\notin\beta$

	THE $\exists \src{A}\in \backtrup{\trg{I},\trg{\OB{\alpha}}}$
		
	TH1 $\SInits{\src{A\hole{C}}} \Xtos{\OB{\alpha}} \src{\Omega}$
		
	THA $\src{\OB{\alpha}}\relatebeta\trg{\OB{\alpha}}$
		
	THS $\src{\Omega}\relatebeta\trg{\Omega}$

	THC $\src{\Omega}.\src{H}.\src{\ell_i}=\card{\trg{\OB{\alpha}}}+1$

	The proof proceeds by induction on \trg{\OB{\alpha'}}.

	\begin{description}
		\item[Base case:]

		We perform a case analysis on \trg{\alpha?}
		\begin{description}
			\item[\trg{\clh{f}{v}{H}}] 	\hfill

			Given

			HP1 $\SInitt{A\hole{\compup{C}}} \Xtot{\clh{f}{v}{H}} \trg{\Omega}$

			We need to show that

			THE $\exists \src{A}\in \backtrup{\trg{I},\trg{\OB{\alpha}}}$
		
			TH1 $\SInits{\src{A\hole{C}}} \Xtos{\OB{\alpha}} \src{\Omega}$

			THA $\src{\clh{f}{v}{H}}\relatebeta\trg{\clh{f}{v}{H}}$
		
			THS $\src{\Omega}\relatebeta\trg{\Omega}$

			THC $\src{\Omega}.\src{H}.\src{\ell_i}=\card{\trg{\OB{\alpha}}}+1$

			By \Cref{tr:backtr-call} the back-translated context executes this code inside \src{main}:
			
			\src{
			\begin{aligned}
					&
					\iftes{!\ell_i == n
					}{
						\\
						&\ \ 
						\src{incrementCounter()}
						\\
						&\ \ 
						\letnews{\src{x1}~}{~\src{v_1}}{\src{register(\pair{x1,n_1})}}
						\\
						&\ \ 
						\cdots
						\\
						&\ \ 
						\letnews{\src{xj}~}{~\src{v_j}}{\src{register(\pair{xj,n_j})}}
						\\
						&\ \ 
						\src{\calls{f}{~v}}
						\\
						&
					}{
					\skips
					}
				\end{aligned}
			}

			As \trg{H_{pre}} is \trge, no \src{update}s are added.

			Given that \src{\ell_i} is initialised to 1 in \Cref{tr:backtr-skel}, this code is executed and it generates action \src{\clh{f}{v}{H}} where \src{H}=\src{\ell_1\mapsto v_1;\cdots;\ell_n\mapsto v_n} for all $\trg{n_i}\in\dom{\trg{H}}$ such that $\src{\ell_i}\relatebeta\trg{\pair{n_i,\_}}$ and:

			HPHR $\src{H}\relatebeta\trg{H}$

			By HPHR, \Thmref{thm:backtr-vals} and \Thmref{thm:action-det} with HPF we get THA, THE and TH1. 

			By \Cref{tr:state-rel-proof}, THS holds too.

			Execution of \src{incrementCounter()} satisfies THC.

			\item[\trg{\rbh{v}{H}}]
			This cannot happen as by \Cref{tr:et-retb} there needs to be a running process with a non-empty stack and by \Cref{tr:ini-t} the stack of initial states is empty and the only way to add to the stack is performing a call via \Cref{tr:et-call}, which would be a different label.

		\end{description}
		\item[Inductive case:] \hfill

		We know that (eliding conditions HP that are trivially satisfied):

		IHP1 $\SInitt{A\hole{\compup{C}}} \Xtot{\OB{\alpha}} \trg{\Omega'} \Xtot{\alpha!} \trg{\Omega''} \Xtot{\alpha?} \trg{\Omega}$

		And we need to prove:
		
		ITH1 $\SInits{\backtrup{\trg{I},\trg{\OB{\alpha}}}\src{\hole{C}}} \Xtos{\OB{\alpha}} \src{\Omega'} \Xtos{\alpha!} \src{\Omega''} \Xtos{\alpha?} \src{\Omega}$
			
		ITHA $\src{\OB{\alpha}\alpha!\alpha?}\relatebeta\trg{\OB{\alpha}\alpha!\alpha?}$
			
		ITHS $\src{\Omega}\relatebeta\trg{\Omega}$

		And the inductive HP is (for $\emptyset\subseteq\beta'$):

		\begin{mdframed}[hidealllines=true,backgroundcolor=blue!20]

		IH-HP1 $\SInitt{A\hole{\compup{C}}} \Xtot{\OB{\alpha}} \trg{\Omega'}$

		IH-TH1 $\SInits{\backtrup{\trg{I},\trg{\OB{\alpha}}}\src{\hole{C}}} \Xtos{\OB{\alpha}} \src{\Omega'}$

		IH-THA $\src{\OB{\alpha}}\relate_{\beta'}\trg{\OB{\alpha}}$

		IH-THS $\src{\Omega'}\relate_{\beta'}\trg{\Omega'}$

		\end{mdframed}

		By IHP1 and HPF we can apply \Thmref{thm:action-det} and so we can apply the IH to get IH-TH1, IH-THA and IH-THS.

		We perform a case analysis on \trg{\alpha!}, and show that the back-translated code performs \src{\alpha!}.

		By IH we have that the existing code is generated by \Cref{tr:backtr-listact}: $\backtrup{\trg{\OB{\alpha}}, n, \trg{H_{pre}}, \trg{ak}, \src{\OB{f}} }$.
		
		The next action $\trg{\alpha!}$ produces code according to: 

		HPF $\backtrup{\trg{\alpha!}, n, \trg{H_{pre}}, \trg{ak}, \src{\OB{f}} } $.

		By \Cref{tr:backtr-join}, code of this action is the first \src{if} statement executed.

		\begin{description}
			\item[\trg{\cbh{f}{v}{H}}] 	

			By \Cref{tr:backtr-callback-loc} this code is placed at function \src{f} so it is executed when compiled code jumps there
			
			\src{
				\begin{aligned}
					&
					\iftes{!\ell_i == n
					}{
						\\
						&\ \ 
						\src{incrementCounter()}
						\\
						&\ \ 
						\letins{\src{l1}~}{~\src{e_1}}{\src{register(\pair{l1, n_1})}}
						\\
						&\ \ 
						\cdots
						\\
						&\ \ 
						\letins{\src{lj}~}{~\src{e_j}}{\src{register(\pair{lj, n_j})}}
						\\
						&
					}{
					\skips
					}
				\end{aligned}
			}

			By IH we have that $\src{\ell_i \mapsto n}$, so we get

			IHL $\src{\ell_i\mapsto n+1}$

			By \Thmref{def:reachable} we have for $i\in1..j$ that a reachable location $\trg{n_i}\in\dom{\trg{H}}$ has a related counterpart in $\src{\ell_i}\in\dom{\src{H}}$ such that \src{H \triangleright e_i \redtos \ell_i}.

			By \Thmref{thm:atk-has-all-caps} we know all capabilities to access any \trg{n_i} are in \trg{ak}.

			We use \trg{ak} to get the right increment of the reach.
							
			\item[\trg{\rth{v}{H}}]

			In this case from IHF we know that $\src{\OB{f}} = \src{f'\OB{f'}}$.

			This code is placed at \src{f'}, so we identify the last called function and the code is placed there.
			Source code returns to \src{f'} so this code is executed \Cref{tr:backtr-ret-loc}
			
			\src{
				\begin{aligned}
					&
					\iftes{!\ell_i == n
					}{
						\\
						&\ \ 
						\src{incrementCounter()}
						\\
						&\ \ 
						\letins{\src{l1}~}{~\src{e_1}}{\src{register(\pair{l1, n_1})}}
						\\
						&\ \ 
						\cdots
						\\
						&\ \ 
						\letins{\src{lj}~}{~\src{e_j}}{\src{register(\pair{lj, n_j})}}
						\\
						&
					}{
					\skips
					}
				\end{aligned}
			}

		This case now follows the same reasoning as the one above.

		\end{description}

		So we get (for $\beta'\subseteq\beta''$):

		HP-AC! $\src{\alpha!}\relate_{\beta''}\trg{\alpha!}$

		By IH-THS and \Cref{tr:state-rel-whole} and HP-AC! we get HP-OM2: 

		HP-OM2: $\src{\Omega''}\relate_{\beta''}\trg{\Omega''}$

		The next action $\trg{\alpha?}$ produces code according to: 

		IHF1 $\backtrup{\trg{\alpha?}, n+1, \trg{H_{pre}'}, \trg{ak'}, \src{\OB{f'}} } $.

		We perform a case analysis on \trg{\alpha?} and show that the back-translated code performs \src{\alpha?}:
		\begin{description}
			\item[\trg{\rbh{v}{H}}]

			By \Cref{tr:backtr-retback}, after \src{n} actions, we have from IHF1 that \src{\OB{f'}} = \src{f'\OB{f''}} and inside function \src{f'} there is this code:
			
			\src{
				\begin{aligned}
					&
					\iftes{!\ell_i == n
					}{
						\\
						&\ \ 
						\letnews{\src{x1}~}{~\src{v_1}}{\src{register(\pair{x1,n_1})}}
						\\
						&\ \ 
						\cdots
						\\
						&\ \ 
						\letnews{\src{xj}~}{~\src{v_j}}{\src{register(\pair{xj,n_j})}}
						\\
						&\ \ 
						\src{update(m_1, u_1) }
						\\
						&\ \ 
						\cdots
						\\
						&\ \ 
						\src{update(m_l, u_l) }
						\\
						&
					}{
					\skips
					}
				\end{aligned}
				}

			By IHL, \src{\ell_i\mapsto n+1}, so the \src{if} gets executed.

			By definition, forall $\trg{n}\in\dom{\trg{H}}$ we have that $\trg{n}\in\trg{H_n}$ or $\trg{n}\in\trg{H_c}$ (from the case definition).

			By \Thmref{thm:atk-has-all-caps} we know all capabilities to access any \trg{n} are in \trg{ak}.

			We induce on the size of \trg{H}; the base case is trivial and the inductive case follows from IH and the following:
			\begin{description}
				\item[\trg{H_n}:] and \trg{n} is newly allocated.

				In this case when we execute

				$\src{C;H' \triangleright \letnews{\src{x1}~}{~\src{v_1}}{\src{register(\pair{x1,n_1})}}} \xtos{\epsilon} \src{C; H'; \ell''\mapsto \backtrup{\trg{v_1}} \triangleright \src{register(\pair{\ell'',n_1})} }$

				and we create $\beta''$ by adding $\src{\ell''},\trg{n,\eta'}$ to $\beta'$.

				By \Thmref{thm:reg-no-dup} we have that:

				$\src{C; H'; \ell''\mapsto \backtrup{\trg{v_1}} \triangleright \src{register(\pair{\ell'',n_1})} } \xtos{\epsilon} \src{C; H'; \ell''\mapsto \backtrup{\trg{v_1}} \triangleright \skips }$ 

				and we can lookup \src{\ell''} via \src{n}.

				\item[\trg{H_c}:] and \trg{n} is already allocated.

				In this case

				$\src{C;H' \triangleright update(m_1, u_1)} \xtos{\epsilon} \src{C;H''\triangleright \skips}$

				By \Thmref{thm:update-no-stuck} we know that $\src{H''}=\src{H'}\subs{\ell''\mapsto \_}{\ell''\mapsto u_1}$

				and $\src{\ell''}$ such that $(\src{\ell''},\trg{m_1,\eta''})\in\beta'$.

			\end{description}
			By \Thmref{thm:backtr-vals} on the values stored on the heap, let the heap after these reduction steps be \src{H}, we can conclude 

			HPRH $\src{H}\relate_{\beta''}\trg{H}$.

			As no other \src{if} inside \src{f} is executed, eventually we hit its return statement, which by \Cref{tr:backtr-join} and \Cref{tr:backtr-fun} is \src{incrementCounter();\ret}.

			Execution of \src{incrementCounter()} satisfies THC.

			So we have $\src{\Omega''}\Xtos{\rbh{}{H}}\src{\Omega}$ (by \Cref{thm:backtr-vals}) and with HPRH.
			
			\item[\trg{\clh{f}{v}{H}}] 	

			Similar to the base case, only with \src{update} bits, which follow the same reasoning above.
		\end{description}

		So we get (for $\beta''\subseteq\beta$):
			
		HP-AC? $\src{\alpha?}\relatebeta\trg{\alpha?}$

		By IH-THA and HP-AC! and HP-AC? we get ITHA.

		Now by \Cref{tr:state-rel-whole} again and HP-AC? we get ITHS and ITH1, so the theorem holds.
	\end{description}
\end{proof}

\BREAK

\begin{lemma}[\LA attacker always has access to all capabilities]\label{thm:atk-has-all-caps}
	\begin{align*}
		&
		\forall
		\\
		\text{if }
		&
		\backtrup{\trg{\alpha}, n, \trg{H}, \trg{ak}, \src{\OB{f}} } = \set{\src{s}, \trg{ak'}, \trg{H'}, \src{\OB{f'}}, \src{f}}
		\\
		&
		\trg{k}\ldotp \trg{n\mapsto v:k}\in\trg{H}
		\\
		\text{then }
		&
		\trg{n}\in\fun{reach}{\trg{ak'}.\mtt{loc},\trg{ak'}.\mtt{cap} ,\trg{H}}
	\end{align*}
\end{lemma}
\begin{proof}
	Trivial case analyisis on \Cref{tr:backtr-ret-loc,tr:backtr-retback,tr:backtr-call,tr:backtr-callback-loc}.
\end{proof}

\BREAK

\begin{lemma}[Backtranslated values are related]\label{thm:backtr-vals}
	\begin{align*}
		&
		\forall\trg{v}, \beta \ldotp 
		\\
		&
		\backtrup{\trg{v}}\relatebeta\trg{v}
	\end{align*}
\end{lemma}
\begin{proof}
	Trivial analysis of \Cref{sec:val-bt}.
\end{proof}

\BREAK

\subsection{Proof of \Thmref{thm:src-ty-impl-safe}}
\begin{proof}

	HP1 $\vdash\src{C}:\UNS$

	HP2 $\src{C\agree M}$

	TH $\src{M}\vdash\src{C}:\src{rs}$

	We expand TH: $\forall\src{A},\src{\OB{\alpha}}. \src{M}\vdash\src{A}:\src{attacker}$ and $\vdash\src{A\hole{C}}:\src{whole}$ if HPR $\SInits{A\hole{C}}\Xtos{\OB{\alpha}}\src{\_}$ then THM $\src{M}\vdash{\src{\OB{\alpha}}}$

	By definition of \strip{\src{\OB{\alpha}}} and by \Cref{tr:mts-valid} we get a \src{\OB{H}} to induce on.

	The base case holds by \Cref{tr:mts-t-s-b}.

	In the inductive case we are considering \src{{\OB{H}}\cdot{H}} and the IH covers the first part of the trace.

	By \Thmref{thm:alpha-red-ok}, given that the state generating the action is \src{C, H\triangleright \Pi} we know that, HPH $\vdash\monh{\src{H},\src{\Delta}}:\src{\Delta}$

	By \Cref{tr:mts-valid} and by \Cref{tr:mts-s} we need to show that $\vdash\src{H}:\src{\Delta}$.

	This follows by HPH. %

	We thus need to prove that the initial steps are related heaps are secure.

	By \Cref{tr:plug-s} we need to show that the heaps consituting the initial heap -- both \src{H} and \src{H_0} -- are well typed. 

	The latter, $\vdash\src{H_0}:\src{\Delta}$, holds by \Cref{tr:plug-s}.

	The former holds by definition of the attacker: \Cref{tr:tsu-base,tr:tsu-loc}.
\end{proof}

\BREAK

\subsubsection{Proof of \Thmref{thm:atk-src-coincide}}
\begin{proof}
	This is proven by trivial induction on the syntax of \src{A}.

	By the rules of \Cref{sec:un-ty}, points 1 and 3 follow, point 2 follows from the HP \Cref{tr:plug-s}.
\end{proof}

\BREAK

\begin{lemma}[\LA-\src{\alpha} reductions respect heap typing]\label{thm:alpha-red-ok}
	\begin{align*}
		\text{if }
		&
		\src{C}\equiv\src{\Delta ; \cdots}
		\\
		&
		\vdash\monh{\src{H},\src{\Delta}}:\src{\Delta}
		\\
		&
		\src{C, H\triangleright \Pi\rho} \Xtos{\OB{\alpha}} \src{C', H'\triangleright \Pi'\rho'} 
		\\
		\text{then } 
		&
		\vdash\monh{\src{H'},\src{\Delta}}:\src{\Delta}
	\end{align*}
\end{lemma}
\begin{proof}
	The proof proceeds by induction on \src{\OB{\alpha}}.

	\begin{description}
		\item[Base case]
			This trivially holds by HP.
		\item[Inductive case]
			This holds by IH plus a case analysis on the last action:
			\begin{description}
				\item[$\src{\cl{f}{v}}$]
					This holds by \Thmref{thm:src-?-rel}.
				\item[$\src{\cb{f}{v}}$]
					This holds by \Thmref{thm:src-!-rel}
				\item[$\src{\rt{}}$]
					This holds by \Thmref{thm:src-!-rel}
				\item[$\src{\rb{}}$] 
					This holds by \Thmref{thm:src-?-rel}
			\end{description}
	\end{description}
\end{proof}

\BREAK

\begin{lemma}[\LA An attacker only reaches \UNS locations]\label{thm:atk-reach-uns-loc}
	\begin{align*}
		&
		\forall
		\\
		\text{if }
		&
		\src{\ell\mapsto v:\UNS} \in \src{H}
		\\
		\text{then }
		&
		\nexists\src{e}
		\\
		&
		\src{H\triangleright e \redtos \ell'}
		\\
		&
		\src{\ell'\mapsto v:\tau}\in\src{H}
		\\
		&
		\src{\tau}\neq\UNS
	\end{align*}
\end{lemma}
\begin{proof}
	This proof proceeds by contradiction.

	Suppose \src{e} exists, there are two cases for \src{\ell'}
	\begin{itemize}
			\item \src{\ell'} was allocated by the attacker:

			This contradicts the judgements of \Cref{sec:un-ty}.

			\item \src{\ell'} was allocated by the compiled code:

			The only way this was possible was an assignment of \src{\ell'} to \src{\ell}, but \Cref{tr:ts-ass} prevents it.
	\end{itemize}	
\end{proof}

\BREAK

\begin{lemma}[\LA attacker reduction respects heap typing]\label{thm:epsi-atk-red-pres-rel}
	\begin{align*}
		\text{if }
		&
		\src{C}\equiv\src{\Delta ; \cdots}
		\\
		&
		\src{C}\vdashatts\src{\Pi}\xtos{}\src{\Pi'}
		\\
		&
		\src{C, H\triangleright \Pi\rho} \xtos{\lambda} \src{C, H'\triangleright \Pi\rho'} 
		\\
		\text{then } 
		&
		\monh{\src{H},\src{\Delta}}=\monh{\src{H'},\src{\Delta}}
	\end{align*}
\end{lemma}
\begin{proof}
	Trivial induction on the derivation of \src{\Pi}, which is typed with $\vdash_\UNS$ and by \Thmref{thm:atk-reach-uns-loc} has no access to locations in \src{\Delta} or with a type $\src{\tau}\vdash\circ$. %
\end{proof}

\BREAK

\begin{lemma}[\LA typed reduction respects heap typing]\label{thm:epsi-red-resp-rel}
	\begin{align*}
		\text{if }
		&
		\src{C}\equiv\src{\Delta ; \cdots}
		\\
		&
		\src{C,\Gamma}\vdash\src{s}
		\\
		&
		\src{C,\Gamma}\vdash\src{s'}
		\\
		&
		\vdash\monh{\src{H},\src{\Delta}}:\src{\Delta}
		\\
		&
		\src{C, H\triangleright s\rho} \xtos{\lambda} \src{C', H'\triangleright s'\rho'} 
		\\
		\text{then } 
		&
		\vdash\monh{\src{H'},\src{\Delta}}:\src{\Delta}
	\end{align*}
\end{lemma}
\begin{proof}
	This is done by induction on the derivation of the reducing statement.%

	There, the only non-trivial cases are:
		\begin{description}
			\item[\Cref{tr:ts-new}]

			By IH we have that 

			$\src{H\triangleright e\rho \redtos v}$

			So

			$\src{C;H \triangleright \letnewty{x}{e}{\tau}{s}} \xtos{\epsilon} \src{C;H \ell\mapsto v:\tau \triangleright s\subs{\ell}{x} }$

			By IH we need to prove that $\vdash\monh{\src{\ell\mapsto v:\tau},\src{\Delta}}:\src{\Delta}$

			As $\src{\ell}\notin\dom{\src{\Delta}}$, by \Cref{tr:heap-ok} this case holds.

			\item[\Cref{tr:ts-ass}] 

			By IH we have (HPH) $\src{H\triangleright e \redtos v}$

			such that $\src{\ell}:\src{\Refs{\tau}}$ and $\src{v}:\src{\tau}$.

			So

			$\src{C;H \triangleright x:=e\rho} \xtos{\epsilon} \src{C;H' \triangleright \skips }$

			where $\subs{x}{\ell}\in\src{\rho}$ and

			$\src{H}=\src{H_1; \ell\mapsto v':\tau ; H_2}$
		
			$\src{H'}=\src{H_1; \ell\mapsto v:\tau ; H_2}$

			There are two cases
			\begin{description}
				\item[$\src{\ell}\in\dom{\src{\Delta}}$] 

				By \Cref{tr:heap-ok} we need to prove that $\src{\ell}:\src{\Refs{\tau}}\in\src{\Delta}$.

				This holds by HPH and \Cref{tr:ini-s}, as the initial state ensures that location \src{\ell} in the heap has the same type as in \src{\Delta} .

				\item[$\src{\ell}\notin\dom{\src{\Delta}}$] 
				This case is trivial as for allocation.
			\end{description}

			\item[\Cref{tr:ts-coe}] 

			We have that $\src{C,\Gamma}\vdash \src{e} : \src{\tau} $ and  HPT $\src{\tau} \vdash \circ$.

			By IH $\src{H\triangleright e \redtos v} $ such that $\vdash\monh{\src{H'},\src{\Delta}}:\src{\Delta}$.

			By HPT we get that $\monh{\src{H}}=\monh{\src{H'}}$ as by \Cref{tr:h-sec-s} function \monh{\cdot} only considers locations whose type is $\src{\tau} \nvdash \circ$, so none affected by \src{e}.

			So this case by IH.

			\item[\Cref{tr:ts-end}]

			By \Cref{tr:es-end} we have that $\src{H\triangleright e \redtos v}$ and that $\src{C, H \triangleright \myendorse{x}{e}{\varphi}{s} } \redtos \src{C, H \triangleright s\subs{v}{x}}$.

			So this holds by IH.

		\end{description}
\end{proof}

\BREAK

\begin{lemma}[\LA any non-cross reduction respects heap typing]\label{thm:epsi-red-ok}
	\begin{align*}
		\text{if }
		&
		\src{C}\equiv\src{\Delta ; \cdots}
		\\
		&
		\vdash\monh{\src{H},\src{\Delta}}:\src{\Delta}
		\\
		&
		\src{C, H\triangleright \Pi\rho} \xtos{\lambda} \src{C', H'\triangleright \Pi'\rho'} 
		\\
		\text{then } 
		&
		\vdash\monh{\src{H'},\src{\Delta}}:\src{\Delta}
	\end{align*}
\end{lemma}
\begin{proof}
	By induction on the reductions and by application of \Cref{tr:es-par}.
	The base case follows by the assumptions directly.
	In the inductive case we have the following:
	\begin{align*}
		&
		\src{C, H\triangleright \Pi\rho} \xtos{\lambda} \src{C'', H''\triangleright \Pi''\rho''} \xtos{\lambda} \src{C', H'\triangleright \Pi'\rho'} 
	\end{align*}
	This has 2 sub-cases, if the reduction is in an attacker function or not.
	\begin{description}
		\item[$\src{C}\vdashatts\src{\Pi''\xtos{}\Pi}$:] this follows by induction on \src{\Pi''} and from IH and \Thmref{thm:epsi-atk-red-pres-rel}.
		\item[$\src{C}\not\vdashatts\src{\Pi''\xtos{}\Pi}$:]
		In this case we induce on \src{\Pi''}.

		The base case is trivial.

		The inductive case is \src{\proc{s}{\OB{f}} \parallel \Pi}, which follows from IH and \Thmref{thm:epsi-red-resp-rel}.
	\end{description}
\end{proof}

\BREAK

\begin{lemma}[\LA-? actions respect heap typing]\label{thm:src-?-rel}
	\begin{align*}
		\text{if }
		&
		\src{C}\equiv\src{\Delta ; \cdots}
		\\
		&
		\src{C, H\triangleright \Pi\rho} \Xtos{\alpha?} \src{C, H'\triangleright v'} 
		\\
		\text{then } 
		&
		\monh{\src{H},\src{\Delta}}=\monh{\src{H'},\src{\Delta}}
	\end{align*}
\end{lemma}
\begin{proof}
	By \Thmref{thm:epsi-red-ok}, and a simple case analysis on \src{\alpha?} (which does not modify the heap).
\end{proof}

\BREAK

\begin{lemma}[\LA-! actions respect heap typing]\label{thm:src-!-rel}
	\begin{align*}
		\text{if }
		&
		\src{C}\equiv\src{\Delta ; \cdots}
		\\
		&
		\src{C, H\triangleright \Pi\rho} \Xtos{\alpha!} \src{C', H'\triangleright v'} 
		\\
		&
		\src{\vdash\monh{\src{H},\src{\Delta}}:\src{\Delta}}
		\\
		\text{then } 
		&
		\src{\vdash\monh{\src{H'},\src{\Delta}}:\src{\Delta}}
	\end{align*}
\end{lemma}
\begin{proof}
	By \Thmref{thm:epsi-red-ok} and a simple case analyis on \src{\alpha!} (which does not modify the heap).
\end{proof}

\BREAK

\subsection{Proof of \Thmref{thm:comp-ap-cc}}\label{sec:proof-compap-cc}

\begin{proof}
	By definition initial states have related components, related heaps and well-typed, related starting processes, 
	for $\beta_0=(\dom{\src{\Delta}},\dom{\trg{H_0}},\trg{H_0}.\trg{\eta})$
	so we have:

	HRS $\SInits{C}\relate_{\beta_0} \SInitt{\compap{\src{C}}}$.

	As the languages have no notion of internal nondeterminism we can apply \Thmref{thm:gen-cc} with HRS to conclude.
\end{proof}

\BREAK

\begin{lemma}[Expressions compiled with \compap{\cdot} are related]\label{thm:expr-rel-ap}
	\begin{align*}
		&
		\forall
		\\
		\text{if }
		&\ 
		\src{H}\relatebeta\trg{H}
		\\
		&\
		\src{H\triangleright e\rho \redtos v}
		\\
		\text{then }
		&\
		\trg{H\triangleright \compap{e}\compap{\src{\rho}} \redtot \compap{v}}
	\end{align*}
\end{lemma}
\begin{proof}
	The proof is analogous to that of \Thmref{thm:expr-rel-up} as the compilers perform the same steps and expression reductions are atomic.
\end{proof}

\BREAK

\begin{lemma}[Generalised compiler correctness for \compap{\cdot}]\label{thm:gen-cc}
	\begin{align*}
		&
		\forall ... \exists \beta'
		\\
		\text{ if }&\
		\src{C;\Gamma}\vdash\src{\Pi},
		\\
		&\
		\vdash\src{C}:\src{whole}
		\\
		&\
		\src{C}=\src{\Delta ; \OB{F} ; \OB{I}} 
		\\
		&
		\compap{\src{C}} = \trg{H_0 ; \OB{F} ; \OB{I}} =\trg{C}
		\\
		&\
		\src{C, H \triangleright \Pi} \relatebeta \trg{C, H \triangleright \compap{\src{C;\Gamma}\vdash\src{\Pi}}}
		\\
		&\
		\src{C,H\triangleright \Pi\rho}\Xtos{} \src{C, H'\triangleright \Pi'\rho'}
		\\
		\text{ then }&\
		\trg{C,H\triangleright \compap{\src{C;\Gamma}\vdash\src{\Pi}}\compap{\src{\rho}}} \Xtot{} \trg{C, H' \triangleright \compap{\src{C;\Gamma}\vdash\src{\Pi'}}\compap{\src{\rho'}}}
		\\
		&\
		\src{C, H \triangleright \Pi'\rho'} \relate_{\beta'} \trg{C, H \triangleright \compap{\src{C;\Gamma}\vdash\src{\Pi'}}\compap{\src{\rho'}}}
		\\
		&\
		\beta\subseteq\beta'
	\end{align*}
\end{lemma}
\begin{proof}
	This proof proceeds by induction on the typing of \src{\Pi} and then of \src{\pi}.
	\begin{description}
		\item[Base Case] 
		\begin{description}
			\item[\skips] Trivial by \Cref{tr:compap-skip}.
		\end{description}

		\item[Inductive Case] \hfill

			In this case we proceed by induction on the typing of \src{s}
			\begin{description}
				\item[Inductive Cases]
				\begin{description}
					\item[\Cref{tr:ts-new}]

					There are 2 cases, they are analogous.
					\begin{description}
						\item[$\src{\tau}=\UNS$] 

							By HP

							$\src{\Gamma}\vdash \src{e} : \src{\tau}$

							$\src{H \triangleright e\redtos v}$

							$\src{C,H \triangleright\letnewty{x}{e}{\tau}{s}\rho} \xtos{\epsilon} \src{C,H;\ell\mapsto v:\tau \triangleright s\subs{\ell}{x}\rho}$

							By \Cref{thm:expr-rel-ap} we have:

							IHR1 $\trg{H\triangleright \compap{\src{\Gamma}\vdash \src{e} : \src{\tau}}\compap{\src{\rho}}} \redtot \trg{ \compap{\src{\Gamma}\vdash \src{v} : \src{\tau}}}$

							By \Cref{tr:compap-new} we get

							\trg{
							\begin{aligned}
								&
								\letnewt{\trg{xo}~}{\compap{\src{\Gamma}\vdash \src{e} : \src{\tau}} 
								\\
								&\
								}{ 
									\letint{\trg{x}~}{~\trg{\pair{xo,0}}
									\\
									&\ \ 
									}{
										\compap{\src{C,\Gamma;x:\Refs{\tau}}\vdash \src{s}} 
									}
								}
							\end{aligned}
							}

							So:
							\begin{align*}
								&
								\trg{C,H\triangleright 
									\begin{aligned}[t]
										&
										\letnewt{\trg{xo}~}{\compap{\src{\Gamma}\vdash \src{e} : \src{\tau}} 
										\\
										&\
										}{ 
											\letint{\trg{x}~}{~\trg{\pair{xo,0}}
											\\
											&\ \ 
											}{
												\compap{\src{C,\Gamma;x:\Refs{\tau}}\vdash \src{s}} 
											}
										}
									\end{aligned}
								}
								\\
								\xtot{\epsilon}
								&
								\trg{C,H ; n\mapsto \compap{\src{\Gamma}\vdash \src{v} : \src{\tau}}:\bot \triangleright 
									\begin{aligned}[t]
										&
										\letint{\trg{x}~}{~\trg{\pair{n,0}}
										\\
										&\ 
										}{
											\compap{\src{C,\Gamma;x:\Refs{\tau}}\vdash \src{s}} 
										}
									\end{aligned}
								}
								\\
								\xtot{\epsilon}
								&
								\trg{C,H ; n\mapsto \compap{\src{\Gamma}\vdash \src{v} : \src{\tau}}:\bot \triangleright 
									\compap{\src{C,\Gamma;x:\Refs{\tau}}\vdash \src{s}}\subt{\trg{\pair{n,0}}}{x} 
								}
							\end{align*}

							For $\beta'=\beta\cup(\src{\ell},\trg{n,\bot})$, this case holds.

						\item[else] 

							The other case holds follows the same reasoning but 

							for $\beta'=\beta\cup(\src{\ell},\trg{n,k})$ and for \trg{H'}=\trg{H;n\mapsto \compap{\src{C,\Gamma}\vdash \src{v} : \src{\tau}}:k;k}.
					\end{description}

					\item[\Cref{tr:ts-seq}] 

					By HP

					$\src{\Gamma}\vdash \src{s}$; $\src{\Gamma}\vdash \src{s'}$

					$\src{C,H \triangleright s\rho}\Xtos{}\src{C',H' \triangleright s''\rho''}$

					There are two cases
					\begin{description}
						\item[\src{s''}=\skips] \Cref{tr:eus-seq}

						$\src{C',H' \triangleright \skips\rho'';s'\rho} \xtos{\epsilon} \src{C',H' \triangleright s'\rho}$

						By IH

						$\trg{C,H\triangleright \compap{\src{\Gamma}\vdash \src{s}}\compap{\src{\rho}}} \Xtot{} \trg{C',H'\triangleright \compap{\src{\Gamma}\vdash \skips}\compap{\src{\rho''}}}$

						By \Cref{tr:compap-seq}

						$\trg{\compap{\src{C,\Gamma}\vdash \src{s}} ; \compap{\src{C,\Gamma}\vdash \src{s'}}}$

						So
						\begin{align*}
							&
							\trg{C,H\triangleright \compap{\src{C,\Gamma}\vdash \src{s}}\compap{\src{\rho}} ; \compap{\src{C,\Gamma}\vdash \src{s'} }} \compap{\src{\rho}}
							\\
							\Xtot{}
							&
							\trg{
							\trg{C',H'\triangleright \compap{\src{\Gamma}\vdash \skips} \compap{\src{\rho''}}}
							;
							\compap{\src{C,\Gamma}\vdash \src{s'} }\compap{\src{\rho}}
							}
							\\
							\xtot{\epsilon}
							&
							\trg{C',H'\triangleright \compap{\src{C,\Gamma}\vdash \src{s'} }\compap{\src{\rho}} }
						\end{align*}

						At this stage we apply IH and the case holds.

						\item[else] By \Cref{tr:eus-step} we have 

						$\src{C,H \triangleright s;s'}\Xtos{}\src{C',H' \triangleright s'';s'}$

						This case follows by IH and HPs.
					\end{description}

					\item[\Cref{tr:ts-fun}] 

					Analogous to the cases above.
					
					\item[\Cref{tr:ts-vardef}]

					Analogous to the cases above.

					\item[\Cref{tr:ts-ass}] 

					Analogous to the cases above.

					\item[\Cref{tr:ts-if}] 

					Analogous to the cases above.

					\item[\Cref{tr:ts-fork}]

					Analogous to the cases above.

					\item[\Cref{tr:ts-coe}]

					By \Cref{tr:compap-coe}, this follows from IH directly.

					\item[\Cref{tr:ts-end}] 

					This has a number of trivial cases based on \Cref{tr:compap-end} that are analogous to the ones above.
				\end{description}
			\end{description}
	\end{description}
\end{proof}

\BREAK

\subsection{Proof of \Thmref{thm:comp-ap-rsc}}

\begin{proof}
	Given:

	HP1: $\src{M}\vdash \src{C} : \src{rs}$

	HPM: $\src{M}\relatebeta\trg{M}$

	We need to prove:

	TP1: $\trg{M}\vdash \compgen{\src{C}} : \trg{rs}$
	
	We unfold the definitions of \com{rs} and obtain:

	$\forall\src{A}. \src{M}\vdash\src{A}:\src{attacker}, \vdash\src{A\hole{C}}:\src{whole}$

	HPE1: if $\src{\SInits{A\hole{C}} \Xtos{\OB{\alpha}} \_}$ then $\src{M}\vdash\strip{\src{\OB{\alpha}}}$

	$\forall\trg{A}. \trg{M}\vdash\trg{A}:\trg{attacker}, \vdash\trg{A\hole{\compap{C}}}:\trg{whole}$

	THE1: if HPRT $\trg{\SInitt{A\hole{\compap{C}}} \Xtot{\OB{\alpha}} \_}$ then THE1 $\trg{M}\vdash\strip{\src{\OB{\alpha}}}$

	By definition of the compiler we have that

	HPISR: $\SInits{A\hole{C}}\relatetbeta\SInitt{A\hole{\compap{C}}}$

	for $\beta = \dom{\src{\Delta}},\trg{H_0}$ such that $\src{M} = \src{(\set{\sigma},\monred,\sigma_0,\Delta,\sigma_c)}$ and $\trg{M} = \trg{(\set{\sigma},\monred,\sigma_0,H_0,\sigma_c)}$

	By \strip{\trg{\OB{\alpha}}} and \Cref{tr:mtt-valid} we get a \trg{\OB{H}} to induce on.

	\begin{description}
		\item[Base case:] this holds by \Cref{tr:mtt-t-s-b}.
		\item[Inductive case:] 

		By \Cref{tr:mtt-t-s}, $\trg{M;\OB{H}\monred M''}$ holds by IH, we need to prove \trg{M'';H\monred M'}.

		By \Cref{tr:ms-t-2} e need to prove that THMR: $\exists \trg{\sigma'}. (\trg{\sigma}, \monh{\trg{H},\trg{H_0}}, \trg{\sigma'}) \in \trg{\monred}$.

		By HPISR and with applications of \Cref{thm:comp-act-pres-relatet,thm:att-act-pres-relatet} we know that states are always related with $\relatetbeta$ during reduction.

		So by \Thmref{thm:relatet-impl-high-heap-rel} we know that HPHH $\monh{\src{H},\src{\Delta}}\relatebeta\monh{\trg{H},\trg{H_0}}$, for \src{H}, \trg{H} being the last heaps in the reduction.

		By HPM and \Cref{tr:mon-rel-ap} we have $\beta_0,\src{\Delta}\vdash\trg{M}$.

		By this and \Cref{tr:ok-mon} we have that HPHR $\forall \monh{\src{H},\src{\Delta}} \relatebeta \monh{\trg{H},\trg{H_0}}. $ $
			\text{ if } \vdash \src{H}: \src{\Delta} $ $
			\text{ then } \exists \trg{\sigma'}. (\trg{\sigma}, \monh{\trg{H},\trg{H_0}}, \trg{\sigma'}) \in \trg{\monred}$
		so by HPHH we can instantiate this with \src{H} and \trg{H}.

		By \Thmref{thm:src-ty-impl-safe} applied to HPE1, as \compap{\cdot} operates on well-typed components, we know that HPMR: $\src{M}\vdash\strip{\src{\OB{\alpha}}}$ for all \src{\OB{\alpha}}.

		So by \Cref{tr:mts-s} with HPMT we get HPHD $\vdash \src{H}: \src{\Delta}$ for the \src{H} above.

		By HPHD with HPHR we get THMR $\exists \trg{\sigma'}. (\trg{\sigma}, \monh{\trg{H},\trg{H_0}}, \trg{\sigma'}) \in \trg{\monred}$, so this case holds.
	\end{description}

\end{proof}

\BREAK

\begin{lemma}[$\relatetbeta$ implies relatedness of the high heaps]\label{thm:relatet-impl-high-heap-rel}
	\begin{align*}
		\text{if }
		&\
		\src{\Omega}=\src{\Delta ; \OB{F},\OB{F'} ; \OB{I} ; H \triangleright \Pi}
		\\
		&
		\trg{\Omega}=\trg{H_0 ; \OB{F},\compap{\OB{F'}} ; \OB{I} ; H \triangleright \Pi}
		\\
		&
		\src{\Omega}\relatetbeta\trg{\Omega}
		\\
		\text{then }
		&
		\monh{\src{H},\src{\Delta}}\relatebeta\monh{\trg{H},\trg{H_0}}
	\end{align*}
\end{lemma}
\begin{proof}
	By point 2a in \Cref{tr:state-rel-proof}.
\end{proof}

\BREAK

\begin{lemma}[\LA-compiled actions preserve $\relatetbeta$]\label{thm:comp-act-pres-relatet}
	\begin{align*}
		&
		\forall ...
		\\
		\text{ if }
		&\
		\src{C, H\triangleright \Pi\rho} \xtos{\lambda} \src{C, H'\triangleright \Pi'\rho'} 
		\\
		&\
		\trg{C,H\triangleright \compap{\src{C;\Gamma}\vdash\src{\Pi}}\compap{\src{\rho}}} \Xtot{\lambda} \trg{C, H' \triangleright \compap{\src{C;\Gamma}\vdash\src{\Pi'}}\compap{\src{\rho'}}}
		\\
		&\
		\src{C, H \triangleright \Pi\rho} \relatetbeta \trg{C, H \triangleright \compap{\src{C;\Gamma}\vdash\src{\Pi}}\rho}
		\\
		&\
		\src{C;\Gamma}\vdash\src{\Pi}
		\\
		\text{ then }
		&\
		\src{C, H'\triangleright \Pi'\rho'} \relatetbeta \trg{C, H' \triangleright \compap{\src{C;\Gamma}\vdash\src{\Pi'}}\compap{\src{\rho'}}}
	\end{align*}
\end{lemma}
\begin{proof}
	Trivial induction on the derivation of \src{\Pi}, analogous to \Thmref{thm:gen-cc}.

	\begin{description}
		\item[\Cref{tr:ts-new}] 
		There are 2 cases, they are analogous.
			\begin{description}
				\item[$\src{\tau}=\UNS$] 

					By HP

					$\src{\Gamma}\vdash \src{e} : \src{\tau}$

					$\src{H \triangleright e\redtos v}$

					$\src{C,H \triangleright\letnewty{x}{e}{\tau}{s}\rho} \xtos{\epsilon} \src{C,H;\ell\mapsto v:\tau \triangleright s\subs{\ell}{x}\rho}$

					By \Thmref{thm:expr-rel-ap} we have:

					IHR1 $\trg{H\triangleright \compap{\src{\Gamma}\vdash \src{e} : \src{\tau}}\compap{\src{\rho}}} \redtot \trg{ \compap{\src{\Gamma}\vdash \src{v} : \src{\tau}}}$

					By \Cref{tr:compap-new} we get

					\trg{
					\begin{aligned}
						&
						\letnewt{\trg{xo}~}{\compap{\src{\Gamma}\vdash \src{e} : \src{\tau}} 
						\\
						&\
						}{ 
							\letint{\trg{x}~}{~\trg{\pair{xo,0}}
							\\
							&\ \ 
							}{
								\compap{\src{C,\Gamma;x:\Refs{\tau}}\vdash \src{s}} 
							}
						}
					\end{aligned}
					}

					So:
					\begin{align*}
						&
						\trg{C,H\triangleright 
							\begin{aligned}[t]
								&
								\letnewt{\trg{xo}~}{\compap{\src{\Gamma}\vdash \src{e} : \src{\tau}} 
								\\
								&\
								}{ 
									\letint{\trg{x}~}{~\trg{\pair{xo,0}}
									\\
									&\ \ 
									}{
										\compap{\src{C,\Gamma;x:\Refs{\tau}}\vdash \src{s}} 
									}
								}
							\end{aligned}
						}
						\\
						\xtot{\epsilon}
						&
						\trg{C,H ; n\mapsto \compap{\src{\Gamma}\vdash \src{v} : \src{\tau}}:\bot \triangleright 
							\begin{aligned}[t]
								&
								\letint{\trg{x}~}{~\trg{\pair{n,0}}
								\\
								&\ 
								}{
									\compap{\src{C,\Gamma;x:\Refs{\tau}}\vdash \src{s}} 
								}
							\end{aligned}
						}
						\\
						\xtot{\epsilon}
						&
						\trg{C,H ; n\mapsto \compap{\src{\Gamma}\vdash \src{v} : \src{\tau}}:\bot \triangleright 
							\compap{\src{C,\Gamma;x:\Refs{\tau}}\vdash \src{s}}\subt{\trg{\pair{n,0}}}{x} 
						}
					\end{align*}
						
					For $\beta'=\beta$, this case holds.

				\item[else] 

					The other case holds follows the same reasoning but 

					for $\beta'=\beta\cup(\src{\ell},\trg{n,k})$ and for \trg{H'}=\trg{H;n\mapsto \compap{\src{C,\Gamma}\vdash \src{v} : \src{\tau}}:k;k}.

					We need to show that this preserves \Cref{tr:state-rel-proof}, specifically it preserves point $(2a)$:
					$\src{\ell}\relatebeta\trg{\pair{n,k}} \text{ and } \src{\ell\mapsto v:\tau}\in\src{H} \text{ and } \src{v}\relatebeta\trg{v} $

					These follow all from the observation above and by \Thmref{thm:expr-rel-ap}.
			\end{description}
	\end{description}
\end{proof}

\BREAK

\begin{lemma}[\LP Attacker actions preserve $\relatet$]\label{thm:att-act-pres-relatet}
	\begin{align*}
		&
		\forall ...
		\\
		\text{ if }
		&\
		\src{C, H\triangleright \Pi\rho} \xtos{\lambda} \src{C, H'\triangleright \Pi'\rho'} 
		\\
		&\
		\trg{C,H\triangleright \Pi\rho} \xtot{\lambda} \trg{C,H'\triangleright \Pi'\rho'}
		\\
		&\
		\src{C, H \triangleright \Pi\rho} \relatetbeta \trg{C,H\triangleright \Pi\rho}
		\\
		&\
		\src{C}\vdashatts\src{\Pi\rho}\xtos{\lambda}\src{\Pi'\rho'}
		\\
		&\
		\trg{C}\vdashattt\trg{\Pi\rho}\xtot{\lambda}\trg{\Pi'\rho'}
		\\
		\text{ then }
		&\
		\src{C, H'\triangleright \Pi'\rho'} \relatetbeta \trg{C, H'\triangleright \Pi'\rho'}
	\end{align*}
\end{lemma}
\begin{proof}

	For the source reductions we can use \Thmref{thm:epsi-red-ok} to know that $\monh{\src{H}}=\monh{\src{H'}}$, so they don't change the interested bits of the $\relatetbeta$.

	Suppose this does not hold by contradiction, %
		there can be three clauses that do not hold based on \Cref{tr:state-rel-proof}:
		\begin{itemize}
			\item violation of $(1)$: $\exists \trg{\pi}\in\trg{\Pi}\ldotp\src{C}\vdash\trg{\pi}:\trg{attacker} \text{ and } \trg{k}\in\fun{fv}{\trg{\pi}}$.

			By HP5 this is a contradiction.

			\item violation of $(2a)$: $\trg{n\mapsto v:k} \in \trg{H} \text{ and } \src{\ell}\relatebeta\trg{\pair{n,k}} \text{ and } \src{\ell\mapsto v:\tau}\in\src{H} \text{ and } \lnot(\src{v}\relatebeta\trg{v})$

			To change this value the attacker needs \trg{k} which contradicts points $(1)$ and $(2b)$.

			\item violation of $(2b)$: either of these:
				\begin{itemize}
					\item $\src{H},\trg{H}\nvdash\lowloc{\trg{n'}}$
					
					Since \Cref{tr:highloc} does not hold, by \Cref{thm:targ-loc-h-or-l} this is a contradiction.

					\item $\trg{v}=\trg{k'}$ for $\src{H},\trg{H}\vdash\highcap{\trg{k'}}$

					This can follow from another two cases
					\begin{itemize}
						\item forgery of \trg{k;}: an ispection of the semantics rules contradicts this
						\item update of a location to \trg{k'}: however \trg{k'} is not in the code (contradicts point $(1)$) and by induction on the heap \trg{H} we have that \trg{k'} is stored in no other location, so this is also a contradiction.
					\end{itemize}
				\end{itemize}
		\end{itemize}
\end{proof}

\BREAK

\subsection{Proofs for the Non-Atomic Variant of \LA (\Cref{sec:nonatom-new-hide})}\label{sec:proofs-for-nonat}
The only proof that needs changing is that for \Cref{thm:comp-act-pres-relatet}: there is this new case.

For this we weaken $\relatetbeta$ and define $\altrelatetbeta$ as follows:

\mytoprule{\src{\Omega}\altrelatetbeta\trg{\Omega}}
\begin{center}
	\typerule{Non Atomic State Relation}{
		\src{\Omega}\relatetbeta\trg{\Omega}
	}{
		\src{\Omega}\altrelatetbeta\trg{\Omega}
	}{state-nonat}
	\typerule{Non Atomic State Relation -stuck}{
		\src{\Omega}=\src{C, H \triangleright \Pi}
		&
		\src{C}=\src{\Delta,\OB{F},\OB{I}}
		&
		\trg{\Omega}=\trg{C, H\triangleright \Pi}
		\\
		\exists\trg{\pi}\in\trg{\Pi}\ldotp \trg{C}\nvdash\trg{\pi} : \trg{attacker}
		\\
		\trg{\pi} = \trg{\proc{\hide{n};s}{\OB{f};f}}
		&
		\trg{C, H\triangleright \pi}\stut
		&
		\exists\src{f}\in\dom{\src{\OB{F}}}\ldotp \src{f}\relatebeta\trg{f}
		\\
		\forall \src{\ell}\ldotp \src{\ell}\in \dom{\vdash\sech{\src{H}}}
		&
		\trg{n \mapsto v; k}\in\trg{H}
		&
		\src{\ell}\notaltrelatetbeta\trg{\pair{n,k}}
		&
		\src{\ell}\altrelatetbeta\trg{\pair{n,0}}
	}{
		\src{\Omega}\altrelatetbeta\trg{\Omega}
	}{state-nonat-st}
\end{center}
\botrule

Two states are now related if:
\begin{itemize}
	\item either they are related by $\relatetbeta$
	\item or the red process is stuck on a \trg{\hide{n}} where \trg{n \mapsto v; k} but $\src{\ell}\sim\trg{\pair{n,k}}$ does not hold for a \src{\ell} that is secure, and we have that $\src{\ell}\sim\trg{\pair{n,0}}$ (as this was after the \trg{new}).
	And the \trg{hide } on which the process is stuck is not in attacker code.
\end{itemize}
Having this in proofs would not cause problems because now all proofs have an initial case analysis whether the state is stuck or not, but because it steps it's not stuck.

This relation only changes the second case of the proof of \Cref{thm:comp-act-pres-relatet} for \Cref{tr:compap-new-nonat} as follows:

\begin{proof}
	\src{\newty{\cdot}{\cdot}} is implemented as defined in \Cref{tr:compap-new-nonat}.

	\begin{description}
			\item[$\src{\tau}\neq\UNS$] 

				By HP

				$\src{\Gamma}\vdash \src{e} : \src{\tau}$

				$\src{H \triangleright e\redtos v}$

				$\src{C,H \triangleright\letnewty{x}{e}{\tau}{s}\rho} \xtos{\epsilon} \src{C,H;\ell\mapsto v:\tau \triangleright s\subs{\ell}{x}\rho}$

				By \Cref{thm:expr-rel-ap} we have:

				IHR1 $\trg{H\triangleright \compap{\src{\Gamma}\vdash \src{e} : \src{\tau}}\compap{\src{\rho}}} \redtot \trg{ \compap{\src{\Gamma}\vdash \src{v} : \src{\tau}}}$

				By \Cref{tr:compap-new-nonat} we get

				\trg{
				\begin{aligned}[t]
					&
					\letnewt{\trg{x}~}{\trg{0}}{
					\\
					&\
						\lethide{\trg{xk}~}{\trg{x}}{
						\\
						&\ \ 
							\letint{\trg{xc}~}{~ \compap{\src{\Delta,\Gamma}\vdash \src{e} : \src{\tau}} 
							}{
							\\
							&\ \ \ 
								\trg{x := xc \with{xk};}
							\\
							&\ \ \ 
								\compap{\src{C,\Delta,\Gamma}\vdash\src{s}}
							}
						}
					}
				\end{aligned}
				}

				So:
				\begin{align*}
					&
					\trg{C,H\triangleright 
						\begin{aligned}[t]
							&
							\letnewt{\trg{x}~}{\trg{0}}{
							\\
							&\
								\lethide{\trg{xk}~}{\trg{x}}{
								\\
								&\ \ 
									\letint{\trg{xc}~}{~ \compap{\src{\Delta,\Gamma}\vdash \src{e} : \src{\tau}} 
									}{
									\\
									&\ \ \ 
										\trg{x := xc \with{xk};}
									\\
									&\ \ \ 
										\compap{\src{C,\Delta,\Gamma}\vdash\src{s}}
									}
								}
							}
						\end{aligned}
					}
					\\
					\xtot{\epsilon}
					&
					\trg{C,H, n\mapsto0:\bot\triangleright 
						\begin{aligned}[t]
							&
								\lethide{\trg{xk}~}{\trg{n}}{
								\\
								&\ 
									\letint{\trg{xc}~}{~ \compap{\src{\Delta,\Gamma}\vdash \src{e} : \src{\tau}} 
									}{
									\\
									&\ \ 
										\trg{x := xc \with{xk};}
									\\
									&\ \ \ 
										\compap{\src{C,\Delta,\Gamma}\vdash\src{s}}
									}
								}
						\end{aligned}
					}
				\end{align*}

				And $\beta'=\beta\cup(\src{\ell},\trg{n,0})$.

				Now there are two cases:
				\begin{itemize}
					\item 
					A concurrent attacker reduction performs \trg{\hide{n}}, so the state changes.

					\begin{align*}
						&
						\trg{C,H, n\mapsto0:k;k\triangleright 
							\begin{aligned}[t]
								&
									\lethide{\trg{xk}~}{\trg{n}}{
									\\
									&\ 
										\letint{\trg{xc}~}{~ \compap{\src{C,\Gamma}\vdash \src{e} : \src{\tau}} 
										}{
										\\
										&\ \ 
											\trg{x := xc \with{xk};}
										\\
										&\ \ \ 
											\compap{\src{C,\Delta,\Gamma}\vdash\src{s}}
										}
									}
							\end{aligned}
						}
					\end{align*}

					At this stage the state is stuck: \Cref{tr:et-hi} does not apply.

					Also, we have that this holds by the new $\beta'$: $(\src{\ell}\altrelatet_{\beta'}\trg{\pair{n,0}})$

					And so this does not hold: $(\src{\ell}\altrelatet_{\beta'}\trg{\pair{n,k}})$

					As the stuck statement is not in attacker code, we can use \Cref{tr:state-nonat-st} to conclude.

					\item  The attacker does not.
					In this case the proof continues as in \Cref{thm:comp-act-pres-relatet}.
				\end{itemize}
		\end{description}
\end{proof}

\BREAK

\subsection{Proof of \Thmref{thm:comp-ai-cc}}

\begin{proof}
	Analogous to that of \Cref{sec:proof-compap-cc}.
\end{proof}

\subsection{Proof of \Thmref{thm:comp-ai-rsc}}\label{sec:proof-comp-ai-rsc}

\begin{proof}
	Given:

	HP1: $\src{M}\vdash\src{C} : \src{rs}$

	HPM: $\src{M}\relatephi\oth{M}$

	We need to prove:

	TP1: $\oth{M}\vdash \compai{\src{C}} : \oth{rs}$
	
	We unfold the definitions of \com{rs} and obtain:

	$\forall\src{A}. \src{M}\vdash\src{A}:\src{attacker}, \vdash\src{A\hole{C}}:\src{whole}$

	HPE1: if $\src{\SInits{A\hole{C}} \Xtos{\OB{\alpha}} \_}$ then $\src{M}\vdash\strip{\src{\OB{\alpha}}}$

	$\forall\oth{A}. \oth{M}\vdash\oth{A}:\oth{attacker}, \vdash\oth{A\hole{\compap{C}}}:\oth{whole}$

	THE1: if HPRT $\oth{\SInito{A\hole{\compap{C}}} \Xtoo{\OB{\alpha}} \_}$ then THE1 $\oth{M}\vdash\strip{\src{\OB{\alpha}}}$

	By definition of the compiler we have that

	HPISR: $\SInits{A\hole{C}}\relatetphi\SInito{A\hole{\compap{C}}}$

	for $\varphi = \dom{\src{\Delta}},\oth{H_0}$ such that $\src{M} = \src{(\set{\sigma},\monred,\sigma_0,\Delta,\sigma_c)}$ and $\oth{M} = \oth{(\set{\sigma},\monred,\sigma_0,H_0,\sigma_c)}$

	By \strip{\oth{\OB{\alpha}}} and \Cref{tr:mo-valid} we get a \oth{\OB{H}} to induce on.

	\begin{description}
		\item[Base case:] this holds by \Cref{tr:mo-t-s-b}.
		\item[Inductive case:] 

		By \Cref{tr:mo-t-s}, $\oth{M;\OB{H}\monred M''}$ holds by IH, we need to prove \oth{M'';H\monred M'}.

		By \Cref{tr:mo-t-2} e need to prove that THMR: $\exists \oth{\sigma'}. (\oth{\sigma}, \monh{\oth{H},\oth{H_0}}, \oth{\sigma'}) \in \oth{\monred}$.

		By HPISR and with applications of \Cref{thm:comp-act-pres-relatet-o,thm:att-act-pres-relatet-o} we know that states are always related with $\relatetphi$ during reduction.

		So by \Thmref{thm:relatet-impl-high-heap-rel-2} we know that HPHH $\monh{\src{H},\src{\Delta}}\relatephi\monh{\oth{H},\oth{H_0}}$, for \src{H}, \oth{H} being the last heaps in the reduction.

		By HPM and \Cref{tr:mon-rel-ap} (adjusted for \LI) we have $\varphi_0,\src{\Delta}\vdash\oth{M}$.

		By this and \Cref{tr:ok-mon} (adjusted for \LI) we have that 

		\noindent HPHR $\forall \monh{\src{H},\src{\Delta}} \relatephi \monh{\oth{H},\oth{H_0}}. $ if $\vdash \src{H}: \src{\Delta} $ then 

		\noindent $\exists \oth{\sigma'}. (\oth{\sigma}, \monh{\oth{H},\oth{H_0}}, \oth{\sigma'}) \in \oth{\monred}$ so by HPHH we can instantiate this with \src{H} and \oth{H}.

		By \Thmref{thm:src-ty-impl-safe} applied to HPE1, as \compap{\cdot} operates on well-typed components, we know that HPMR: $\src{M}\vdash\strip{\src{\OB{\alpha}}}$ for all \src{\OB{\alpha}}.

		So by \Cref{tr:mts-s} with HPMT we get HPHD $\vdash \src{H}: \src{\Delta}$ for the \src{H} above.

		By HPHD with HPHR we get THMR $\exists \oth{\sigma'}. (\oth{\sigma}, \monh{\oth{H},\oth{H_0}}, \oth{\sigma'}) \in \oth{\monred}$, so this case holds.
	\end{description}
\end{proof}

\BREAK

\begin{lemma}[$\relatetphi$ implies relatedness of the high heaps]\label{thm:relatet-impl-high-heap-rel-2}
	\begin{align*}
		\text{if }
		&\
		\src{\Omega}=\src{\Delta ; \OB{F},\OB{F'} ; \OB{I} ; H \triangleright \Pi}
		\\
		&
		\oth{\Omega}=\oth{H_0 ; \OB{F},\compai{\OB{F'}} ; \OB{I} ; \OB{E} ; H \triangleright \Pi}
		\\
		&
		\src{\Omega}\relatetphi\oth{\Omega}
		\\
		\text{then }
		&
		\monh{\src{H},\src{\Delta}}\relatephi\monh{\oth{H},\oth{H_0}}
	\end{align*}
\end{lemma}
\begin{proof}
	By \Cref{tr:state-rel-proof-o}.
\end{proof}

\BREAK

\begin{lemma}[\LA-compiled actions preserve $\relatetphi$]\label{thm:comp-act-pres-relatet-o}
	\begin{align*}
		&
		\forall ...
		\\
		\text{ if }
		&\
		\src{C, H\triangleright \Pi\rho} \xtos{\lambda} \src{C, H'\triangleright \Pi'\rho'} 
		\\
		&\
		\oth{C,H\triangleright \compai{\src{C;\Gamma}\vdash\src{\Pi}}\compai{\src{\rho}}} \Xtoo{\lambda} \oth{C, H' \triangleright \compai{\src{C;\Gamma}\vdash\src{\Pi'}}\compai{\src{\rho'}}}
		\\
		&\
		\src{C, H \triangleright \Pi\rho} \relatetphi \oth{C, H \triangleright \compai{\src{C;\Gamma}\vdash\src{\Pi}}\rho}
		\\
		&\
		\src{C;\Gamma}\vdash\src{\Pi}
		\\
		\text{ then }
		&\
		\src{C, H'\triangleright \Pi'\rho'} \relatetphi \oth{C, H' \triangleright \compai{\src{C;\Gamma}\vdash\src{\Pi'}}\compai{\src{\rho'}}}
	\end{align*}
\end{lemma}
\begin{proof}
	Trivial induction on the derivation of \src{\Pi}, analogous to \Thmref{thm:gen-cc} and \Thmref{thm:comp-act-pres-relatet}.
\end{proof}

\BREAK

\begin{lemma}[\LP Attacker actions preserve $\relatet$]\label{thm:att-act-pres-relatet-o}
	\begin{align*}
		&
		\forall ...
		\\
		\text{ if }
		&\
		\src{C, H\triangleright \Pi\rho} \xtos{\lambda} \src{C, H'\triangleright \Pi'\rho'} 
		\\
		&\
		\oth{C,H\triangleright \Pi\rho} \xtoo{\lambda} \oth{C,H'\triangleright \Pi'\rho'}
		\\
		&\
		\src{C, H \triangleright \Pi\rho} \relatetphi \oth{C,H\triangleright \Pi\rho}
		\\
		&\
		\src{C}\vdashatts\src{\Pi\rho}\xtos{\lambda}\src{\Pi'\rho'}
		\\
		&\
		\oth{C}\vdashatto\oth{\Pi\rho}\xtoo{\lambda}\oth{\Pi'\rho'}
		\\
		\text{ then }
		&\
		\src{C, H'\triangleright \Pi'\rho'} \relatetphi \oth{C, H'\triangleright \Pi'\rho'}
	\end{align*}
\end{lemma}
\begin{proof}
	For the source reductions we can use \Thmref{thm:epsi-red-ok} to know that $\monh{\src{H}}=\monh{\src{H'}}$, so they don't change the interested bits of the $\relatetphi$.

	Suppose this does not hold by contradiction, there can be one clause that does not hold based on \Cref{tr:state-rel-proof}:
	\begin{itemize}
		\item 
		two related high-locations \src{\ell} and \oth{n} point to unrelated values.

		Two cases arise: creation and update of a location to an unrelated value.

		Both cases are impossible because \Cref{tr:eo-ac-k} and \Cref{tr:eo-hi} check $\oth{C}\vdash\oth{f}:\oth{prog}$ and \Cref{tr:whole-o} ensures that the attacker defines different names from the program, so the attacker can never execute them.
	\end{itemize}
\end{proof}

 \newpage

\section{\facomp and Inefficient Compiled Code}\label{sec:fac-short}
We illustrate various ways in which \facomp forces inefficiencies in compiled code via a running example.
Consider a password manager written in an object-oriented language that is compiled to an assembly-like language.
We elide most code details and focus only on the relevant aspects.

\begin{lstlisting}[mathescape]
private db: Database;

public testPwd( user: Char[8], pwd: BitString): Bool{
  if( db.contains( user )){ return db.get( user ).getPassword() == pwd; }
}
...
private class Database{ ... }
\end{lstlisting}

The source program exports the function \lst{testPwd} to check whether a \lst{user}'s stored password matches a given password \lst{pwd}. The stored password is in a local database, which is represented by a piece of \emph{local state} in the variable \lst{db}. The details of \lst{db} are not important here, but the database is marked private, so it is not directly accessible to the context of this program in the source language.

\begin{example}[Extensive checks]\label{ex:checks}
A fully-abstract compiler for the program above must generate code that checks that the arguments passed to \lst{testPwd} by the context are of the right type~\cite{KULeuven-358154,scoo-j,fstar2js,catalin,mfac}. 
In fact, the code expects an array of characters of length 8, any other parameter (e.g., an array of objects) cannot be passed in the source, so it must also be prevented to be passed in the target.
More precisely, a fully abstract compiler will generate code similar to the following for \lst{testPwd} (we assume that arrays are passed as pointers into the heap).
\lstset{language=Asm,numbersep=5pt,frame=single}
\begin{lstlisting}[mathescape]
label testpwd
  for i = 0; i<8; i++  // 8 is the legth of the $\lst{user}$ field in the previous snippet
    add r0 i
    load the memory word stored at address r0 into r1
    test that r1 is a valid char encoding
  ...
\end{lstlisting} 
\lstset{language=Java,numbersep=5pt,frame=single}
Basically, this code dynamically checks that the first argument is a character array.
Such a check can be very inefficient.
\end{example}
The problem here is that \facomp forces these checks on all arguments, even those that have no security relevance.
In contrast, \rscomp does not need these checks. 
Indeed, neither of our earlier compiler, \compup{\cdot} nor \compap{\cdot}, insert them. 
Note that any robustly safe source program will have programmer-inserted checks for all parameters that are relevant to the safety property of interest, and these checks will be compiled to the target. 
For other parameters, the checks are irrelevant, both in the source and the target, so there is no need to insert them.

\begin{example}[Component size in memory]\label{ex:size}
Let us now consider two different ways to implement the \lst{Database} class: as a \lst{List} and as a \lst{RedBlackTree}.
As the class is \lst{private}, its internal behaviour and representation of the database is invisible to the outside.
Let  \Cs{_{list}} be the program with the \lst{List} implementation and \Cs{_{tree}} be the program with the \lst{RedBlackTree} implementation; in the source language, these are equivalent.

However, a subtlety arises when considering the assembly-level, compiled counterparts of \Cs{_{list}} and \Cs{_{tree}}: the \emph{code} of a \lst{RedBlackTree} implementation consumes more memory than the code of a \lst{List} implementation.
Thus, a target-level context can distinguish \Cs{_{list}} from \Cs{_{tree}} by just inspecting the sizes of the code segments.
So, in order for \compgen{\cdot} to be fully abstract, it must produce code of a fixed size~\cite{KULeuven-358154,scoo-j}. This wastes memory and makes it impossible to compile some components.
An alternative would be to spread the components in an overly-large memory at random places i.e., use address-space layout randomization or ASLR, so that detecting different code sizes has a negligible chance of success~\cite{Abadi:2012,Jagadeesan:2011:LMV:2056311.2056556}. However, ASLR is now known to be broken~\cite{break-aslr2,break-aslr1}.
\end{example}
Again, we see that \facomp introduces an inefficiency in compiled code (pointless code memory consumption) even though this has no security implication here.
In contrast, \rscomp does not require this unless the safety property(ies) of interest care about the size of the code (which is very unlikely in a security context, since security by code obscurity is a strongly discouraged security practice). 
In particular, the monitors considered in this paper cannot depend on code size.

\begin{example}[Wrappers for heap resources]\label{ex:wrap}
Assume that the \lst{Database} class is implemented as a \lst{List}.
Shown below are two implementations of the \lst{newList} method inside \lst{List} which we call \Cs{_{one}} and \Cs{_{two}}.
The only difference between \Cs{_{one}} and \Cs{_{two}} is that \Cs{_{two}} allocates two lists internally; one of these (\lst{shadow}) is used for internal purposes only.

\noindent\begin{minipage}{.45\textwidth}
\begin{lstlisting}[mathescape]
public newList(): List{

  ell = new List();
  return ell;
} 
\end{lstlisting} 
\end{minipage}\hfill
\begin{minipage}{.45\textwidth}
\begin{lstlisting}[mathescape]
public newList(): List{
  shadow = new List();
  ell = new List();
  return ell;
}
\end{lstlisting} 
\end{minipage}

Again, \Cs{_{one}} and \Cs{_{two}} are equivalent in a source language that does not allow pointer comparison. To attain {\facomp} when the target allows pointer comparisons, the pointers returned by \lst{newList} in the two implementations must be the same, but this is very difficult to ensure since the second implementation does more allocations.
A simple solution to this problem is to wrap \lst{ell} in a proxy object and return the proxy~\cite{scoo-j,mfac,KULeuven-358154,Morris:1973:PPL:361932.361937}.
Compiled code needs to maintain a lookup table mapping the proxy to the original object. 
Proxies must have allocation-independent addresses.
Proxies work but they are inefficient due to the need to look up the table on every object access.

Another way to attain \facomp is to weaken the source language, introducing an operation to distinguish object identities in the source~\cite{gcFA}.
However, this is a widely discouraged practice, as it changes the source language from what it really is and the implication of such a change may be difficult to fathom for programmers and verifiers.
\end{example}
In this example, \facomp forces all privately allocated locations to be wrapped in proxies, however \rscomp does not require this. 
Our target languages \LP and \LC support address comparison (addresses are natural numbers in their heaps) but \compup{\cdot} and \compap{\cdot} just use capabilities to attain security efficiently.
On the other hand, for attaining \facomp, capabilities alone would be insufficient since they do not hide addresses; proxies would still be required (this point is concretely demonstrated in \Cref{sec:facomp-instance}).

\begin{example}[Strict termination vs divergence]\label{ex:termina}
Consider a  source language that is strictly terminating while a target language that is not.
Below is an extension of the password manager to allow database encryption via an externally-defined function.
As the database is not directly accessible from external code, the two implementations below \Cs{_{enc}} (which does the encryption) and \Cs{_{skip}} which skips the encryption are equivalent in the source.

\noindent\begin{minipage}{.45\textwidth}
\begin{lstlisting}[mathescape]
public encryptDB( func : Database -> Bitstring) : void {
  func( this.db );
  return;
}
\end{lstlisting} 
\end{minipage}\hfill
\begin{minipage}{.45\textwidth}
\begin{lstlisting}[mathescape]
public encryptDB( func : Database -> Bitstring) : void {

  return;
}
\end{lstlisting} 
\end{minipage}

If we compile \Cs{_{enc}} and \Cs{_{skip}} to an assembly language, the compiled counterparts \emph{cannot} be equivalent, since the target-level context can detect which function is compiled by passing a \lst{func} that diverges. Calling the compilation of \Cs{_{enc}} with such a \lst{func} will cause divergence, while calling the compilation of \Cs{_{skip}} will immediately return.
\end{example}
This case presents a situation where \facomp is outright \emph{impossible}. 
The only way to get \facomp is to make the source language artificially non-terminating, see the work of Devriese \etal~\citet{lambda-seal} for more details of this particular problem.
On the other hand, \rscomp can be easily attained even in such settings since it is completely independent of termination in the languages (note that program termination and nontermination are both different from the monitor getting stuck on an action, which is what \rscomp cares about). 
Indeed, if our source languages \LU and \LA were restricted to terminating programs only, the same compilers and the same proofs of \rscomp would still work.

\paragraph{Remark}
It is worth noting that many of the inefficiencies above could be resolved by just replacing contextual equivalence with a different equivalence in the statement of \facomp.
However, it is not known how to do this generally for arbitrary sources of inefficiency and, further, it is unclear what the security consequences of such instantiations of \facomp would be. 
On the other hand, {\rscomp} is \emph{uniform} and it does address all these inefficiencies.

An issue that can normally not be addressed just by tweaking equivalences is side-channel leaks, as they are, by definition, not expressible in the language.
Neither \facomp nor \rscomp deals with side channels, but recent results describe how to account for side channels in secure compilers~\cite{BartheGL18}.

\section{Towards a Fully Abstract Compiler from \LU to \LP}\label{sec:facomp-instance}
This section sketches a fully abstract compiler from \LU to \LP.

\subsection{Language Extensions to \LU and \LP}\label{sec:fac-langs}
This section lists the language extensions required by the compiler. It is not possible to motivate all of them before explaining the details of the compiler, so some of the justification is postponed to \Cref{sec:fac-comp}.

A first concern for full abstraction is that a target context can always determine the memory consumption of two compiled components, analogously to \Cref{ex:size}. 
To ensure that this does not break full abstraction, we add a source expression \src{size} that returns the amount of locations \src{\ell} allocated in the current heap \src{H}.

In the target language \LP, we need to know whether an expression is a pair, whether it is a location, and we need to be able to compare two capabilities. 
For this, we add the expression constructs  \trg{isloc(e)}, \trg{ispair(e)} and \trg{eqcap(e,e)}, respectively.

Finally, compiled code needs private functions for its runtime checks that must not be visible to the context.
\LP does not have this functionality: all functions defined by a component can be called by the context.
Now we modify \LP so that all functions \trg{\OB{F}} defined in a component are by default private to it.
Additionally, each component must explicitly define the list of functions it exports (typically a subset of \trg{\OB{F}}), so that those are the only ones that can be called by the context and the rest are private to the component.

\subsection{The \compfac{\src{\cdot}} Compiler}\label{sec:fac-comp}
\compfac{\src{\cdot}} is similar to \compup{\cdot} but with critical differences.
We know that fully abstract compilation preserves \emph{all} source abstractions in the target language. Here, the only abstraction that distinguishes \LP from \LU is that locations are abstract in \LP, but concrete natural numbers in \LU.
Thus, locations allocated by compiled code must not be passed directly to the context as this would reveal the allocation order (as seen in \Cref{ex:wrap}).
Instead of passing the location \trg{\pair{n,k}} to the context, the compiler arranges for an opaque handle \trg{\pair{n',k_{com}}} (that cannot be used to access any location directly) to be passed.
Such an opaque handle is often  called a \emph{mask} or \emph{seal} in the literature.

To ensure that masking is done properly, \compfac{\src{\cdot}} inserts code at entry points and at exit points to compiled code, \emph{wrapping} the compiled code in a way that enforces masking.
This notion of wrapping is standard in literature on fully abstract compilation~\cite{fstar2js,mfac}.
The wrapper keeps a list \trg{\OB{L}} of component-allocated locations that are shared with the context in order to know their masks.
When a component-allocated location is shared, it is added to the list \trg{\OB{L}}. The mask of a location is its index in this list. If the same location is shared again it is not added again but its previous index is used.
So if \trg{\pair{n,k}} is the 4th element of \trg{\OB{L}}, its mask is \trg{\pair{4,k_{com}}}.
To implement lookup in \trg{\OB{L}} we must compare capabilities too, so we rely on \trg{eqcap}.
To ensure capabilities do not leak to the context, the second field of the pair is a constant capability \trg{k_{com}} whose location the compiled code does not use otherwise.
Technically speaking, this is exactly how existing fully abstract compilers operate (e.g.,~\cite{scoo-j}).

As should be clear, this kind of masking is very inefficient at runtime. 
However, even this masking is not sufficient for full abstraction. Next, we explain additional things the compiler must do.

\paragraph{Determining when a Location is Passed to the Context.}
A component-allocated location can be passed to the context not just as a function argument but on the heap.
So before passing control to the context the compiled code needs to scan the whole heap where a location can be passed and mask all found component-allocated locations.
Dually, when receiving control the compiled code must scan the heap to unmask it.
The problem now is determining what parts of the heap to scan and how.
Specifically, the compiled code needs to keep track of all the locations (and related capabilities) that are shared, i.e., (i) passed from the context to the component and (ii) passed from the component to the context.
These are the locations on which possible communication of locations can happen.
Compiled code keeps track of these shared locations in a list \trg{\OB{S}}.
Intuitively, on the first function call from the context to the compiled component, assuming the parameter is a location, the compiled code will register that location and all other locations reachable from it in \trg{\OB{S}}.
On subsequent ? (incoming) actions, the compiled code will register all new locations available as parameters or reachable from \trg{\OB{S}}.
Then, on any ! (outgoing) action, the compiled code must scan whatever locations (that the compiled code has created) are now reachable from \trg{\OB{S}} and add them to \trg{\OB{S}}. We need the new instructions \trg{isloc} and \trg{ispair} in \LP to compute these reachable locations.
Of course, this kind of scanning of locations reachable from \trg{\OB{S}} at every call/return between components can be extremely costly.

\paragraph{Enforcing the Masking of Locations}
The functions \trg{mask} and \trg{unmask} are added by the compiler to the compiled code. 
The first function takes a location (which intuitively contains a value \trg{v}) and replaces (in \trg{v}) any pair \trg{\pair{n,k}} of a location protected with a component-created capability \trg{k} with its index in the masking list \trg{\OB{L}}.
The second function replaces any pair \trg{\pair{n,k_{com}}} with the $n$th element of the masking list \trg{\OB{L}}.
These functions should not be directly accessible to the context (else it can \trg{unmask} any \trg{mask}'d location and break full abstraction). 
This is why \LP needs private functions.

\paragraph{Letting the Context use Masked Locations}
Masked locations cannot be used directly by the context to be read and written.
Thus, compiled code must provide a \trg{read} and a \trg{write} function (both of which are public) that implement reading and writing to masked locations.

\smallskip

As should be clear, code compiled through \compfac{\src{\cdot}} has a lot of runtime overhead in calculating the heap reachable from \trg{\OB{S}} and in \trg{mask}ing and \trg{unmask}ing locations.
Additionally, it also has code memory overhead: the functions \trg{read}, \trg{write}, \trg{mask}, \trg{unmask} and list manipulation code must be included.
Finally, there is data overhead in maintaining  \trg{\OB{S}}, \trg{\OB{L}} and other supporting data structures to implement the runtime checks described above. 
In contrast, the code compiled through \compup{\cdot} (which is just robustly safe and not fully abstract) has none of these overheads.

\subsection{Proving that \compfac{\src{\cdot}} is a Fully Abstract Compiler}\label{sec:fac-inst-proof}
Using \compfac{\src{\cdot}} as a concrete example, we now discuss why \emph{proving}  \facomp is harder than proving \rscomp.
Consider the hard part of \facomp, the forward implication,
  $\Cs{_1}\ceqs\Cs{_2}\Rightarrow\compgen{\Cs{_1}}\ceqt\compgen{\Cs{_2}}$. The contrapositive of this statement is
  $\compgen{\Cs{_1}}\nceqt\compgen{\Cs{_2}}\Rightarrow\Cs{_1}\nceqs\Cs{_2}$.
By unfolding the definition of $\nceqs$ we see that, given a target context \ctxt{} that distinguishes \compgen{\Cs{_1}} from \compgen{\Cs{_2}}, it is necessary to show that there exists a source context \ctxs{} that distinguishes \Cs{_1} from \Cs{_2}.
That source context \ctxs{} must be built (backtranslated) starting from the already given target context~\ctxt{} that differentiates \compgen{\Cs{_1}} from \compgen{\Cs{_2}}.

A backtranslation directed by the syntax of the target context \ctxt{} is hopeless here since  the target expressions \trg{iscap} and \trg{isloc} cannot be directly backtranslated to valid source expressions.
Hence, we resort to another well-known technique~\cite{KULeuven-358154,mfac}. 
First, we define a \emph{fully abstract (labeled) trace semantics} for the target language. 
A trace semantics is fully abstract when its notion of equivalence coincides with contextual equivalence, and thus can be used in place of the latter.
Specifically, this means that two components are contextually inequivalent iff their trace semantics differ in at least one trace. 
We write \trt{\trg{C}} to denote the traces of the component \trg{C} in this fully abstract semantics. 
Given this trace semantics, the statement of the forward implication of full abstraction reduces to: 
\begin{align*}
  \trt{\compfac{\src{C_1}}}\neq\trt{\compfac{\src{C_2}}}\Rightarrow\Cs{_1}\nceqs\Cs{_2}.
\end{align*}
The advantage of this formulation over the original one is that now we can construct a distinguishing source context for \src{C_1} and \src{C_2} using the \emph{trace} on which \trt{\compfac{\src{C_1}}} and \trt{\compfac{\src{C_2}}} disagree.
While this proof strategy of constructing a source context from a trace is similar to our proof of \rscomp, %
it is fundamentally much harder and much more involved. There are two reasons for this.

First, fully abstract trace semantics are much more complex than our simple trace semantics of \LP from earlier sections. 
The reason is that our earlier trace semantics include the entire heap in every action, but this breaks full abstraction of the trace semantics: such trace semantics also distinguish contextually equivalent components that differ in their internal private state. 
In a fully abstract trace semantics, the trace actions must record only those heap locations that are shared between the component and the context. Consequently, the definition of the trace semantics must inductively track what has been shared in the past. In particular, the definition must account for locations reachable indirectly from explicitly shared locations. This complicates both the definition of traces and the proofs that build on the definition.

Second, the source context that the backtranslation constructs from a target trace must simulate the shared part of the heap at every context switch. Since locations in the target may be masked, the source context must maintain a map with the source locations corresponding to the target masked ones, which complicates it substantially.
We call this map \src{B}.
Now, this affects two patterns of target traces that need to be handled in a special way: \trg{\clh{read}{v}{H}\cdot\rth{}{H'}{}} and \trg{\clh{write}{v}{H}\cdot\rth{}{H'}{}}.
Normally, these patterns would be translated in source-level calls to the same functions (\src{read} and \src{write}), but this is not possible.
In fact, the source code has no \src{read} nor \src{write} function, and the target-level calls to those functions need to be backtranslated to the corresponding source constructs (\src{!} and \src{:=}, respectively).
The locations used by these constructs must be looked up from \src{B} as these are reads and writes to masked locations.
Moreover, calls and returns to \trg{read}  can be simply ignored since the effects of reads are already captured by later actions in traces.
Calls and returns to \trg{write} cannot be ignored as they set up a component location (albeit masked) in a certain way and that affects the behaviour of the component.
We show in \Cref{ex:backtr-fatr} how to backtranslate calls and returns to \trg{write}.
\begin{example}[Backtranslation of traces]\label{ex:backtr-fatr}
Consider the trace below and its backtranslation.

{\centering
		\begin{tikzpicture}[remember picture]
			\node[align=left](trace){ 
				${\scriptstyle (1)}\left. \trg{~~\clh{f}{0}{ \trg{1\mapsto4} }} \right. $
				\\
				${\scriptstyle (2)}\left. \trg{~~\rth{}{1\mapsto\pair{1,k_{com}} }{}} \right. $
				\\
				${\scriptstyle (3)}\left[
					\begin{aligned}
						&\trg{\clh{write}{\pair{\pair{1,k_{com}},5}}{ 1\mapsto\pair{1,k_{com}} }}
						\\
						&\trg{\rth{}{ 1\mapsto\pair{1,k_{com}} }{} }
					\end{aligned}
				\right.$
				};

			\node[align=left, right of=trace, xshift = 17em](code){
				\src{main(x)\mapsto}
				\\
				$\left.
				\begin{aligned}[c]
					&\src{\tikz\node(ch1){}; \letnew{x}{4}{L::\pair{x,1}}}\quad
					\\
					&\src{\tikz\node(zzz){}; \call{f}~0}
				\end{aligned}
				\right]{\scriptstyle (1)}$
				\\
				$\left.\src{\tikz\node(ch3){};\, \letin{x}{!L(1)}{B::\pair{x,1}}} \qquad \right]{\scriptstyle (2)}$
				\\
				$\left.
				\begin{aligned}[c]
					&\src{\tikz\node(ch4){}; !B(1) := 5 }\quad
				\end{aligned}
				\right]{\scriptstyle (3)}$
			};

		\end{tikzpicture}
}
	The first action, where the context registers the first location in the list \src{L}, is as before. %
	Then in the second action the compiled component passes to the context (in location \trg{1}) a masked location with index \trg{1} and, later, the context writes \trg{5} to it.
	The backtranslated code must recognise this pattern and store the location that, in the source, corresponds to the mask \trg{1} in the list \src{B} (action 2).
	In action 3, when it is time to write \src{5} to that location, the code looks up the location to write to from \src{B}.
\end{example}
It should be clear that this proof of \facomp is substantially harder than our corresponding proof of \rscomp, which needed neither fully abstract traces, nor tracking any mapping in the backtranslated source contexts.

\section{A Fully Abstract Compiler from \LU to \LP}\label{sec:fa}
We perform the aforementioned changes to languages.

\subsection{The Source Language \LU}\label{sec:fa-src}
In \LU we need to add a functionality to get the size of a heap, as that is an observable that exists in the target.
In fact, in the target if one allocates something, that reveals how much it's been allocated entirely.

\begin{align*}
	\mi{Components}~\src{C} \bnfdef&\ \src{\OB{F} ; \OB{I} ; \OB{E}}
	\\
	\mi{Exports}~\src{E} \bnfdef&\ \src{f}
	\\
	\mi{Expressions}~\src{e} \bnfdef&\ \cdots \mid \src{size}
\end{align*}
\begin{center}
	\typerule{\LU-Size}{
		\card{\src{H}} = n
	}{
		\src{H\triangleright size \reds n}
	}{es-size}
\end{center}

\mytoprule{\text{Helpers}}
\begin{center}
	\typerule{\LU-Jump-Internal}{
		((\src{f'}\in\src{\OB{I}} \wedge \src{f}\in\src{\OB{I}}) \vee
				\\
		(\src{f'}\in\src{\OB{E}} \wedge \src{f}\in\src{\OB{E}}))
	}{
		\src{\OB{I},\OB{E}}\vdash\src{f,f'}:\src{internal}
	}{us-aux-intern}
	\typerule{\LU-Jump-IN}{
		\src{f}\in\src{\OB{I}} \wedge \src{f'}\notin\src{\OB{E}}
	}{
		\src{\OB{I},\OB{E}}\vdash\src{f,f'}:\src{in}
	}{us-aux-in}
	\typerule{\LU-Jump-OUT}{
		\src{f}\in\src{\OB{E}} \wedge \src{f'}\in\src{\OB{I}}
	}{
		\src{\OB{I},\OB{E}}\vdash\src{f,f'}:\src{out}
	}{us-aux-out}
	\typerule{\LU-Plug}{
		\src{A} \equiv \src{H ; \OB{F}\hole{\cdot}}
		&
		\src{C}\equiv\src{\OB{F'} ; \OB{I}; \OB{E}} 
		&
		\vdash\src{C,\OB{F}}:\src{whole}
		&
		\src{main}\in\fun{names}{\src{\OB{F}}}
		\\
		\forall \src{f}\in\src{\OB{E}}, \src{f}\notin\fun{fn}{\src{\OB{F}}}
		&
		\forall \src{f}\in\fun{fn}{\src{\OB{F'}}}, \src{f}\in\src{\OB{I}} \vee \src{f}\in\src{\OB{F'}}
	}{
		\src{A\hole{C}} = \src{H; \OB{F;F'}; \OB{I} ; \OB{E}}
	}{plug-us}
	\typerule{\LU-Whole}{
		\src{C}\equiv\src{\OB{F'} ; \OB{I} ; \OB{E}} 
		\\
		\fun{names}{\src{\OB{F}}}\cap\fun{names}{\src{\OB{F'}}}=\emptyset
		\\
		\fun{names}{\src{\OB{I}}}\subseteq \fun{names}{\src{\OB{F}}}\cup\fun{names}{\src{\OB{F'}}}
		\\
		\fun{fv}{\src{\OB{F}}}\cup\fun{fv}{\src{\OB{F'}}}=\srce
	}{
		\vdash\src{C,\OB{F}}:\src{whole}
	}{whole-us}
	\typerule{\LU-Initial State}{
		\src{P}\equiv\src{H ; \OB{F} ; \OB{I} ; \OB{E}}
		\\
		\src{C}\equiv\src{\OB{F} ; \OB{I} ; \OB{E}} 
	}{
		\SInits{P} = \src{C ; H \triangleright \call{main}~ 0}
	}{ini-us}
\end{center}
\botrule
The semantics is unchanged, it only relies on the new helper functions above.

\subsection{The Target Language \LP}\label{sec:fa-trg}
\subsubsection{Syntax Changes}
\begin{align*}
	\mi{Components}~\trg{C} \bnfdef&\ \trg{\OB{F} ; \OB{I}; \OB{E} ; k_{root}, k_{com}}
	\\
	\mi{Exports}~\trg{E} \bnfdef&\ \trg{f}
	\\
	\mi{Expressions}~\trg{e} \bnfdef&\ \cdots \mid \trg{isloc(e)} \mid \trg{ispair(e)} \mid \trg{eqcap(e,e)}
	\\
	\mi{Trace\ states}~\trg{\Theta} \bnfdef&\ \trg{(C; H ; \OB{n} \triangleright \proc{t}{\OB{f}})} 
	\\
	\mi{Trace\ bodies}~\trg{t} \bnfdef&\ \trg{s} \mid \trg{\unk}
	\\
	\mi{Trace\ labels}~\trg{\delta} \bnfdef&\ \trg{\epsilon} \mid \trg{\beta}
	\\
	\mi{Trace\ actions}~\trg{\beta} \bnfdef&\  \trg{\clh{f}{v}{H}} \mid \trg{\cbh{f}{v}{H}} \mid \trg{\rth{v}{H}} \mid \trg{\rbh{v}{H}} \mid \tert \mid \divt \mid \trg{\wrl{v,n}}
	\\
	\mi{Traces}~\trg{\OB{\beta}} \bnfdef&\ \trge \mid \trg{\OB{\beta}\cdot \beta}
\end{align*}
We assume programs are given two capabilities they own: \trg{k_{root}} and \trg{k_{com}} and that the attacker does not have.
The former is used to create a part of the heap for component-managed datastructures.
The latter does not even hide a location, we need it as a placeholder.

Traces in this case have the same syntactic structure as before, but they do not carry the whole heap.
So we use a different symbol (\trg{\beta}), to visually distinguish between the two traces and the kind of information carried by them.

We need a write label \trg{\wrl{v,n}} that tells that masked location \trg{n} has been updated to value \trg{v}.
This captures the usage of compiler-inserted functions to read and write masked locations (concepts that will be clear once the compiler is defined).
The read label is not needed because its effect are captured anyway by call/return.

Trace states are either operational semantics states or an unknown state, mimicking the execution in a context.
The former has an addtional element \trg{\OB{n}}, the list of locations shared with the context.
The latter carries the information about the component and the heap comprising the one private to the component and the one shared with the context.
It also carries the stack of function calls, where we add symbol \trg{\unk} to indicate when the called function was in the context.

Helper functions are as above.

\subsubsection{Semantics Changes}
In \LP we need functionality to tell if a pair is a location or not and to traverse values in order to extract such locations.
\begin{center}
	\typerule{\LP-isloc}{
		(\trg{H\triangleright e \redt \pair{n,v}}
		&
		\trg{n\mapsto \_;\eta}\in\trg{H} 
		&
		\trg{\eta}=\trg{v} \text{ or } \trg{\eta}=\trg{\bot}
		)\Rightarrow \trg{b}=\truet
		\\
		\text{ otherwise }\trg{b}=\falset
	}{
		\trg{H\triangleright isloc(e) \redt \trg{b}}
	}{et-isloc}
	\typerule{\LP-ispair}{
		\trg{H\triangleright e \redt \pair{v,v}}
		\Rightarrow \trg{b}=\truet
		\\
		\text{ otherwise }\trg{b}=\falset
	}{
		\trg{H\triangleright ispair(e) \redt \trg{b}}
	}{et-ispair}
\end{center}
These are used to traverse the value stored at a location and extract all sublocations stored in there.
There may be pairs containing pairs etc, and thus when we need to know if something is a pair before projecting out.
Also, we need to know if a pair is a location or not, in order to know whether or not we can dereference it.

Additionally, we need a functionality to tell if two capabilities are the same.
Now, this could be problematic because it could reveal capability allocation order and thus introduce observations that we do not want.
However, the compiler will ensure that the context only receives \trg{k_{com}} as a capability and never a newly-allocated capability.
So the context will not be able to test equality of capabilities generated by the compiled component as it will effectively see only one.

\begin{center}
	\typerule{\LP-eqcap-true}{
		\trg{H\triangleright e \redt k}
		&
		\trg{H\triangleright e' \redt k}
	}{
		\trg{H\triangleright eqcap(e,e') \redt \truet}
	}{et-eqcap-t}
	\typerule{\LP-eqcap-false}{
		\trg{H\triangleright e \redt k}
		&
		\trg{H\triangleright e' \redt k'}
		&
		\trg{k}\neq\trg{k'}
	}{
		\trg{H\triangleright eqcap(e,e') \redt \falset}
	}{et-eqcap-t}
\end{center}

\subsubsection{A Fully Abstract Trace Semantics for \LP}\label{src:fa-traces-trg}
\begin{align*}
	\trg{\Theta \xtol{\beta} \Theta'} &&& \text{State \trg{\Theta} emits visible action \trg{\beta} becoming \trg{\Theta'}.}
	\\
	\trg{\Theta \Xtol{\OB{\beta}} \Theta'} &&& \text{State \trg{\Theta} emits trace \trg{\OB{\beta}} becoming \trg{\Theta'}.}
\end{align*}

\mytoprule{ \textit{Helper functions} }
\begin{center}
	\typerule{Reachable}{
		\trg{n} \in \fun{reach}{ \trg{n_{st}}, \trg{k_{st}}, \trg{H} } 
		&
		\trg{n_{st}\mapsto \_:\_}\in\trg{H'}
		\\
		\trg{k_{st}}\in\trg{k_{root}}\cup\trg{H'}
		&
		\trg{n\mapsto v:\eta}\in\trg{H}
	}{
		\trg{H}\vdash\fun{reachable}{\trg{n},\trg{H'}}
	}{reachable-2}
	\typerule{Valid value}{
		\forall \trg{k}\in\trg{H}.~ \trg{k}\notin\trg{v}
	}{
		\vdash\fun{valid}{\trg{v},\trg{H}}
	}{valid-v}
	\typerule{Valid heap}{
		\trg{H}=\trg{H_{priv} \cup H_{sha}}
		&
		\trg{H'}=\trg{H_{priv} \cup H_{sha}' \cup H_{new}}
		\\
		\trg{H''}=\trg{H_{sha}' \cup H_{new}}
		&
		\dom{\trg{H}}=\trg{\OB{n}}
		&
		\dom{\trg{H''}}=\trg{\OB{n'}}
		\\
		\trg{k}\in\trg{H_{sha}} \iff \trg{k}\in\trg{H_{sha}'}
		\\
		\forall\trg{k'}\in\trg{H_{new}}.~ \trg{k'}\notin\trg{H_{priv}\cup H_{sha}}
		\\
		\forall \trg{n\mapsto v;\eta}\in\trg{H_{sha}}.~ \trg{n\mapsto v'';\eta}\in\trg{H_{sha}'} \wedge \vdash\fun{valid}{\trg{v''},\trg{H}}
		\\
		\forall \trg{n'\mapsto v':\eta'}\in\trg{H_{new}}.~ 
			\vdash\fun{valid}{\trg{v'},\trg{H_{priv}\cup H_{sha}'}} \wedge 
			\\
			& \trg{H''}\vdash\fun{reachable}{\trg{v'},\trg{H_{priv}\cup H_{sha}'}}
	}{
		\vdash\fun{validHeap}{\trg{H},\trg{H'},\trg{H''},\trg{\OB{n}},\trg{\OB{n'}}}
	}{valid-h}
\end{center}
\botrule

\mytoprule{ \trg{\Theta \xtol{\OB{\beta}} \Theta'} }
\begin{center}
	\typerule{\LP-Traces-Silent}{
		\trg{ (C; H ; \OB{n} \triangleright \proc{s}{\OB{f}}) \xtot{\epsilon} (C; H' ; \OB{n} \triangleright \proc{s'}{\OB{f'}}) }
	}{
		\trg{ (C; H ; \OB{n} \triangleright \proc{s}{\OB{f}}) \xtol{\epsilon} (C; H' ; \OB{n} \triangleright \proc{s'}{\OB{f'}}) }
	}{etr-ter-s}
	\typerule{\LP-Traces-Call}{
		\trg{C}=\trg{\OB{F};\OB{I};\OB{E}}
		&
		\trg{f}\in\trg{\OB{E}}
		&
		\trg{f(x)\mapsto s;\ret}\in\trg{\OB{F}}
		\\
		\trg{\OB{f'}}=\trg{\OB{f}\cdot f}
		&
		\trg{H}\vdash\fun{valid}{\trg{v}}
		\\
		\vdash\fun{validHeap}{\trg{H},\trg{H''},\trg{H'},\trg{\OB{n}},\trg{\OB{n'}}}
	}{
		\trg{(C; H ; \OB{n} \triangleright \proc{\unk}{\OB{f}}) \xtol{\clh{f}{v}{H'}} (C; H'' ; \OB{n'} \triangleright \proc{s;\ret}{\OB{f'}})}
	}{etr-call}
	\typerule{\LP-Traces-Returnback}{
		\trg{\OB{f}}=\trg{\OB{f'}\cdot f}
		\\
		\vdash\fun{validHeap}{\trg{H},\trg{H''},\trg{H'},\trg{\OB{n}},\trg{\OB{n'}}}
	}{
		\trg{(C; H ; \OB{n} \triangleright \proc{\unk}{\OB{f}}) \xtol{\rbh{}{}{H'}} (C; H'' ; \OB{n'} \triangleright \proc{\skipt}{\OB{f'}})}
	}{etr-call}
	\typerule{\LP-Traces-Callback}{
		\trg{s} = \trg{\call{f}~e}
		&
		\trg{H\triangleright e \redt v}
		\\
		\trg{C}=\trg{\OB{F};\OB{I};\OB{E}}
		&
		\trg{\OB{f'}}=\trg{\OB{f}\cdot f}
		&
		\trg{f}\in\trg{\OB{I}}
		\\
		\trg{\OB{n}}\subseteq\trg{\OB{n'}}= \myset{\trg{n}}{ \trg{H}\vdash\fun{reachable}{\trg{n},\trg{H}}}
		&
		\trg{H'}=\trg{H|_{\trg{\OB{n'}}}}
	}{
		\trg{ (C; H ; \OB{n} \triangleright \proc{s}{\OB{f}}) \xtol{\cbh{f}{v}{H'}} (C; H ; \OB{n'} \triangleright \proc{\unk}{\OB{f'}}) }
	}{etr-call}
	\typerule{\LP-Traces-Return}{
		\trg{C}=\trg{\OB{F};\OB{I};\OB{E}}
		&
		\trg{\OB{f}}=\trg{\OB{f'}\cdot f}
		&
		\trg{f}\in\trg{\OB{E}}
		\\
		\trg{\OB{n}}\subseteq\trg{\OB{n'}}= \myset{\trg{n}}{ \trg{H}\vdash\fun{reachable}{\trg{n},\trg{H}}}
		&
		\trg{H'}=\trg{H|_{\trg{\OB{n'}}}}
	}{
		\trg{ (C; H ; \OB{n} \triangleright \proc{\ret}{\OB{f}}) \xtol{\rth{}{}{H'}} (C; H ; \OB{n'} \triangleright \proc{\unk}{\OB{f'}}) }
	}{etr-call}
	\typerule{\LP-Traces-Terminate}{
		\trg{ (C; H ; \OB{n} \triangleright \proc{s}{\OB{f}}) \nxtot{\epsilon} \_ }
	}{
		\trg{ (C; H ; \OB{n} \triangleright \proc{s}{\OB{f}}) \xtol{\tert} (C; H ; \OB{n} \triangleright \proc{s}{\OB{f}}) }
	}{etr-ter}
	\typerule{\LP-Traces-Diverge}{
		\forall n.~ \trg{ (C; H ; \OB{n} \triangleright \proc{s}{\OB{f}}) \xtot{\epsilon}\redapp{n} (C; H' ; \OB{n'} \triangleright \proc{s'}{\OB{f'}}) }
	}{
		\trg{ (C; H ; \OB{n} \triangleright \proc{s}{\OB{f}}) \xtol{\divt} (C; H ; \OB{n} \triangleright \proc{s}{\OB{f}}) }
	}{etr-div}
	\typerule{\LP-Traces-Write}{
		\trg{C}=\trg{\OB{F};\OB{I};\OB{E}}
		&
		\trg{write}\in\trg{\OB{E}}
		&
		\trg{write(x)\mapsto s;\ret}\in\trg{\OB{F}}
		\\
		\trg{C;H\triangleright s\subt{n}{x};\ret \xtot\redapp{*} C;H'\triangleright \ret}
	}{
		\trg{ (C; H ; \OB{n} \triangleright \proc{\unk}{\OB{f}}) \xtol{\wrl{v,n}} (C; H' ; \OB{n} \triangleright \proc{\unk}{\OB{f}}) }
	}{etr-wr}
	\typerule{\LP-Traces-Read}{
		\trg{C}=\trg{\OB{F};\OB{I};\OB{E}}
		&
		\trg{read}\in\trg{\OB{E}}
		&
		\trg{read(x)\mapsto s;\ret}\in\trg{\OB{F}}
		\\
		\trg{C;H\triangleright s;\ret \xtot\redapp{*} C;H'\triangleright \ret}
	}{
		\trg{ (C; H ; \OB{n} \triangleright \proc{\unk}{\OB{f}}) \xtol{\epsilon} (C; H' ; \OB{n} \triangleright \proc{\unk}{\OB{f}}) }
	}{etr-rd}
\end{center}
\botrule

\mytoprule{ \trg{\Theta \Xtol{\OB{\beta}} \Theta'} }
\begin{center}
	\typerule{E\LP-single}{
		\trg{\Omega}\Xtot{}\trg{\Omega''}
		&
		\trg{\Omega''}\xtot{\beta}\trg{\Omega'}
	}{
		\trg{\Omega}\Xtot{\beta}\trg{\Omega'}
	}{et-tr-sin}
	\typerule{E\LP-silent}{
		\trg{\Omega}\xtot{\epsilon}\trg{\Omega'}
	}{
		\trg{\Omega}\Xtot{}\trg{\Omega'}
	}{et-tr-silent}
	\typerule{E\LP-trans}{
		\trg{\Omega}\Xtot{\OB{\beta}}\trg{\Omega''}
		&
		\trg{\Omega''}\Xtot{\OB{\beta'}}\trg{\Omega'}
	}{
		\trg{\Omega}\Xtot{\OB{\beta}\cdot\OB{\beta'}}\trg{\Omega'}
	}{et-tr-trans}
\end{center}
\botrule

\begin{center}
	\typerule{\LP-Traces-Initial}{
		\trg{n}\in\trg{\OB{n}} \iff \trg{n\mapsto v;\eta}\in\trg{H}
		& 
		\trg{main}\notin\dom{\trg{\OB{F}}}
		&
		\trg{C}=\trg{\OB{F};\OB{I};\OB{E}}
	}{
		\TInitt{C} = \trg{(C; H ; \OB{n} \triangleright \proc{\unk}{main})}
	}{etr-ini}	
\end{center}
\begin{align*}
	\trt{C} =&\ \myset{\trg{\OB{\beta}}}{ \trg{\TInitt{C} \Xtol{\OB{\beta}} \_ }}
\end{align*}

\subsubsection{Results about the Trace Semantics}
The following results hold for $\trg{C_1}=\compfac{\src{C_1}}$ and $\trg{C_2}=\compfac{\src{C_2}}$.

\begin{property}[Heap locations]\label{prop:heap loc}
	AS mentioned, the trace semantics carries the whole shared heap: locations created by the compiled component and then passed to the context and locations created by the context and passed to the compiled component.
	We can really partition the heap as follows then:
	\begin{center}
		\begin{tabular}{c | c| c| }
			location \textbackslash creator
			& \compfac{\src{C}} 
			& $\ctxt{}$
			\\
			\hline
			private
			&(1) to \compfac{\src{C}}
			&(2) to $\ctxt{}$
			\\
			\hline
			shared
			&(3) with $\ctxt{}$
			&(4) with \compfac{\src{C}}
		\end{tabular}
	\end{center}

	Now, for compiled components there never are locations of kind 3.
	That is because those locations are masked and never passed, never made accessible to the context.
	So really, the trace semantics only collects locations of kind 4 on the traces.
\end{property}

\begin{lemma}[Correctness]\label{thm:traces-corr}
	\begin{align*}
		\text{if }
		&
		\trg{C_1 \ceqt C_2}
		\\
		\text{then }
		&
		\trg{\trt{C_1}=\trt{C_2}}
	\end{align*}
\end{lemma}
\begin{proofsketch}
	By contraposition:
	\begin{align*}
		\text{if }
		&
		\trg{\trt{C_1}\neq\trt{C_2}}
		\\
		\text{then }
		&
		\trg{\exists\trg{A}.~ A\hole{C_1}\termt \wedge A\hole{C_2}\divrt} (\mi{wlog})
	\end{align*}

	We are thus given $\trg{\OB{\beta_1}} = \trg{\OB{\beta}\cdot\beta_1}$ and $\trg{\OB{\beta_2}} = \trg{\OB{\beta}\cdot\beta_2}$ and $\trg{\beta_1}\neq\trg{\beta_2}$.

	We can construct a context \trg{A} that replicates the behaviour of \trg{\OB{\beta}} and then performs the differentiation.
	
	This is a tedious procedure that is analogous to existing results~\cite{javajr,llfatr-j} and analogous to the backtranslation of \Cref{sec:general-backtr-trace-based}.

	The actions only share the heap that is reachable from both sides, the heap that is private to the component is never touched, so reconstructing the heap is possible.
	The reachability conditions on the heap also ensure this.

	The differentiation is based on differences on the actions which are visible and reachable, so that is also possible.
\end{proofsketch}

\begin{lemma}[Completeness]\label{thm:traces-compl}
\begin{align*}
		\text{if }
		&
		\trg{\trt{C_1}=\trt{C_2}}
		\\
		\text{then }
		&
		\trg{C_1 \ceqt C_2}
	\end{align*}
\end{lemma}
\begin{proofsketch}
	By contradiction let us assume that
	\begin{align*}
		\trg{\exists\trg{A}.~ A\hole{C_1}\Downarrow \wedge A\hole{C_2}\divrt}  (\mi{wlog})
	\end{align*}

	Contexts are deterministic, so they cannot behave differently based on the values of locations that are never shared with \trg{C_1} or \trg{C_2}.

	The semantics forbids guessing, so a context will never have access to the locations that \trg{C_1} or \trg{C_2} do not share.

	Thus a context can exhibit a difference in behaviour by relying on something that \trg{C_1} modified unlike \trg{C_2} and that can be:
	\begin{itemize}
		\item a parameter passed in a call.

			This contradicts the hypothesis that the trace semantics is the same as that parameter is captured in the \trg{\cbh{f}{v}{H}} label.

		\item the value of a shared location.

			This contradicts the hypothesis that the trace semantics is the same as all locations that are reachable both by the context and by \trg{C_1} and \trg{C_2} are captured on the labels
	\end{itemize}
	Having reached a contradiction, this case holds.
\end{proofsketch}

\begin{lemma}[Full abstraction of the trace semantics for compiled components]\label{thm:fatr-compiled-comp}
	\begin{align*}
		\trg{\trt{\compfac{\src{C_1}}}=\trt{\compfac{\src{C_2}}}}
		\iff&\
		\trg{\compfac{\src{C_1}} \ceqt \compfac{\src{C_2}}}
	\end{align*}
\end{lemma}
\begin{proof}
	By \Cref{thm:traces-corr,thm:traces-compl}.
\end{proof}

\subsection{The Compiler \compfac{\src{\cdot}}}

\begin{align*}
	\tag{\compfac{\src{\cdot}}-Comp}
		\compfac{
			\src{\OB{F} ; \OB{I} ; \OB{E}}
		} &= \trg{
			\begin{aligned}[t]
				&
				\compfac{\src{\OB{F}}}, \trg{read(x)\mapsto s_{read}}, \trg{write(x)\mapsto s_{write}},
				\\
				&\ 
					\trg{mask(x)\mapsto s_{mask}}, \trg{unmask(x)\mapsto s_{unmask}}, \cdots ; 
				\\
				&
				\compfac{\src{\OB{I}}} ; 
				\\
				&
				\compfac{\src{\OB{E}}}, \trg{read}, \trg{write} ; 
				\\
				&
				\trg{k_{root}},\trg{k_{com}}
			\end{aligned}
		}
	\\
	\tag{\compfac{\src{\cdot}}-Function}
		\compfac{
			\src{f(x) \mapsto s;\ret}
		} &= \trg{
			f(x)\mapsto
				\begin{aligned}[t]
					&
					\trg{s_{add}(x);}
					\\
					&
					\trg{s_{pre} ;}
					\\
					&
					\compup{\src{s}};
					\\
					&
					\trg{s_{post} ;}
					\\
					&
					\trg{\ret}
				\end{aligned}
		}
	\\
	\tag{\compfac{\src{\cdot}}-Interfaces}
		\compfac{\src{f}} &= \trg{f}
	\\
	\tag{\compfac{\src{\cdot}}-Exports}
		\compfac{\src{f}} &= \trg{f}
	\\
	\text{Expression translation}
	&\
	\text{ unmodified:} \compfac{\src{e}} = \compup{e}
	\\
	\text{Statement translation}
	&\
	\text{ unmodified except for }
	\\
	\tag{\compfac{\src{\cdot}}-New}  \label{tr:compfac-new}
		\compfac{
			\src{\letnew{x}{e}{s}}
		}
		=&\
		\trg{
			\begin{aligned}[t]
				&
				\letnewt{\trg{x_{loc}}}{\compfac{\src{e}}}{
				\\
				&\
					\lethide{\trg{x_{cap}}}{x_{loc}}{
					\\
					&\ \
					\trg{s_{register}(x_{loc},x_{cap});}
					\\
					&\ \ 
						\letint{\trg{x} ~}{~ \trg{\pair{x_{loc},x_{cap}}}}{\compfac{\src{s}}}
					}
				}
			\end{aligned}
		}	
	\\
	\tag{\compfac{\src{\cdot}}-Call}  \label{tr:compfac-call}
		\compfac{
			\src{\call{f}~{e}}
		}
		=&\
		\trg{
			\letin{x}{\compfac{e}}{
			\trg{s_{add}(x);s_{post}; \call{f}~x; s_{pre}}}
		}	
\end{align*}
So the compiler is mostly unchanged.

The compiled code will maintain the following invariant:
\begin{itemize}
	\item no locations (even though protected by capabilities) are ever made accessible ``in clear'' to the context;

	\item ``made accessible'' means either passed as a parameter or through a shared location;

	\item instead, before passing control to the context, all component-created locations that are shared with the context are masked, i.e., their representation \trg{\pair{n,k}} is replaced with \trg{\pair{n',k_{com}}}, where \trg{n'} is their index in the list of shared locations that the compiled component keeps.

	\item when receiving control from the context, the compiled component ensures that all component-created locations that are shared are unmasked, i.e., upon regaining control the component replaces all values \trg{\pair{n',k_{com}}} that are sub-values of reachable locations with \trg{\pair{n,k}}, which is  the \trg{n'}th pair in the list of component-allocated locations;

	\item what is a ``component-shared'' location?
	A shared location is a pair \trg{\pair{n,k}} where (i) \trg{k} is a capability created by the compiled component and (ii) the pair is stored in the heap at a location that the context can dereference (perhaps not directly).
	
	\item In order to define what is a shared location, the compiled component keeps a list of all the locations that have been passed to it and that the context created.
	These locations can only be in \trg{\pair{n,\_}} form, where \trg{\_} is either a capability or not depending whether the context hid the location.
	These locations can only be pairs since we know that a compiled component will only use pairs as locations, mimicking the source semantics.

	We normally do not know what locations will be accessed, but given a parameter that is a location, we can scan its contents to understand what new locations are passed.
	
	\item The compiled component thus can keep a list of ``shared'' locations: those whose contents are accessible both by the context and by itself.
	These locations created by the context are acquired as parameters or as locations reachable by a parameter.
	These locations created by the component are tracked as those hidden with a component-created capability and reachable from a shared location.

	\item The only concern that can arise is if we create location \trg{n} and then add it to the list of shared locations at index \trg{n'}.
	That location \trg{\pair{n,k}} would be masked as \trg{\pair{n',k_{com}}}, which grants the context direct access to it.
	This is where we need to use \trg{k_{com}} as leaking different capabilities would lead to differentiation between components.
	Fortunately, the context starts execution and, in order to call the compiled component, it must allocate at least one location, so this problem cannot arise.
\end{itemize}

\subsubsection{Syntactic Sugar}
The languages we have do not let us return directly a value.
In the following however, for readability, we write 
\begin{align*}
	\trg{\letin{x}{func~v}{s}}
\end{align*}
to intend: call function \trg{func} with parameter \trg{v} and store its returned value in \trg{x} for use in \trg{s}.

We indicate how that statement can be expressed in our language with the following desugaring:
\begin{align*}
	&
	\trg{\letnew{y}{0}{\letin{z}{\pair{v,y}}{\call{func}~z;\letin{x}{!\projtwo{z}~with~0}{s}}}}
\end{align*}

\subsubsection{Support Data Structures}
The compiler relies on a number of data structures it keeps starting from location \trg{0}, which is accessible via \trg{k_{root}}.

These data structures are:
\begin{itemize}
	\item a list of capabilities, which we denote with \trg{\OB{K}}.
	These capabilities are those that the compiled component has allocated.

	\item a list of component-allocated locations, which we denote with \trg{\OB{L}}.
	These are locations \trg{\pair{n,k}} that are created by the compiled component and whose \trg{k} are elements of \trg{\OB{K}}

	\item a list of shared locations, which we denote with \trg{\OB{S}}.
	These are either (i) locations that are created by the context and passed to the compiled component or (ii) locations that are created by the compiled component and passed to the context.

\end{itemize}
Given a list $L$ of elements $e$, we use these helper functions:
\begin{itemize}
	\item $\fun{indexof}{L,e}$ returns $n$, the index of $e$ in $L$, or $0$ if $e$ is not in $L$;
	\item $L(n)$ returns the nth element $e$ of $L$ or $0$ if the list length is shorter than $n$;
	\item $L::e$ if $e$ is not in $L$, it adds element $e$ to the list, increasing its length by 1;
	\item $\fun{rem}{L,e}$ removes element $e$ from $L$;
	\item $e\in L$ returns true or false depending on whether $e$ is in $L$ or not.
\end{itemize}
We keep this abstract syntax for handling lists and do not write the necessary recursive functions as they would only be tedious and hardly readable.
Realistically, we would also need a temporary list for accumulating results etc, again, this is omitted for simplicity and readability.

\subsubsection{Support Functions}
\paragraph{Read}
\begin{align*}
	\trg{s_{read}} =&\
		\begin{aligned}[t]
			&
			\letint{\trg{x_{n}}}{\trg{\projone{\projone{x}}}}{
			\\
			&\ 
				\letint{\trg{x_{k}}}{\trg{\projtwo{\projone{x}}}}{
			\\
			&\ \ 
					\letint{\trg{x_{real}}}{\trg{\OB{L}(x_n)}}{
			\\
			&\ \ \ 
						\letint{\trg{x_{dest}}}{\trg{\projone{\projtwo{x}}}}{
						\\
						&\ \ \ \
							\letint{\trg{x_{dcap}}}{\trg{\projtwo{\projtwo{x}}}}{
							\\
							&\ \ \ \ \ 
								\letint{\trg{x_{val}}}{\trg{ !x_{real}~with~x_{k}}}{
								\\
								&\ \ \ \ \ \ 
								\trg{x_{dest} := x_{val}~with~x_{dcap}}
								}
							}
						}
					}
				}
			}
		\end{aligned}
\end{align*}
In order to read a location \trg{\pair{n,k}}, we receive that as the first projection of parameter \trg{x}.
Because we do not explicitly return values, we need the second projection of \trg{x} to contain the destination where to target receives the result of the read.

We split the pair in the masking index \trg{x_n} and in the capability to access the location \trg{x_k}.
Then we lookup the location in the list of component-created locations and return its value.
We do not need to mask its contents as we know that they have already been masked when this location was shared with the context (line 5 of the postamble).
We do not need to add its contents to the list of shared locations as that is already done in lines 2 and 3 of the postamble.

\paragraph{Write}
\begin{align*}
	\trg{s_{write}} =&\
		\begin{aligned}[t]
			&
			\letint{\trg{x_{n}}}{\trg{\projone{\projone{x}}}}{
			\\
			&\ 
				\letint{\trg{x_{k}}}{\trg{\projtwo{\projone{x}}}}{
			\\
			&\ \ 
					\letint{\trg{x_{real}}}{\trg{\OB{L}(x_n)}}{
			\\
			&\ \ \ 
						\trg{x_{real}:= \projtwo{x} ~with~x_{k};}
					}
				}
			}
		\end{aligned}
\end{align*}
In order to write value \trg{v} a location \trg{\pair{n,k}} we receive a parameter structured as follows: $\trg{x}\equiv\trg{\pair{n,k},v}$.
Then we unfold the elements of the parameter and lookup element \trg{n} in the list of component-defined locations.
We use this looked-up element to write the value \trg{v} there.

We do not need to mask \trg{v} because it cannot point to locations that are created by the compiled component.

At this stage, \trg{v} may contain new locations created by the context and that are now shared.
We do not add them now to the list of shared locations because we know that upon giving control again to the compiled component, the preamble will do this.

\paragraph{Mask}
\begin{align*}
	\trg{s_{mask}} =&\
		\begin{aligned}[t]
			&
			\forall\trg{\pair{n,k}}\in\trg{x.~ isloc(\pair{n,k})}
			\\
			&
			\trg{if~k\in\OB{K}}
			\\
			&\ 
			\trg{replace~\pair{n,k}~with~ \pair{\fun{indexof}{\OB{L},n},k_{com}}}
		\end{aligned}
\end{align*}
We use the abstract construct \trg{replace ...} to indicate the following.
We want to keep the value passed as parameter \trg{x} unchanged but replace its subvalues that are pairs and, more specifically, component-created locations, with a pair with its location masked to be the index in the list of component-allocated locations.

This can be implemented by checking the sub-values of a value via the \trg{ispair} and \trg{isloc} expressions, we omit its details for simplicity.
To ensure $\in\trg{\OB{K}}$ is implementable, we use the \trg{eqcap} expression.

Masked locations cannot mention their capability or they would leak this information and generate different traces for equivalent compiled programs.

\paragraph{Unmask}
\begin{align*}
	\trg{s_{unmask}} =&\
		\begin{aligned}[t]
			&
			\forall\trg{\pair{n,k}}\in\trg{x}
			\\
			&
			\trg{if~k == k_{com}}
			\\
			&\ 
			\trg{replace~\pair{n,k}~with~ \OB{L}(n)}
		\end{aligned}
\end{align*}
In the case of unmasking, we receive a value through parameter \trg{x} and we know that there may be subvalues of it of the form \trg{\pair{n,k}} where \trg{n} is an index in the component-created shared locations.
So we lookup the element from that list and replace it in \trg{x}.

\subsubsection{Inlined Additional Statements (Preamble, Postamble, etc)}
\paragraph{Adding}
\begin{align*}
	\trg{s_{add}(x) } =&\
		\begin{aligned}[t]
			&
			\trg{if ~{isloc(x)}~ then}
			\\
			&\
			\trg{\OB{S}::x};
			\\			
			&\
			\trg{\ifte{\projtwo{x}\in\OB{K}}{\OB{L}::x}{skip}}
		\end{aligned}
\end{align*}
This common part ensures that the parameter \trg{x} is added to the list of shared locations (line 1) and then, if the capability is locally-created, it is also added to the list of locally-shared locations (line 2).

The second line is for when this code is called before a \compfac{\src{\call{f}}}.

\paragraph{Registration}
\begin{align*}
	\trg{s_{register}(x_{loc},x_{cap})} =&\
		\begin{aligned}[t]
			&
			\trg{\OB{K}::x_{cap};}
		\end{aligned}
\end{align*}
This statement registers capability \trg{x_{cap}} in the list of component-created capabilities. %

\paragraph{Preamble}
The preamble is responsible of adding all context-created locations to the list of shared locations and to ensure that all contents of shared locations are unmasked, as the compiled code will operate on them.
\begin{align*}
	\trg{s_{pre}} =&\
		\begin{aligned}[t]
			&
			\forall \trg{\pair{n,k}}\in\fun{reach}{\trg{\OB{S}}}.~ \trg{isloc(\pair{n,k})}
			\\
			&
			\trg{\ifte{\pair{n,k}\notin\OB{S}}{\OB{S}::\pair{n,k};}{skip}}
			\\
			&
			\forall\trg{\pair{n,k}}\in\trg{\OB{S}.~ isloc(\pair{n,k})}
			\\
			&\ \ 
			\trg{\letin{x}{unmask(!n~with~k)}{n := x~with~k}}
		\end{aligned}
\end{align*}
First any location that is reachable from the shared locations (line 1) and that is not a shared location already is added to the list of shared locations (line 2).
By where this code is placed we know that these new locations can only be context-created.

Then, for all shared locations (line 3), we unmask their contents using the \trg{unmask} function (line 4).

\paragraph{Postamble}
The postamble is responsible of adding all component-created locations to the list of shared locations and of component-created shared locations and to ensure that all shared locations are masked as the context will operate on them.
\begin{align*}
	\trg{s_{post}(x)} =&\
		\begin{aligned}[t]
			&
			\forall \trg{\pair{n,k}}\in\fun{reach}{\trg{\OB{S}}}.~ \trg{isloc(\pair{n,k})}
			\\
			& \ \
			\trg{\ifte{\pair{n,k}\notin\OB{S}}{\OB{S}::\pair{n,k};\OB{L}::\pair{n,k};}{skip}}
			\\
			&
			\forall \trg{\pair{n,k}}\in\trg{\OB{S}.~ isloc(\pair{n,k})}
			\\
			&\ \ 
			\trg{\letin{x}{mask(!n~with~k)}{n := x~with~k}}
		\end{aligned}
\end{align*}
Then for all locations that are reachable from a shared location (line 1), and that are not already there (line 2), we add those locations to the list of shared locations and to the list of component-created shared locations (line 2).
Then for all shared locations (line 3), we mask their contents using the \trg{mask} function (line 4).

\subsection{The Trace-based Backtranslation: \backtrfac{\cdot}}

Value backtranslation is the same, so $\backtrfac{\trg{v}}=\backtrup{\trg{v}}$.

\subsubsection{The Skeleton}
The skeleton is almost as before (\Cref{sec:backtr-skel}), with the only addition of another list \src{B} explained below.

The only additions are two functions \src{terminate} and \src{diverge}, which do what their name suggests:
\begin{align*}
	&\src{terminate(x)\mapsto fail}
	\\
	&\src{diverge(x)\mapsto \call{diverge}~0}
\end{align*}

\subsubsection{The Common Prefix}
\begin{description}
	\item[\trg{\clh{f}{v}{H}}] 

	As in \Cref{tr:backtr-call}, we keep a list of the context-allocated locations and we update them.
	Also, we extend that list.

	\item[\trg{\rbh{}{}{H}}] 

	As above.

	\item[\trg{\cbh{f}{v}{H}}] 

	This is analogous to \Cref{tr:backtr-callback-loc} but with a major complication.

	Now this is complex because in the target we don't receive locations \trg{\pair{n,k}} from the compiled component, but masked indices \trg{\pair{i,k_{com}}}. (using \trg{i} as a metavariable for natural numbers outputted by the masking function)
	We need to extract them based on where they are located in memory, knowing that the same syntactic structure is maintained in the source.
	So what before was relying on the relation on values $\src{\ell}\relatebeta\trg{\pair{n,k}}$ now is no longer true because we have $\src{\ell}\relatebeta\trg{\pair{i,k_{com}}}$ which cannot hold.
	We need to keep a this relation as a runtime argument in the backtrnanslation and base it solely on the syntactic occurrencies of \trg{\pair{i,k_{com}}}.
	So this runtime relation maps target masking indices to source locations.

	So this relation is really a list \src{B} where each entry has the form \src{\pair{\backtrfac{\trg{i}},\ell}}.

	Intuitively, consider heap \trg{H} from the action.
	For all of its content $\trg{n\mapsto v:\eta}$, we do a structural analysis of \trg{v}.
	This happens at the meta-level, in the backtranslation algorithm.
	\trg{v} may contain subvalues of the form \trg{\pair{i,k_{com}}}, and accessing this subvalue we know is a matter of \trg{\projone{\cdot}} etc.
	So we produce an expression \src{e} with the same instructions (\src{\projone{\cdot}} etc) in the source in order to scan \emph{at runtime} the heap \src{H} we receive after the callback is done.
	(so after the action here is executed and where backtranslation code executes)

	Given that expression \src{e} evaluate to location \src{\ell}, we now need to add to \src{B} the pair \src{\pair{i,\ell}} (also given that \src{i}=\backtrfac{\trg{i}}).

	\item[\trg{\rth{}{}{H}}] 

	As above.

	\item[\trg{\wrl{v,i}}] 

	In this case we need to make use of the runtime-kept relation \src{B}.
	We need to know what source location \src{\ell} corresponds to \trg{i} so we can produce the correct code: \src{\ell:=\backtrfac{\trg{v}}}.

	\src{\ell} is looked up as \src{B(\backtrfac{i})}.

\end{description}

\subsubsection{The Differentiator}
The differentiator needs to put the right code at the right place.
The backtranslation already carries all necessary information to know what the right place is, this is as in previous work: the index of the action $i$ (at the meta level) stored in location \src{\ell_i} (at runtime) and the call stack \src{\OB{f}}

We now go over the various cases of trace difference and see that the differentiation code exists.
We consider \trg{\alpha_1} to be the last action in the trace of \compgen{\src{C_1}} while \trg{\alpha_2} is the last one of \compgen{\src{C_2}}, both made after a common prefix.

\begin{description}
	\item[$\trg{\alpha_1}= \trg{\cbh{f}{v}{H}}$ and $\trg{\alpha_2}= \trg{\cbh{g}{v}{H}}$] 

		Code \src{\ifte{!\ell_i == i }{ \call{terminate}~0 }{skip}} is placed in the body of \src{f} while the code \src{\ifte{!\ell_i == i }{\call{diverge}~0}{skip}} is placed in the body of \src{g}.

	\item[$\trg{\alpha_1}= \trg{\cbh{f}{v}{H}}$ and $\trg{\alpha_2}= \trg{\cbh{f}{w}{H}}$] 

		Code 
		\begin{align*}
			\src{
			\begin{aligned}
				&
				\iftes{\src{!\ell == i }}{
				\\
				&\ \src{\ifte{x == \backtrfac{\trg{v}}}{\call{terminate}~0}{\call{diverge}~0}}}{\skips}	
			\end{aligned}
			}
		\end{align*}
		 is placed in \src{f}.

	\item[$\trg{\alpha_1}= \trg{\cbh{f}{v}{H}}$ and $\trg{\alpha_2}= \trg{\cbh{f}{v}{H'}}$] 
		Here few cases can arise, consider $\trg{H}=\trg{H_1,n\mapsto v:\eta,H_2}$ and $\trg{H'}=\trg{H_1,n'\mapsto v':\eta',H_2'}$:
		\begin{description}
			\item[$\trg{v}\neq\trg{v'}$] 
				We use shortcut $\src{L_{glob}(n)}$ to indicate the location bound to name \src{n} in the list of shared locations (same as in \Cref{sec:backtr-single-action}).

				Code 
				\begin{align*}
					\src{
						\begin{aligned}
							&
							\iftes{\src{!\ell_i==i}}{ 
							\\
							&\
							\letins{\src{x}}{\src{L_{glob}(\backtrfac{\trg{n}})}}{
							\\
							&\ \ 
							\iftes{\src{x==\backtrfac{\trg{v}}}}{\src{\call{terminate}~0}}{\src{\call{diverge}~0}}} 
							\\
							&}{\skips}		
						\end{aligned}
					}
				\end{align*} is placed in the body of \src{f}.

			\item[$\trg{n}\neq\trg{n'}$] 
				In this case one of the two addresses must be bigger than the other.
				Wlog, let's consider $\trg{n}=\trg{n'+1}$.

				So \trg{H_1}= \trg{H_1',n'\mapsto v';\eta'}	and \trg{H_2'}=\trge (otherwise we'd have a binding for \trg{n} there).

				The code in this case must access the location related to \trg{n}, it will get stuck in one case and succeed in the other:

				\src{\ifte{!\ell_i==i}{ \letin{x}{update(\backtrfac{\trg{n}},0)}{\call{diverge}~0} }{skip}}

			\item[$\trg{\eta}\neq\trg{\eta'}$] 
				Two cases arise:
				\begin{itemize}
					\item the location is context-created: in this case the tag must be the same, so we have a contradiction;
					\item the location is component-created, but in this case we know that no such location is ever passed to the context (see \Cref{prop:heap loc}), so we have a contradiction.
				\end{itemize}
		\end{description}
	\item[$\trg{\alpha_1}= \trg{\rth{}{H}{}}$ and $\trg{\alpha_2}= \trg{\rth{H'}{}{}}$] 
		As above.
	\item[$\trg{\alpha_1}= \trg{\cbh{f}{v}{H}}$ and $\trg{\alpha_2}= \trg{\rth{H'}{}{}}$] 

		Code \src{\ifte{!\ell_i==i}{\call{terminate}~0}{\skips}} is placed at \src{f} while \src{\ifte{!\ell_i==i}{\call{diverge}~0}{\skips}} is placed at the top of \src{\OB{f}}.

	\item[$\trg{\alpha_1}= \trg{\cbh{f}{v}{H}}$ and $\trg{\alpha_2}= \trg{\tert}$] 

		Code \src{\ifte{!\ell_i==i}{\call{diverge}~0}{\skips}} is placed at \src{f}.

	\item[$\trg{\alpha_1}= \trg{\cbh{f}{v}{H}}$ and $\trg{\alpha_2}= \trg{\divt}$] 

		Code \src{\ifte{!\ell_i==i}{\call{terminate}~0}{\skips}} is placed at \src{f}.

	\item[$\trg{\alpha_1}= \trg{\rth{}{H}{}}$ and $\trg{\alpha_2}= \trg{\tert}$] 

		Code \src{\ifte{!\ell_i==i}{\call{diverge}~0}{\skips}} is placed at the top of \src{\OB{f}}.

	\item[$\trg{\alpha_1}= \trg{\rth{}{H}{}}$ and $\trg{\alpha_2}= \trg{\divt}$] 

		Code \src{\ifte{!\ell_i==i}{\call{terminate}~0}{\skips}} is placed at the top of \src{\OB{f}}.

	\item[$\trg{\alpha_1}= \trg{\tert}$ and $\trg{\alpha_2}= \trg{\divt}$] 
		Nothing to do, the compiled component performs the differentiation on its own.
\end{description}

\bibliography{./bibliobib.bib}
\end{document}